\newif\ifextabs
\def\llncs{0}
\def\fullpage{1}
\def\anonymous{0}
\def\authnote{0}
\def\notxfont{0}
\def\submission{0}
\def\reply{0}
\def\cameraready{0}

\ifnum\submission=1
\def\anonymous{1}
\def\authnote{0}
\def\llncs{1}
\else
\fi

\ifnum\cameraready=1
\def\llncs{1}
\def\anonymous{0}
\def\authnote{0}
\else
\fi

\ifnum\llncs=1
	\documentclass[runningheads]{llncs}
\else
	\documentclass[letterpaper,hmargin=1.05in,vmargin=1.05in]{article}
			\ifnum\fullpage=1
		\usepackage{fullpage}
		\fi
\fi

\ifnum\reply=1
\usepackage{ulem}
\renewcommand{\emph}{\textit}
\else
\fi

\usepackage[%
  colorlinks=true,
  citecolor=darkgreen,
  linkcolor=darkblue,
  urlcolor=darkblue,
  pagebackref=true
]{hyperref}

\usepackage{amsmath, amsfonts, amssymb, mathtools,amscd}

\usepackage{amsthm}

\usepackage{lmodern}
\usepackage[T1]{fontenc}
\usepackage[utf8]{inputenc}

\usepackage{etex}

\usepackage{arydshln} % In order to use \hdashline
\usepackage{url}
\usepackage{ifthen}
\usepackage{bm}
\usepackage{multirow}
\usepackage[dvips]{graphicx}
\usepackage[usenames]{color}
\usepackage{xcolor,colortbl} % Leave out in case the usepackage cannot be found. Not critical
\usepackage{threeparttable}
\usepackage{comment}
\usepackage{paralist,verbatim}
\usepackage{cases}
\usepackage{booktabs}
\usepackage{braket}
\usepackage{cancel} 
\usepackage{ascmac} 
\usepackage{framed}
\usepackage{authblk}
\usepackage{pifont}
\usepackage{physics}
\definecolor{darkblue}{rgb}{0,0,0.6}
\definecolor{darkgreen}{rgb}{0,0.5,0}
\definecolor{maroon}{rgb}{0.5,0.1,0.1}
\definecolor{dpurple}{rgb}{0.2,0,0.65}
\definecolor{chocolate}{rgb}{0.8,0.4,0.1}

\usepackage[capitalise,noabbrev]{cleveref}
\usepackage[absolute]{textpos}
\usepackage[final]{microtype}
\usepackage[absolute]{textpos}
\usepackage{everypage}
\usepackage{autonum}

\usepackage{enumitem}

\usepackage{dsfont}
\DeclareMathAlphabet{\mathpzc}{OT1}{pzc}{m}{it}

\ifnum\submission=1
\renewcommand*{\backref}[1]{}
\def\notxfont{1}
\else
\fi

\ifnum\llncs=1
\renewcommand{\subparagraph}{\paragraph}
\else
\ifnum\notxfont=1
\else
\usepackage{mathpazo}
\usepackage{newtxtext}
\usepackage{helvet}
\fi
\fi

\newtheoremstyle{thicktheorem}%
{\topsep}
{\topsep}
{\itshape}{}%
{\bfseries}%
{.}
{ }%
{\thmname{#1}\thmnumber{ #2}%
		\thmnote{ (#3)}%
}

\newtheoremstyle{remark}%name
{\topsep}
{\topsep}
	{}%body font
	{}%indent amount
	{}%theorem head font
	{.}%punctuation after theorem head
	{ }%space after theorem head
	{\textit{\thmname{#1}}\thmnumber{ #2}%theorem head specs
			\thmnote{ (#3)}%
	}

\ifnum\llncs=0
	\theoremstyle{thicktheorem}
	\newtheorem{theorem}{Theorem}[section]
	\newtheorem{lemma}[theorem]{Lemma}

	\newtheorem{definition}[theorem]{Definition}

	\theoremstyle{remark}
	\newtheorem{claim}[theorem]{Claim}
	\newtheorem{remark}[theorem]{Remark}

\fi
	\crefname{theorem}{Theorem}{Theorems}
	\crefname{assumption}{Assumption}{Assumptions}
	\crefname{construction}{Construction}{Constructions}
	\crefname{corollary}{Corollary}{Corollaries}
	\crefname{conjecture}{Conjecture}{Conjectures}
	\crefname{definition}{Definition}{Definitions}
	\crefname{exmaple}{Example}{Examples}
	\crefname{experiment}{Experiment}{Experiments}
	\crefname{counterexample}{Counterexample}{Counterexamples}
	\crefname{lemma}{Lemma}{Lemmata}
	\crefname{observation}{Observation}{Observations}
	\crefname{proposition}{Proposition}{Propositions}
	\crefname{remark}{Remark}{Remarks}
	\crefname{claim}{Claim}{Claims}
	\crefname{fact}{Fact}{Facts}
	\crefname{note}{Note}{Notes}

\ifnum\llncs=1
 \crefname{appendix}{App.}{Appendices}
 \crefname{section}{Sec.}{Sections}
\else
\fi

\ifnum\llncs=1
\renewcommand*{\backref}[1]{}
\else
	\renewcommand*{\backref}[1]{(Cited on page~#1.)}
	\ifnum\notxfont=1
	\else
		\usepackage{newtxtext}
	\fi
\fi

\newcommand*{\keys}[1]{\mathsf{#1}}

\newcommand*{\algo}[1]{\ensuremath{\mathsf{#1}}}
\newcommand*{\qalgo}[1]{\ensuremath{\mathpzc{#1}}}
\newcommand*{\qstate}[1]{\mathpzc{#1}}

\newcounter{expitem}

\usepackage{mathtools}
\newcommand{\chosen}{\leftarrow}
\newcommand{\sample}{\leftarrow}
\newcommand{\lrun}{\leftarrow}
\newcommand{\rrun}{\rightarrow}
\newcommand{\la}{\leftarrow}
\newcommand{\ra}{\rightarrow}
\renewcommand{\gets}{\leftarrow}
\newcommand{\defeq}{\mathbin{\stackrel{\textrm{def}}{=}}}
\newcommand{\seteq}{\coloneqq}

\newcommand{\tensor}{\otimes}

\newcommand{\setbracket}[1]{\{#1\}}
\newcommand{\setbk}[1]{\{#1\}}

\newcommand{\cA}{\mathcal{A}}

\newcommand{\cM}{\mathcal{M}}

\newcommand{\cX}{\mathcal{X}}
\newcommand{\cY}{\mathcal{Y}}
\newcommand{\cZ}{\mathcal{Z}}

\newcommand{\qA}{\qalgo{A}}
\newcommand{\qB}{\qalgo{B}}

\newcommand{\qD}{\qalgo{D}}

\def\makeuppercase#1{
\expandafter\newcommand\csname sf#1\endcsname{\mathsf{#1}}
\expandafter\newcommand\csname frak#1\endcsname{\mathfrak{#1}}
\expandafter\newcommand\csname bb#1\endcsname{\mathbb{#1}}
\expandafter\newcommand\csname bf#1\endcsname{\textbf{#1}}
}

\def\makelowercase#1{
\expandafter\newcommand\csname frak#1\endcsname{\mathfrak{#1}}
\expandafter\newcommand\csname bf#1\endcsname{\textbf{#1}}
}

\newcounter{char}
\loop
   \edef\letter{\alph{char}}
   \edef\Letter{\Alph{char}}
   \expandafter\makelowercase\letter
   \expandafter\makeuppercase\Letter
   \stepcounter{char}
   \unless\ifnum\thechar>26
\repeat

\def\makeuppercase#1{
\expandafter\newcommand\csname tl#1\endcsname{\widetilde{#1}}
}

\def\makelowercase#1{
\expandafter\newcommand\csname tl#1\endcsname{\widetilde{#1}}
}

\newcommand{\bit}{\{0,1\}}

\newcommand{\Fs}{\mathcal{F}}
\newcommand{\Ms}{\mathcal{M}}

\newcommand{\secp}{\lambda}
\newcommand{\secpar}{\secp}
\newcommand{\sep}{\lambda}
\newcommand{\coin}{\keys{coin}}

\newcommand{\cert}{\keys{cert}}

\newcommand{\crs}{\mathsf{crs}}

\newcommand{\aux}{\mathsf{aux}}

\newcommand{\A}{\qA}

\newcommand*{\xProblem}[1]{\ensuremath{\mathrm{#1}}}

\newcommand*{\LWE}{\xProblem{LWE}}

\newcommand{\adva}[2]{\mathsf{Adv}_{#1}^{\mathsf{#2}}}
\newcommand{\advb}[3]{\mathsf{Adv}_{#1}^{\mathsf{#2} \mbox{-} \mathsf{#3}}}
\newcommand{\advc}[4]{\mathsf{Adv}_{#1}^{\mathsf{#2} \mbox{-} \mathsf{#3} \mbox{-} \mathsf{#4}}}

\newcommand{\expa}[2]{\mathsf{Expt}_{#1}^{\mathsf{#2}}}
\newcommand{\expb}[3]{\mathsf{Exp}_{#1}^{ \mathsf{#2} \mbox{-} \mathsf{#3}}}
\newcommand{\expc}[4]{\mathsf{Exp}_{#1}^{ \mathsf{#2} \mbox{-} \mathsf{#3} \mbox{-} \mathsf{#4}}}

\newcommand{\Hyb}{\mathsf{Hyb}}
\newcommand{\hyb}{\mathsf{Hyb}}

\newcommand*{\pk}{\keys{pk}}
\newcommand*{\sk}{\keys{sk}}
\newcommand*{\dk}{\keys{dk}}
\newcommand*{\ek}{\keys{ek}}
\newcommand*{\dvk}{\keys{dvk}}
\newcommand*{\vk}{\keys{vk}}

\newcommand*{\svk}{\keys{svk}}

\newcommand*{\msk}{\keys{msk}}

\newcommand*{\pp}{\keys{pp}}

\newcommand*{\xk}{\keys{xk}}

\newcommand*{\ct}{\keys{ct}}

\newcommand*{\msg}{\mathsf{msg}}

\newcommand{\qsk}{\qstate{sk}}

\newcommand{\qdk}{\qstate{dk}}

\newenvironment{boxfig}[2]{\begin{figure}[#1]\fbox{\begin{minipage}{0.97\linewidth}
                        \vspace{0.2em}
                        \makebox[0.025\linewidth]{}
                        \begin{minipage}{0.95\linewidth}
            {{
                        #2 }}
                        \end{minipage}
                        \vspace{0.2em}
                        \end{minipage}}
                        }
                        {\end{figure}}

\newcommand{\Setup}{\algo{Setup}}

\newcommand{\KeyGen}{\algo{KeyGen}}

\newcommand{\KG}{\algo{KG}}
\newcommand{\Enc}{\algo{Enc}}
\newcommand{\Dec}{\algo{Dec}}

\newcommand{\Sign}{\algo{Sign}}
\newcommand{\Vrfy}{\algo{Vrfy}}
\newcommand{\DelVrfy}{\algo{DelVrfy}}

\newcommand{\SigVrfy}{\algo{SigVrfy}}
\newcommand{\qSign}{\qalgo{Sign}}

\newcommand{\qKG}{\qalgo{KG}}

\newcommand{\qDec}{\qalgo{Dec}}

\newcommand{\TrapGen}{\algo{TrapGen}}

\newcommand\PKE{\algo{PKE}}

\newcommand{\GC}{\algo{GC}}
\newcommand{\Garble}{\algo{Grbl}}

\newcommand{\Sim}{\algo{Sim}}
\newcommand{\lab}{\mathsf{lab}} % labels for GC

\newcommand{\Eval}{\algo{Eval}}

\newcommand{\Mark}{\mathsf{Mark}}

\newcommand{\Extract}{\mathsf{Ext}}

\newcommand{\negl}{{\mathsf{negl}}}

\newcommand{\poly}{{\mathrm{poly}}}

\newcommand{\zo}[1]{\{0,1\}^{#1}}
\newcommand{\bin}{\{0,1\}}

\newcommand{\xor}{\oplus}

\newcommand{\Ext}{\mathrm{Ext}}

\newcommand{\calM}{\mathcal{M}}

\newcommand{\calO}{\mathcal{O}}

\newcommand{\calZ}{\mathcal{Z}}

\newcommand{\Ctilde}{{\widetilde{C}}}

\newcommand{\Domprf}{D_{\mathsf{prf}}}
\newcommand{\Ranprf}{R_{\mathsf{prf}}}
\newcommand{\qLEval}{\qalgo{LEval}}
\newcommand{\prfinp}{s}
\newcommand{\prfout}{t}

\newcommand{\UPFSKL}{\mathsf{UPFSKL}}

\newcommand{\chall}{\mathsf{chall}}

\newcommand{\ans}{\mathsf{ans}}
\newcommand{\PKESKL}{\algo{PKESKL}}
\newcommand{\intPKESKL}{\PKESKL_{\mathsf{int}}}
\newcommand{\qIntKeyGen}{\qalgo{IntKeyGen}}
\newcommand{\PRFSKL}{\algo{PRFSKL}}
\newcommand{\DSSKL}{\algo{DSSKL}}

\newcommand{\OT}{\mathsf{OT}}

\newcommand{\st}{\mathsf{st}}

\newcommand{\NTCF}{\algo{NTCF}}

\newcommand{\GoodSet}{\algo{GoodSet}}
\newcommand{\Supp}{\mathrm{Supp}}

\newcommand{\lenx}{w}
\newcommand{\lenst}{u}
\newcommand{\lenmsg}{v}
\newcommand{\lenm}{\ell}

\newcommand{\FuncGen}{\mathsf{FuncGen}}
\newcommand{\StateGen}{\qalgo{StateGen}}
\newcommand{\Chk}{\algo{Check}} % The name check not allowed
\newcommand{\mode}{\mathsf{mode}}

\newcommand{\sigvk}{\mathsf{svk}}

\newcommand{\qsigk}{\qstate{sigk}}

\newcommand{\pke}{\mathsf{pke}}

\newcommand{\SFE}{\mathsf{SFE}}
\newcommand{\StaRcv}{\mathsf{StaRcv}}

\newcommand{\WPKE}{\mathsf{WPKE}}
\newcommand{\WPRF}{\mathsf{WUPF}}
\newcommand{\WDS}{\mathsf{WDS}}
\newcommand{\CS}{\mathsf{CS}}
\newcommand{\qP}{\qalgo{P}}
\newcommand{\TEPRF}{\mathsf{TEPRF}}
\newcommand{\hatm}{\hat{m}}
\newcommand{\hats}{\hat{s}}
\newcommand{\hatt}{\hat{t}}
\newcommand{\hatsigma}{\hat{\sigma}}

\newcommand{\CRSGen}{\mathsf{CRSGen}}
\newcommand{\Receive}[1]{\mathsf{Rec}_{#1}}
\newcommand{\qReceive}[1]{\qalgo{Rec}_{#1}}
\newcommand{\Send}{\mathsf{Send}}

\newcommand{\key}{\mathsf{key}}
\newcommand{\extk}{\mathsf{xk}}
\newcommand{\Constrain}{\mathsf{Constrain}}

\newcommand{\qF}{\qalgo{F}}

\newcommand{\qDel}{\qalgo{Del}}
\newcommand{\Del}{\algo{Del}}

\newcommand{\IsMessy}{\algo{IsMessy}}

\newcommand{\qExt}{\qalgo{Ext}}

\newcommand{\event}{\mathsf{E}}

\newcommand{\verdict}{\mathsf{vrd}}

\newcommand{\Invert}{\mathsf{Invert}}

\newcommand{\sigk}{\mathsf{sigk}}
\newcommand{\td}{\mathsf{td}}

\newcommand{\QSign}{\qalgo{QSign}}

\newcommand{\NN}{\mathbb{N}}   
\newcommand{\ZZ}{\mathbb{Z}}
\newcommand{\RR}{\mathbb{R}}  
\newcommand{\CC}{\mathbb{C}} 
\newcommand{\mA}{\mathbf{A}}    \newcommand{\mB}{\mathbf{B}}
    
\newcommand{\mH}{\mathbf{H}}    
    
\newcommand{\mT}{\mathbf{T}}    
    
\newcommand{\mC}{\mathbf{C}}

\newcommand{\ve}{\mathbf{e}}    
    
\newcommand{\vs}{\mathbf{s}}    \newcommand{\vt}{\mathbf{t}}
\renewcommand{\vu}{\mathbf{u}} 
\newcommand{\vv}{\mathbf{v}} 
\newcommand{\vx}{\mathbf{x}}    
    \newcommand{\vw}{\mathbf{w}}

 \newcommand{\vr}{\mathbf{r}}

\newcommand{\trans}{\top}

\newcommand{\bd}{\begin{description}}
\newcommand{\ed}{\end{description}}
\newcommand{\bt}{\begin{theorem}}
\newcommand{\et}{\end{theorem}}
\newcommand{\bc}{\begin{claim}}
\newcommand{\ec}{\end{claim}}
\newcommand{\bl}{\begin{lemma}}
\newcommand{\el}{\end{lemma}}
\newcommand{\br}{\begin{remark}}
\newcommand{\er}{\end{remark}}
\newcommand{\bde}{\begin{definition}}
\newcommand{\ede}{\end{definition}}
\newcommand{\bi}{\begin{itemize}}
\newcommand{\ei}{\end{itemize}}
\newcommand{\be}{\begin{enumerate}}
\newcommand{\ee}{\end{enumerate}}

\newcommand{\authornote}[3]{\textcolor{#3}{[\textbf{#1:} {#2}]}}
\ifnum\authnote=1
\newcommand{\fuyuki}[1]{\authornote{fuyuki}{#1}{chocolate}}
\newcommand{\takashi}[1]{\authornote{takashi}{#1}{dpurple}}
\newcommand{\shota}[1]{\authornote{shota}{#1}{blue}}
\newcommand{\synote}[1]{\authornote{shota}{#1}{blue}}
\newcommand{\jiahui}[1]{\authornote{jiahui}{#1}{magenta}}
\else
\newcommand{\fuyuki}[1]{}
\newcommand{\takashi}[1]{}
\newcommand{\shota}[1]{}
\newcommand{\synote}[1]{}
\newcommand{\jiahui}[1]{}
\fi

\ifnum\llncs=1
\let\oldvec\vec% Store \vec in \oldvec
\let\vec\oldvec% Restore \vec from \oldvec
\makeatletter
\renewcommand*\l@author[2]{}
\renewcommand*\l@title[2]{}
\makeatletter
\fi    
\theoremstyle{remark}

\ifnum\submission=1
\title{
\textbf{A Unified Approach to Quantum Key Leasing\\ with a Classical Lessor}
}
\else
\title{
\textbf{A Unified Approach to Quantum Key Leasing\\ with a Classical Lessor}
}
\fi

\begin{document}
%\author{}
%\institute{}

\ifnum\anonymous=1 
\ifnum\llncs=1
% for anonymized LNCS version
\author{\empty}\institute{\empty}
\else
% for anonymized full version
\author{}
\fi
\else
%
%  For camera ready version.
%
\ifnum\llncs=1
\author{
Fuyuki Kitagawa\inst{1,2} \and 
Jiahui Liu\inst{3} \and
Shota Yamada\inst{4} \and
Takashi Yamakawa\inst{1,2}
}
\institute{ 
NTT Social Informatics Laboratories, Tokyo, Japan \and
NTT Research Center for Theoretical Quantum Information, Atsugi, Japan \and
Fujitsu Research, Santa Clara, USA \and
AIST, Tokyo, Japan
}
\else
%
%   For full/eprint version, etc.
%
\author[1,2]{\hskip 1em Fuyuki Kitagawa}
\author[3]{\hskip 1em Jiahui Liu}
\author[4]{\hskip 1em Shota Yamada}
\author[1,2]{\hskip 1em Takashi Yamakawa}
\affil[1]{{\small NTT Social Informatics Laboratories, Japan}\authorcr{\small \{fuyuki.kitagawa,takashi.yamakawa\}@ntt.com}}
\affil[2]{{\small NTT Research Center for Theoretical Quantum Information, Japan}}
\affil[3]{{\small Fujitsu Research of America, USA}\authorcr{\small  jiahuiliu.crypto@gmail.com}}
\affil[4]{{\small AIST, Japan}\authorcr{\small  yamada-shota@aist.go.jp}}

\fi %%%%% END OF LNCS branch
\fi

\ifnum\llncs=1
\date{}
\else
\ifnum\anonymous=0
\date{\today}
\else
\date{}
\fi
\fi

\maketitle

\begin{abstract}
Secure key leasing allows a cryptographic key to be leased as a quantum state in such a way that the key can later be revoked in a verifiable manner. In this work, we propose a modular framework for constructing secure key leasing with a classical-lessor, where the lessor is entirely classical and, in particular, the quantum secret key can be both leased and revoked using only classical communication. Based on this framework, we obtain classical-lessor secure key leasing schemes for public-key encryption (PKE), pseudorandom function (PRF), and digital signature. We adopt the strong security notion known as security against verification key revealing attacks (VRA security) proposed by Kitagawa et al. (Eurocrypt 2025) into the classical-lessor setting, and we prove that all three of our schemes satisfy this notion under the learning with errors assumption. Our PKE scheme improves upon the previous construction by Chardouvelis et al. (Eurocrypt 2025), and our PRF and digital signature schemes are respectively the first PRF and digital signature with classical-lessor secure key leasing property.

Along the way, we also construct a watermarking scheme and a dual-mode secure function evaluation scheme that satisfy certain useful properties, which may be of independent interest.

\fuyuki{Do we need to say something about watermarking or SFE??}\jiahui{added one sentence}
\fuyuki{\cite{AC:PWYZ24} achieved classical communication PKE-SKL assuming poly modulus LWE, so their PKE-SKL and ours are incomparable. Is it misleading to say ours improves \cite{chardouvelis2025quantum} ignoring \cite{AC:PWYZ24}?} \jiahui{just added \cite{AC:PWYZ24} to the table as well}
\end{abstract}

\ifnum\llncs=1
\else
\newpage
\setcounter{tocdepth}{2}
\tableofcontents

\newpage
\fi
\ifnum\llncs=0
% !TEX root = main.tex

\section{Introduction}
\paragraph{Delegation and Revocation of Cryptographic Functionalities}
Delegation of computation is a central theme in cryptography and the practice of cloud computing. Beyond efficiency and correctness, many real-world scenarios demand revocability: a client who temporarily outsources a capability must be able to withdraw it cleanly, without incurring operational hazards or reconfiguring public parameters. Consider a manager who leases signing capability to a delegate while on vacation; upon return, she must revoke the delegatee's signing power. The classical solution to this problem requires not only updating the secret signing key, but also updating and re-announcing public verification parameters.  This procedure can be costly and can create vulnerability during propagation. This motivates us to look for a cryptographic mechanism that can revoke functionality without updating public parameters or incurring propagation delays

A formalization of the above protocol is called secure key leasing (also known as key-revocable cryptography), first put forward by \cite{EC:AKNYY23,TCC:AnaPorVai23}: a lessor issues a quantum secret key that enables some functionality (e.g., decryption, PRF evaluation, or signing) to a lessee (server), and later executes a deletion protocol: the lessee (server) must return the key or a proof of deleting the key; this deletion protocol is supposed to provably strip the lessee of the capability of computing the functionality. Naturally, secure key leasing (SKL) also captures temporary licensing for general programs apart from delegation of computing: a software lessor leases quantum software whose functionality can later be invalidated upon successful return and verification, for instance, when the lessee's subscription to the software has expired. Furthermore, a secure leasing scheme for a cryptographic function can potentially lead to a secure leasing scheme for
a software entity of which this cryptographic function is a component.
%\jiahui{Will add later a little discussions on applications of leasing cryptographic functions }

However, this goal is classically impossible as the lessee/server can always make a copy of the key or software. As shown in a sequence of previous works, \takashi{I'm not sure which paper we should cite since it seems that we are talking about the general concept of unclonability here. Perhaps, we do not need to cite anything here.} \jiahui{added a generic saying}leveraging the unclonability of quantum information  will help us construct protocols that satisfy all the above requirements.

\paragraph{Importance of a Classical Client/Lessor} While we need to utilize the property of quantum information to realize the secure key leasing protocol, minimizing the cost and quantum resources used is also indispensable. Especially in the near-term quantum computing era and practical deployments, it is crucial to minimize the  required quantum computing resources on the client's side.
For instance, in real-world scenarios of cloud computing, the client will ideally use much less computing resources than the server uses when interacting with the server, otherwise the task of delegation is not very meaningful. 
When designing a protocol for cloud quantum computing or software leasing, 
 the lessor (client) ideally \emph{remains entirely classical and communicates over classical channels} throughout the protocol, while the lessee (server) performs the quantum work locally on its end. 
%This classical-lessor setting has been realized for public-key encryption (PKE)—and extended to fully homomorphic encryption (FHE) from standard lattice assumptions (LWE), with security that no polynomial-time quantum adversary can both furnish a valid (classical) deletion certificate and still distinguish ciphertexts. 

\paragraph{Towards a Modular Approach}
A sequence of works has realized secure key leasing (SKL) for several functionalities—PKE, FHE, PRFs, and digital signatures—with the most recent results based on standard cryptographic assumptions \cite{chardouvelis2025quantum,TCC:AnaHuHua24,kitagawa2025simple,kitagawa2025pke}. However, these constructions are largely ad hoc and tailored to each functionality rather than arising from a common, modular blueprint. Moreover, the only known scheme with a fully classical lessor (client) in \cite{chardouvelis2025quantum} applies to PKE and FHE; their underlying approach seems relatively tailored to encryption schemes and is not known to generalize to other primitives. \cite{kitagawa2025simple} provides a relatively unified approach to construct SKL protocols for several cryptographic functionalities, but the lessor still needs to be fully quantum and there does not seem to be a straightforward way to dequantize the lessor.

Motivated by these limitations, we ask the following question: 
\begin{center}
\emph{Can we develop a modular SKL framework that supports a classical lessor and spans a broad range of widely used cryptographic functionalities?}    
\end{center}

\subsection{Our Results}
\label{sec:our results}

\paragraph{Main Result: Classical-Lessor SKL Framework}
We first state our main result, which answers the above question in a positive way. 

We propose a simple framework for constructing secure key leasing schemes with a fully classical lessor.
Based on our framework, we obtain the following schemes for ``upgrading'' some very generic cryptographic primitives to their corresponding classical-lessor SKL schemes, using only the post-quantum security of standard lattice assumptions. \takashi{We should clarify that we rely on LWE with subexp modulus while PWYZ24 only uses LWE with polynomial modulus.}\jiahui{Added the emphasis in the above form and added it to related works section}
\begin{itemize}
\item PKE-SKL with classical lessor based on (sub-exponential) LWE. The underlying PKE scheme is, however, \emph{not} limited to any specific LWE-based construction.
More concretely, we can ``upgrade'' \emph{any post-quantum IND-CPA secure classical PKE scheme} to a PKE-SKL with classical lessor, using LWE. 
%\footnote{

While we still need the LWE assumption, we use the other two primitives in our construction, the noisy trapdoor claw-free families (NTCFs) and secure function evaluation (SFE) also 
\ifnum\cameraready=1 \emph{in a black-box way}.
\else
\emph{in a black-box way} 
(see Theorems \ifextabs 6.2, 7.2, 8.2 in the technical manuscript.\else ~\ref{thm-proof-PKESKL},~\ref{thm-proof-skl-prf},~\ref{thm-proof-skl-ds}.\fi). \fi 
They have potentially non-LWE-based constructions\footnote{For instance, NTCFs have a candidate construction from group-actions (\cite{alamati2022candidate}); the dual-mode SFE could potentially be constructed from code-based assumptions \cite{dowsley2008oblivious} and LPN variants \cite{david2014universally}.}.  Therefore, it is meaningful to modify the highly LWE-structured PKE-SKL protocol in \cite{chardouvelis2025quantum,AC:PWYZ24} into a more modular protocol.
%} 

Moreover, it is unlikely that we can obtain PKE-SKL with a  classical lessor from any PKE alone without reliance on any other structures, otherwise we will obtain a classically verifiable quantum advantage from any PKE (\cite{kitagawa2025simple}). 
\takashi{As pointed out by Reviewer C, I believe we should not sell  generality of the instantiation.
That is, I'm wondering if we should simply remove the second and third sentences here. The same comment applies to PRF and DS.
} \jiahui{I believe it is still meaningful that our underlying PKE is independent of LWE. It's very probable that the other primitives (NTCF and SFE) we use do not need to rely on LWE either. }

\item PRF-SKL with classical lessor from LWE. The underlying PRF scheme used can be any post-quantum watermarkable PRF, which we upgrade to a classical lessor PRF-SKL with LWE. 

\item DS-SKL with classical lessor based on LWE. %quantum hardness of the 
%short integer solution (SIS) assumption. 
For digital signatures, our construction is less black-box.
We lift a digital signature to a classical-lessor DS-SKL through the same framework as we apply to the PKE and PRF above. But we need a special property of the underlying digital signature scheme, which we can obtain from LWE.

%Our scheme has static signing keys, i.e., the states of a signing key almost does not change before and after signing. In particular, this implies that the sizes of signing keys and signatures are independent of the number of issued signatures.    
%It does not require a state of size depending on the number of issued signatures.
\end{itemize}
Finally, all of our schemes remain secure even if a deletion verification key used by the deletion verification algorithm is leaked after the adversary submits a valid certificate of deletion. This security is called security against verification key revealing attack (VRA security), as proposed in \cite{kitagawa2025simple}.

For a comparison with existing results on the aforementioned properties of SKL, we present Tables~\ref{table:pkeskl},\ref{table:prfskl},\ref{table:dsskl}.

\paragraph{Non-interactive key generation} 
We further improve our classical-lessor SKL scheme to acquire a non-interactive quantum key generation 
algorithm (whereas the vanilla protocol requires a constant round interaction) through a clean technique. This is a non-trivial property for SKL with a fully classical lessor. While this property is easy to achieve for a quantum lessor SKL--simply letting the lessor prepare the quantum key and send it to the lessee, a classical lessor can only send classical information to the lessee 
%to allow it prepare a resulting quantum state that the lessee cannot ``clone''.  
and may need interactive messages from the lessee to prevent its malicious behavior during key generation.\footnote{\cite{chardouvelis2025quantum,AC:PWYZ24} also have non-interactive key generation algorithms for their PKE-SKL with classical lessor, but their non-interaction relies on the LWE-structure of the scheme. Our approach of removing interaction is more black-box. See the technical overview for more discussions. }
\shota{I think PRF and DS cannot be non-interactive. Shall we remove the following?}\jiahui{removed}
%Our approach of removing interaction is more black-box and therefore applicable to our PRF-SKL and DS-SKL protocols. We believe it may be applicable for reducing communication rounds in broader unclonable cryptography protocols in a classical client-quantum server mode.}. %All previous works on classical-lessor SKL (\cite{chardouvelis2025quantum,AC:PWYZ24}) require interactions for key generation\footnote{Their need for interaction is implicit but necessary for full security. See \cite{chardouvelis2025quantum} Section A.1}. Therefore, ours is the first to achieve this property. \takashi{I believe this was not true; the previous works were also non-interactive.}

\paragraph{Side results}
Along the way of constructing our classical-lessor SKL framework, we construct a special dual-mode Secure Function Evaluation (SFE) scheme  and a watermarking scheme with parallel extraction. They satisfy many useful properties, which we believe can be of independent interest.

\begin{table}[t]
\centering
%\begin{tabular}{cc}
%\begin{minipage}{0.5\textwidth}
%\centering
\caption{Comparison of PKE-SKL}
\label{table:pkeskl}
\begin{tabular}{ccccc}
\toprule
& Classical-Lessor &VRA security  &Assumption \\
\midrule
\cite{EC:AKNYY23}& & & PKE\\
\midrule
\cite{TCC:AnaPorVai23,TCC:AnaHuHua24}&  & &LWE\\
\midrule
\cite{chardouvelis2025quantum,AC:PWYZ24} & \checkmark & & LWE (poly-modulus LWE in \cite{AC:PWYZ24}) \\
\midrule
\cite{kitagawa2025simple} &  & \checkmark &  PKE \\
\midrule
Ours& \checkmark & \checkmark & LWE\\
\bottomrule
\end{tabular}
%\end{minipage}
%&
%\begin{minipage}{0.5\textwidth}
\end{table}

\begin{table}[t]
\centering
\caption{Comparison of PRF-SKL}
\label{table:prfskl}
\begin{tabular}{ccccc}
\toprule
& Classical-Lessor &VRA security &Assumption\\
\midrule
\cite{TCC:AnaPorVai23,TCC:AnaHuHua24}& & &LWE\\
\midrule
\cite{kitagawa2025simple} &  & \checkmark & OWF\\
\midrule
Ours & \checkmark & \checkmark & LWE \\
\bottomrule
\end{tabular}
%\end{minipage}
%\end{tabular}
\end{table}

\begin{table}[t]
\centering
\caption{Comparison of DS-SKL}
\label{table:dsskl}
\begin{tabular}{ccccccc}
\toprule
& Classical-Lessor &VRA  security  &Signing key size&Signature size &Assumption\\
\midrule
\cite{TQC:MorPorYam24}& & &$n\cdot \poly(\secp)$& $n\cdot \poly(\secp)$& LWE\\
\midrule
\cite{TQC:MorPorYam24}& & &$n\cdot \poly(\secp)$& $n\cdot \poly(\secp)$& OWGA\\
\midrule
\cite{kitagawa2025simple}& &\checkmark &$\poly(\secp)$&$\poly(\secp)$& SIS\\
\midrule
Ours&\checkmark &\checkmark &$\poly(\secp)$&$\poly(\secp)$& LWE\\
\bottomrule
\end{tabular}\\
$n$: the number of issued signatures, $\secp$: the security parameter.\jiahui{Check this table?}
\synote{Changed our assumption to LWE. However, LWE is quantumly as hard as SIS in certain parameter regime \cite{AC:SSTX09}. \fuyuki{An explanation for the term OWGA would be required? OWGA is introduced in the related work section.}}
\end{table}

\subsection{Related Works}

\paragraph{More on SKL}
The work of~\cite{EC:AKNYY23} realized the first construction of PKE-SKL from any PKE scheme.
They also realized unbounded (resp. bounded) collusion-resistant PKFE-SKL by combining PKE-SKL with unbounded (resp. bounded) collusion-resistant classical PKFE.
Independently, \cite{TCC:AnaPorVai23} proposed constructions of PKE-SKL, FHE-SKL, and PRF-SKL under the LWE assumption, though their security proofs relied on an additional unproven conjecture.
This reliance on unproven conjecture was later eliminated by \cite{TCC:AnaHuHua24}.
Moreover, \cite{EC:BGKMRR24} presented a publicly verifiable PKFE-SKL scheme, where the verification key can be made public, based on indistinguishability obfuscation (iO).

\cite{chardouvelis2025quantum} constructed PKE-SKL and FHE-SKL with a classical-lessor under the LWE assumption.
Their schemes assumed sub-exponential hardness of the LWE problem with sub-exponential modulus, which was subsequently improved to polynomial hardness of the LWE problem with polynomial modulus by~\cite{AC:PWYZ24}.
\fuyuki{Should we say that our work assumes polynomial hardness of the LWE problem with sub-exponential modulus somewhere in the intro?}

\cite{TQC:MorPorYam24} developed two DS-SKL schemes: one based on the LWE assumption and another on one-way group actions (OWGA)~\cite{TCC:JQSY19}.
A limitation of their constructions is that the signing key evolves by each signature generation, leading to signing keys and signatures whose sizes grow linearly in the total number of signatures.

Finally, \cite{kitagawa2025pke} proposed collusion-resistant PKE-SKL and ABE-SKL schemes that support leasing an a-priori unbounded number of quantum secret keys, under the LWE assumption.

\paragraph{Secure software leasing.}
Secure software leasing (SSL) was introduced by \cite{EC:AnaLaP21} as a method to encode classical programs into quantum states that can be leased in a secure manner.
The idea of SKL is directly inspired by SSL and can be regarded as its specialization for cryptographic functionalities.
Much of the existing research on SSL~\cite{EC:AnaLaP21,TCC:KitNisYam21,TCC:BJLPS21} has focused on a weaker security model, where pirated programs are restricted to employ the \emph{honest} evaluation procedure.
In contrast, \cite{EC:BGKMRR24} proposed an SSL construction for all differing-input circuit pairs\footnote{Informally, a pair of circuits $(C_0,C_1)$ is said to be differing input if it is computationally difficult to find an input $y$ such that $C_0(y)\ne C_1(y)$.}, under a stronger model that allows pirated programs to use arbitrary evaluation strategies similarly to SKL.
Their scheme, however, relies on iO.

\paragraph{Encryption with certified deletion}
An encryption scheme with certified deletion~\cite{TCC:BroIsl20} allows a quantum ciphertext to be deleted in a verifiable way.
This notion is closely connected to SKL, as one may view SKL as providing certified deletion of cryptographic keys rather than ciphertexts.
In their initial work, \cite{TCC:BroIsl20} presented an unconditionally secure one-time secret-key encryption scheme with certified deletion.
Subsequent works~\cite{AC:HMNY21,ITCS:Poremba23,C:BarKhu23,EC:HKMNPY24,EC:BGKMRR24} have explored certified deletion for richer forms of encryption, including public-key encryption, attribute-based encryption, functional encryption, fully homomorphic encryption, and witness encryption.
Also, \cite{TCC:KitNisYam23,TCC:BKMPW23} extended these results and obtained publicly verifiable certified deletion schemes under minimal assumptions.

\subsection{Concurrent Work}
Concurrent work by Takeuchi and Xu~\cite{takeuchi2025computationalcertifieddeletionproperty} obtains a result similar to ours. In particular, they also construct PKE-SKL, PRF-SKL, and DS-SKL with classical lessors from the LWE assumption. However, their security proof relies on an additional conjecture concerning parallel repetition of a certified-deletion property of the magic square game. Moreover, their key generation procedure requires polynomially many rounds, whereas ours is non-interactive for PKE-SKL and requires only a constant number of rounds for PRF-SKL and DS-SKL.

\subsection{Technical Overview}
\label{sec:tech overview}

\paragraph{Security notion for secure key leasing}
The security of the secure key leasing scheme roughly guarantees that if the adversary who has performed the key generation algorithm with the challenger to obtain a quantum secret key $\qsk$, later outputs a deletion certificate that is accepted by %$\DelVrfy(\dvk,\cdot)$, 
the deletion verification algorithm $\DelVrfy$, 
the adversary can no longer perform the task that requires $\qsk$.

%As stated in \cref{sec-our-results}, 
Like \cite{kitagawa2025simple}, our schemes remain secure even if the secret deletion verification key is given to the adversary after the adversary outputs a deletion certificate.
Such security definition is called security against verification key revealing attacks (VRA security),  defined in \cite{kitagawa2025simple}.  %VRA security is stronger than security under the existence of the verification oracle for secure key leasing schemes with classical revocation. 
We refer the readers to \cite{kitagawa2025simple} for more detailed discussions on the definitions.

%that VRA security is stronger than security under the existence of the verification oracle for secure key leasing schemes with classical revocation, which was studied in a previous work~\cite{EC:AKNYY23}. \ifnum\llncs=0(See \Cref{sec:comparison_def}.)\fi

\jiahui{check this part}\fuyuki{Checked and modified. (The upgrade from UPF-SKL to PRF-SKL using quantum GL was shown in \cite{kitagawa2025simple}.)}
We can define both an unpredictability style security definition and an indistinguishability style one for PKE-SKL and PRF-SKL, though it is inherently unpredictability style for DS-SKL.
The former can be upgraded into the latter by using the standard (quantum) Goldreich-Levin technique as shown in previous works \cite{EC:AKNYY23,kitagawa2025simple}.
Thus, in this work, we primarily focus on constructing secure key leasing primitives satisfying unpredictability style security notions.

The unpredictability style  security experiment for X-SKL is abstracted by using a predicate $P$ that determines the adversary's goal.
%In the experiment, an adversary is first given a quantum secret key $\qsk$ and then outputs a deletion certificate $\cert$.
%If $\cert$ is accepted, the challenger generates a challenge instance $\chall$ and corresponding auxiliary information $\aux$, and sends $\chall$ to the adversary together with the verification key $\dvk$.
%The adversary's goal is to output an answer $\ans$ to the challenge that satisfies the predicate $P(\aux,\cdot)$.
It is described as follows.

\begin{enumerate}
\item The challenger and the adversary $\qA$ generate $\qsk$ %and $\dvk$ 
by executing the (possibly interactive) key generation algorithm $\qIntKeyGen$. If $\qA$ follows the protocol, it gets $\qsk$, but it may deviate from the honest protocol.\synote{I changed the explanation a bit. Previously, it is written $\qsk$ is sent to $\qA$. } (If X is PKE or DS, $\qA$ is also given the corresponding public key output by $\qIntKeyGen$.)
\item $\qA$ outputs a classical string $\cert$. If $\DelVrfy(\tau, \dvk,\cert)=\bot$, where $\tau$ is the transcript of the interaction and $\dvk$ is the deletion verification key, the experiment ends. Otherwise, the challenger generates a challenge $\chall$ and corresponding auxiliary information $\aux$, and sends $\dvk$ and $\chall$ to $\qA$. Note that $\aux$ is used by the predicate $P$ to determine whether $\qA$ wins the experiment or not.
\item $\qA$ outputs $\ans$.
\end{enumerate}
$\qA$ wins the experiment if both $\DelVrfy(\tau,\dvk,\cert)=\top$ and $P(\aux,\ans)=1$ holds simultaneously.
For example, if X is PKE, $\chall$ is a ciphertext of a uniformly random message $m$, $\aux$ is $m$, and $\ans$ is the guess $m^\prime$ for $m$ by $\qA$. Also, $P(\aux=m,\ans=m^\prime)$ outputs $1$ if and only if $m^\prime=m$. %\takashi{It may be useful to explain it also for PRF and DS to understand the intution of what $\ans$ and $P$ mean.}
It guarantees that the winning probability of any QPT $\qA$ is negligible in the security parameter.

 For the next few paragraphs in the technical overview, we first focus on secure key leasing for decryption keys for a PKE scheme and discuss how we generalize to other cryptographic functionalities later.

\paragraph{BB84 states based quantum key leasing}
We start by basing our initial ideas on the secure key leasing framework in \cite{kitagawa2025simple}, which is built on the properties of BB84 states~\cite{BB84}. 
BB84 states are states with the following format: for strings $x\in\bit^n$ and $\theta\in\bit^n$, we let a BB84 state  $\ket{x^\theta} \seteq H^{\theta[1]}\ket{x[1]}\tensor\cdots \tensor H^{\theta[n]}\ket{x[n]}$, where  
$H$ is the Hadamard operator, and $x[i]$ and $\theta[i]$ are the $i$-th bits of $x$ and $\theta$, respectively.
%A state of the form $\ket{x^\theta}$ is called a BB84 state~\cite{BB84}. 
%\jiahui{simplify this part since we are not using this property}
To put in a very informal manner, BB84 states are unclonable and have been used in many quantum cryptographic protocols. For example, \cite{kitagawa2025simple} deployed the following uncertainty-principle-like property of BB84 states (called certified deletion property in \cite{EPRINT:BehSatShi21}): if an adversary $\qA$ without the knowledge of $x,\theta$ can correctly output all the value $x[i]$ encoded in Hadamard basis (i.e. $x[i]$ for all $i$ where $\theta[i] = 1$), it is impossible for $\qA$ to obtain the value of $x[i]$ for the positions $i$ encoded in the computational basis (i.e. $x[i]$ for all $i$ where $\theta[i] = 0$). \footnote{As we will discuss later, we will not directly use the property of BB84 states as \cite{kitagawa2025simple} does, so we will not go into more details here.}

\iffalse
Here, they take use of a specific unclonable property of BB84 states called \emph{certified deletion property} in \cite{EPRINT:BehSatShi21}.
Consider the following experiment.
\begin{enumerate}
\item The challenger picks uniformly random strings $x\in\bit^n$ and $\theta\in\bit^n$ and sends the BB84 state $\ket{x^\theta}$ to $\qA$. 
\item $\qA$ outputs a classical string $y\in\bit^n$.
\item $\qA$ is then given $\theta$ and $(x[i])_{i\in[n]:\theta[i]=1}$, and outputs a classical string $z\in\bit^n$.
\end{enumerate}
$\qA$ wins the experiment if $y[i]=x[i]$ for every $i\in[n]$ such that $\theta[i]=1$ and $z[i]=x[i]$ for every $i\in[n]$ such that $\theta[i]=0$.
The certified deletion property of BB84 states guarantees that the winning probability of any adversary $\qA$ is negligible in $n$.

%The string $y$ output by $\qA$ can be seen as a deletion certificate that is accepted if it agrees with $x$ in the Hadamard basis positions.  The verification of the deletion certificate is done by the deletion verification key $(\theta, (x[i])_{i\in[n]:\theta[i]=1})$.
Note that $\qA$ can generate an accepting deletion certificate, the string $y$, by just measuring the given BB84 state $\ket{x^\theta}$ in the Hadamard basis. Anyone with the information of $(\theta, (x[i])_{i\in[n]:\theta[i]=1})$ can verify the outcome.
Then, the property says that once $\qA$ outputs an accepting deletion certificate, it is impossible for $\qA$ to obtain the value of $x$ in all of the computational basis positions even given the deletion verification key $(\theta,(x[i])_{i\in[n]:\theta[i]=1})$.
\fi

However, as we can see, BB84 states alone are simply random strings encoded in random bases. To formulate a secure key leasing scheme for decryption keys of a PKE, one needs to carefully embed a decryption key into the structure of BB84 states  so that we can obtain an arbitrary-polynomially many-time  reusable quantum decryption key and meanwhile maintain the unclonable property.
To achieve this, \cite{kitagawa2025simple} generates $2n$ i.i.d classical keys for a classical PKE and entangles them respectively with each qubit in the BB84 states $\ket{x^\theta}$ in the following manner:
%certified deletion
Generates $2n$ classical secret keys $(\sk_{i,b})_{i\in[n],b\in\bit}$ and 
$(x,\theta)\gets\bit^n\times\bit^n$. 
It then applies the map $M_i$ that acts as
$\ket{b}\ket{c} \ra \ket{b}\ket{c \oplus \sk_{i,c}}$
to the $i$-th qubit of the BB84 state $\ket{x^\theta}$ and $\ket{0\cdots0}$, and obtains a quantum state $\qsk_i$ for every $i\in[n]$.
Namely, we have
\begin{equation}
\label{eqn:qsk_format}
\begin{aligned}
    \qsk_i\seteq
    \begin{cases}
    \ket{x[i]}\ket{\sk_{i,x[i]}} & (\textrm{if~}\theta[i]=0)\\   
    \frac{1}{\sqrt{2}}\left(\ket{0}\ket{\sk_{i,0}}+(-1)^{x[i]}\ket{1}\ket{\sk_{i,1}}\right) & (\textrm{if~}\theta[i]=1). 
    \end{cases}
\end{aligned}
\end{equation}
The resulting secret key is set to $\qsk:=(\qsk_1,\ldots,\qsk_n)$. The corresponding deletion verification key is the basis information, and the value of $x$ and the classical secret keys at the indices encoded in Hadamard basis: $\dvk\seteq (\theta,(x[i])_{i\in[n]:\theta[i]=1},(\sk_{i,0},\sk_{i,1})_{i\in[n]:\theta[i]=1})$.

\jiahui{add a little on how the encryption and decryption work}\fuyuki{Slightly modified.}
To encrypt a message, anyone with access to the public key $\pk=(\pke.\pk_{i,b})_{i\in [n],b\in\bit}$ and a message $m=m_1\|m_2\|\ldots \|m_{n}$ (padding the message w.l.o.g.), where $\pke.\pk_{i,b}$ is the public key corresponding to $\sk_{i,b}$,
can generate $\pke.\ct_{i,b} \la \PKE.\Enc(\pke.\pk_{i,b},m_i)$ for $i\in[n]$ and $b\in\bit$ and output $\ct=(\pke.\ct_{i,b})_{i\in [n],b\in\bit}$.

To decrypt with the quantum key $\qsk:=(\qsk_1,\ldots,\qsk_n)$ on a ciphertext $\ct=\{\pke.\ct_{i,b}\}_{i\in [n],b\in\bit}$: for each $i\in[n]$, decrypt $(\pke.\ct_{i,0},\pke.\ct_{i,1})$ coherently with $\qsk_i$\footnote{In more detail, we execute the circuit that has $(\pke.\ct_{i,0},\pke.\ct_{i,1})$ hardwired, takes $(b,\sk)$ as input, and outputs $\PKE.\Dec(\sk,\pke.\ct_{i,b})$ in superposition and measures the outcome written onto an additional register.} to get $m_i$ and output  $m=m_1\|m_2\|\ldots \|m_{n}$. 

We will specify more details on deletion, revocation and verification in the following paragraphs after we get a full picture of our construction.

%We can then observe that the above construction in \cite{kitagawa2025simple} can  already help one easily dequantize the \emph{deletion part} of the protocol based on the above BB84 state quantum key. The deletion algorithm $\qDel$, given the above $\qsk=(\qsk_1,\ldots,\qsk_n)$, measures $\qsk_i$ in the Hadamard basis and obtains $(e_i,d_i)$ for every $i\in[n]$. The resulting deletion certificate is set to $\cert:=(e_i,d_i)_{i\in[n]}$. The deletion verification algorithm, given a deletion verification key $\dvk\seteq (\theta,(x[i])_{i\in[n]:\theta[i]=1},\allowbreak (\sk_{i,0},\sk_{i,1})_{i\in[n]:\theta[i]=1})$ and a deletion certificate $\cert^\prime=(e^\prime_i,d^\prime_i)_{i\in[n]}$, outputs $\top$ if and only if $e^\prime_i=x[i]\oplus d^\prime_i \cdot (\sk_{i,0}\oplus\sk_{i,1})$ holds for every $i\in[n]$ such that $\theta[i]=1$. One can check the deletion verification correctness of the construction by simple calculations.

\paragraph{Dequantizing the lessor entirely}
%\jiahui{I will first not describe how \cite{kitagawa2025simple} did the classical revocation part in detail, since the properties we reduce to are different}
The next major question is how we can dequantize the preparation of the quantum key as well as the revocation of the key for the lessor, making the lessee do all the quantum work. In particular, the step of preparing the key is challenging to dequantize--the lessee has to prepare a ``computationally unclonable'' quantum state from only classical information and without knowing its description.

Naturally, we look for inspiration from methods of dequantizing the client in the context of delegating quantum computation and state preparation. 
\jiahui{discussed more barriers here from the perspective of RSP}
One natural approach is to run a remote state preparation (RSP) protocol that allows a classical client to instruct a quantum server to prepare a quantum state whose description is known to the client but not to the server.\footnote{More precisely, what we need is an RSP protocol with verification, in which the quantum client can verify that the quantum server has indeed prepared the desired state up to isometry.} However, existing RSP protocols~\cite{FOCS:GheVid19,FOCS:Zhang22,ICALP:GheMetPor23,ITCS:Zhang25} suffer from the drawback that they achieve only inverse-polynomial security. Consequently, directly using these protocols to replace quantum communication in existing SKL schemes yields a scheme whose security bound is fixed in advance to be inverse-polynomial.  It remains unclear how to amplify this to the standard notion of negligible security. \takashi{added this issue.} In addition, existing RSP protocols are highly costly in both communication and rounds, typically requiring polynomially many rounds. This contradicts the goal of designing a lightweight protocol. The cost of generic-purpose RSP protocols stems from their very strong security requirement: more concretely, a rigidity 
statement asserting that the server must exactly perform certain quantum operations up to isometry. However, such strong guarantees are not necessary for our purposes.  We simply need the server (lessee) to prepare a classically functional quantum state and later delete the function.
\fuyuki{I think not only RSP is inefficient, but also it is highly non-trivial to adapt existing RSP protocols for our purpose mainly because they achieve only inverse polynomial soundness (and soundness amplification seems hard to perform due to polynomially many rounds). We may move this paragraph to related work section and add more detailed discussions.}\takashi{added}

Another possible approach is delegating the computation via a quantum fully homomorphic encryption (QFHE) scheme. By sending QFHE's ciphertext of $x,\theta, \{\sk_{i,0},\sk_{i,1}\}_{i \in [n]}$, the lessor can let the lessee prepare the above quantum key by homomorphically evaluating on the ciphertext. However, in our case, the obtained key on the server (lessee)'s end will be encrypted and unusable for the decryption functionality. On the other hand, revealing the QFHE decryption information will result in insecurity.

The seminal works \cite{FOCS:BCMVV18,FOCS:Mahadev18a} introduced a protocol where a classical client can verify the quantumness and even any BQP computation of a quantum server.
    In particular, during the protocol of \cite{FOCS:BCMVV18,FOCS:Mahadev18a}, after receiving only classical information of a function $f$ from the client, the server has to create a (computationally) unclonable state on its own end, with the following structure: 
\begin{align}
    \frac{1}{\sqrt{2}}(\ket{0,a_0}+\ket{1,a_1}), f(a_0) = f(a_1) = y
\end{align}  
    Such a function family $f$ is called a \emph{(noisy) trapdoor claw-free family} \footnote{The pre-images $a_0,a_1$ are conventionally denoted as $x_0, x_1$. We do the same in the actual construction and proof. Here, we use $a$ to distinguish from the string $x$ in the previously mentioned BB84 states.}. After preparing the above state (which we will refer to from now on as the \emph{claw state}), the server has to "commit" to this state by sending back the classical image $y$. 
    
    In the next step of the protocol, the client will request the server to measure this state randomly in either the computational or Hadamard basis and send back the result. The client, using a trapdoor for the corresponding $f$ (which it generates together with $f$) can verify if the server has done the requested measurement.
    
Informally, the trapdoor claw-free function has the following properties: 
\begin{enumerate}
    \item It is computationally hard to find two pre-images $a_0, a_1$ such that $f(a_0) = f(a_1) = y$;

    \item Once the server has measured the state in the Hadamard basis, and obtains a string $(e,d) \in \{0,1\}^{n+1}$ such that $e = d \cdot ( a_0 \oplus  a_1)$ mod 2 that it sends back to the client for verification, it will (computationally) be deprived of the preimage information $a_0$ and $a_1$ encoded in the computational basis. 

    \item It can also be extended to have a dual-mode property: in the two-to-one mode, there are two pre-images $a_0,a_1$ for each image $y$ (as shown above); in the injective mode, there is only one preimage $a$ for each image $y$.
    In other words, the claw state prepared by the server in this mode will be $\ket{b,a_b}$, where $ f(a_b) = y$ for some random $b\in \{0,1\}$. %\jiahui{check this?}
    The two modes are computationally indistinguishable.
\end{enumerate}

The second property is formalized as \emph{adaptive hardcore bit property} in \cite{FOCS:BCMVV18,FOCS:Mahadev18a}.
We can observe that this adaptive hardcore bit property aligns with the property of the BB84 states described above: if one measures the keys in the Hadamard basis, one loses all the information, i.e. the classical description for the secret keys, in the computational basis.

To make this intuition work,  we will look for an approach where we can let the server prepare the secret key $\qsk$ in the above BB84 format \eqref{eqn:qsk_format}. The lessor will pick uniformly random $\theta \gets \{0,1\}^n$; it then 
generates $n$ NTCF functions (with corresponding trapdoors) $(f_1,\cdots,f_n)$ such that $f_i$ is in injective mode for positions $\theta[i] = 0$ and  $f_i$ is in two-to-one mode for positions $\theta[i] = 1$; then it sends $(f_1,\cdots,f_n)$  to the lessee. 
After some key generation procedure that helps the lessee go from preparing claw states to preparing quantum keys (which we will go into later), the lessee obtains the $\qsk$ where: for indices $i$ with $\theta[i] = 0$, the quantum key $\qsk_i := \ket{b_i}\ket{\sk_{i, b_i}}$ is prepared by the injective mode $f_i$; for indices $i$, $\theta[i] = 1$, the quantum key $\qsk_i := \frac{1}{\sqrt{2}}(\ket{0}\ket{\sk_{i, 0}}+ \ket{1}\ket{\sk_{i,1}})$ is prepared by the two-to-one  mode $f_i$.
\fuyuki{I thought this paragraph might be confusing; $x$ is not chosen by the lessor, it is fixed as a result of the computation by lessee. For the two-to-one mode position, the phase $(-1)^{x[i]} $ is not added to $\qsk_i$.} \jiahui{sorry I got myself confused with the notations. Changed}

To delete the key, the lessor asks the lessee to measure all the keys  $\qsk := (\qsk_1, \cdots, \qsk_n)$ in the Hadamard basis and send back the measurement outcomes. Using the verification key which contains the NTCF trapdoors for the indices encoded in Hadamard basis, the lessor can verify the deletion (measurement) outcomes sent by the lessee\footnote{As an informed reader may notice, the classical-lessor PKE-SKL in \cite{chardouvelis2025quantum} is also based on the NTCF approach in  \cite{JACM:BCMVV21,FOCS:Mahadev18a}. In fact, they directly use the claw state as the quantum secret key. But we will soon discuss in the upcoming technical overview that using exactly the claw state may not be enough for a modular SKL scheme and may be difficult  to apply on more generic cryptographic primitives. }.
\jiahui{add a comment on what \cite{chardouvelis2025quantum} does}

%The hope is that we can leverage the properties of the \cite{FOCS:BCMVV18,FOCS:Mahadev18a} protocol in an intricate way so that we get a classically preparable decryption key in the above format and show the security 

\paragraph{Preparing the quantum key with secure function evaluation}
The next step is to allow the lessee to go from preparing the claw state to preparing the actual quantum key $\qsk$. That is, we start from generating a state of the form $\ket{0}\ket{a_0}+\ket{1}\ket{a_1}$ by using an NTCF, and convert it into a state of the form $\ket{0}\ket{\sk_0}+\ket{1}\ket{\sk_1}$ for secret keys $\sk_0$ and $\sk_1$ \footnote{In the technical overview, for the sake of presentation simplicity, we will sometimes only focus on the $\ket{+}$-format of the claw state and secret key (i.e. $\ket{0}\ket{a_0}+\ket{1}\ket{a_1}$  and correspondingly $\ket{0}\ket{\sk_0}+\ket{1}\ket{\sk_1}$ ) . But it is easy to see that the following discussions apply to the $\ket{0}, \ket{1}$-versions of the state.}. %Jiahui described an idea to do that based on the encrypted CNOT technique. 
%\jiahui{modified from the notes}
Our idea is to achieve this by secure function evaluation (SFE). Roughly, SFE is a two-round protocol  between a sender and receiver that works as follows:
\begin{itemize}
\item The receiver takes a message $m$ as input and sends a ciphertext $\ct$ to the sender while keeping a state information $\st$. %\takashi{I quit to explicitly mention the randomness $r$ and instead use a state $\st$ (which I believe is a more common formalization). This is because the randomness in the LWE-based construction isn't actually unique if we go back to the randomness of the Gaussian noise sampler.}
\item The sender takes a circuit $C$ as input, receives $\ct$, computes $\ct'$ and sends $\ct'$ to the receiver.
\item The receiver recovers $C(m)$ from $\st$ and $\ct'$. 
\end{itemize}

In particular, we consider SFE in the CRS model. 
The security roughly requires that the sender learns nothing about $m$, and the receiver learns nothing about $C$ beyond $C(m)$.  We will specify here some of the concrete security properties we require from the SFE. Notably, we will show that we can obtain SFE with all the properties we need from LWE.
In particular, we require the following \emph{dual-mode} property. \footnote{State recoverability in the hiding mode is not a standard requirement for dual-mode SFE, but the LWE-based construction satisfies it.}%\footnote{There is actually an additional correctness requirement about the possibility of running the receiver in superposition without causing entanglement with the randomness. The LWE-based instantiation satisfies it (essentially because we can generate Gaussian superposition without having any entanglement with "randomness" of the sampler). See also \Cref{footnote:superposition_execution_SFE}.}
%\takashi{I made it clear the dual-mode requirement for SFE.}
\begin{itemize} 
\item {\bf Mode indistinguishability.}
There are two modes for generation of CRS, 
\emph{hiding mode} and
\emph{extraction mode}. 
In both modes, an associated trapdoor is generated along with CRS. 
CRSs generated in different modes are computationally indistinguishable. 
\item {\bf Statistical security against malicious senders in the hiding mode}:  
In the hiding mode, $\ct$ statistically hides $m$.
\item {\bf State recoverability in the hiding mode:}
In the hiding mode, 
$\ct$ and $m$ uniquely determine the corresponding $\st$. Moreover, the trapdoor enables us to efficiently recover $\st$ from $\ct$ and $m$. 
\item 
{\bf Extractability in the extractable mode}: 
In the extractable mode, we can extract $m$ from (possibly maliciously generated) $\ct$ using the trapdoor, so that $\ct'$ can be simulated only using $C(m)$ (and not using $C$). 
Note that this automatically implies (computational) security against malicious receivers in the extractable mode. \fuyuki{''statistical'' is replaced with ``computational''. (We use GC.)}
\end{itemize}
In addition to the above, we need two additional properties: one is the "efficient superposition generation property" and the other is the "decomposability of the state". 
We explain those additional properties when we need them.\footnote{The former appears in the description of the sub-protocol below; the latter appears in the idea for the security proof.}

Using SFE, consider the following sub-protocol to generate the key. For simplicity, we consider the "stand-alone" generation of a single key $\qsk_i$ as part of the final $\qsk$, where $\theta[i] = 1$ and $x[i] = 0$.
\begin{enumerate} 
    \item The lessor generates 
    CRS of SFE in the hiding mode and a two-to-one function $f$ from NTCF  and sends $(CRS,f)$ to the lessee.  
    \item The lessee prepares the claw state $\ket{0}\ket{a_0}+\ket{1}\ket{a_1}$ along with $y = f(a_0) = f(a_1)$ by the standard usage of NTCF and sends $y$ to the lessor. 
    \item The lessee 
     runs the receiver's algorithm of SFE in superposition, where the second register of   $\ket{0}\ket{a_0}+\ket{1}\ket{a_1}$ is treated as the message,  
    and measures the resulting $\ct$.  
   Since $\ct$ statistically hides the message (which is either $a_0$ or $a_1$), this measurement almost never collapses the two branches. Moreover, since $\ct$ and the message uniquely determines the corresponding state, the resulting state (un-normalized for convenience) looks like: 
   \[
\ket{0}\ket{a_0}\ket{\st_0}+\ket{1}\ket{a_1}\ket{\st_1}
   \]
   where $\st_0$ and $\st_1$ are state information corresponding to $(a_0,\ct)$ and $(a_1,\ct)$, respectively.
   \footnote{
   While this might be unclear from the notation, 
   $\st_b$ depends on $x_b$ (and not only on $b$.)
   We write $\st_b$ instead of $\st_{x_b}$ for simplicity.
   }
Here, we are relying on an additional property of our SFE "efficient superposition generation property", which makes sure that we can generate the above state without remaining entanglement with the randomness used by the receiver. 
%The LWE-based instantiation satisfies this.\footnote{This is essentially due to the fact that we can generate Gaussian superposition.}  
    
   The lessee sends the measured ciphertext $\ct$ to the lessor.

\item In the next step, the lessor can define a circuit $C$ that maps $a_b$ to some secret key $\sk_b$ for $b \in \{0,1\}$ respectively.  Then the lessor runs $C$ on $\ct$ and sends the result $\ct'$ to the lessee and the lessee runs the output derivation algorithm coherently, writing the outcome on an additional register.
By the correctness of the SFE, the lessee will finally obtain $\ket{0}\ket{a_0}\ket{\st_0}\ket{\sk_0}+\ket{1}\ket{a_1}\ket{\st_1}\ket{\sk_1}$. 

However, in the actual construction, the circuit $C$ will be slightly more sophisticated than described here. We will elaborate in the next subsection when we discuss the security proof.

\end{enumerate}
While this is not exactly the state of the form  $\ket{0}\ket{\sk_0}+\ket{1}\ket{\sk_1}$ as we originally discussed, it is already good enough for decryption since both branches contain the information of a secret key, so that during decryption, the lessee can use the secret key in the last register to decrypt coherently. 
Moreover, if we define the deletion algorithm as just a Hadamard basis measurement of the above state, then we can still define a deletion verification algorithm, where we verify that the certificate $(e,d,c)$ satisfies:
$e = (d \| c) \cdot (a_0\|\st_0\|\sk_0 \oplus a_1\|\st_1\|\sk_1)$. This can also be verified with the trapdoor of the NTCF, the trapdoor of the SFE and the circuit $C$. We will make some further minor changes to the deletion and verification algorithms, but for the purpose of intuition, we leave them as they are now.
%is orthogonal to $0\|a_0\|\st_0\|\sk_0\oplus 1\|a_1\|\st_1\|\sk_1$.\footnote{Since the lessor knows the trapdoor, it can recover $\st_0$ and $\st_1$.}

 \paragraph{Watermarking: extraction of the pre-images}
For the security proof, we would like to make the following argument. Recall the security states that the lessee, after successfully passing the revocation verification, will not be able to distinguish the encryptions of two messages (an IND-CPA style security game). 
Suppose there exists a quantum adversary that passes both the revocation verification and the IND-CPA game afterwards, then we should be able to break the security of the NTCF used, namely, the adaptive hardcore bit property. We can argue in a relatively straightforward manner that the adversary provides the reduction with the measurement outcomes in the Hadamard basis, if it passes the revocation verification. The trickier part is to argue that by passing the IND-CPA game, we can extract the preimage $a_0$ or $a_1$ from this quantum decryptor. Finally, by gathering both the measurements in the Hadamard basis and the preimages $a_0$ or $a_1$, the reduction can break the adaptive hardcore bit property of the NTCF.

The previous work on fully classical lessor SKL \cite{chardouvelis2025quantum} also uses an NTCF based construction and need to extract the preimages in order to break the adaptive hardcore bit security. They design an algorithm that carefully extracts the preimage information from a possibly destructive quantum decryptor, by using quantum extraction techniques as well as LWE properties. However, this approach is likely limited to analyzing the LWE based PKE scheme in \cite{chardouvelis2025quantum}, and does not generalize to any PKE and other cryptographic functionalities. 

In this work we provide a simple approach through \emph{watermarking}.
This approach is not only is cleaner, but also works for any PKE scheme and is generalizable to other cryptographic functionalities such as PRF and signatures presented in this work.
A PKE with watermarking (often called watermarkable PKE) has syntax similar to standard PKE except that the $\KeyGen$ algorithm generates a marking key along with the encryption/decryption keys and we have an additional "marking" algorithm $\Mark$. $\Mark$ will embed a watermark into an input cryptographic function; the watermarking security roughly states that trying to remove the watermark will destroy the functionality and thereby any functional program output by an adversary who receives a marked program will contain the original mark with high probability. We don't need additional assumptions for a watermarkable PKE scheme: any PKE scheme can be upgraded to have the watermarking security requirements we need.

Next, we go back to modify step 4 of the construction when the lessor receives the ciphertext $\ct$ of SFE in step 3 from the lessee. It does the following instead:
\begin{itemize}
        \item 
    The lessor defines a circuit $C$ that takes $a$ as input and outputs a "watermarked" secret key $\sk(a)$ in which a "mark" $a$ is embedded. %\takashi{Just as a matter of the notation, I'm wondering which of $\sk(x)$ or $\sk[x]$ (or anything else) is good to denote a watermarked secret key. I avoided $\sk[x]$ since we often use the notation of $x[i]$ to mean the $i$-th bit of $x$.} 
    The precise requirement for the watermarking will become clear later. %\takashi{I clarified that $C$ is a mark embedding function.} 
    Then runs the sender's algorithm of SFE on $C$ and $\ct$ to generate $\ct'$ and sends it to the lessee.
    \item The lessee runs the output derivation algorithm of SFE in superposition to get
\begin{align}
&\ket{0}\ket{a_0}\ket{\st_0}\ket{C(a_0)}+\ket{1}\ket{a_1}\ket{\st_1}\ket{C(a_1)}\\
=&\ket{0}\ket{a_0}\ket{\st_0}\ket{\sk(a_0)}+\ket{1}\ket{a_1}\ket{\st_1}\ket{\sk(a_1)}
\label{eqn:qsk_format_ours}.
\end{align}
\end{itemize}

%Watermarkable PKE (with parallel extraction): The syntax is similar to standard PKE except that the KeyGen generates $(\pk,\msk)$ and we have an additional "marking" algorithm that takes $\msk$ and $x$ as input and outputs a marked secret key $\sk(x)$. 

Now going back to the security proof: the watermarking security requires that if an adversary given $\sk(a)$ (where the mark $a$ is chosen by the adversary) generates a (possibly quantum) "decryptor" that can correctly decrypt a ciphertext of a uniformly random message with probability $\epsilon$, then we can extract $a$ from the decryptor with probability $\epsilon-\negl(\secpar)$.\footnote{This security guarantee is weaker than standard watermarking security, but is sufficient for our purpose.\fuyuki{I added this footnote.}}
This even holds if the decryptor is a QPT algorithm that has an auxiliary  quantum state, as in our scenario. Therefore, we can use a successful quantum decryptor to extract the marked pre-image $a_0$ or $a_1$.

Moreover, we show that a similar extractability holds even in the parallel setting where many instances of the scheme under independently generated keys. 
 That is, if the decryptor succeeds in decryption simultaneously for all instances with probability $\epsilon$, then we should succeed in simultaneously extracting the marks with probability $\epsilon-\negl(\secp)$. Recall that in our final SKL protocol, the secret key $\qsk$ contains $n$ independently generated keys $\qsk_i$. We will extract the watermarked preimages $a_{i,b_i}$ from all of them using the parallel extraction property. %In the next paragraph, we briefly discuss how this reduction breaks a 

\paragraph{Cut-and-choose adaptive hardcore bit}
To make the final security reduction go through, we need to perform a reduction where the reduction plays as the adversary in a security game regarding the NTCF and hopes to leverage the PKE-SKL adversary to break the security of NTCF, namely the adaptive hardcore bit property previously mentioned.

One subtlety is that not all our quantum keys $\qsk_i, i\in [n]$ are  "claw states" that contain a superposition of two claws, but some are classical states that only contain one pre-image, where the adaptive hardcore bit security is not well defined.
On the other hand, let us 
recall that the BB84 states have a certified deletion property shown in \cite{EPRINT:BehSatShi21}): if an adversary $\qA$ without the knowledge of $x,\theta$ can correctly output all the value $x[i]$ encoded in Hadamard basis (i.e. $x[i]$ for all $i$ where $\theta[i] = 1$), it is impossible for $\qA$ to obtain the value of $x[i]$ for the positions $i$ encoded in the computational basis.
%Our observation is that the cut-and-choose adaptive hardcore bit property is very 
We therefore would like to find a property of the NTCF/claw states similar to the certified deletion property of BB84 states used in \cite{kitagawa2025simple}, and thus we may be able to replace the latter with the former in the security proof.

Here, we use a property called \emph{cut-and-choose adaptive hardcore bit} property for dual-mode NTCF
introduced by
Hiroka et al.~\cite{AC:HMNY21}. 
They prove that \emph{any} dual-mode NTCF that satisfies the adaptive hardcore bit property also satisfies the   cut-and-choose adaptive hardcore bit property. 
Our final construction is designed in such a way that its security can be reduced to the cut-and-choose adaptive hardcore bit property. 

%In the cut-and-choose adaptive hardcore bit property, 
We consider the following experiment between an adversary  $\qA$ and challenger: 
\begin{enumerate}
    \item The challenger chooses a uniform subset $S\subseteq [2n]$ such that $|S|=n$. 
    \item 
    For each $i\in [2n]$, the challenger generates a function $f_i$ along with its trapdoor, where:
    \begin{itemize}
    \item If $i\in S$, $f_i$ is in the injective mode,
    \item If $i\notin S$, $f_i$ is in the two-to-one mode. 
    \end{itemize}
    The challenger 
    sends $\{f_i\}_{i\in[2n]}$ to $\qA$. 
    \item $\qA$ sends 
    $\{y_i,d_i,e_i\}_{i\in[2n]}$ to the challenger. 
    \item 
    For each $i\notin S$, 
    the challenger recovers the preimages $(a_{i,0},a_{i,1})$, and checks if $e_i=d_i\cdot(a_{i,0}\oplus a_{i,1})$ and $d_i\ne \mathbf{0}$ hold. 
    If they do not hold for some $i\notin S$, the challenger immediately aborts and the experiment returns $0$. 
    \item  \label{step:reveal_S} 
    The challenger sends $S$ to $\qA$.
    \item $\qA$ sends $\{b_i,a_i\}_{i\in S}$ to the challenger. 
    \item 
    The challenger checks if $f_{i}(b_i,a_i)=y_i$ holds for all $i\in S$.
    If this holds for all $i\in S$, the experiment returns $1$. Otherwise, it returns $0$.
\end{enumerate}
The cut-and-choose adaptive hardcore bit property claims that 
for any QPT adversary $\qA$, the probability that the above experiment returns $1$ is negligible (for sufficiently large $n$).

\paragraph{Putting things together for the security analysis}
%\jiahui{Taken from the notes for now. Modify later}
Now our goal is to reduce the PKE-SKL security to the cut-and-choose adaptive hardcore bit property, in a manner similar to that of \cite{kitagawa2025simple} with the certified deletion property of BB84 states. 
To do so, we mainly have to prove the following two steps:
\begin{enumerate}
\item 
If one knows the trapdoors of the SFE (but not that of NTCF), 
any valid deletion certificate
$\cert=(e_i,d_i,c_i)_{i\in [2n]}$ can be used to compute $(\hat{e}_i, \hat{d}_i)_{i\notin S}$ such that $\hat{e}_i=\hat{d}_i\cdot (a_{i,0}\oplus a_{i,1})$ and $\hat{d}_i\ne \mathbf{0}$. 
\item 
If a  "decryptor" (the state after passing the deletion verification)  correctly decrypts the special ciphertext with probability $\epsilon$, we can extract $\{a_i\}_{i\in S}$ with probability at least $\epsilon-\negl(\secp)$.  
\end{enumerate}

Recall that in the PKE-SKL security game, the adversary has to hand in a valid deletion certificate and then succeeds in decrypting. Therefore, the goal is to use such an adversary
to obtain both of the above items; then we can use them to break the cut-and-choose adaptive hardcore bit property of the NTCF.

%\jiahui{add a few words why}
%We refer to \cite{cryptoeprint:2024/1564} on why the above are sufficient for completing the reduction. 

Now we argue how to show each item above.

As we have discussed in the watermarking paragraph, the second item straightforwardly follows from the security of watermarkable PKE. 

%noting that the adversary is given just one watermarked key $\sk_i(x_i)$ for each $i\in S$.  

However, the first item is less obvious. First, recall that for all $i \notin S$ (indices for 2-to-1 $f_i$), the (unnormalized) quantum keys will look like 
\begin{align}
   \qsk_i &= \ket{0}\ket{a_{i,0}}\ket{\st_{i,0}}\ket{\sk(a_{i,0})}+\ket{1}\ket{a_{i,1}}\ket{\st_{i,1}}\ket{\sk(a_{i,1})}  \\
\end{align}
We make some minor modifications on how we delete a key to make the analysis simpler: in the deletion algorithm, the honest lessee first uncomputes the register containing the secret keys $\ket{\sk(a_{i,b})}, b\in\{0,1\}$ using the circuit $C$ from SFE which generates the key.
Afterwards, it measures all registers in the Hadamard basis.
Now we know that when a deletion certificate $\cert=(e_i,d_i,c_i)_{i\in [2n]}$ received from the adversary is valid, we have the following for all $i \notin S$:
$$e_i= d_i\cdot(a_{i,0}\oplus a_{i,1})\oplus c_i\cdot(\st_{i,0}\oplus \st_{i,1}).$$ 

%\jiahui{I changed all $x$ to $a$ for consistency and avoid the confusion with previous $x$ in BB84 states}
Here, we need to use another property of the SFE, named "decomposability of the state", which we have mentioned but did not go into details. Elaborated below, this property allows us to obtain item 1 in a relatively clean way.

Recall that the receiver in the SFE protocol takes a message $a$ and a common reference string CRS, then  outputs a ciphertext $\ct$ to the sender while also keeping a (classical) secret state information $\st$.
This        ``decomposability of the state'' property is with respect to the state $\st$ of the SFE: intuitively, it says that $\st$ corresponding to input message $a$ and output $\ct$ (sometimes we say the state $\st$ corresponds to $(a,\ct)$ for simplicity) is of the following form
\[
\st=\st_{1,a[1]}\|\st_{2,a[2]}\|\ldots\|\st_{m,a[m]}
\]
\jiahui{The index $m$ here may appear a little sudden. Is there a better way to explain it?}
where 
$m$ is the bit-length of $a$, 
$a[j]$ is the $j$-th bit of $a$, 
and $(\st_{j,b})_{j\in [m],b\in \bit}$ is a family of strings that only depend on $\ct$ and CRS (and not on $a$). 
Using this, 
for each $b\in \bit$ we can write $\st_{i,b}$ as
\[
\st_{i,b}=\st_{i,1,a_{i,b}[1]}\|\st_{i,2,a_{i,b}[2]}\|\ldots\|\st_{i,m,a_{i,b}[m]}.
\]
For each $i\in [2n]$, $b\in \bit$, and $j \in [m]$, let $\st_{i,b}[j]$ be the $j$-th block of $\st_{i,b}$, i.e.,
\[
\st_{i,b}[j]=\st_{i,j,a_{i,b}[j]}.
\]
%takashi{Sorry for the super confusing notation. We may need to improve it later.}
We divide the string $c_i$ in the deletion certificate $\cert = \{(e_i,c_i,d_i)\}_{i \in [2n]} $into $m$ blocks:
\[
c_i=c_{i,1}\|c_{i,2}\|\ldots\|c_{i,m}. 
\]
For $i\in [2n]$ and $j\in [m]$, we have
\begin{align}
&~~~c_{i,j}\cdot (\st_{i,0}[j]\oplus \st_{i,1}[j])\\
&=
c_{i,j}\cdot (\st_{i,j,a_{i,0}[j]}\oplus \st_{i,j,a_{i,1}[j]})\\
&=
\begin{cases}
    0&\text{~if~}a_{i,0}[j]=a_{i,1}[j]\\
    c_{i,j}\cdot (\st_{i,j,0}\oplus \st_{i,j,1})&\text{~if~}a_{i,0}[j]\ne a_{i,1}[j]
\end{cases}\\
&=( c_{i,j}\cdot (\st_{i,j,0}\oplus \st_{i,j,1}))\cdot(a_{i,0}[j]\oplus a_{i,1}[j]).
\end{align}

By observing the above calculation, we define a string $d^*_i$ as 
\[
d^*_i[j]=c_{i,j}\cdot (\st_{i,j,0}\oplus \st_{i,j,1})
\]
for $j\in [m]$, then we have 
\[
c_{i,j}\cdot (\st_{i,0}[j]\oplus \st_{i,1}[j])=d_i^*[j]\cdot(a_{i,0}[j]\oplus a_{i,1}[j]).
\]
This implies
\[
c_i\cdot (\st_{i,0}\oplus \st_{i,1})=d^*_i\cdot (a_{i,0}\oplus a_{i,1}). 
\]
%\shota{Should it be $d^*_i(a_0\oplus a_1) = c_i(\st_0\oplus \st_1)$? Also, should $\st$ be indexed by $i$?}\takashi{Yes, you are right. I corrected it and added more explanation.}
Substituting this relation into $e_i= d_i\cdot(a_{i,0}\oplus a_{i,1})\oplus c_i\cdot(\st_{i,0}\oplus \st_{i,1})$,  we have 
$$e_i= (d_i\oplus d^*_i)\cdot(a_{i,0}\oplus a_{i,1}).$$ 
We can then observe that the reduction works as long as $d_i\ne d^*_i$, which happens with overwhelming probability. We therefore obtain a string $(e_i, d_i\oplus d_i^*)$ to help us break the cut-and-choose adaptive hardcore bit property, along with the computational basis information $\{a_{i,b}\}_{i \in S}$ we extract using the watermarking security from the successful decryptor.
%for each $i$ with the indices $i \notin S$ 

\paragraph{Summary so far: classical-lessor SKL from ``parallel extractable'' watermarking}
To summarize, we obtain a compiler that transforms a parallel-extractable watermarkable PKE scheme into a classical-lessor PKE-SKL scheme using NTCF and our special SFE, both of which can be based on the LWE assumption.
In this compiler, we rely on the implicit fact that the functionality of PKE is preserved under parallel composition (i.e., the encryption and decryption of the resulting classical-lessor PKE-SKL scheme respectively correspond to $2n$ parallel executions of the encryption and decryption of the underlying watermarkable PKE).
More generally, this compiler can be applied not only to watermarkable PKE but also to any cryptographic primitive equipped with parallel-extractable watermarking, provided that the functionality of the primitive is preserved under parallel composition.

\paragraph{Classical-lessor SKL for PRF and digital signature}
Using the above compiler, we obtain classical-lessor SKL for unpredictable functions (UPFs) and digital signature, where the former can then be transformed into a classical-lessor PRF-SKL via the quantum Goldreich-Levin technique, as discussed earlier.
Following the previous works~\cite{TCC:AnaPorVai23,TCC:AnaHuHua24,TQC:MorPorYam24,kitagawa2025simple}, we consider the security notion for UPF-SKL and DS-SKL where the challenge input/message is chosen uniformly at random and after the adversary outputs the deletion certificate and is given the challenge, it is not given oracle access to the leased functionality (e.g., signing oracle in the DS-SKL case).
Under this ``non-adaptive'' setting, the functionality of UPFs and digital signature is preserved under parallel composition.
Thus, relying on our compiler, to achieve classical-lessor UPF-SKL and DS-SKL, the main task is to construct a parallel-extractable watermarkable UPF and digital signature.
\begin{description}
\item[UPFs:] Building on the technique of~\cite{kitagawa2025simple}, we achieve parallel extractable watermarkable UPFs using two-key equivocal PRFs (TEPRFs)~\cite{C:HJOSW16}, which can be based on OWFs.
By combining it with our compiler, we obtain the first classical-lessor PRF-SKL based on the LWE assumption.
\item[Digital signature:] 
To apply our compiler, we require a parallel extractable watermarkable signature scheme with an additional property known as coherent signability, introduced by~\cite{kitagawa2025simple} in the context of constrained signatures.
Coherent signability is defined for digital signatures in which a signing key (e.g., a constrained or watermarked key) can be generated from a master secret key.
Intuitively, it ensures that for any message $m$ and two signing keys $\sk_0$ and $\sk_1$, one can generate a valid classical signature for $m$ using the superposed signing key $\alpha\ket{0}\ket{\sk_{0}}\ket{\psi_0}+ \beta\ket{1}\ket{\sk_{1}}\ket{\psi_1}$ with $|\alpha|^2 + |\beta|^2 = 1$, almost without disturbing the superposed state, provided that both $sk_0$ and $sk_1$ can individually generate valid signatures for $m$ using the standard (classical) signing algorithm.
This property is necessary when our compiler is applied because, in a $2$-to-$1$ mode position, a signature must be generated using such a superposed signing key (see~\cref{eqn:qsk_format_ours}).

\cite{kitagawa2025simple} showed how to construct a constrained signature scheme with coherent signability under the SIS assumption.
We extend this by showing that a parallel-extractable watermarkable signature scheme with coherent signability can be obtained by combining any constrained signature scheme with coherent signability and TEPRFs.
The construction builds on that of parallel-extractable watermarkable UPFs.
By further applying our compiler, we obtain the first classical-lessor DS-SKL based on the LWE assumption.
\end{description}

\paragraph{Removing interaction during key generation for PKE-SKL} 
Finally, note that our constructions above require a constant-round interaction between the lessor and the lessee for key generation. This is not ideal for practice.
We improve this aspect by making the key generation algorithm non-interactive in the following way: instead of letting the lessee prepare the full key in the key generation procedure, it only prepares the claw state and executes the first SFE receiver operation.  We delay the generation for the watermarkable PKEs%(or PRF, DS, resp.) 
\synote{I edited here a bit, since the trick for making non-interactive cannot applied to DS and PRF.}
keys to the \emph{encryption} algorithm $\Enc$. The ciphertext produced by $\Enc$ will contain some auxiliary information about the watermarkable PKE so that the lessee, once receiving the ciphertext, can use the auxiliary information to prepare the full key (the same one as in the interactive construction), by running the second SFE receiver operation on its previous half-baked quantum key. The resulting key will be able to decrypt the corresponding ciphertext.

\paragraph{Construction of the special SFE}
In general, SFE can be constructed from oblivious transfer (OT) via Yao's well-known garbled-circuits protocol~\cite{FOCS:Yao86}, which in turn can be based on OWFs.
The same applies in our setting: the task of realizing the required SFE in our compiler reduces to constructing an OT scheme that satisfies all the properties introduced above. (Since OT is a special case of SFE, the properties for SFE can be naturally translated into OT setting.)
We show that the OT scheme of~\cite{C:PeiVaiWat08}, based on the LWE assumption, can be naturally extended to achieve these required properties.
This gives us the special SFE needed for our compiler under the LWE assumption.

% !TEX root = main.tex

\subsection{Discussion}

One might expect that the constructions of PKE-SKL with a classical lessor in \cite{chardouvelis2025quantum,AC:PWYZ24} can be directly generalized to other primitives, such as PRF-SKL and DS-SKL. Here, we explain why such a generalization may be difficult.

We begin by briefly recalling the basic idea of \cite{chardouvelis2025quantum,AC:PWYZ24}. Using an NTCF $f$, one generates a quantum state of the form 
\begin{align}
    \frac{1}{\sqrt{2}}(\ket{0,a_0}+\ket{1,a_1}), f(a_0) = f(a_1) = y.
\end{align}
Then they observe that, for a specific LWE-based NTCF, 
each of $a_0$ and $a_1$ can be viewed as a decryption key of a certain PKE scheme. Thus, to encrypt a message $m$, one generates ciphertexts $\ct_0$ and $\ct_1$ under the corresponding public keys. Given $\ct_0$ and $\ct_1$, the state $\frac{1}{\sqrt{2}}(\ket{0,a_0}+\ket{1,a_1})$ can then be used to decrypt to $m$. The crucial property of PKE here is that decrypting $\ct_0$ with $a_0$ and decrypting $\ct_1$ with $a_1$ yield the same message $m$, by correctness of the PKE scheme. Therefore, decryption can be performed coherently in superposition without collapsing the state $\frac{1}{\sqrt{2}}(\ket{0,a_0}+\ket{1,a_1})$. 

However, this situation does not extend to general primitives such as PRFs or signature schemes. Since the issue is essentially the same for both, let us focus on signatures. Suppose that $a_0$ and $a_1$ are interpreted as secret keys of some classical signature scheme. One might then hope to regard $\frac{1}{\sqrt{2}}(\ket{0,a_0}+\ket{1,a_1})$ as a quantum signing key for a DS-SKL scheme. However, this does not work in general, because signing the same message using $a_0$ and $a_1$ may produce different signatures. In that case, the output reveals the branch, and measuring it collapses the state to one side, making the state unusable for generating the deletion certificate later. 

To circumvent this issue in the quantum communication setting, \cite{kitagawa2025simple} relies on a special signature scheme that has special pairs $(\sk_0, \sk_1)$ such that $\frac{1}{\sqrt{2}}(\ket{0,\sk_0}+\ket{1,\sk_1})$ can be used to generate a signature without collapsing the state, for almost all messages. In addition, their security proof essentially relies on another special property: a signature on a hidden special message reveals which of $\sk_0$ or $\sk_1$ was used. In our work, we abstract a signature scheme with these properties as a watermarkable signature scheme.

This makes the difficulty of directly applying the techniques of \cite{chardouvelis2025quantum,AC:PWYZ24} to signatures clear. The main obstacle is that it is unclear how to construct a watermarkable signature scheme in which the preimages $(a_0, a_1)$ of the LWE-based NTCF themselves serve as the required special signing-key pair. This is why we need an additional mechanism based on SFE: it allows us to further entangle auxiliary strings in each branch and thereby embed the signing keys of the watermarkable signature scheme into the superposition.

On the other hand, it is plausible that the constructions of \cite{chardouvelis2025quantum,AC:PWYZ24} can be directly adapted to achieve VRA security using the idea, developed in this paper, of reducing security to the cut-and-choose adaptive hardcore-bit property of NTCFs. More precisely, one could modify the constructions of \cite{chardouvelis2025quantum,AC:PWYZ24} so that the NTCF is used in injective mode on half of the coordinates and in two-to-one mode on the remaining coordinates. If this idea works for the construction of \cite{AC:PWYZ24}, it would yield a PKE-SKL scheme with a classical lessor from the \emph{polynomial}-modulus LWE assumption, which would improve on our assumptions. In this paper, we do not pursue this direction, since our goal is to provide a general framework applicable to various primitives, rather than to optimize the construction of PKE-SKL.

\ifnum\submission=1
\subsection{Paper Organization}
Due to page limitations, in the main body we present only our construction of classical-lessor PKE-SKL together with its security proof.
All other contents are provided in the supplemental materials, organized as follows.
\begin{itemize}
\item In \cref{sec-prelim}, we provide preliminaries and review the definitions of standard cryptographic primitives.
\item In \cref{sec-watermarking}, we introduce our definition of watermarkable cryptographic primitives and present concrete constructions.
\item In \cref{sec:dual_mode_SFE}, we formally define our special SFE.
\item In \cref{sec-def-SKL}, we provide the definitions of classical-lessor SKL.
\item In \cref{sec-PRFSKL} and \cref{sec-DSSKL}, we present our constructions of classical-lessor PRF-SKL and DS-SKL, together with their security proofs.
\item In \cref{sec:dual-mode-SFE-from-LWE}, we show how to realize our special SFE.
\end{itemize}
\else\fi

% !TEX root = main.tex

\section{Preliminaries}\label{sec-prelim}
\paragraph{Notations and conventions.}
In this paper, standard math or sans serif font stands for classical algorithms (e.g., $C$ or $\algo{Gen}$) and classical variables (e.g., $x$ or $\keys{pk}$).
Calligraphic font stands for quantum algorithms (e.g., $\qalgo{Gen}$) and calligraphic font and/or the bracket notation for (mixed) quantum states (e.g., $\qstate{q}$ or $\ket{\psi}$).

Let $[\ell]$ denote the set of integers $\{1, \cdots, \ell \}$, $\secp$ denote a security parameter, and $y \seteq z$ denote that $y$ is set, defined, or substituted by $z$.
For a finite set $X$ and a distribution $D$, $x \chosen X$ denotes selecting an element from $X$ uniformly at random, $x \chosen D$ denotes sampling an element $x$ according to $D$. Let $y \gets \algo{A}(x)$ and $y \gets \qalgo{A}(\qstate{x})$ denote assigning to $y$ the output of a probabilistic or deterministic algorithm $\algo{A}$ and a quantum algorithm $\qalgo{A}$ on an input $x$ and $\qstate{x}$, respectively. When we explicitly show that $\algo{A}$ uses randomness $r$, we write $y \gets \algo{A}(x;r)$.
For a classical algorithm $\algo{A}$, we write $y\in \Supp(\algo{A}(x))$ or $y\in \algo{A}(x)$ to mean that $y$ is in the range of $\algo{A}(x)$, i.e., $\Pr[\algo{A}(x)=y]>0$. 
PPT and QPT algorithms stand for probabilistic polynomial-time algorithms and polynomial-time quantum algorithms, respectively.
Let $\negl$ denote a negligible function and $\poly$ denote a positive polynomial. 
For a string $x\in\bit^n$, $x[i]$ is its $i$-th bit. 
For strings $x,y\in \bit^n$, $\langle x, y \rangle$ denotes $\bigoplus_{i\in[n]} x[i]\cdot y[i]$. For quantum states $\rho$ and $\sigma$, $\|\rho-\sigma\|_{tr}$ means their trace distance. 
When we consider security, we consider non-uniform QPT adversary by default. 
%where $x_i$ and $y_i$ denote the $i$th bit of $x$ and $y$, respectively. 
%Similarly, for vectors $x,y\in \mathbb{F}_p^n$, $x\cdot y$ denotes $\sum_{i\in[n]} x_i y_i$ where $x_i$ and $y_i$ denote the $i$th component of $x$ and $y$, respectively. 
%For a projection $\Pi$, a binary-outcome measurement w.r.t. $\Pi$ means a projective measurement $\{I-\Pi,\Pi\}$ that returns $1$ if the state is projected onto $\Pi$ and otherwise $0$.  
%\takashi{I commented out notations which I believe are not used in the current version.}

\subsection{Cryptographic Primitives}\label{sec:standard_crypto}

\paragraph{Public-key encryption.}
\begin{definition}[PKE]
A PKE scheme $\PKE$ is a tuple of three algorithms $(\KG, \Enc, \Dec)$. 
Below, let $\Ms$  be the message space of $\PKE$. 
\begin{description}
%\item[$\Setup(1^\secp,1^{\numkey})\ra\msk$:] The setup algorithm takes a security parameter $1^\lambda$ and a collusion bound $1^{\numkey}$, and outputs a master secret key $\msk$.
\item[$\KG(1^\secp)\ra(\ek,\dk)$:] The key generation algorithm is a PPT algorithm that takes a security parameter $1^\lambda$, and outputs an encryption key $\ek$ and a decryption key $\dk$. 

\item[$\Enc(\ek,m)\ra\ct$:] The encryption algorithm is a PPT algorithm that takes an encryption key $\ek$ and a message $m \in \Ms$, and outputs a ciphertext $\ct$.

\item[$\Dec(\dk,\ct)\ra\tilde{m}$:] The decryption algorithm is a deterministic classical 
polynomial-time algorithm that takes a decryption key $\dk$ and a ciphertext $\ct$, and outputs a value $\tilde{m}$.

\item[Correctness:]For every $m \in \Ms$, we have
\begin{align}
\Pr\left[
\Dec(\dk, \ct) \allowbreak = m
\ \middle |
\begin{array}{ll}
(\ek,\dk)\gets\KG(1^\secp)\\
\ct\gets\Enc(\ek,m)
\end{array}
\right] 
=1.
\end{align}
\end{description}
\end{definition}

\begin{definition}[IND-CPA Security]\label{def:IND-CPA_PKE}
We say that a PKE scheme $\PKE$ with the message space $\Ms$ is IND-CPA secure if it satisfies the following requirement, formalized from the experiment $\expb{\PKE,\qA}{ind}{cpa}(1^\secp,\coin)$ between an adversary $\qA$ and the challenger:
        \begin{enumerate}
            \item  The challenger runs $(\ek,\dk)\gets\KG(1^\secp)$ and sends $\ek$ to $\qA$. 
            \item $\qA$ sends $(m_0^*,m_1^*)\in \Ms^2$ 
            to the challenger. 
            \item The challenger generates $\ct^*\la\Enc(\ek,m_\coin^*)$ and sends $\ct^*$ to $\qA$.
            \item $\qA$ outputs a guess $\coin^\prime$ for $\coin$. The challenger outputs $\coin'$ as the final output of the experiment.
        \end{enumerate}
        For any QPT $\qA$, it holds that
\begin{align}
\advb{\PKE,\qA}{ind}{cpa}(\secp) \seteq \abs{\Pr[\expb{\PKE,\qA}{ind}{cpa} (1^\secp,0) = 1] - \Pr[\expb{\PKE,\qA}{ind}{cpa} (1^\secp,1) = 1] }\leq \negl(\secp).
\end{align}
\end{definition}

\paragraph{Two-key equivocal PRF.}
We rely on a primitive called two-key equivocal PRF (TEPRF) introduced in  \cite{C:HJOSW16}.
%\takashi{The following definition is taken from \cite{EC:KitMorYam25}, which is slightly different from the original one in \cite{C:HJOSW16}.}
The following definition is taken verbatim from  \cite{kitagawa2025simple}. 
% \takashi{Here is a notational collision; $\ell$ in the following definition is different from the one in the description of watermarkable UPF.
% (In TEPRF, $\ell$ is the length of the input, in watermarkable UPF, $\ell$ is the length of the mark.) We need to resolve the notational collision later.
% }
\begin{definition}[Two-Key Equivocal PRF]\label{def:TEPRF} 
A two-key equivocal PRF (TEPRF) with input length $\ell$ (and output length $1$)\footnote{Though we can also define it for larger output length, we assume the output length to be $1$ by default similarly to \cite{C:HJOSW16}.} is a tuple of two algorithms $(\KG,\Eval)$.
\begin{description}
\item[$\KG(1^\secp,s^*) \ra (\key_0,\key_1)$:] The key generation algorithm is a PPT algorithm that takes as input the security parameter $1^\secp$ and a string $s^*\in \bit^\ell$, and outputs two keys $\key_0$ and $\key_1$. 
\item[$\Eval(\key,s) \ra b$:] The evaluation algorithm is a deterministic classical polynomial-time algorithm that takes as input a key $\key$ and an input $s\in \bit^\ell$, and outputs a bit $b\in \bit$. 
\end{description}
\begin{description}
    \item[Equality:]
    For all $\secp\in \mathbb{N}$, $s^*\in \bit^\ell$, $(\key_0,\key_1)\gets \KG(1^\secpar,s^*)$, $s\in \bit^\ell\setminus \{s^*\}$, 
    \[
    \Eval(\key_0,s)=\Eval(\key_1,s).
    \]
    \item[Different values on target:] 
     For all $\secp\in \mathbb{N}$, $s^*\in \bit^\ell$, $(\key_0,\key_1)\gets \KG(1^\secpar,s^*)$, 
    \[
    \Eval(\key_0,s^*)\ne \Eval(\key_1,s^*).
    \]
  \item[Differing point hiding.] 
For any (stateful) QPT adversary $\A$, 
    \begin{align}
    \abs{
    \Pr\left[
    \A(\key_b)=1:
    \begin{array}{l}
    (s^*_0,s^*_1,b) \gets \A(1^\secpar)\\
    (\key_0,\key_1)\gets \KG(1^\secpar,s^*_0)\\
    \end{array}
    \right]
    -
   \Pr\left[
    \A(\key_b)=1:
    \begin{array}{l}
    (s^*_0,s^*_1,b) \gets \A(1^\secpar)\\
    (\key_0,\key_1)\gets \KG(1^\secpar,s^*_1)\\
    \end{array}
    \right]
    }\le \negl(\secp).
    \end{align}
\end{description}
\end{definition}
\begin{theorem}[\cite{C:HJOSW16}]\label{thm:OWF_to_TEPRF}
Assuming the existence of OWFs, there exist TEPRFs with input length $\ell$  for any polynomial $\ell=\ell(\secpar)$.
\end{theorem}
We also need the following property, which follows from the differing point hiding property defined above for $\ell = \omega(\log \secp)$.
\begin{description}
\item[Hard-to-find differing-point.] 
For any $b\in \bin$ and for any QPT adversary $\qA$, we have
\begin{align}
   \Pr\left[
    \A(\key_b)=s^*:
    \begin{array}{l}
    s^*\lrun \bin^\ell \\
    (\key_0,\key_1)\gets \KG(1^\secpar,s^*)\\
    \end{array}
    \right]
    \le \negl(\secp).
\end{align}
\end{description}

\paragraph{Coherently-signable constrained signatures.}
The following definition is taken from  \cite{kitagawa2025simple}. 

\begin{definition}[Constrained signatures]\label{def:cs}
A constrained signature (CS) scheme with the message space $\Ms$ and constraint class $\Fs= \setbk{f \colon \Ms \ra \zo{}}$ is a tuple of four algorithms $(\Setup,\Constrain,\Sign,\Vrfy)$.
\begin{description}
\item[$\Setup(1^\secp) \ra (\vk,\msk)$:] The setup algorithm is a PPT algorithm that takes as input the security parameter $1^\secp$, and outputs a master signing key $\msk$ and a verification key $\vk$. 
\item[$\Constrain(\msk,f)\ra \sigk_f$:] The constrain algorithm is a PPT algorithm that takes as input the master signing key $\msk$ and a constraint $f\in \Fs$, and outputs a constrained signing key $\sigk_f$.
\item[$\Sign(\sigk_f,m) \ra \sigma$:] The signing algorithm is a PPT algorithm that takes as input a constrained key $\sigk_f$ and an message $m \in \Ms$, and outputs a signature $\sigma$.
\item[$\Vrfy(\vk,m,\sigma) \ra \top/\bot$:] The verification algorithm is a deterministic classical polynomial-time algorithm that takes as input a verification key $\vk$,  message $m\in\Ms$, and signature $\sigma$, and outputs $\top$ or $\bot$. 
\end{description}
\begin{description}
\item[Correctness:] 
For any $m\in \Ms$ and $f\in \Fs$ such that $f(m)=1$,  we have 
\begin{align}
    \Pr\left[\Vrfy(\vk,m,\sigma)=\top:
    \begin{array}{l}
    (\vk,\msk)\gets \Setup(1^\secp)\\
    \sigk_f \gets \Constrain(\msk,f)\\
    \sigma \gets \Sign(\sigk_f,m)
    \end{array}
    \right]\ge 1-\negl(\secp).
\end{align}
\end{description}
\end{definition}

\begin{definition}[Function Privacy]
%    \synote{To be checked}
    For any $m\in \Ms$, $f_0, f_1\in \Fs$ such that $f_0(m)=f_1(m)=1$, $(\vk,\msk)\in \Setup(1^\secp)$, 
    $\sigk_{f_0} \in \Constrain(\msk,f_0)$, and 
    $\sigk_{f_1} \in \Constrain(\msk,f_1)$, 
    we require the following statistical indistinguishability: 
    \[
     \sigma_0 
    \approx_s 
    \sigma_1,
    \quad \mbox{where} \quad 
    \sigma_0 \gets \Sign(\sigk_{f_0},m), \quad 
    \sigma_1 \gets \Sign(\sigk_{f_1},m).
    \]
\end{definition}

\begin{definition}[Selective single-key security]
We say that a CS scheme satisfies selective single-key security if for any stateful QPT $\qA$, we have 
\begin{align}
    \Pr\left[\Vrfy(\vk,m,\sigma)=\top
    \land f(m)=0
    :
    \begin{array}{l}
    f \gets \qA(1^\secpar)\\
    (\vk,\msk)\gets \Setup(1^\secp)\\
    \sigk_f \gets \Constrain(\msk,f)\\
    (m,\sigma) \gets \qA(\vk,\sigk_f)
    \end{array}
    \right]\le \negl(\secp).
\end{align} 
\end{definition}

We introduce an additional property which we call coherent signability. Roughly, this ensures that a superposition of two constrained signing keys can be used to generate a signature almost without collapsing the state. 
\begin{definition}[Coherent-signability]\label{def:coherent_signing} 
We say that a CS scheme is coherently-signable if for any polynomial $L=L(\secp)$,  there is a QPT algorithm $\QSign$ that takes a quantum state $\ket{\psi}$ and a classical message $m\in \Ms$ and outputs a quantum state $\ket{\psi'}$ and a signature $\sigma$, satisfying the following: 
\begin{enumerate}
\item \label{item:coherent_one}
For any family $\{ f_z\in \Fs\}_{z\in \bit^L}$, 
for any $z\in \bin^L$, 
$(\vk,\msk)\in \Setup(1^\secp)$,    
$\sigk_{f_z}\in \Constrain(\msk,f_z)$, and $m\in \cM$,  
the output distribution of $\QSign\left(\ket{z
}\ket{\sigk_{f_z}},m\right)$ is identical to that of $\Sign(\sigk_{f_z},m)$. 
\item \label{item:coherent_two}

For any family $\{ f_z\in \Fs\}_{z\in \bit^L}$, 
any amplitudes $\{\alpha_z\}_{z\in \bit^L}$ such that $\sum_{z\in \bit^L}|\alpha_z|^2=1$, 
any $(\vk,\msk)\in \Setup(1^\secp)$,  any $\sigk_{f_z}\in \Constrain(\msk,f_z)$,   and any message $m\in \Ms$ such that $f_z(m)=1$ for every $z$ with $\alpha_z\ne 1$, let $\ket{\psi}$ be the state $\ket{\psi}=\sum_{z\in \bit^L}\alpha_z\ket{z}\ket{\sigk_{f_z}}$. 
Suppose that we run $(\ket{\psi'},\sigma)\gets\QSign(\ket{\psi},m)$. 
Then we have 
\begin{align}
    \|\ket{\psi}\bra{\psi}-\ket{\psi'}\bra{\psi'}\|_{tr}=\negl(\secp).
\end{align}
\end{enumerate}
\end{definition}

\begin{lemma}[Adapted from \cite{kitagawa2025simple}]
If the SIS assumption holds for a subexponential modulus-to-target-norm-bound ratio, 
then there exists coherently-signable constrained signatures with selective security, function privacy, and coherent-signability
that supports any constrained class representable by circuits of bounded arbitrary polynomial depth.
\end{lemma}

We note that \cite{kitagawa2025simple} does not explicitly prove function privacy of their construction. 
However, this property follows almost immediately from the way their scheme is obtained: they apply the general conversion from homomorphic signatures to constrained signatures due to \cite{TCC:Tsabary17} to the homomorphic signature scheme of \cite{STOC:GorVaiWic15}. 
Since their homomorphic signature already enjoys the context-hiding property, this yields the function private constrained signatures. 
More concretely, the construction satisfies function privacy because the signature distribution is negligibly close to one that depends only on the verification key and the message—denoted as $D_{\mathbb{Z}_q^m \cap \Lambda^\perp_{\mA|\mB\mH + \mC}}$ in their notation (see Section~8.2 of \cite{kitagawa2025simple}).

\subsection{Noisy Trapdoor Claw-Free Generators}\label{sec:NTCF}
Here, we introduce noisy trapdoor claw-free (NTCF) generators, which is an abstraction of NTCF functions~\cite{FOCS:BCMVV18,FOCS:Mahadev18a} and its associated algorithms. NTCF generators with input space $\cX = \{\cX_\secp\}_\secp$ consist of the following algorithms: 
\begin{description}
  \item[$\FuncGen(1^\lambda,\mode) \to (\pp,\td)$:] 
  The function generation algorithm is a PPT algorithm that takes a security parameter $1^\lambda$ and $\mode\in \{1,2\}$ and outputs a public parameter $\pp$ and a corresponding trapdoor $\td$,
  where $\mode=1$ indicates that the system is in ``injective mode", and $\mode=2$ indicates ``two-to-one mode".
  \item[$\StateGen(\pp) \to (y, \ket{\psi})$ :]
  The state generation algorithm is a QPT algorithm that takes a public parameter $\pp$ and outputs a classical string $y$ and a quantum state $\ket{\psi}$.
  \item[$\Chk(\pp,x,y) \to \top \mbox{ or } \bot$ :]  
  The verification algorithm is a deterministic algorithm that takes a public parameter $\pp$ and classical strings $x\in \cX$ and $y$ and outputs $\top$ indicating accept or $\bot$ indicating reject.
  \item[$\Invert(\pp, \td, y) \to x \mbox{ or } (x_0,x_1) \mbox{ or } \bot:$] 
  The inversion algorithm is a PPT algorithm that takes as input a public parameter $\pp$, a trapdoor $\td$, and a classical string $y$ and outputs a single classical string $x\in \cX$, two classical strings $(x_0,x_1)\in \cX\times \cX$, or $\bot$ indicating that the inversion failed.
\end{description}
We highlight two differences from the original formulation of NTCF functions in \cite{FOCS:BCMVV18,FOCS:Mahadev18a}.
(i) In their definition, an input $x$ is mapped to a distribution over outputs $y$ via a function description (in our case, $\pp$). In contrast, our formulation disregards the definition of the distribution and retains only a verification algorithm $\Chk$, which checks whether a pair $(x,y)$ is consistent, meaning that $x$ maps to $y$ in the noiseless case, or that $y$ belongs to the support associated with $x$ in the noisy case.
(ii) The original works define injective and two-to-one modes as separate families. However, we unify them into a single family, using the parameter $\mode$ to specify the operational mode.

\begin{definition}[NTCF Generator]\label{def:NTCF-generator}
Let $(\FuncGen, \StateGen, \Chk, \Invert)$ be a set of algorithms. 
It is a NTCF generator if it satisfies the following requirements.
In the following, let $\cY_\pp$ be 
\[
\cY_\pp = \{ y : \exists x\in \cX  \mbox{~s.t.~} \Chk(\pp,x,y) = \top \}.
\]
\begin{itemize}

\item{\bf Correctness of function generation and inversion in injective mode:}
For all $(\pp,\td) \in \FuncGen(1^\secp,1)$ and $y \in \cY_\pp$, there exists a unique $x \in \cX$ such that $\Chk(\pp,x,y)=\top$.
Furthermore, $\Invert(\pp,\td,y) $ outputs $ x$.

\item{\bf Correctness of function generation and inversion in 2-to-1 mode:}
For all $(\pp,\td) \in \FuncGen(1^\secp,2)$ and $y \in \cY_\pp$, there exist distinct $x_0, x_1 \in \cX$ such that
$\Chk(\pp,x_0,y) = \Chk(\pp,x_1,y) = \top$, 
and no other $x \notin \{x_0,x_1\}$ satisfies $\Chk(\pp,x,y)=\top$.
Furthermore, $\Invert(\pp,\td,y)$ outputs $(x_0,x_1)$. 

\item{\bf Correctness of state generation:}
For all $\pp$ generated by $\FuncGen(1^\lambda,\mode)$, we have 
\[
\Pr[ y\in \cY_\pp :\StateGen(\pp)\to (y, \ket{\psi})  ] = 1-\negl(\lambda).
\]
Furthermore, with overwhelming probability over the randomness of $\StateGen(\pp)\to (y, \ket{\psi})$, we have 
\begin{align}
\parallel |\psi\rangle\langle\psi|- 
|\psi'\rangle\langle\psi'|\parallel_{tr}\leq \negl(\lambda)
\end{align}
where 
\[
\ket{\psi'} \defeq 
\begin{cases}
  \ket{x} & \text{if } y\in \cY_\pp \text{ and } \mode =1, \text{ where } x= \Invert(\pp, \td, y)\\
  \frac{1}{\sqrt{2}} (\ket{x_0} + \ket{x_1}) & \text{if } y\in \cY_\pp \text{ and } \mode=2, \text{ where } (x_0,x_1)=\Invert(\pp, \td, y)\\
  \bot & \text{Otherwise.} \\
\end{cases}
\] 

\item{\bf Indistinguishability of the modes}
For any QPT algorithm $\cA$, we have 
    \begin{align}
    \abs{
    \Pr[\cA(\pp)=1:(\pp,\td)\lrun \FuncGen(1^\secp,1)]-
    \Pr[\cA(\pp)=1:(\pp,\td)\lrun \FuncGen(1^\secp,2)]
    }\leq  \negl(\secp).
    \end{align}
\item{\bf Adaptive hardcore bit security}
 For all $(\pp,\td)\in \FuncGen(1^\secp, 2)$ and $\lenx=\poly(\secp)$ such that  $\cX \subseteq \bin^w$, the following holds. 
\begin{enumerate}
\item 
There exists an efficient deterministic algorithm 
$\GoodSet$ that takes as input $\td$, $y$, and $d\in \bin^\lenx$ and outputs $\top$ or $\bot$.
We require that for all $(\pp,\td)\in \FuncGen(1^\secp,2)$, $y \in \cY_\pp$,
\[
\Pr_{d\chosen \bin^\lenx }[ \GoodSet(\td, y, d) = \top ] \ge 1- \negl(\secp). 
\]
In the case of $y\not\in \cY_\pp$, $\GoodSet(\td, y, d) $ outputs $\bot$ for any $d\in \bin^\lenx$.

%For all $y\in \cY_\pp$, there exists a set $\Goodset_{\pp,y}\subseteq \cX$ such that $\Pr_{d\chosen \bin^\lenx }[d \notin \Goodset_{\pp,y} ] \le \negl(\secp)$. In addition, there exists a PPT algorithm that checks for membership in $\Goodset_{\pp,y}$ given $\pp,y$, and $\td$.
\item  Let
\begin{align}
 H_{\pp,0}  & \seteq \setbracket{(x_b,y,d) \mid b\in\zo{},(x_0,x_1)=\Invert(\pp,\td,y), \GoodSet(\td, y, d) = \top,~ \langle d, (x_0\xor x_1) \rangle =0 },\\
 H_{\pp,1} & \seteq \setbracket{(x_b,y,d )\mid b\in\zo{}, (x_0,x_1)=\Invert(\pp,\td,y), \GoodSet(\td, y, d) = \top,~ \langle d, (x_0\xor x_1) \rangle =1 },
\end{align}
then for any QPT $\A$,\footnote{Note that the original definition by \cite{FOCS:Mahadev18a} applies a bijective map $J$ to an element of $\cX$ when taking the hardcore bit. We can remove this indirection caused by this map by 
directly defining their $ J(\cX) $ as our $\cX $. This slightly simplifies the notation.
} 
it holds that
\begin{align}
\abs{\Pr_{(\pp,\td)\gets \FuncGen (1^\secp,2)}[\A(\pp)\in H_{\pp,0}] - \Pr_{(\pp,\td)\gets \FuncGen (1^\secp,2)}[\A(\pp)\in H_{\pp,1}]} \le \negl(\secp).
\end{align}
\end{enumerate}
 \end{itemize}
\end{definition}

\begin{lemma}[\cite{FOCS:Mahadev18a}]
If the LWE assumption holds against QPT adversaries, there exists an NTCF generator.
\takashi{May be better to mention parameter for LWE.}
\end{lemma}

The following lemma is a variant of \cite{AC:HMNY21}. 
In our variant, the challenger provides $\A$ with the trapdoors corresponding to the set $\overline{S}$, whereas in \cite{AC:HMNY21}, these trapdoors are not given.
Therefore, the following statement is slightly stronger, but can be proven by the same proof.

\begin{lemma}[Cut-and-Choose Adaptive Hardcore Property, adapted from Lemma 5.6 in \cite{AC:HMNY21}]
Let $\NTCF = (\FuncGen, \allowbreak \StateGen, \Chk, \Invert)$ be an NTCF generator that satisfies the adaptive hardcore bit security. 
Then it satisfies what we call the \emph{cut-and-choose  adaptive hardcore property} defined below.
For a QPT adversary $\cA$ and a positive integer $n$, 
we consider the following  experiment $\expc{\NTCF,\cA}{cut}{and}{choose}(\secp,n)$.

\begin{enumerate}
    \item The challenger chooses a uniform subset $S\subseteq [2n]$ such that $|S|=n$.
    In the following, $\overline{S}$ denotes $[2n]\backslash S$.
    %\footnote{We can also take $S\subseteq [2n]$ such that $|S|=n$, but we do as above just for convenience  in the proof.} 
    \item The challenger generates 
    $(\pp_i,\td_i)\lrun \FuncGen(1^\secp,1)$ for all $i\in S$ and 
     $(\pp_i,\td_i)\lrun \FuncGen(1^\secp,2)$ for all $i\in \overline{S}$ and
    sends $\{\pp_i\}_{i\in[2n]}$ to $\cA$. 
    \item $\cA$ sends 
    $\{y_i,d_i\}_{i\in[2n]}$ to the challenger. 
    \item The challenger
    computes $(x_{i,0},x_{i,1})\lrun \Invert(\pp_i, \td_i, y_i)$ for all $i\in \overline{S}$ and checks if 
     $\GoodSet({\td_i, y_i,d_i}) = \top$ and
    $d_i\cdot (x_{i,0}\oplus x_{i,1})=0$ hold for all $i\in \overline{S}$.
    If they do not hold for some $i\in \overline{S}$, the challenger immediately aborts and the experiment returns $0$. 
    \item  \label{step:reveal_S} 
    The challenger sends $S$ and $\{ \td_i \}_{i\in \overline{S}}$ to $\cA$.
    \item $\cA$ sends $\{x_i\}_{i\in S}$ to the challenger. 
    \item 
    The challenger checks if $\Chk(\pp_i,x_i,y_i)=1$ holds for all $i\in S$.
    If this holds for all $i\in S$, the experiment returns $1$. Otherwise, it returns $0$.
\end{enumerate}
Then for any $n$ such that
$n\leq \poly(\secp)$ and 
$n=\omega(\log \secp)$, it holds that 
\begin{align}
\advc{\NTCF,\cA}{cut}{and}{choose}(\secp,n)\seteq \Pr[ \expc{\NTCF,\cA}{cut}{and}{choose}(\secp,n)=1] \leq \negl(\secp).
\end{align}
\end{lemma}

% !TEX root = main.tex

\section{Watermarkable Cryptographic Functionalities}\label{sec-watermarking}
Here, we introduce and construct several watermarkable cryptographic primitives. 
Specifically, we build watermarkable public-key encryption from plain public-key encryption, 
watermarkable unpredictable functions from one-way functions, 
and watermarkable digital signatures from SIS. 
Our constructions are inspired by the ideas implicit in \cite{kitagawa2025simple}; 
for the latter two primitives, however, we add new twists to handle evaluation and signing oracles in the security proofs.
\subsection{Watermarkable Public Key Encryption}
\label{sec:WPKE}
%\takashi{The following is just taken from the note. We should make it formal.}
\paragraph{Syntax}
A watermarkable public key encryption scheme $ \WPKE=(\KG,\Mark, \Enc,\Dec)$ with the message space $\cM=\{\cM_\lambda\}_\lambda$ and the mark space $\cX=\{\cX_\lambda\}_\lambda$ consists of the following algorithms.
In the following, we omit the subscript for $\cX_\secp$ and $\cM_\secp$ to denote $\cX$ and $\cM$ when it is clear from the context.
\begin{description}
\item[$\KG(1^\secp)\ra(\ek,\msk)$:] The key generation algorithm is a PPT algorithm that takes a security parameter $1^\lambda$, and outputs a public key $\ek$ and a master secret key $\msk$. 

\item[$\Mark(\msk,x)\ra \dk(x)$:] The marking algorithm is a deterministic classical polynomial time algorithm that takes a master secret key $\msk$ and a mark $x\in \cX$, and outputs a marked secret key $\dk(x)$. 

\item[$\Enc(\ek, m )\ra\ct$:] The encryption algorithm is a PPT algorithm that takes a public key $\ek$ and a message $m \in \Ms$, and outputs a ciphertext $\ct$.

\item[$\Dec(\dk(x),\ct)\ra m$:] The decryption algorithm is a deterministic classical 
polynomial-time algorithm that takes a marked decryption key $\dk(x)$ or a master secret key $\msk$ and a ciphertext $\ct$, and outputs a value $\tilde{m}$.

\item[Correctness:]For every $m \in \Ms$ and $x\in \cX$, we have
\begin{align}
\Pr\left[
\begin{array}{ll}
\Dec(\dk(x), \ct) \allowbreak = m
\\
\qquad \quad \land 
\\
\Dec(\msk, \ct) \allowbreak = m
\end{array}
\ \middle |
\begin{array}{ll}
(\ek,\msk)\gets\KG(1^\secp)\\
\dk(x)\gets \Mark(\msk,x)\\ 
\ct\gets\Enc(\ek,m)
\end{array}
\right] 
=1.
\end{align}
\item[One-wayness:] For all QPT adversary $\qA$, we need
\begin{align}
\Pr\left[
\qA(\ek, \ct) \allowbreak = m
\ \middle |
\begin{array}{ll}
(\ek,\msk)\gets\KG(1^\secp)\\
m \lrun \cM \\ 
\ct\gets\Enc(\ek,m)
\end{array}
\right] 
=\negl(\secp).
\end{align}
\item[Parallel mark extractability:]
Here, we introduce the notion of ``parallel mark extractability". 
To do so, we first define the following parallel one-way inversion experiment involving an adversary \( \A \) and a challenger, parameterized by an arbitrary polynomial \( n = \text{poly}(\secp) \), denoted as \( \expb{\WPKE,\qA,n}{par}{ow}(1^\secp) \).
\begin{enumerate}
\item The challenger generates $(\ek_i,\msk_i)\gets \KG(1^\secp)$ for $i\in [n]$ and sends $\{\ek_i \}_{i\in[n]}$ to $\A$.
\item Given $\{\ek_i \}_{i\in[n]}$, the adversary $\qA$ chooses arbitrary $x_1,...,x_n \in \cX$ and sends them to the challenger. 
\item 
The challenger computes $\dk_i(x_i)\gets\Mark(\msk_i,x_i)$ for $i\in [n]$ and sends $\{ \dk_i(x_i) \}_{i\in[n]}$ to $\A$.
\item\label{step:par-ow-3} 
For $i\in [n]$, the challenger generates $m_i\gets \calM$ and  $\ct_i\gets \Enc(\ek_i,m_i)$ and sends $\{\ct_i\}_{i\in [n]}$ to $\A$.
\item\label{step:par-ow-4}  
Given $\{\ct_i\}_{i\in [n]}$, $\A$ outputs $\{m'_i\}_{i\in [n]}$.
\item\label{step:par-ow-5} 
The experiment outputs $1$ (indicating that $\A$ wins) if $m'_i=m_i$ for all $i\in [n]$ and $0$ otherwise. 
\end{enumerate}
We define the advantage of a QPT adversary $\cA$ to be 
\begin{align}
\advb{\WPKE,\qA,n}{par}{ow}(\secp) \seteq \Pr[\expb{\WPKE,\qA,n}{par}{ow}(1^\secp)= 1].
\end{align}
We also define parallel mark extraction experiment defined by a QPT extractor $\qExt$, an adversary \( \A \), and a challenger, parameterized by an arbitrary polynomial \( n = \text{poly}(\secp) \), denoted as $\expb{\WPKE,\qA,\qExt,n}{par}{ext}(1^\secp)$. 
The parallel extraction experiment is the same as parallel inversion experiment except for the last 3 steps. Namely, we replace Step \ref{step:par-ow-3}, \ref{step:par-ow-4}, and \ref{step:par-ow-5} with the following:
\begin{enumerate}[label=\arabic*', start=4]
\item
Take the adversary $\qA$ right before Step \ref{step:par-ow-3} of \( \expb{\WPKE,\qA,n}{par}{ow}(1^\secp) \). 
Since $\qA$ is a quantum non-uniform QPT adversary, we can 
parse $\qA$ as a pair of a unitary circuit $U_\qA$ and a quantum state $\rho_\A$, where to run $\qA$ on a set of ciphertexts $\{\ct_i\}_{i\in [n]}$, 
we apply $U_\qA$ to $\{\ct_i\}_{i\in [n]}$ and $\rho_\qA$ and then measure the corresponding registers to obtain the output $\{m'_i\}_{i\in [n]}$.
\item We run $\qExt(\{\ek_i\}_{i\in [n]}, (U_\qA, \rho_\qA) ) \to \{x'_i\}_{i\in [n]}$.  
%\takashi{A decryptor may be formalized as a quantum circuit with initial state, similarly to Definition 6.2 of \url{https://eprint.iacr.org/2021/946.pdf}.}
\item The experiment outputs $1$ (indicating that the adversary wins) if $x'_i=x_i$ for all $i\in [n]$ and $0$ otherwise. 
\end{enumerate}
We define the advantage of a QPT adversary $\cA$ with respect to $\qExt$ to be 
\begin{align}
\advb{\WPKE,\qA,\qExt,n}{par}{ext}(\secp) \seteq \Pr[\expb{\WPKE,\qA,\qExt,n}{par}{ext}(1^\secp)= 1].
\end{align}
We say that $\WPKE$ satisfies parallel mark extractability if there exists a QPT algorithm $\qExt$ such that for all $n=\poly(\secp)$, for all QPT algorithm $\cA$,
the following holds:
\[
\advb{\WPKE,\qA,\qExt,n}{par}{ext}(\secp) \geq \advb{\WPKE,\qA,n}{par}{ow}(\secp) -\negl(\secp).
\] 
\end{description}

\paragraph{Construction}
The above defined watermarkable PKE can be constructed from any IND-CPA secure PKE scheme $\PKE.(\KG,\Enc,\Dec)$ for single-bit messages.
Let $\ell=\omega(\secp)$ be the bit-length of marks. Then we construct a watermarkable PKE scheme for message length $\ell$ as follows. 
\begin{description}
\item[$\KG(1^\secp)$:] 
On input the security parameter $1^\secp$, the key generation algorithm proceeds as follows.
\begin{itemize}
\item Generate $(\ek_{j,b},\dk_{j,b})\gets \PKE.\KG(1^\secp)$ for $j\in [\ell]$ and $b\in \bit$. 
\item Output $\ek=\{ \ek_{j,b}\}_{j\in[\ell],b\in\bit}$ and $\msk=\{ \dk_{j,b}\}_{j\in[\ell],b\in\bit}$. 
\end{itemize}
\item[$\Mark(\msk,x)$:]
On input $\msk=\{ \dk_{j,b}\}_{j\in[\ell],b\in\bit}$ and $x\in \bit^\ell$, output 
$\dk(x)=(x,\{ \dk_{j,x[j]}\}_{j\in [\ell]})$ where $x[j]$ is the $j$-th bit of $x$.  

\item[$\Enc(\ek,m)$:] 
On input $\ek=\{ \ek_{j,b}\}_{j\in[\ell],b\in\bit}$ and
$m=m_1\|m_2\|\ldots\|m_\ell$, the encryption algorithm proceeds as follows.
\begin{itemize}
\item Generate $\ct_{j,b}\gets \PKE.\Enc(\ek_{j,b},m_j)$ for $j\in [\ell]$ and $b\in \bit$. 
\item Output $\ct=\{ \ct_{j,b}\}_{j\in[\ell],b\in\bit}$. 
\end{itemize}
\item[$\Dec(\dk(x),\ct)\ra\tilde{m}$:]
On input a marked key $\dk(x)=(x,\{ \dk^\prime_j\}_{j\in [\ell]})$ or a master secret key $\msk=\{ \dk_{j,b}\}_{j\in[\ell],b\in\bit}$ 
and $\ct=\{ \ct_{j,b}\}_{j\in[\ell],b\in\bit}$, the decryption algorithm proceeds as follows.
\begin{itemize}
\item If the first input is a master secret key, set $\dk'_j = \dk_{j,0}$ and $x[j]=0$ for all $j\in [\ell]$.
    \item Run $m^\prime_j\gets \PKE.\Dec(\dk^\prime_j,\ct_{j,x[j]})$ for $j\in[\ell]$. 
    \item Output $m^\prime=m^\prime_1\|m^\prime_2\|\ldots\|m^\prime_\ell$
\end{itemize}
\end{description}
Correctness and one-wayness of the above scheme immediately follows from correctness and IND-CPA security of $\PKE$, respectively. Below, we prove parallel mark extractability.

\paragraph{\bf Security proof} 
\begin{theorem}
    The above construction satisfies parallel mark extractability if $\PKE$ satisfies IND-CPA security.
\end{theorem}
\begin{proof}
To prove the theorem, we first define the extraction algorithm $\qExt$ as follows.
\begin{description}
    \item[$\qExt(\{\ek_i\}_{i\in[n]}, (U_\qA, \rho_\qA) ) \to \{x'_i\}_{i\in [n]}$:] 
  On input $\{\ek_i\}_{i\in[n]}$  and the description of the adversary $\qA = (U_\qA, \rho_\qA)$, it proceeds as follows.  
    \begin{itemize}
     \item Parse the encryption key as
        $\ek_i=\{ \ek_{i,j,b}\}_{j\in[\ell],b\in\bit}$ for each $i\in [n]$. 
    \item Choose $m_{i,j,0}\gets \bit$ and set $m_{i,j,1}= 1-m_{i,j,0}$ for $i\in[n]$ and $j\in [\ell]$. 
\item Generate $\ct_{i,j,b}\gets \PKE.\Enc(\ek_{i,j,b},m_{i,j,b})$ for $i\in[n]$, $j\in [\ell]$, and $b\in \bit$. 
\item Set $\ct_i=\{\ct_{i,j,b}\}_{j\in[\ell],b\in\bit}$ for $i\in[n]$. 
        \item Run the adversary $\qA$ on input $\{\ct_i\}_{i\in[n]}$ (i.e., run the unitary $U_\qA$ on input $(\{\ct_i\}_{i\in[n]}, \rho_\qA)$ and measure the output register) to obtain $\{m^\prime_i\}_{i\in [n]}$
        \item 
      Parse $m^\prime_i = m^\prime_i[1] \| \cdots \|m^\prime_i[\ell]$ for $i\in[n]$. 
      \item Let $x^\prime_{i,j}$ be the unique bit such that $m^\prime_i[j]=m_{i,j,x^\prime_{i,j}}$ for $i\in[n]$ and $j\in[\ell]$.  
\item Let $x^\prime_i=x^\prime_{i,1}\|\ldots\|x^\prime_{i,\ell}$ for $i\in[n]$. 
    \item Output $ \{x'_i\}_{i\in [n]}$.
    \end{itemize}
\end{description}
    We then prove the theorem by a sequence of hybrids.
    For the sake of the contradiction, let us consider an adversary $\A$ against the parallel one-way inversion experiment \( \expb{\WPKE,\qA,n}{par}{ow}(1^\secp) \). 
    In the following, the event that $\qA$ wins in $\Hyb_{\rm xx}$ is denoted by $\event_{\rm xx}$.

\begin{description}
      \item[$\Hyb_0$:] This is the parallel predictability game between the challenger and an adversary $\A$.  
      Namely, the game proceeds as follows.
\begin{enumerate}
\item The challenger runs $\KG(1^\secp)$ to obtain $\ek_i$ for $i\in[n]$ as follows. 
\begin{itemize}
    \item Generate $(\ek_{i,j,b},\dk_{i,j,b})\gets \PKE.\KG(1^\secp)$ for $j\in [\ell]$, and $b\in \bit$. 
\item Set $\ek_i=\{ \ek_{i,j,b}\}_{j\in[\ell],b\in\bit}$.
\end{itemize}
It then sends $\{\ek_i\}_{i\in [n]}$ to $\qA$.
\item Given $\{\ek_i\}_{i\in [n]}$, the adversary $\qA$ sends arbitrarily chosen $x_1,...,x_n \in \bit^\ell$ to the challenger. 
\item 
The challenger sets $\dk_i(x_i)=(x_i,\{ \dk_{i,j,x_i[j]}\}_{j\in[\ell]} )$ for $i\in [n]$ and sends $\{ \dk_i(x_i) \}_{i\in[n]}$ to $\A$.
\item\label{step:par-ow-3-const} 
For $i\in [n]$, the challenger generates $m_i\gets \bit^\ell$ and 
$\ct_{i,j,b}\gets \PKE.\Enc(\ek_{i,j,b},m_i[j])$ for $j\in[\ell]$ and $b\in\bit$, sets
$\ct_i=\{ \ct_{i,j,b}\}_{j\in[\ell],b\in\bit}$ and sends $\{\ct_i\}_{i\in [n]}$ to $\A$.
\item
Receiving $\{\ct_i\}_{i\in [n]}$ from the challenger, $\A$ outputs $\{m'_i\}_{i\in [n]}$.
\item \label{step:par-ow-5-const}
The challenger sets the output of the game to be  $1$ (indicating that $\A$ wins) if $m'_i=m_i$ for all $i\in [n]$ and $0$ otherwise. 
\end{enumerate}

By definition, we have $\Pr[\event_0=1]=\advb{\WPKE,\qA,n}{par}{ow}(\secp)$. 
\item[$\Hyb_1$:] In this hybrid, we change the way of generating $\ct_{i,j,1-x_i[j]}$ for $i\in[n]$ and $j\in [\ell]$. 
Specifically, it is generated as 
\[
\ct_{i,j,1-x_i[j]}\gets \PKE.\Enc(\ek_{i,j,1-x_i[j]},1-m_i[j]).
\]
Note that $\ct_{i,j,x_i[j]}$ is generated as $\ct_{i,j,x_i[j]}\gets \PKE.\Enc(\ek_{i,j,x_i[j]},m_i[j])$ similarly to the previous hybrid. 

Since $\dk_{i,j,1-x_i[j]}$ is not given to the adversary for any $i\in[n]$ and $j\in[\ell]$, by the IND-CPA security of $\PKE$, we have 
$\abs{\Pr[\event_0]-\Pr[\event_1]}=\negl(\secp)$.
This can be shown using a hybrid argument over all positions of $i$ and $j$. The reduction algorithm embeds its challenge instance into $\ek_{i,j,1-x_i[j]}$ by randomly guessing $x_i[j]$, where guess is correct with probability of $1/2$.

\item[$\Hyb_2$:]
In this hybrid, we change the way of generating $\ct_{i,j,b}$ for $i\in[n]$, $j\in [\ell]$, and $b\in\bin$, and the winning condition. 
In Step \ref{step:par-ow-3-const}, the challenger chooses $m_{i,j,0}\gets \bit$ and sets $m_{i,j,1}=1-m_{i,j,0}$ for $i\in[n]$ and $j\in[\ell]$. Then the challenger generates 
\[
\ct_{i,j,b}\gets \PKE.\Enc(\ek_{i,j,b},m_{i,j,b})
\]
for $i\in[n]$, $j\in[\ell]$, and $b\in \bit$. 
In Step \ref{step:par-ow-5-const}, the challenger sets the output of the game to be $1$ if $m'_i[j]=m_{i,j,x_i[j]}$ for all $i\in[n]$ and $j\in[\ell]$ and $0$ otherwise.  

This is just a conceptual change, and thus we have $\Pr[\event_1=1]=\Pr[\event_2=1]$. 
It is also easy to see that we have 
$\Pr[\event_2=1]=\advb{\WPKE,\qA,\qExt,n}{par}{ext}(\secp)$.
        \end{description}
Combining the above, we obtain $\advb{\WPKE,\qA,\qExt,n}{par}{ext}(\secp) \geq \advb{\WPKE,\qA,n}{par}{ow}(\secp) -\negl(\secp)$ as desired.
\end{proof}

%The extractor runs a decryptor on a "special ciphertext" where $\ct_{j,0}$ and $\ct_{j,1}$ encrypt different messages. Then, by looking at the decryption result of the decryptor, one can decide each bit of $x$. This works in the quantum and parallel setting since the extractor just runs the decryptor once without any rewinding etc.

\subsection{Watermarkable Unpredictable Functions}\label{watermarkable_UPF}
%We can easily extend the above construction of PKE-SKL to any "watermarkable" functionality. Thus, wehenever we have a suitable watermarkable functionality, we can securely lease (parallel repetition of) the functionality. 
Here, we define and construct watermarkable unpredictable functions (UPF).
%
%We can focus on constructing UPF (unpredictable function) since that can be easily upgraded into PRF by quantum Goldreich-Levin. %(even in the setting of SKL with a classical lessor).
%Thus, we only have to instantiate a watermarkable UPF. 
Its definition is very similar to watermarkable PKE above; if 
an adversary is given $\key(x)$ and generates a (possibly quantum) ``predictor" that can guess the output on a uniformly random input with probability $\epsilon$, then we can extract $x$ from the predictor with probability $\epsilon-\negl(\secpar)$, and moreover the extraction should work even in the parallel setting.  
Based on the idea of \cite{kitagawa2025simple}, we can construct such a watermarkable UPF from two-key equivocal PRFs \cite{C:HJOSW16}, which in turn can be built from OWFs. 
To deal with evaluation queries in the security proof, which were not considered in \cite{kitagawa2025simple}, we employ a simple trick
using random bits. 

%\takashi{I may sketch the full description of UPF-SKL with a classical lessor, though that would be largely a copy-and-paste of the PKE case except that watermarkable PKE with replaced with watermarkable UPF.}

\paragraph{Syntax}
A watermarkable unpredictable functions $ \WPRF=(\KG,\Mark, \Eval )$ with the input space $\bin^{\ell(\secp)}$ and the mark space $\cX=\{\cX_\lambda\}_\lambda$ consists of the following algorithms.
In the following, we omit the security parameter $\secp$ for $\cX_\secp$ and $\ell(\secp)$ to denote $\cX$ and $\ell$ when it is clear from the context.
\begin{description}
\item[$\KG(1^\secp)\ra (\msk,\extk)$:] The key generation algorithm is a PPT algorithm that takes a security parameter $1^\lambda$, and outputs a master secret key $\msk$ and an extraction key $\extk$. 

\item[$\Mark(\msk,x)\ra \key(x)$:] The marking algorithm is a deterministic polynomial time algorithm that takes a master secret key $\msk$ and a mark $x\in \cX$, and outputs a marked key $\key(x)$. 

\item[$\Eval(\key(x),s)\ra t$:] The evaluation algorithm is a PPT algorithm that takes a 
master secret key $\msk$ or
marked key $\key(x)$ and an input $s\in \bin^\ell$, and outputs $t$. 

\item[Correctness:]For every $s\in \bin^\ell$, we require that
\begin{align}
\Pr\left[
\begin{array}{ll}
\Eval(\key(x), s) = \Eval(\msk, s) 
\\ \mbox{holds for all~} x\in \cX
\end{array}
\ \middle |
\begin{array}{ll}
(\msk,\extk) \gets\KG(1^\secp)\\
\key(x)= \Mark(\msk,x)~ \mbox{for all } x\in \cX
\end{array}
\right] 
=1-\negl(\secp).
\end{align}
One can consider a relaxed correctness requirement, where the equation $\Eval(\key(x), s) = \Eval(\msk, s)$ only needs to hold for a fixed $x$, as opposed to all $x$ simultaneously as above. Nonetheless, our construction fulfills the stronger requirement.

\item[Correctness against adversarially chosen mark and input:]
We also consider a strengthened version of the correctness. Here, we require the correctness to hold, even if the adversary chooses the mark and the input to the function adversarially. To define the notion, we define the correctness experiment involving an adversary $\qA$, denoted as \( \expa{\WPRF,\qA}{corre}(1^\secp) \). 
\begin{enumerate}
\item The challenger runs $\KG(1^\secp)\ra (\msk,\extk)$.
\item The adversary $\qA$ is given input $1^\secp$ and starts to run.
Throughout the game, it has an oracle access to an oracle $\Eval(\msk, \cdot) $, which takes as input $\hats\in \bin^\ell$ and returns $ \Eval(\msk, \hats)$.
\item At some point, $\qA$ chooses arbitrary $x\in \cX$. 
\item 
Then, the challenger runs $\Mark(\msk,x)\rrun \key(x)$ and returns $\key(x)$ to $\qA$. 

\item 
The adversary continues to have access to the oracle. 
Finally, it outputs an input $s$. The output of the experiment is then set to be $1$ if $\Eval(\msk,s) \neq \Eval(\key(x),s)$.
\end{enumerate}
We say that $\WPRF$ has correctness against adversarially chosen input if for all QPT adversary $\qA$, we have 
\begin{align}
\expa{\WPRF,\qA}{corre}(1^\secp) \leq \negl(\secp).
\end{align}

\item[Parallel mark extractability:]
Here, we introduce the notion of ``parallel mark extractability", which is similar to that for watermarkable PKE. 
We first define the following parallel prediction experiment involving an adversary \( \A \) and a challenger, parameterized by an arbitrary polynomial \( n = \text{poly}(\secp) \), denoted as \( \expb{\WPRF,\qA,n}{par}{pre}(1^\secp) \).
%
% We require the following "parallel extractability", which requires the existence of a PPT extractor $\Ext$ satisfying the following:
% Consider the following game where $n$ is an arbitrary polynomial:
\begin{enumerate}
\item The challenger runs $\KG(1^\secp)\ra (\msk_i,\extk_i)$ for $i\in [n]$. 
\item The adversary $\qA$ is given input $1^\secp$ and $1^n$ and starts to run.
The adversary has an oracle access to the function $\Eval(\msk_i, \cdot)$ for each $i\in [n]$, 
which on input $\hat{s} \in \bin^\ell$ returns $\Eval(\msk_i, \hat{s})$.
The adversary can make a query on any input at any timing arbitrarily many times, until $\qA$ receives the challenge at Step \ref{Step:after-the-challenge-wprf}. 
\item 
At some point, $\qA$ chooses arbitrary $x_1,...,x_n \in \cX$ and sends them to the challenger.  
\item For $i\in [n]$, the challenger runs $\key_i(x_i)\gets\Mark(\msk_i,x_i)$ and sends $\{ \key_i(x_i)\}_{i\in[n]}$ to the adversary. 
%\takashi{We may additionally give an evaluation oracle to the adversary. I believe our construction also satisfies the security in this setting. This will be inherited to the eventual construction of UPF-SKL with a classical lessor. }
\item \label{Step:WPRF-declare-ready}
At some point, $\qA$ declares that it is ready for the challenge. 
\item \label{Step:wprf-challenge-is-sent}
The challenger picks uniformly random $s_i \lrun \bin^\ell$ for $i\in [n]$ and sends $\{s_i\}_{i\in [n]}$ to $\qA$. 
\item \label{Step:after-the-challenge-wprf} 
$\qA$ receives $\{s_i\}_{i\in [n]}$ from the challenger and loses access to the evaluation oracles. 
It then outputs $\{t'_i\}_{i\in [n]}$.

\item \label{Step:WPRF-decision}
The challenger sets the output of the experiment to be $1$ (indicating that $\qA$ wins) if $t'_i=\Eval(\msk_i,s_i)$ for all $i\in [n]$ and otherwise outputs $0$. 
\end{enumerate}
\jiahui{Should we emphasize that this is a non-uniform QPT adv?}
\shota{Added a remark below.}
We define the advantage of a QPT adversary $\cA$ to be 
\begin{align}
\advb{\WPRF,\qA,n}{par}{pre}(\secp) \seteq \Pr[\expb{\WPRF,\qA,n}{par}{pre}(1^\secp)= 1].
\end{align}
We also define parallel mark extraction experiment defined by a QPT extractor $\qExt$, an adversary \( \A \), and a challenger, parameterized by an arbitrary polynomial \( n = \text{poly}(\secp) \), denoted as $\expb{\WPRF,\qA,\qExt,n}{par}{ext}(1^\secp)$. 
The parallel extraction experiment is the same as parallel prediction experiment except for the last 3 steps. Namely, we replace Step \ref{Step:wprf-challenge-is-sent}, \ref{Step:after-the-challenge-wprf}, and \ref{Step:WPRF-decision} with the following:
\begin{enumerate}
    \item[6'] Take the adversary $\qA$ right before Step \ref{Step:wprf-challenge-is-sent} of \( \expb{\WPRF,\qA,n}{par}{pre}(1^\secp) \). 
    \synote{The following emphasizes about being non-uniform QPT:}
    Since $\qA$ is a quantum non-uniform QPT adversary, we can
parse $\qA$ as a pair of a unitary circuit $U_\qA$ and a quantum state $\rho_\A$, where to run $\qA$ on the inputs $\{s_i\}_{i\in [n]}$, 
we apply $U_\qA$ to $\{s_i\}_{i\in [n]}$ and $\rho_\qA$ and then measure the corresponding registers to obtain the output $\{t'_i\}_{i\in [n]}$.
\item[7'] We run $\qExt(\{\extk_i\}_{i\in [n]}, (U_\qA, \rho_\qA) ) \to \{x'_i\}_{i\in [n]}$.  
%\takashi{A decryptor may be formalized as a quantum circuit with initial state, similarly to Definition 6.2 of \url{https://eprint.iacr.org/2021/946.pdf}.}
\item[8'] The experiment outputs $1$ (indicating that the adversary wins) if $x'_i=x_i$ for all $i\in [n]$ and $0$ otherwise. 
\end{enumerate}

We define the advantage of a QPT adversary $\cA$ with respect to $\qExt$ to be 
\begin{align}
\advb{\WPRF,\qA,\qExt,n}{par}{ext}(\secp) \seteq \Pr[\expb{\WPRF,\qA,\qExt,n}{par}{ext}(1^\secp)= 1].
\end{align}
We say that $\WPRF$ satisfies parallel mark extractability if there exists a QPT algorithm $\qExt$ such that for all $n=\poly(\secp)$, for all QPT algorithm $\cA$,
the following holds:
\[
\advb{\WPRF,\qA,\qExt,n}{par}{ext}(\secp) \geq \advb{\WPRF,\qA,n}{par}{pre}(\secp) -\negl(\secp).
\] 

% For any QPT adversary, we have
% \[
% \Pr\left[
% \forall~i\in[n],~ \Ext(\extk_i,P_i)=x_i
% \right]\ge
% ~\Pr[\text{The~adversary~wins~the~above~experiment}]-\negl(\secp)
% \]
% where, in the probability in the LHS, $\extk_i,x_i,P_i$ are generated as in the above experiment.  
% \if0
% With an overwhelming probability over the execution of the above experiment, the following holds: 
% If all the predictors are simultaneously good, i.e.,
% \[
% \Pr_{s_1,...,s_n}\left[
% \forall~i\in[n], P_i(s_i)=\Eval(\msk_i,s_i)
% \right]\ge \epsilon
% \]
% then we can extract all the marks simultaneously, i.e., 
% \[
% \Pr\left[
% \forall~i\in[n],~ \Ext(P_i)=x_i
% \right]\ge \epsilon-\negl(\secp).
% \]
% \fi
\end{description}
\begin{remark}[Public extraction vs. secret extraction]
%\takashi{The difference from watermarkable PKE is that we have an "extraction key" that should be hidden from the adversary.This is often referred to as "secret extraction" in the watermarking literatures. For the case of PKE, secret extraction would suffice for our purpose, but we formalized it with the public extraction since our construction satisfies it.  }
In our definition, unlike in watermarkable PKE in \cref{sec:WPKE}, the extraction key $\extk$ must be kept hidden from the adversary.
This is often referred to as ``secret extraction" in the watermarking literature and suffices for our purpose. 
While secret extraction would suffice also for PKE, we chose to formalize it with the public extraction since our construction satisfies it.   
\end{remark}
\begin{remark}[Security as a plain UPF]\label{rem:security-as-plain-UPF}
We note that the above definition of a watermarkable UPF does not necessarily imply the standard notion of unpredictability as a plain UPF. 
Nevertheless, unpredictability can be added by a simple conversion: run a plain UPF in parallel and append its output alongside the original one. 
The marked key of the resulting scheme consists of that of the original one along with the secret key of the plain UPF.
\end{remark}
We also introduce the notion of equivocality. This property will be used for the construction of watermarkable digital signatures using watermarkable UPF.
\begin{description}
\item[Equivocality:]
We define the equivocality experiment involving an adversary $\qA$, an equivocation simulation algorithm $\Sim$, and a challenger, denoted as \( \expa{\WPRF,\qA}{equiv}(1^\secp, \coin) \). 
\begin{enumerate}
\item The challenger runs $\KG(1^\secp)\ra (\msk,\extk)$ if $\coin=0$ and $\Sim(1^\secp) \to \key^*$ if $\coin =1$.
\item The adversary $\qA$ is given input $1^\secp$ and starts to run.
Throughout the game, it has an oracle access to an oracle $\calO(\cdot)$, which takes as input $s\in \bin^\ell$ and works as follows:
\[
\calO(s) =
\begin{cases}
  \Eval(\msk, s) & \text{if } \coin =0  \\
  \Eval(\key^*, s), & \text{if } \coin = 1 \\
\end{cases}
\]
\item At some point, $\qA$ chooses arbitrary $x\in \cX$. 
\item 
If $\coin =0$, the challenger runs $\Mark(\msk,x)\rrun \key(x)$ and returns $\key(x)$ to $\qA$. 
If $\coin =1$, the challenger returns $\key^*$ to $\qA$. 

\item 
The adversary continues to have access to the oracle. 
Finally, it outputs its guess $\coin'$ for $\coin$. The output of the experiment is then set to be $\coin'$.
\end{enumerate}
We say that $\WPRF$ has equivocality if there exists a classical PPT algorithm $\Sim$ such that for all QPT adversary $\qA$, we have 
\begin{align}
\adva{\WPRF,\qA}{equiv}(\secp) \seteq \abs{\Pr[\expa{\WPRF,\qA}{equiv} (1^\secp,0) = 1] - \Pr[\expa{\WPRF,\qA}{equiv} (1^\secp,1) = 1]  }\leq \negl(\secp).
\end{align}
\end{description}

\begin{remark}\label{rem:equivocality-implies-correctness}
    We can see that the equivocality implies the correctness against adversarially chosen mark and input, since if the adversary is able to find $x$ and $s$ such that $\Eval(\msk,s) \neq \Eval(\key(x),s)$ with non-negligible probability, these can be used for distinguishing the case of $\coin =0$ and $\coin=1$ in the equivocality experiment.
    On the other hand, we do not know whether the correctness against adversarially chosen mark and input implies the correctness,
    since the latter requires the equation to hold for all $x\in \cX$ simultaneously, which is somewhat strong condition.
\end{remark}

\paragraph{Construction}
The above defined watermarkable UPF can be constructed from any two-key equivocal PRF $\TEPRF=\TEPRF.(\KG,\Eval)$, which in turn can be constructed from any OWF. (See \Cref{def:TEPRF} for the definition of two-key equivocal PRF.)
Let $\ell$ be the bit-length of the inputs for $\TEPRF$.
We assume $\ell =\omega(\log \secp)$.
The bit-length of the marks for the construction can be set to an arbitrary polynomial $\lenx = \lenx(\secp)$. 
The bit-length of the input of the construction then becomes $\ell\lenx$.
\begin{description}
\item[$\KG(1^\secp)\ra(\msk,\extk)$:] 
On input the security parameter $1^\secp$, the key generation algorithm proceeds as follows. 
\begin{itemize}
    \item Sample $s^*_1,s^*_2,...,s^*_\lenx \lrun \bin^\ell$ 
    and $c \lrun \bin^\lenx$.
    \item Run $\TEPRF.\KG(1^\secp, s^*_j) \rrun (\TEPRF.\key_{j,0}, \TEPRF.\key_{j,1})$ for $j\in [\lenx]$.
    \item Output $\msk = (c,\{ \TEPRF.\key_{j,b}\}_{j\in[\lenx],b\in\bit})$ 
    and $\extk=(c, \{ s^*_j\}_{j\in [\lenx ]}, \{ \TEPRF.\key_{j,b}\}_{j\in[\lenx],b\in\bit} )$.
\end{itemize}
% By running the key generation algorithm of the two-key equivocal PRF $\ell$ times with uniformly random strings $s^*_1,s^*_2,...,s^*_\ell$ as input, generate $(\sk_{j,0},\sk_{j,1})$ for $j\in [\ell]$. 
% Output a master secret key $\msk=(\sk_{j,b})_{j\in[\ell],b\in\bit}$
% and an extraction key $\extk=(s^*_i)_{i\in [\ell]}$.

\item[$\Mark(\msk,x)\ra \key(x)$:] 
Upon input the master secret key $\msk = (c,\{ \TEPRF.\key_{j,b}\}_{j\in[\lenx],b\in\bit})$ and a mark $x\in \bit^\lenx$, it outputs 
$\key(x)=\{\TEPRF.\key_{j,x[j] \oplus c[j] }\}_{j\in [\lenx]}$ where $x[j]$ and $c[j]$ are the $j$-th bit of $x$ and $c$ respectively.  

\item[$\Eval(\key(x),s)\ra t$:] 
On input the marked key $\key(x)$ or the master secret key $\msk$, the evaluation algorithm proceeds as follows. 
\begin{itemize}
    \item If the first input is a master secret key, parse it as $\msk =(c,\{ \TEPRF.\key_{j,b}\}_{j\in[\lenx],b\in\bit})$. Then, set $\TEPRF.\key'_{j} = \TEPRF.\key_{j,0}$ for all $j\in [\lenx]$.
    \item If the first input is a marked key $\key(x)$, parse it as $\key(x)=\{\TEPRF.\key'_{j}\}_{j\in [w]}$.
    \item Parse the input as $s=s_1\|s_2\|\ldots\|s_\lenx$. 
    \item For $j\in [\lenx]$, compute $t_j\gets \TEPRF.\Eval(\TEPRF.\key'_j,s_j)$, where each block is of length $\ell$. Then, output $t=t_1\|t_2\|\ldots\|t_\lenx$. 
\end{itemize}
\end{description}
% The extractor runs a predictor on $s^*_j$. Then, by looking at the output of the predictor, one can decide each bit of $x$.

\paragraph{Correctness}
% We omit the proof of correctness since it is implied by the equivocality as we note in \cref{rem:equivocality-implies-correctness}. We will prove equivocality of the above construction in \cref{th:equivocality-of-WPRF}. 
For an input $s=s_1\| \cdots \|s_\lenx$, unless $s_i=s^*_i$ for some $i\in [\lenx]$, it holds that $\Eval(\key(x),s)=\Eval(\msk, s)$ for all $x$. By the union bound, the event $s_i=s^*_i$ for some $i$ occurs with probability at most $w 2^{-\ell}$, which is negligible if $\ell =\omega(\log \secp)$.

\paragraph{Security Proof}
The following theorem asserts the parallel mark extractability of the above construction.
\begin{theorem}
    The above construction satisfies parallel mark extractability if $\ell(\secp)=\omega(\log \secp)$ and $\TEPRF$ satisfies the different values on target input property and the differing point hiding property.
\end{theorem}
Before presenting the proof, we briefly outline the intuition.
Our construction builds on \cite{kitagawa2025simple}, with the main difference that we introduce ``equivocating bits'' $c$ to simulate evaluation oracles. 
The core idea for simulating the evaluation oracle is to evaluate the function using the marked key instead of the master secret key, where the latter is not available to the reduction. 
This approach seems to work, since the adversary cannot issue evaluation queries on points that would distinguish the simulated oracle from the real one by the differing point hiding property of the underlying $\TEPRF$. 
A difficulty arises, however, because the marks are determined by the adversary through its evaluation queries. 
If we follow the naive strategy, this makes it impossible to simulate the evaluation oracle before the marks are fixed. 
To overcome this, we employ the ``equivocating bits'' $c$. 
They allow us to use $\{ \TEPRF.\key_{i,j, c^*_i[j]} \}_{j \in [\lenx]}$ for random $c^*$ as if it were a marked key, even before the actual mark $x$ is chosen. 
Once $x$ is fixed, we can equivocate by interpreting $c^*$ as $c \oplus x$.

\begin{proof}
To prove the theorem, we first define the extraction algorithm $\qExt$ as follows.
\begin{description}
    \item[$\qExt(\{\extk_i\}_{i\in [n]}, (U_\qA, \rho_\qA) ) \to \{x'_i\}_{i\in [n]}$:] 
    Given the extraction keys $\{ \extk_i \}_{i\in [n]}$ and the description of the adversary $\qA = (U_\qA, \rho_\qA)$, it proceeds as follows.  
    \begin{itemize}
        \item Parse the extraction key as
        $\extk_i=(c_i, s^*_i= \{ s^*_{i,j}\}_{j\in [\lenx ]}, \{ \TEPRF.\key_{i,j,b}\}_{j\in[\lenx],b\in\bit} )$ for each $i\in [n]$,
        where $c_i \in \bin^{\lenx}$, $s^*_{i,j}\in \bin^\ell$.
        \item Run the adversary $\qA$ on input $(s^*_1,\ldots, s^*_n)$ (i.e., run the unitary $U_\qA$ on input $((s^*_1,\ldots, s^*_n), \rho_\qA)$ and measure the output register) to obtain $t'_1,\ldots, t'_n$. 
        \item Define $x'_i = x'_i[1] \| \cdots \|x'_i[\lenx]$ as 
    \begin{align}
    x'_i[j] = d_i[j]\oplus c_i[j],
    \quad \mbox{where} \quad 
    d_i[j] =
    \begin{cases}
      0 & \text{if } t'_{i,j}= 
      \TEPRF.\Eval(\TEPRF.\key_{i,j, 0 },s^*_{i,j}) \\
      1 & \text{otherwise }. 
    \end{cases} 
\end{align}
    \item Output $ \{x'_i\}_{i\in [n]}$.
    \end{itemize}
\end{description}
    We then prove the theorem by a sequence of hybrids.
    For the sake of the contradiction, let us consider an adversary $\A$ against the parallel prediction game \( \expb{\WPRF,\qA,n}{par}{pre}(1^\secp) \). 
    In the following, the event that $\qA$ wins in $\Hyb_{\rm xx}$ is denoted by $\event_{\rm xx}$.
\begin{description}
      \item[$\Hyb_0$:] This is the parallel prediction game between the challenger and an adversary $\A$.  
      Namely, the game proceeds as follows.
    \begin{enumerate}
        \item The challenger runs $\KG(1^\secp)$ to obtain $(\msk_i,\extk_i)$
        for $i\in [n]$ as follows. 
\begin{itemize}\label{Step1:WPRF}
    \item Sample $s^*_{i,1},s^*_{i,2},...,s^*_{i,\lenx} \lrun \bin^\ell$ 
    and $c_i \lrun \bin^\lenx$.
    \item Run $\TEPRF.\KG(1^\secp, s^*_{i,j}) \rrun (\TEPRF.\key_{i,j,0}, \TEPRF.\key_{i,j,1})$ for $j\in [\lenx]$.
    \item Set $\msk_i = (c_i,\{ \TEPRF.\key_{i,j,b}\}_{j\in[\lenx],b\in\bit})$ 
    and $\extk_i=(c_i, \{ s^*_{i,j}\}_{j\in [\lenx ]}, \allowbreak \{ \TEPRF.\key_{i,j,b}\}_{j\in[\lenx],b\in\bit} )$.
\end{itemize}
    \item The adversary $\qA$ is given input $1^\secp$ and $1^n$ and starts to run.
    Throughout the game, the adversary has an oracle access to the function $\Eval(\msk_i, \cdot)$, until it receives the challenge inputs.
    Recall that given an input $\hat{s}=\hat{s}_1\|\hat{s}_2\|\ldots\|\hat{s}_\lenx$, $\Eval(\msk_i, \hat{s})$ computes 
    \begin{align}\label{eq:PRF-eval-oracle-honest}
            \hat{t}_j = \TEPRF.\Eval(\TEPRF.\key_{i,j,0},\hat{s}_j)
    \end{align}  
    for $j\in [\lenx]$ and returns $\hat{t}=\hat{t}_1\|\hat{t}_2\|\ldots\|\hat{t}_\lenx$. 
    
    \item 
    At some point, $\qA$ chooses arbitrary $x_1,...,x_n \in \cX$ and sends them to the challenger.  
    \item \label{Step:WPRF-Step-4}
    For $i\in [n]$, the challenger sets $\key_i = \{ \TEPRF.\key_{i,j, x_i[j]\oplus c_i[j] } \}_{j\in [\lenx]}$ and sends $\{ \key_i\}_{i\in[n]}$ to the adversary. 
    %\takashi{We may additionally give an evaluation oracle to the adversary. I believe our construction also satisfies the security in this setting. This will be inherited to the eventual construction of UPF-SKL with a classical lessor. }
    \item 
    At some point, the adversary declares that it is ready for the challenge. 
    \item \label{Step:WPRF-Step-6}
    The challenger picks uniformly random $s_i=s_{i,1}\|\cdots \| s_{i,\lenx} \lrun \bin^{\ell\lenx }$ for $i\in [n]$ and sends $\{s_i\}_{i\in [n]}$ to the adversary. 
    \item 
    Receiving $\{s_i\}_{i\in [n]}$ from the challenger, $\qA$
     outputs $\{t'_i\}_{i\in [n]}$.
     Note that after receiving $\{s_i\}_{i\in [n]}$ from the challenger, $\qA$ loses access to the oracles.
    \item \label{Step:WPRF-Step-8}
    The challenger sets the output of the game to be $1$ (indicating that the adversary wins) if $t'_i= t_i  $ for all $i\in [n]$ and otherwise outputs $0$, where we define $t_i = \Eval(\msk_i,s_i) $.
    By the definition of $\Eval(\msk_i,s_i) $, 
    \[
    t_i  = t_{i,1}\|t_{i,2}\|\ldots\|t_{i,\lenx} 
    \quad \mbox{where} \quad 
    t_{i,j} = \TEPRF.\Eval(\TEPRF.\key_{i,j,0 },s_{i,j}).
    \]

\end{enumerate}

\item[$\Hyb_1$:]
In this hybrid, we change the game so that the challenger does not choose $c_i$ 
at Step \ref{Step1:WPRF} any more. 
Therefore, $\msk_i$ and $\extk_i$ are undefined.
Instead, $c^*_i \lrun \bin^\lenx$ is chosen in Step \ref{Step1:WPRF} and $\key_i$ is defined as
$\key_i = \{ \TEPRF.\key_{i,j, c^*_i[j] } \}_{j\in [\lenx]}$ in 
Step~\ref{Step:WPRF-Step-4}. 
We can see that $\hyb_1$ is equivalent to $\hyb_0$ if $c^*_i$ is replaced with $c^*_i = c_i \oplus x_i$. 
This is a conceptual change and thus we have $\Pr[\event_1]=\Pr[\event_0]$.

\item[$\Hyb_2$:]
In this hybrid, we change the winning condition of the adversary. Namely, in Step \ref{Step:WPRF-Step-8}, the challenger outputs $1$ if 
$t'_i = t_i$ for all $i\in [n]$, where we now define $t_i$ as $ t_i= \Eval(\key_i,s_i)$. 
By the definition of $\Eval(\key_i,s_i)$, we have 
\[
    t_i  = t_{i,1}\|t_{i,2}\|\ldots\|t_{i,\lenx} 
    \quad \mbox{where} \quad 
    t_{i,j} = \TEPRF.\Eval(\TEPRF.\key_{i,j, c^*_i[j] },s_{i,j}).
\]
By the equality of $\TEPRF$, the value of $t_{i,j}$ is different from the one in the previous hybrid only when $ c^*_{i}[j] =1$ and $s_{i,j}=s^*_{i,j}$. 
However, the latter occurs only with negligible probability since $s_{i,j}$ is chosen uniformly at random, independently from anything else. 
We therefore have $|\Pr[\event_{1}] - \Pr[\event_{2}]| = \negl(\secp)$.
\item[$\Hyb_3$:]
In this hybrid, we change the way the evaluation oracles answer the queries.
Namely, in this hybrid, the $i$-th evaluation oracle on input $\hat{s} = \hat{s}_{1}\|\hat{s}_{2}\|\ldots\|\hat{s}_{\lenx} $ computes $\hat{t}_j$ for $j\in [\lenx]$ as 
\[\hat{t}_j = \TEPRF.\Eval(\TEPRF.\key_{i,j, c^*_{i}[j] },\hat{s}_j)\] 
instead of as \cref{eq:PRF-eval-oracle-honest}, and returns
$\hat{t}=\hat{t}_1\|\hat{t}_2\|\ldots\|\hat{t}_\lenx$ to $\qA$.

We claim that this hybrid is indistinguishable from the previous hybrid.
To see this, let us observe that $\hyb_2$ and $\hyb_3$ differ only when $\qA$ makes an evaluation query for $\hat{s}$ with $\hat{s}_{j}=s^*_{i,j}$ to the $i$-th evaluation oracle for some $i\in [n]$ and $j\in [\lenx]$.
However, we can show that such an event occurs only with negligible probability by a straightforward reduction to the hard-to-find differing-point property of $\TEPRF$, by noting that we only need $ \TEPRF.\key_{i,j, c^*_{i}[j] }$ for each $i,j$ for simulating $\hyb_3$ (i.e., $ \TEPRF.\key_{i,j, 1-c^*_{i}[j] }$ is not needed).
%and it is not necessary to simulate the challenge input for $\qA$, which the simulator does not know, to bound the probability.  
We therefore have $|\Pr[\event_{2}] - \Pr[\event_{3}]| = \negl(\secp)$.

\item[$\Hyb_4$:]
In this hybrid, we set the challenge inputs to be $s^*_i = s^*_{i,1}\| s^*_{i,2}\|\cdots\|s^*_{i,\lenx} $ for $i\in [n]$ at Step~\ref{Step:WPRF-Step-6}.
Accordingly, to check the winning condition in Step~\ref{Step:WPRF-Step-8}, $t_i$ is set as $ t_i= \Eval(\key_i,s^*_i)$ for $i\in [n]$. 
Note that $s_{i,1},s_{i,2},...,s_{i,\lenx} $ are no longer necessary to run the game.

We claim that this hybrid is indistinguishable from the previous hybrid.
This can be shown by a reduction to the differing-point-hiding property of $\TEPRF$, where we replace the challenge input for each position of $i,j$ one by one.
The reduction starts by choosing random $(s_{i,j}^0,s^1_{i,j}, c^*_{i}[j])$ and submitting this to its challenger and is given $\TEPRF.\key_{i,j,c^*_{i}[j]}$, where the differing point is either $s_{i,j}^0$ or $s^1_{i,j}$.
The reduction uses $s_{i,j}^1$ as a challenge input.
In the case $\TEPRF.\key_{i,j,c^*_{i}[j]}$ is generated using $s_{i,j}^0$, the challenge input for the position is simulated as $\hyb_3$ (i.e., the differing point and challenge input are independent for the position), while $\hyb_4$ (i.e., the differing point and challenge input are the same) otherwise.
The rest of the reduction is straightforward, since $\Hyb_3$ and $\Hyb_4$ can be run only using $\TEPRF.\key_{i,j,c^*_{i}[j]}$ for each $i,j$. 
We therefore have $|\Pr[\event_{3}] - \Pr[\event_{4}]| = \negl(\secp)$.

\item[$\Hyb_5$:]
In this hybrid, we switch the way the evaluation oracles are answered to the original one. Namely, on input $\hat{s}$, the $i$-th evaluation oracle computes $\hat{t}_j$ as \cref{eq:PRF-eval-oracle-honest} again.

Similarly to the case of $\hyb_2$ and $\hyb_3$, we can show that this hybrid is indistinguishable from the previous hybrid by the hard-to-find differing-point property of $\TEPRF$.  
Here again, we have to bound the probability that $\qA$ makes an evaluation query for $\hat{s}$ with $\hat{s}_{j}=s^*_{i,j}$ for some $j\in [\lenx]$. One additional subtlety that arises here is that the reduction algorithm does not know $s^*_{i,j}$ and thus it cannot simulate the challenge input for the adversary.
However, this is not a problem because it is sufficient for the reduction to simulate the game right before the challenge input is given to $\qA$, since the evaluation queries are made only before that. 
We therefore have $|\Pr[\event_{4}] - \Pr[\event_{5}]| = \negl(\secp)$.

\item[$\Hyb_6$:]
In this hybrid, we undo the changes introduced in $\hyb_1$.
Namely, the challenger chooses $c_i$ for $i\in [n]$ in Step~\ref{Step1:WPRF} again. The challenger then sets $\key_i = \{ \TEPRF.\key_{i,j, x_i[j]\oplus c_i[j] } \}_{j\in [\lenx]}$ for $i\in [n]$ in Step~\ref{Step:WPRF-Step-4} and uses it for checking the winning condition in Step~\ref{Step:WPRF-Step-8}.
We can easily see that this is just a conceptual change and thus we have 
$\Pr[\event_6]=\Pr[\event_5]$.

\item[$\Hyb_7$:]
In this hybrid, we further change the winning condition in Step~\ref{Step:WPRF-Step-8} again.
Namely, from the output $\{t'_i = t'_{i,1}\| \cdots \|t'_{i,\lenx}  \}_{i\in [n]}$ of $\qA$, it computes 
 $\{ x'_i = x'_i[1] \| \cdots \|x'_i[\lenx] \}_{i\in [n]}$ as follows:
    \begin{align}
    x'_i[j] = d_i[j]\oplus c_i[j],
    \quad \mbox{where} \quad 
    d_i[j] =
    \begin{cases}
      0 & \text{if } t'_{i,j}= 
      \TEPRF.\Eval(\TEPRF.\key_{i,j, 0 },s^*_{i,j}) \\
      1 & \text{otherwise }. 
    \end{cases} 
\end{align}
Then, the challenger sets the output of the game to be $1$ if $x'_i = x_i$
for all $i\in [n]$.

We claim that this change does not reduce the winning probability of $\qA$.
To see this, we recall that when the adversary wins the game in the previous hybrid, $t'_{i,j}=  \TEPRF.\Eval(\TEPRF.\key_{i,j, x_i[j]\oplus c_i[j] },s^*_{i,j})$ holds for all $i$ and $j$. 
In such a case, we have $x'_i=x_i$ for all $i$, since we have $d_{i}[j] = x_i[j]\oplus c_i[j]$ by the different values on target property of $\TEPRF$. 
We therefore have $\Pr[\event_7] \geq \Pr[\event_6]$. 
\end{description}
From the above discussion, we can see that $\Pr[\event_7] \geq \Pr[\event_0] -\negl(\secp)$ holds. 
Furthermore, one can see by inspection that $\hyb_7$ corresponds to the parallel extraction experiment $\advb{\WPRF,\qA,\qExt,n}{par}{ext}(\secp)$ for $\qA$ with respect to $\qExt$ we defined.
The theorem therefore follows.
\end{proof}
We also prove the equivocality of the construction. 
\begin{theorem}\label{th:equivocality-of-WPRF}
    The above construction satisfies equivocality if $\TEPRF$ satisfies hard-to-find differing-point property.
\end{theorem}
\begin{proof}
To prove the theorem, we first define the equivocation simulator $\Sim$ as follows.
\begin{description}
    \item[$\Sim(1^\secp) \to \key^*$:] 
    Given the security parameter $1^\secp$, it proceeds as follows.  
    \begin{itemize} 
        \item Sample $s^*_1,s^*_2,\ldots,s^*_{\lenx} \lrun \bin^\ell$ and $c^* \lrun \bin^\lenx$. 
        \item Run $\TEPRF.\KG(1^\secp, s^*_j) \rrun (\TEPRF.\key_{j,0}, \TEPRF.\key_{j,1})$ for $j\in [\lenx]$. 
        \item Set $\key^* = \{ \TEPRF.\key_{j,c^*[j]}\}_{j\in[\lenx]}$. 
    \end{itemize}
\end{description}

    We then prove the theorem by a sequence of hybrids.
    For the sake of the contradiction, let us consider an adversary $\A$ against the equivocality game \( \expa{\WPRF,\qA}{equiv}(1^\secp) \). 
    In the following, the event that $\qA$ outputs $1$ in $\Hyb_{\rm xx}$ is denoted by $\event_{\rm xx}$.
\begin{description}
    \item[$\Hyb_0$:] This is the equivocality game between the challenger and an adversary $\A$ with $\coin =0$.  
      Namely, the game proceeds as follows.
    \begin{enumerate}
        \item \label{Step:WPRF-equiv-step1}
        The challenger runs $\KG(1^\secp)$ to obtain $(\msk,\extk)$ as follows. \begin{itemize} 
        \item Sample $s^*_1,s^*_2,\ldots,s^*_{\lenx} \lrun \bin^\ell$ and $c \lrun \bin^\lenx$. 
        \item Run $\TEPRF.\KG(1^\secp, s^*_j) \rrun (\TEPRF.\key_{j,0}, \TEPRF.\key_{j,1})$ for $j\in [\lenx]$. 
        \item Set $\msk = (c,\{ \TEPRF.\key_{j,b}\}_{j\in[\lenx],b\in\bit})$ and $\extk=(c, \{ s^*_j\}_{j\in [\lenx ]} )$. 
        \end{itemize}

    \item The adversary $\qA$ is given input $1^\secp$ and starts to run.
    Throughout the game, the adversary has an oracle access to the function $\Eval(\msk, \cdot)$.
    Recall that given an input $\hat{s}=\hat{s}_1\|\hat{s}_2\|\ldots\|\hat{s}_\lenx$, $\Eval(\msk, \hat{s})$ computes 
    \begin{align} %\label{eq:PRF-eval-oracle-honest}
            \hat{t}_j = \TEPRF.\Eval(\TEPRF.\key_{j,0},\hat{s}_j)
    \end{align}  
    for $j\in [\lenx]$ and returns $\hat{t}=\hat{t}_1\|\hat{t}_2\|\ldots\|\hat{t}_\lenx$. 
    
    \item 
    At some point, $\qA$ chooses arbitrary $x \in \cX$ and sends it to the challenger.  
    \item \label{Step:WPRF-equiv-Step-4}
    The challenger sets $\key = \{ \TEPRF.\key_{j, x[j]\oplus c[j] } \}_{j\in [\lenx]}$ and sends $ \key$ to the adversary. 
    \item 
    The adversary continues to have access to the oracle. 
Finally, it outputs its guess $\coin'$ for $\coin$. The output of the experiment is then set to be $\coin'$.
\end{enumerate}
\item[$\Hyb_1$:] 
In this hybrid, we choose $c^*\lrun \bin^{\lenx}$ at Step~\ref{Step:WPRF-equiv-step1} of the game. Furthermore, at Step~\ref{Step:WPRF-equiv-Step-4} of the game, the challenger returns $\key = \{ \TEPRF.\key_{j, c^*[j] } \}_{j\in [\lenx]}$ to the adversary. 
We can see that this change is conceptual, since no information of $c$ is revealed except for Step~\ref{Step:WPRF-equiv-Step-4} of the game in the previous hybrid and thus $c$ effectively works as a one-time pad. 
We therefore have $\Pr[\event_1] = \Pr[\event_0]$.
\item[$\Hyb_2$:]
In this hybrid, we change the evaluation oracle so that $\hatt_j$ for $j\in [\lenx]$ is defined as 
\[
\hatt_j = \TEPRF.\Eval(\TEPRF.\key_{j,c^*[j]},\hat{s}_j).
\]
We claim that this change is unnoticed by $\qA$ except for a negligible probability. To see this, we first observe that by the equality of $\TEPRF$, the response from the oracle in this hybrid differs from that of the previous game only when  $\qA$ queries on input $\hat{s}$ such that $\hat{s}_j = s^*_j$ for some $j\in [\lenx]$.
However, we can show that such an event occurs only with negligible probability by a straightforward reduction to the hard-to-find differing-point property of $\TEPRF$, by noting that we only need $ \TEPRF.\key_{j, c^*[j] }$ for each $j$ for simulating $\hyb_2$ (i.e., $ \TEPRF.\key_{j, 1-c^*[j] }$ is not needed).
We therefore have $|\Pr[\event_{1}] - \Pr[\event_{2}]| = \negl(\secp)$.
\end{description}
We can observe that $\hyb_2$ is identical to the equivocality game between the challenger and an adversary $\A$ with $\coin =1$.
Furthermore, from the above discussion, we know $|\Pr[\event_{0}] - \Pr[\event_{2}]| = \negl(\secp)$. Therefore, the theorem follows.
\end{proof}

\subsection{Watermarkable Digital Signatures}\label{watermarkable_DS}
The definition of watermarkable digital signatures (DS) is very similar to that of watermarkable PKE in \Cref{sec:WPKE}; if 
an adversary is given a marked signing key $\sigk(x)$ and generates a (possibly quantum) "signer" that can sign a  uniformly random message with probability $\epsilon$, then we can extract $x$ from the signer with probability $\epsilon-\negl(\secpar)$, and moreover the extraction should work even in the parallel setting.  
We generalize the idea of \cite{kitagawa2025simple}
and show a generic construction of watermarkable DS from coherently signable constrained signatures and the watermarkable UPF, which in turn can be based on the SIS assumption. 
Similarly to the case of watermarkable UPF, we must also account for signing queries in the security proof, a scenario not addressed in \cite{kitagawa2025simple}. 
This introduces some technicalities in the proof, which can be resolved by exploiting the equivocality property of watermarkable UPF.

\paragraph{Syntax}
A watermarkable DS scheme $ \WDS=(\KG,\Mark, \Sign, \Vrfy )$ with the message space $\cM=\{\cM_\lambda\}_\lambda$ and the mark space $\cX=\{\cX_\lambda\}_\lambda$ consists of the following algorithms.
In the following, we omit the security parameter $\secp$ for $\cX_\secp$ and $\ell(\secp)$ to denote $\cX$ and $\ell$ when it is clear from the context.
%The syntax of watermarkable digital signatures is as follows:
\begin{description}
\item[$\KG(1^\secp)\ra (\vk,\msk,\extk)$:] The key generation algorithm is a PPT algorithm that takes a security parameter $1^\lambda$, and outputs a verification key $\vk$, master secret key $\msk$, and an extraction key $\extk$. 
%\takashi{Similarly to the case of UPF, it only supports secret extraction.

\item[$\Mark(\msk,x)\ra \sigk(x)$:] The marking algorithm is a PPT algorithm that takes a master secret key $\msk$ and a mark $x\in \cX$, and outputs a marked signing key $\sigk(x)$. 

\item[$\Sign(\sigk(x),m)\ra \sigma$:] The signing algorithm is a PPT algorithm that takes a master secret key $\msk$ or marked signing key $\sigk(x)$ and a message $m\in \cM$, and outputs a signature $\sigma$. 

\item[$\Vrfy(\vk,m,\sigma)\ra \top/\bot$:] The verification algorithm is a PPT algorithm that takes a verification key $\vk$, a message $m\in \cM$, and a signature $\sigma$, and outputs $\top$ or $\bot$.   

\item[Correctness:]
For every $x\in \cX$ and $m\in\cM$, we have
\begin{align}
\Pr\left[
\begin{array}{ll}
\Vrfy(\vk,m,\sigma) = \top 
\\
\qquad \quad \wedge
\\
\Vrfy(\vk,m,\sigma') = \top 
\end{array}
\ \middle |
\begin{array}{ll}
(\vk, \msk, \extk) \gets\KG(1^\secp)\\
\sigk(x)\gets \Mark(\msk,x)\\ 
\sigma \gets \Sign(\sigk(x), m) \\
\sigma' \gets \Sign(\msk, m)
\end{array}
\right] 
=1-\negl(\secp).
\end{align}

\item[Coherent signability]
We say that a watermarkable DS scheme is coherently-signable if for any polynomial $L=L(\secp)$, there is a QPT algorithm $\QSign$ that takes a quantum state $\ket{\psi}$ and a classical message $m\in \cM$ and outputs a quantum state $\ket{\psi'}$ and a signature $\sigma$, satisfying the following: 
\begin{enumerate}
\item For any family $\{x_z\in \cX \}_{z\in \bin^L}$, for any 
$z\in \bin^L$, 
$m\in \cM$, $(\vk,\msk, \extk) \in \Supp(\KG(1^\secp))$, and $ \sigk(x_z)\in \Supp(\Mark(\msk,x_z))$, the output distribution of $\QSign(\ket{z}\ket{\sigk(x_z)},m)$ is identical to that of $\Sign(\sigk(x_z),m)$.

\item For any families $\{ \alpha_z \in \CC \}_{z\in \bin^L }$ such that $\sum_{z\in \bin^L }|\alpha_z|^2 = 1$ and $\{x_z \in\cX \}_{z\in \bin^L }$ and for any $m\in \cM$, the following holds:
For overwhelming probability over the choice of 
$(\vk, \msk, \xk)\gets \KG(1^\secp)$ and for all $\sigk(x_z) \in \Supp(\Mark(\msk, x_z))$, we have
\begin{align}
\|\ket{\psi}\bra{\psi}-\ket{\psi'}\bra{\psi'}\|_{tr}=\negl(\secp),
\quad 
\mbox{where}
\quad 
\ket{\psi}=\sum_{z\in \bin^L }\alpha_z\ket{z}\ket{\sigk(x_z)},
\quad
\ifnum\llncs=1
\\
\fi
(\ket{\psi'},\sigma)\gets\QSign(\ket{\psi},m).
\end{align}
\end{enumerate}
\end{description}
\begin{remark}
    The coherent-signability is required so that the DS-SKL scheme we construct from the watermarkable DS has reusability of quantum signing key, which states that the quantum signing key can be repeatedly used to sign on multiple messages. 
\end{remark}
\begin{description}
\item[Parallel mark extractability:]
Here, we introduce the notion of ``parallel mark extractability", which is similar to that for watermarkable PKE. 
We first define the following parallel forgery experiment involving an adversary \( \A \) and a challenger, parameterized by an arbitrary polynomial \( n = \text{poly}(\secp) \), denoted as \( \expb{\WDS,\qA, n}{par}{for}(1^\secp) \).

% We require the following "parallel extractability", which requires the existence of a PPT extractor $\Ext$ satisfying the following:
% Consider the following game where $n$ is an arbitrary polynomial:
\begin{enumerate}
\item The challenger runs $\KG(1^\secp)\ra (\vk_i,\msk_i,\extk_i)$ for $i\in [n]$ and
sends $\{ \vk_i \}_{i\in [n]}$ to $\qA$.
\item 
Throughout the game, $\qA$ has an oracle access to the signing oracle $\Sign(\msk_i, ) $ for each $i\in [n]$, until it receives the challenge messages in Step~\ref{Step:WDS-exp-step-7}.
    When queried on input $\hatm\in \cM$, the oracles $\Sign(\msk_i, ) $ runs $\Sign(\msk_i,\hatm) \to \hatsigma$ and returns $\hatsigma$ to $\qA$. 
\item At some point, the adversary sends $x_1,...,x_n\in \cX$ to the challenger. 
\item The challenger computes the marked signing key $\sigk_i(x_i)\gets\Mark(\msk_i,x_i)$ for $i\in [n]$ and sends $\{ \sigk_i(x_i)\}_{i\in[n]}$ to $\qA$.
% \item For $i\in [n]$, the challenger generates $\msk_i\gets \KG(1^\secp)$ and $\sk_i(x_i)\gets\Mark(\msk_i,x_i)$ and sends $(\sk_i(x_i))_{i\in[n]}$ to the adversary. \takashi{We may additionally give an evaluation oracle to the adversary. I believe our construction also satisfies the security in this setting. This will be inherited to the eventual construction of UPF-SKL with a classical lessor. }
% \item The adversary outputs (potentially entangled) "quantum signer" $(S_i)_{i\in [n]}$. 
\item 
    After receiving the marked signing keys, $\qA$ continues to make queries to the signing oracles.
    At some point, it declares that it is ready for the challenge. 
\item  \label{Step:WDS-exp-step-6}
    The challenger picks uniformly random $m_i\lrun \cM$ for $i\in [n]$ and sends $\{m_i\}_{i\in [n]}$ to $\qA$. 
\item \label{Step:WDS-exp-step-7}
    Receiving $\{m_i\}_{i\in [n]}$ from the challenger, $\qA$
     outputs $\{\sigma_i\}_{i\in [n]}$.
     Note that after receiving $\{m_i\}_{i\in [n]}$ from the challenger, $\qA$ loses access to the signing oracles.     
% \item For $i\in [n]$, the challenger runs $\sigma'_i \gets S_i(\msg_i)$.
\item  \label{Step:WDS-exp-step-8}
The experiment outputs $1$ (indicating that $\qA$ wins) if $\Vrfy(\vk_i,m_i,\sigma_i)=\top$ for all $i\in [n]$ and otherwise outputs $0$. 

\end{enumerate}
We define the advantage of a QPT adversary $\qA$ to be 
\begin{align}
\advb{\WDS,\qA,n}{par}{for}(\secp) \seteq \Pr[\expb{\WDS,\qA,n}{par}{for}(1^\secp)= 1].
\end{align}
We also define parallel mark extraction experiment defined by a QPT extractor $\qExt$, an adversary \( \A \), and a challenger, parameterized by an arbitrary polynomial \( n = \text{poly}(\secp) \), denoted as $\expb{\WDS,\qA,\qExt,n}{par}{ext}(1^\secp)$. 
The parallel extraction experiment is the same as parallel forgery experiment except for the last 3 steps. Namely, we replace Step \ref{Step:WDS-exp-step-6}, \ref{Step:WDS-exp-step-7}, and \ref{Step:WDS-exp-step-8} with the following:
\begin{enumerate}[label=\arabic*', start=6]
    \item Take the adversary $\qA$ right before Step \ref{Step:WDS-exp-step-6} of \( \expb{\WDS,\qA,n}{par}{for}(1^\secp) \). 
    Since $\qA$ is a quantum non-uniform QPT adversary,  we can
parse $\qA$ as a pair of a unitary circuit $U_\qA$ and a quantum state $\rho_\A$, where to run $\qA$ on the inputs $\{m_i\}_{i\in [n]}$, 
we apply $U_\qA$ to $\{m_i\}_{i\in [n]}$ and $\rho_\qA$ and then measure the corresponding registers to obtain the output $\{\sigma_i\}_{i\in [n]}$.
\item We run $\qExt(\{\extk_i\}_{i\in [n]}, (U_\qA, \rho_\qA) ) \to \{x'_i\}_{i\in [n]}$.  
\item The experiment outputs $1$ (indicating that the adversary wins) if $x'_i=x_i$ for all $i\in [n]$ and $0$ otherwise. 
\end{enumerate}
We define the advantage of a QPT adversary $\cA$ with respect to $\qExt$ to be 
\begin{align}
\advb{\WDS,\qA,\qExt,n}{par}{ext}(\secp) \seteq \Pr[\expb{\WDS,\qA,\qExt,n}{par}{ext}(1^\secp)= 1].
\end{align}
We say that $\WDS$ satisfies parallel mark extractability if there exists a QPT algorithm $\qExt$ such that for all $n=\poly(\secp)$, for all QPT algorithm $\cA$,
the following holds:
\[
\advb{\WDS,\qA,\qExt,n}{par}{ext}(\secp) \geq \advb{\WDS,\qA,n}{par}{for}(\secp) -\negl(\secp).
\] 
% For any QPT adversary, we have
% \[
% \Pr\left[
% \forall~i\in[n],~ \Ext(\extk_i,S_i)=x_i
% \right]\ge
% ~\Pr[\text{The~adversary~wins~the~above~experiment}]-\negl(\secp)
% \]
% where, in the probability in the LHS, $\extk_i,x_i,S_i$ are generated as in the above experiment.  
% \if0
% With an overwhelming probability over the execution of the above experiment, the following holds: 
% If all the predictors are simultaneously good, i.e.,
% \[
% \Pr_{\msg_1,...,\msg_n}\left[
% \forall~i\in[n], \Vrfy(\vk,\msg_i,S_i(\msg_i))=\top
% \right]\ge \epsilon
% \]
% then we can extract all the marks simultaneously, i.e., 
% \[
% \Pr\left[
% \forall~i\in[n],~ \Ext(S_i)=x_i
% \right]\ge \epsilon-\negl(\secp).
% \]
% \fi
\end{description}
\begin{remark}
    Similarly to the case of UPF, the above definition only supports secret extraction. However, it is sufficient for our purpose.
\end{remark}
\begin{remark}[Security as a plain DS]\label{rem:WDS-security-as-plain-DS}
We note that the above definition of a watermarkable DS does not necessarily imply the standard EUF-CMA security as a plain DS. 
Nevertheless, the EUF-CMA security can be added by a simple conversion: run a plain DS in parallel and append its signature alongside the original one. 
The new verification algorithm verifies both the signatures and accepts only when both are valid. 
The marked key of the resulting scheme consists of that of the original one along with the signing key of the plain DS.
\end{remark}

\paragraph{Construction}
%The above defined watermarkable digital signatures can be constructed from any two-key equivocal PRF and coherently-signable constrained signatures, both which in turn can be constructed from LWE. (See \Cref{sec:coherently_signable_DS} for the definition of coherently-signable constrained signatures.)

The above defined watermarkable digital signature scheme can be constructed from any watermarkable UPF and coherently-signable constrained signatures. 
Our construction abstracts the core idea of the watermarkable DS scheme implicit in \cite{kitagawa2025simple}. 
Observe first that UPF can already be viewed as a MAC; what remains is simply to augment it with a public verification algorithm. 
To achieve this, we equip the signer with a constrained signing key that allows signatures to be generated only on valid input--output pairs of the UPF. 
Consequently, in order to produce a valid signature on a message (treated as an input to the UPF), the signer must also provide a corresponding valid output of the UPF as part of the signature. 
This structure enables us to reuse the same extractor as in the UPF setting.

% Our idea for the construction abstracts out the idea behind the DS-SKL construction in \cite{kitagawa2025simple}. 
% First, note that UPF can already be regarded as a MAC, so we only need to augment it by providing a public verification algorithm. 
% To do so, we provide a signer with a constrained signing key that enables one to generate a signature only on valid input-output pairs of the two-key equivocal PRF. Then, for the signer to generate a valid signature on some message (which is regarded as an input of the two-key equivocal PRF), it must provide a valid output of the  two-key equivocal PRF as part of a signature. Thus, we can use the same extractor as that for the case of UPF. %\takashi{We might be able to formalize this as a generic construction from watermarkable UPF and coherently-signable constrained signatures.}

% The coherent-signability is needed to make sure that the quantum signing key can be repeatedly used to sign on multiple messages. Note that such a property wasn't explicitly required for PKE and UPF since they are deterministic and thus coherent execution is always possible by the gentle measurement lemma. 
%\synote{The above explanation will be removed or edited.}

Formally, our construction of watermarkable digital signature scheme $\WDS=(\KG, \Mark, \Sign, \Vrfy)$ uses watermarkable UPF $\WPRF=\WPRF.(\KG,\Mark,\Eval)$ and coherently-signable constrained signatures $\CS=\CS.(\Setup,\Constrain,\Sign,\Vrfy)$ (See \Cref{def:cs} for the definition of coherently-signable constrained signatures.) as building blocks, both of which in turn can be constructed from SIS. 
As additional requirements, we need $\WPRF$ to satisfy equivocality and $\CS$ to satisfy function privacy. 
Our construction $\WDS$ supports mark space $\bin^\lenx$ and the message space $\bin^\ell$.
For the construction, we need the underlying $\WPRF$ to support the mark space $\bin^\lenx$ and the input space $\bin^\ell$. 
We denote the bit-length of the output of the $\WPRF$ function by $v$.
We then need $\CS$ to support the $\WPRF$ evaluation circuit defined below,
whose input length is $\ell +v$ and the depth is bounded by a fixed polynomial.

\begin{description}
\item[$\KG(1^\secp)\ra(\vk, \msk,\extk)$:] 
On input the security parameter $1^\secp$, the key generation algorithm proceeds as follows. 
\begin{itemize}
    \item Run $\WPRF.\KG(1^\secp) \to (\WPRF.\msk, \WPRF.\extk) $.
    \item Run $\CS.\Setup(1^\secp) \rrun (\CS.\vk, \CS.\msk)$.
    \item Output $\vk = \CS.\vk$, 
    $\msk = (\WPRF.\msk, \CS.\msk)$,  
    and $\extk=  \WPRF.\extk$.
\end{itemize}
% By running the key generation algorithm of the two-key equivocal PRF $\ell$ times with uniformly random strings $s^*_1,s^*_2,...,s^*_\ell$ as input, generate $(\sk_{j,0},\sk_{j,1})$ for $j\in [\ell]$. 
% Output a master secret key $\msk=(\sk_{j,b})_{j\in[\ell],b\in\bit}$
% and an extraction key $\extk=(s^*_i)_{i\in [\ell]}$.

\item[$\Mark(\msk,x)\ra \sigk(x)$:] 
Upon input the master secret key $\msk = (\WPRF.\msk, \CS.\msk)$ and a mark $x\in \bit^\lenx$, it proceeds as follows.
\begin{itemize}
    \item Run $\WPRF.\Mark(\WPRF.\msk,x) = \WPRF.\key(x)$.
    \item Construct a circuit $f[\WPRF.\key(x)]$ that takes as input a pair $(m,t)\in \bin^\ell \times \bin^v$
    and outputs $1$ if $\WPRF.\Eval(\WPRF.\key(x), m) = t$ and $0$ otherwise.
    \item Run $\CS.\Constrain(\CS.\msk, f[\WPRF.\key(x)]) \to \CS.\sigk(x)$.
    \item Output $\sigk(x)= (\WPRF.\key(x), \CS.\sigk(x))$. 
\end{itemize}

\item[$\Sign(\sigk(x),m)\ra \sigma$:] 
On input the marked key $\sigk(x)$ or the master secret key $\msk$ and the message $m\in \bin^\ell$, the signing algorithm proceeds as follows. 
\begin{itemize}
    \item If the first input is a master secret key, parse it as $\msk = (\WPRF.\msk, \CS.\msk)$. Then, set $\WPRF.\key = \WPRF.\msk$ and compute
    $\CS.\Constrain(\CS.\msk, f_{=1}) \to \CS.\sigk$, where $f_{=1}$ is a constant function that always outputs $1$.
    \item If the first input is a marked key $\sigk(x)$, parse it as 
    $\sigk(x)= (\WPRF.\key, \CS.\sigk)$.
    \item Compute $t= \WPRF.\Eval(\WPRF.\key, m)$ and $\CS.\Sign(\CS.\sigk, (m,t)) \to \CS.\sigma$.
    \item Output $\sigma=(t,\CS.\sigma)$.
\end{itemize}
\item[$\Vrfy(\vk,m,\sigma)$:] 
On input the message $m$ and the signature $\sigma=(t,\CS.\sigma)$, it runs $\CS.\Vrfy(\CS.\vk, (m,t), \CS.\sigma)$ and outputs whatever it outputs.
\end{description}

\paragraph{Correctness}
%\shota{To be added}
We first observe that the signature $\sigma=(t,\CS.\sigma)$ generated by the master secret key passes the verification with overwhelming probability by the correctness of $\CS$, since $f_{=1}(m,t)=1$ for any $m$ and $t$.
We then observe that the signature $\sigma=(t,\CS.\sigma)$ generated by the signing key $\sigk(x)$ also passes the verification 
with overwhelming probability by the correctness of $\CS$, since 
$f[\WPRF.\key(x)](m,t)= 1$ holds when $t= \WPRF.\Eval(\WPRF.\key(x), m)$. 

\paragraph{Coherent Signability}
To show the coherent signability, we first define the quantum signing algorithm $\QSign$.
It takes as input a quantum state $\ket{\psi}$ and a message $m$. 
We can assume that $\ket{\psi}$ is of the form $\sum_{z\in \bin^L }\alpha_z\ket{z}\ket{\sigk(x_z)} = \sum_{z\in \bin^L}\alpha_z\ket{z}\ket{\WPRF.\key(x_z)}\ket{\CS.\sigk(x_z)}$.
\begin{enumerate}
    \item Apply the following isometry to $\ket{\psi}$: 
    \begin{align}
    \ket{z}\ket{\WPRF.\key(x_z)}\ket{\CS.\sigk(x_z)} 
    \mapsto
    \ket{z}\ket{\WPRF.\key(x_z)}\ket{\CS.\sigk(x_z)}\ket{\WPRF.\Eval(\WPRF.\key(x_z), m) }.
    \end{align}
    \item Measure the last register and discard it. Let the result of the measurement be $t$ and the resulting state after discarding the register be $\ket{\phi}$.
    \item Run $\CS.\QSign(\ket{\phi}, (m,t)) \to (\ket{\psi'}, \CS.\sigma)$ and output $\ket{\psi'}$ and $\sigma=(t,\CS.\sigma)$. 
\end{enumerate}

We first show that the above defined algorithm satisfies the first property of coherent signability, which states that applying $\QSign$ to a classical message results in a signature distributed identically to when $\Sign$ is applied.
This straightforwardly follows from the fact that $\CS.\QSign$ has the same output distribution as  $\CS.\Sign$ when the input is classical.

We then show the second property, which states that the application of $\QSign$ on a superposition of a marked signing key almost does not change the state.
To see this, we first observe that after applying the first step, the state becomes
\begin{align}
&
\sum_{z\in \bin^L}\alpha_z\ket{z}\ket{\WPRF.\key(x_z)}\ket{\CS.\sigk(x_z)}
\ket{\WPRF.\Eval(\WPRF.\key(x_z), m) }
\\
&=
\sum_{z\in \bin^L}\alpha_z\ket{z}\ket{\WPRF.\key(x_z)}\ket{\CS.\sigk(x_z)}
\ket{\WPRF.\Eval(\WPRF.\msk, m) },
\end{align}
where the equation holds with overwhelming probability over the randomness of $\WPRF.\KG(1^\secp)$ from the correctness of $\WPRF$.\footnote{Recall that we require somewhat strong evaluation correctness condition for $\WPRF$, which requires that with overwhelming probability over the choice of $\WPRF.\msk$ and $\WPRF.\xk$, the evaluation result by the master secret key equals that of the marked key with mark $x$ \emph{for all} $x$.}
Therefore, we have $\ket{\phi}=\ket{\psi}$ with overwhelming probability. 
We next see that the application of $\CS.\QSign$ almost does not change the state $\ket{\phi}$. 
To see this, we consider a function family $\{g_{z'}\}_{z'}$ where for $z'= (z, \WPRF.\key(x_z))$, $g_{z'}$ is defined as $g_{z'} = f[\WPRF.\key(x_z)]$ and apply the coherent signability of $\CS$.
We therefore have that $\ket{\psi'}$ is negligibly close to $\ket{\psi}$ in trace distance as desired.

\paragraph{Security Proof}
The following theorem asserts the parallel mark extractability of the above construction.
\begin{theorem}
Assume that $\WPRF$ satisfies the parallel mark extractability as watermarkable unpredictable function and equivocality and $\CS$ satisfies selective single-key security and function privacy.  
Then, the above construction satisfies parallel mark extractability as watermarkable digital signatures. 
\end{theorem}
Before presenting the proof, we briefly outline the intuition. 
At a high level, one can extract the mark from a signer, 
since a signature $\sigma = (t, \CS.\sigma)$ contains the $\WPRF$ value $t$ and the extractor for $\WPRF$ can be applied to perform the extraction. 
The only way for the adversary to evade the extraction is to produce a signature 
$(t, \CS.\sigma)$ such that $\CS.\sigma$ is a valid signature on $(m,t)$,  
but $t \neq \WPRF.\Eval(\WPRF.\key(x), m)$. 
However, generating such a signature constitutes a forgery against $\CS$ 
and is therefore computationally infeasible, 
since the adversary only obtains a constrained key that permits signing 
on pairs $(m,t)$ satisfying $t = \WPRF.\Eval(\WPRF.\key(x), m)$.

Turning this intuition into a formal proof is nontrivial, 
because the mark is chosen adaptively by the adversary 
(whereas in \cite{kitagawa2025simple} it is sampled uniformly at random). 
In particular, when reducing to selectively single-key forgery, 
the reduction must commit to a function $f$ at the beginning of the game. 
In our setting, however, this function necessarily depends on the mark $x$, 
since it hardwires the $\WPRF$ key for that mark, which is chosen adaptively by the adversary. 
As a result, the reduction cannot commit to the function at the outset. 
To overcome this difficulty, we rely on the equivocality of $\WPRF$, 
which provides a simulated key that is indistinguishable from a marked key 
even before $x$ is known. 
This simulated key can be chosen independently of $x$ 
and then embedded into the function.

An additional subtlety arises from the presence of a signing oracle, 
which is not introduced in \cite{kitagawa2025simple}. 
For simulating the signing oracle, we use the constrained key of $\CS$, similarly to the case of $\WPRF$. 
Here, however, the function hardwired into the constrained key 
differs from the one in the real execution, 
and thus could be detected by the adversary. 
The function privacy of $\CS$ guarantees this is not the case and such a difference is hidden from the adversary.

\begin{proof}
To prove the theorem, we first define the extraction algorithm $\qExt$ as follows.
\begin{description}
    \item[$\qExt(\{\extk_i\}_{i\in [n]}, (U_\qA, \rho_\qA) ) \to \{x'_i\}_{i\in [n]}$:] 
    Given the extraction keys $\{ \extk_i \}_{i\in [n]}$ and the description of the adversary $\qA = (U_\qA, \rho_\qA)$, it proceeds as follows.  
    \begin{itemize}
        \item Parse the extraction key as
        $\extk_i= \WPRF.\extk_i $ for each $i\in [n]$.
        \item From the adversary $\qA$, construct a predictor $\qP=(U_\qP, \rho_\qP)$ for $\WPRF$ as follows.
        \begin{description}
            \item[$\qP(m_1,\ldots, m_n):$]
            Given $m_1,\ldots, m_n$, it first runs $\qA$ on input $m_1,\ldots, m_n$ to obtain $\{ \sigma_i \}_{i\in [n]}$.
            It then parses $ \sigma_i = (t_i, \CS.\sigma_i)$ for each $i\in [n]$.
            It finally outputs $(t_1,\ldots, t_n)$.
        \end{description}
        \item Run $\WPRF.\qExt(\{\WPRF.\xk_i\}_{i\in [n]}, (U_\qP, \rho_\qP) ) \rrun \{x'_i\}_{i\in [n]}$. 
    \item Output $ \{x'_i\}_{i\in [n]}$.
    \end{itemize}
\end{description}
    We then prove the theorem by a sequence of hybrids.
    Let us consider an adversary $\A$ against the parallel forgery game \( \expb{\WDS,\qA,n}{par}{for}(1^\secp) \). 
    We gradually change the game into the parallel extraction game \( \expb{\WDS,\qA,\qExt,n}{par}{ext}(1^\secp) \) without changing the winning probability of $\qA$ more than negligibly.
    In the following, the event that $\qA$ wins in $\Hyb_{\rm xx}$ is denoted by $\event_{\rm xx}$.
\begin{description}
      \item[$\Hyb_0$:] This is the parallel forgery game between the challenger and an adversary $\A$.  
      Namely, the game proceeds as follows.
    \begin{enumerate}
    \item \label{Step:WDS-1}
    The challenger runs $\KG(1^\secp)$ to obtain $\vk_i = \CS.\vk_i$, 
    $\msk_i = (\WPRF.\msk_i, \CS.\msk_i)$,  
    and $\extk_i=  \WPRF.\extk_i$ for $i\in [n]$. $\qA$ is given $\{ \vk_i\}_{i\in [n]}$.
    \item \label{Step:WDS-2}
    Throughout the game, the adversary has an oracle access to the signing oracle $\Sign(\msk_i, ) $ for each $i\in [n]$, until it receives the challenge messages in Step~\ref{Step:WDS-step-7}.
    When queried on input $\hatm$, the oracles $\Sign(\msk_i, ) $ proceeds as follows:
    \begin{itemize}
        \item Compute $\CS.\Constrain(\CS.\msk_i, f_{=1}) \to \CS.\sigk_i$, where $f_{=1}$ is a constant function that always outputs $1$.
        \item Compute $\hatt = \WPRF.\Eval(\WPRF.\msk_i, \hatm )$. 
        \item Run $\CS.\Sign(\CS.\sigk_i, (\hatm, \hatt) ) \to \CS.\hatsigma$.
    \item Return $\hatsigma = (\hatt, \CS.\hatsigma)$. 
    \end{itemize}
    
    \item 
    At some point, $\qA$ chooses arbitrary $x_1,...,x_n \in \cX$ and sends them to the challenger.  
    \item \label{Step:WDS-Step-4}
    For $i\in [n]$, the challenger sets the marked signing keys $\sigk_i \seteq ( \WPRF.\key_{i}(x_i), \CS.\sigk_i(x_i) )$ and sends $\{ \sigk_i \}_{i\in[n]}$ to the adversary, where
    \begin{align}
    \WPRF.\key_i(x_i) = \WPRF.\Mark(\WPRF.\msk_i,x_i) 
    \ifnum\llncs=1
    \\
    \fi
    \text{ and } \CS.\sigk_i(x_i) \lrun \CS.\Constrain(\CS.\msk_i, f[\WPRF.\key_i(x_i)]). 
    \end{align}
    
    \item 
    After receiving the marked signing keys, $\qA$ continues to make queries to the signing oracles.
    At some point, it declares that it is ready for the challenge. 
    \item \label{Step:WDS-Step-6}
    The challenger picks uniformly random $m_i\lrun \bin^{\ell}$ for $i\in [n]$ and sends $\{m_i\}_{i\in [n]}$ to the adversary. 
    \item \label{Step:WDS-step-7}
    Receiving $\{m_i\}_{i\in [n]}$ from the challenger, $\qA$
     outputs $\{\sigma_i\}_{i\in [n]}$.
     Note that after receiving $\{m_i\}_{i\in [n]}$ from the challenger, $\qA$ loses access to the signing oracles.
    \item \label{Step:WDS-Step-8}
    Given $\{\sigma_i\}_{i\in [n]}$, the challenger parses it as $\sigma_i=(t_i,\CS.\sigma_i)$ for $i\in [n]$ and proceeds as follows.
    \begin{itemize}
        \item It checks whether $\CS.\Vrfy(\CS.\vk_i, (m_i,t_i), \CS.\sigma_i)=\top$ holds for all $i\in [n]$.
        If it holds, then it sets the output of the game to be $1$ (indicating that the adversary wins). Otherwise, it outputs $0$.
    \end{itemize}
\end{enumerate}

\item[$\Hyb_1$:]
In this hybrid, we use the simulated key of $\WPRF$ to answer signing queries and to simulate the marked signing keys. 
Concretely, we replace Step~\ref{Step:WDS-1}, Step~\ref{Step:WDS-2}, and Step \ref{Step:WDS-Step-4} with the following: 
\begin{enumerate}[label=\arabic*', start=1]
    \item \label{Step:WDS-1'}
    The challenger runs $\CS.\Setup(1^\secp) \to (\CS.\vk_i,\CS.\msk_i)$ 
    and $\WPRF.\Sim(1^\secp) \to \WPRF.\key^*_i$ for $i\in [n]$.
    Then, $\qA$ is given $\{ \vk_i = \CS.\vk_i\}_{i\in [n]}$.
    \item \label{Step:WDS-2'}
    Throughout the game, the adversary has an oracle access to the signing oracles, until it receives the challenge messages in Step~\ref{Step:WDS-step-7}.
    When queried on input $\hatm$, the $i$-th oracle proceeds as follows:
    \begin{itemize}
        \item Compute $\CS.\Constrain(\CS.\msk_i, f_{=1}) \to \CS.\sigk_i$, where $f_{=1}$ is a constant function that always outputs $1$.
        \item Compute $\hatt = \WPRF.\Eval(\WPRF.\key^*_i, \hatm )$. 
        \item Run $\CS.\Sign(\CS.\sigk_i, (\hatm, \hatt) ) \to \CS.\hatsigma$.
    \item Return $\hatsigma = (\hatt, \CS.\hatsigma)$. 
    \end{itemize}
\setcounter{enumi}{3}
    \item \label{Step:WDS-4'}
    For $i\in [n]$, the challenger sets the marked signing keys $\sigk_i \seteq ( \WPRF.\key^*_{i}, \CS.\sigk^*_i )$ and sends $\{ \sigk_i \}_{i\in[n]}$ to the adversary, where
\begin{align}\label{eq:CS-sigk-with-sim-WPRFkey}
    \CS.\sigk^*_i &\lrun \CS.\Constrain(\CS.\msk_i, f[\WPRF.\key_i^*]).
\end{align}

\end{enumerate}
By a straightforward reduction to the equivocality of $\WPRF$, we have $|\Pr[\event_0]-\Pr[\event_1]|=\negl(\secp)$.

\item[$\Hyb_2$:]
In this hybrid, we change how we handle signing queries. 
We move the sampling of $\CS.\sigk^*_i$ as per \cref{eq:CS-sigk-with-sim-WPRFkey} to Step \ref{Step:WDS-1}, which is possible since it is not dependent on $x_i$ any more, and then use it for answering the signing queries.
Namely, we replace Step~\ref{Step:WDS-1'}, Step~\ref{Step:WDS-2'}, and Step \ref{Step:WDS-4'} with the following: 
\begin{enumerate}[label=\arabic*'', start=1]
    \item \label{Step:WDS-1''}
    The challenger runs $\CS.\Setup(1^\secp) \to (\CS.\vk_i,\CS.\msk_i)$, 
    $\WPRF.\Sim(1^\secp) \to \WPRF.\key^*_i$, and 
    $\CS.\sigk^*_i \lrun \CS.\Constrain(\CS.\msk_i, f[\WPRF.\key_i^*])$ for $i\in [n]$.
    Then, $\qA$ is given $\{ \vk_i = \CS.\vk_i\}_{i\in [n]}$.
    \item \label{Step:WDS-2''}
    Throughout the game, the adversary has an oracle access to the signing oracles, until it receives the challenge messages in Step~\ref{Step:WDS-step-7}.
    When queried on input $\hatm$, the $i$-th oracle proceeds as follows:
    \begin{itemize}
        \item Compute $\hatt = \WPRF.\Eval(\WPRF.\key^*_i, \hatm )$. 
        \item Run $\CS.\Sign(\CS.\sigk^*_i, (\hatm, \hatt) ) \to \CS.\hatsigma$.
    \item Return $\hatsigma = (\hatt, \CS.\hatsigma)$. 
    \end{itemize}
\setcounter{enumi}{3}
    \item \label{Step:WDS-4''}
    For $i\in [n]$, the challenger sets the marked signing keys $\sigk_i \seteq ( \WPRF.\key^*_{i}, \CS.\sigk^*_i )$ and sends $\{ \sigk_i \}_{i\in[n]}$ to the adversary.
\end{enumerate}

We claim that this hybrid is statistically indistinguishable from the previous one, due to the function privacy of $\CS$.
We recall that $\CS.\sigk^*_i$ in this hybrid encodes a function $f[\WPRF.\key^*_i]$ and thus we have $f[\WPRF.\key^*_i](\hatm, \hatt)=1 = f_{=1}(\hatm, \hatt)$. 
Therefore, the distribution of the returned $\CS.\hatsigma$ is statistically close to that of the previous hybrid. 
We therefore have $|\Pr[\event_{1}] - \Pr[\event_{2}]| = \negl(\secp)$.

\item[$\Hyb_3$:]
In this hybrid, we change the winning condition of $\qA$.
Namely, we replace Step \ref{Step:WDS-Step-8} with the following, where additional check regarding $t_i$ is inserted.
\begin{enumerate}[label=\arabic*', start=8]
    \item \label{Step:WDS-8'}
        Given $\{\sigma_i\}_{i\in [n]}$, the challenger parses it as $\sigma_i=(t_i,\CS.\sigma_i)$ for $i\in [n]$ and proceeds as follows.
    \begin{itemize}
        \item It checks whether $t_{i} = \WPRF.\Eval(\WPRF.\key^*_{i}, m_{i})$ holds for all $i\in [n]$. If not, it outputs $0$.
        \item It then checks whether $\CS.\Vrfy(\CS.\vk_i, (m_i,t_i), \CS.\sigma_i)=\top$ holds for all $i\in [n]$.
        If it holds, then it sets the output of the game to be $1$ (indicating that the adversary wins). Otherwise, it outputs $0$.
    \end{itemize}
\end{enumerate}

We claim that this change does not reduce the winning probability of $\qA$ more than negligibly due to the selective single key unforgeability of $\CS$.
To see this, we observe that the only case where $\qA$ wins the game in the previous hybrid, but does not in the current one occurs when $\CS.\Vrfy(\CS.\vk_i, (m_i,t_i), \CS.\sigma_i)=\top$ holds for all $i\in [n]$,
but there exists $i^*\in [n]$ such that 
$t_{i^*} \neq  \WPRF.\Eval(\WPRF.\key^*_{i^*}, m_{i^*})$. 
However, the probability that such an event occurs is negligible due to the selective single-key security of $\CS$.
To see this, we consider a reduction that guesses $i^*$ and embeds the parameter of the $\CS$ into this instance, while the other instances as well as the parameters related to $\WPRF$ are generated by itself. 
We observe that the reduction algorithm can choose its target function $f[\WPRF.\key^*_{i^*}]$ at the beginning of the game, since $\WPRF.\key^*_{i^*}$ can be chosen by $\WPRF.\Sim(1^\secp)$, without depending on any parameter other than the security parameter. 
Furthermore, by the changes we introduced so far, all the signing queries made by $\qA$ can be handled by the single key $\CS.\sigk^*_{i^*}$. 
Finally, since $t_{i^*} \neq  \WPRF.\Eval(\WPRF.\key^*_{i^*}, m_{i^*})$ implies $f[\WPRF.\key^*_{i^*}] (m_{i^*},t_{i^*})=0 $, 
the reduction breaks the selective single-key security if such an event occurs. 
We therefore have 
$\Pr[\event_{3}] \geq  \Pr[\event_{2}] - \negl(\secp)$.

\item[$\Hyb_4$:]
In this hybrid, we switch Step~\ref{Step:WDS-1''}, Step~\ref{Step:WDS-2''}, and Step \ref{Step:WDS-4''} back to those in $\hyb_1$, i.e., Step~\ref{Step:WDS-1'}, Step~\ref{Step:WDS-2'}, and Step \ref{Step:WDS-4'}. 
In particular, we use $\CS.\sigk_i$ derived as 
$\CS.\Constrain(\CS.\msk_i, f_{=1}) \to \CS.\sigk_i$ to sign on $(\hatm,\hatt)$.
Similarly to the change from $\hyb_1$ to $\hyb_2$, we have $|\Pr[\event_{3}] - \Pr[\event_{4}]| = \negl(\secp)$ by the function privacy of $\CS$.

\item[$\Hyb_5$:]
In this hybrid, we switch Step~\ref{Step:WDS-1'}, Step~\ref{Step:WDS-2'}, and Step \ref{Step:WDS-4'} back to those in $\hyb_0$, i.e., Step~\ref{Step:WDS-1}, Step~\ref{Step:WDS-2}, and Step \ref{Step:WDS-Step-4}. 
Namely, the simulated key $\WPRF.\key_i^*$ is not generated any more. 
Accordingly, we also change Step~\ref{Step:WDS-8'} so that it uses $\WPRF.\key_i(x_i)$ instead of $\WPRF.\key_i^*$.
Concretely, we replace Step~\ref{Step:WDS-8'} with the following.
\begin{enumerate}[label=\arabic*'', start=8]
    \item \label{Step:WDS-8''}
        Given $\{\sigma_i\}_{i\in [n]}$, the challenger parses it as $\sigma_i=(t_i,\CS.\sigma_i)$ for $i\in [n]$ and proceeds as follows.
    \begin{itemize}
        \item It checks whether $t_{i} = \WPRF.\Eval(\WPRF.\key_i(x_i), m_{i})$ holds for all $i\in [n]$. If not, it outputs $0$.
        \item It then checks whether $\CS.\Vrfy(\CS.\vk_i, (m_i,t_i), \CS.\sigma_i)=\top$ holds for all $i\in [n]$.
        If it holds, then it sets the output of the game to be $1$ (indicating that the adversary wins). Otherwise, it outputs $0$.
    \end{itemize}
\end{enumerate}
Similarly to the change from $\hyb_0$ to $\hyb_1$, we have 
$|\Pr[\event_{4}] - \Pr[\event_{5}]| = \negl(\secp)$ by the equivocality of $\WPRF$.

\item[$\Hyb_6$:]
In this hybrid, we further change the winning condition so that it does not check the validity of $\CS$ signatures. 
Namely, we replace Step~\ref{Step:WDS-8''} with the following: 
\begin{enumerate}[label=\arabic*''', start=8]
    \item \label{Step:WDS-8'''}
        Given $\{\sigma_i\}_{i\in [n]}$, the challenger parses it as $\sigma_i=(t_i,\CS.\sigma_i)$ for $i\in [n]$ and proceeds as follows.
    \begin{itemize}
        \item It checks whether $t_{i} = \WPRF.\Eval(\WPRF.\key_i(x_i), m_{i})$ holds for all $i\in [n]$. If not, it outputs $0$.
    \end{itemize}
\end{enumerate}
This change relaxes the winning condition for $\qA$ and only increases the chance of it winning. We therefore have $\Pr[\event_{6}] \geq \Pr[\event_{5}]$.

\item[$\Hyb_7$:]
In this hybrid, we further change the winning condition so that the check of the $\WPRF$ values are done using $\WPRF.\msk$.
Namely, we replace Step~\ref{Step:WDS-8'''} with the following: 
\begin{enumerate}[label=\arabic*${}^*$, start=8]
    \item \label{Step:WDS-8''''}
        Given $\{\sigma_i\}_{i\in [n]}$, the challenger parses it as $\sigma_i=(t_i,\CS.\sigma_i)$ for $i\in [n]$ and proceeds as follows.
    \begin{itemize}
        \item It checks whether $t_{i} = \WPRF.\Eval(\WPRF.\msk_i, m_{i})$ holds for all $i\in [n]$. If not, it outputs $0$.
    \end{itemize}
\end{enumerate}
By the correctness against adversarially chosen mark and input of $\WPRF$, which is implied by the equivocality, we have $\WPRF.\Eval(\WPRF.\msk_i, m_{i}) = \WPRF.\Eval(\WPRF.\key_i(x_i), m_{i})$ with overwhelming probability for each $i\in [n]$, even if $m_i$ was adversarially chosen by $\qA$.  
We therefore have $\abs{\Pr[\event_{7}] - \Pr[\event_{6}]} \leq \negl(\secp)$.

\item[$\Hyb_8$:]
In this hybrid, we further change the winning condition. 
Here, we use $\qA$ as a predictor for $\WPRF$ and then extract the marks using the corresponding extractor of $\WPRF$. 
Concretely, we replace Step~\ref{Step:WDS-Step-6}, Step~\ref{Step:WDS-step-7}, and Step~\ref{Step:WDS-8''''} with the following:
\begin{enumerate}[label=\arabic*${}^{**}$, start=6]
    \item The challenger parses the adversary $\qA$ right before Step~\ref{Step:WDS-step-7} of the game as a unitary circuit $U_\qA$ and a quantum state $\rho_\qA$.
      It then constructs a quantum predictor $\qP$ for $\WPRF$ that takes $\{m_i\}_{i\in [n]}$ as input and outputs $\{ t_i \}_{i\in [n]}$ 
      from $\A$ as follows.
     \begin{description}
         \item[$\qP( \{m_i\}_{i\in [n]} )$:]
         It runs $\A$ on input $ \{ m_i \}_{i\in [n]} $ to obtain $\{ \sigma_i \}_{i\in [n]} $.
         It then parses 
         $\sigma_i=(t_i,\CS.\sigma_i)$ for $i\in [n]$ 
         and outputs $\{ t_i\}_{i\in [n]}$.  
    \end{description}
     \item 
     The challenger parses
     $\qP$ as a pair of unitary $U_\qP$ and a quantum state $\rho_\qP$.
     Then, it runs the extraction algorithm as 
         $\WPRF.\qExt( \{\WPRF.\extk_i\}_{i\in [n]} , (U_\qP, \rho_\qP ) ) \to \{x'_i\}_{i\in [n]}$.
    \item \label{Step:PRF-step-9prime}
    If $x_i = x'_i$ holds for all $i\in [n]$, 
    the challenger sets the output of the game to be $1$. 
    Otherwise, it sets it to be $0$.
\end{enumerate}

Let us construct an adversary $\qB$ against the parallel prediction experiment of $\WPRF$
by running $\qA$ internally.
$\qB$ samples the parameters related to $\CS$ by itself and simulates the computation of $\WPRF.\Eval(\WPRF.\msk_i,\cdot)$ necessary for answering the signing queries from $\qA$ by accessing its evaluation oracles. 
$\qB$ then predicts the PRF values by constructing $\qP$ from $\qA$ as above and running it on the challenge input. 
It is straightforward to see that the advantage of $\qB$ against parallel predictability game of $\WPRF$ is the same as $\Pr[\event_7]$.
It is also easy to see that the advantage of $\qB$ against the parallel mark extraction experiment equals $\Pr[\event_8]$.
Therefore, by the parallel mark extractability of $\WPRF$, we have $\Pr[\event_{8}] \geq \Pr[\event_{7}] - \negl(\secp)$.
\end{description}
From the above discussion, we can see that $\Pr[\event_8] \geq \Pr[\event_0] -\negl(\secp)$ holds. 
Furthermore, one can see by inspection that $\hyb_8$ corresponds to the parallel extraction experiment $\advb{\WDS,\qA,\qExt,n}{par}{ext}(\secp)$ for $\qA$ with respect to $\qExt$ we defined.
The theorem therefore follows.
\end{proof}

% !TEX root = main.tex

\section{Special Dual-Mode Secure Function Evaluation}\label{sec:dual_mode_SFE}

In this section, we introduce special dual-mode secure function evaluation (SFE), which is used as a building block of our secure key leasing schemes with classical lessors. We show that special dual-mode SFE exists assuming the LWE assumption.

Below, we define special dual-mode secure SFE. 
%\takashi{I added "special" since we assume more than dual-mode.} 
\begin{definition}[Special Dual-mode SFE]\label{def:dual-mode-SFE}
A special dual-mode SFE scheme $\SFE$ is a tuple of four algorithms $\SFE=\SFE.(\CRSGen, \Receive{1},\allowbreak \Send, \Receive{2})$. 
Below, let $\bin^\lenx$  be the message space of $\SFE$. 
\begin{description}
%\item[$\Setup(1^\secp,1^{\numkey})\ra\msk$:] The setup algorithm takes a security parameter $1^\lambda$ and a collusion bound $1^{\numkey}$, and outputs a master secret key $\msk$.
\item[$\CRSGen(1^\secp,\mode )\ra (\crs,\td)$:] The CRS generation algorithm is a PPT algorithm that takes a security parameter $1^\lambda$ and $\mode\in \{1,2\}$, and outputs a CRS $\crs$ and a trapdoor $\td$. %\takashi{I belive it should also output $\td$, so I added it.}
$\mode=1$ indicates that the system is in the ``extractable mode", while $\mode=2$ indicates ``hiding mode".

\item[$\Receive{1}(\crs,x)\ra(\msg^{(1)},\st)$:] This is a classical PPT algorithm supposed to be run by a receiver to generate the first message of the protocol. 
It takes as input the CRS $\crs$ and an input $x\in \bin^{\lenx}$ and outputs the first message $\msg^{(1)}$ and a secret state $\st$. 
\item[$\Send(\crs,\msg^{(1)}, C )\ra\msg^{(2)}$:] This is a classical PPT algorithm supposed to be run by a sender to generate the second message of the protocol. 
It takes as input the CRS $\crs$, a message $\msg^{(1)}$ from the receiver, 
and a circuit $C$ with input length $\lenx$ and outputs the second message $\msg^{(2)}$.
\item[$\Receive{2}(\crs,x, \st, \msg^{(2)})\ra y$:] This is a classical PPT algorithm supposed to be run by a receiver to derive the  output.  
It takes as input the CRS $\crs$, an input $x\in \bin^{\lenx}$, 
a state $\st$, and a message $\msg^{(2)}$ from the sender, 
and outputs a string $y$. 
\end{description}
We require special dual-mode SFE to satisfy the following properties.
\begin{description}
\item[Evaluation correctness:]For all $x \in \bin^{\lenx}$, $\mode \in \{1,2\}$, and a circuit $C$ with input length $\lenx$, we have
\begin{align}
\Pr\left[
\Receive{2}(\crs,x, \st, \msg^{(2)}) = C(x)
\ \middle |
\begin{array}{ll}
(\crs,\td) \la \CRSGen(1^\secp,\mode ) \\
(\msg^{(1)},\st) \la \Receive{1}(\crs,x)  \\ 
\msg^{(2)} \la \Send(\crs,\msg^{(1)}, C )
\end{array}
\right] 
=1.
\end{align}

\item[Unique state:]
For all 
$\mode\in \{1,2\}$, 
$(\crs,\td) \in \Supp(\CRSGen(1^\secp,\mode))$, $x\in \bin^\lenx$,  and $(\msg^{(1)},\st) \in \Supp(\Receive{1}(\crs,x))$, $\st$ is the unique state that corresponds to 
$\crs$, $x$, and $\msg^{(1)}$,  that is, there is no $\st^\prime\ne \st$ such that $(\msg^{(1)},\st^\prime) \in \Supp(\Receive{1}(\crs,x))$. 
We denote the unique state $\st$ such that  $(\msg^{(1)},\st) \in \Supp(\Receive{1}(\crs,x))$ by $\st_{x\to \msg^{(1)}}$.\footnote{The state also depends on $\crs$, but we omit it from the notation.}
\takashi{I added this; this follows from state recoverability in the hiding mode, but it was unclear in the extraction mode.}

\item[Mode indistinguishability:]
For any QPT algorithm $\cA$, we have 
    \begin{align}
    \abs{
    \Pr[\cA(\crs)=1: (\crs,\td) \la \CRSGen(1^\secp,1) ]-
    \Pr[\cA(\crs)=1: (\crs,\td) \la \CRSGen(1^\secp,2) ]
    }\leq  \negl(\secp).
    \end{align}
\item[Statistical security against malicious senders in the hiding mode:]
For all $\crs$ output by $\CRSGen(1^\secp, 2)$ and for all $x_0, x_1 \in \bin^\lenx$, we require the following statistical indistinguishability:
\[
\msg^{(1)}_0 \approx_s 
\msg^{(1)}_1
\quad
\mbox{where}
\quad 
(\msg^{(1)}_b,\st)\lrun \Receive{1}(\crs,x_b) \quad \mbox{for} \quad b\in \bin.
\]

\item[State recoverability in the hiding mode:]
%We require that the randomness used by $\Receive{1}$ is divided into two parts $\st\in \bin^{\lenst \lenx}$ for some polynomial $u=\poly(\secp)$ and $\ernd$. While the former is output by $\Receive{1}$ as its secret state and reused by $\Receive{2}$, the latter is ephemeral and is discarded.  
%Namely, we have $\Receive{1}(\crs, x; (\st, \ernd)) = (\msg^{(1)},\st)$. \takashi{Is this assumption needed?}
%Furthermore, w  
%\takashi{I omitted the rquirements about the form of randomness.} \shota{ok}
We require that there exists a classical PPT algorithm $\StaRcv$ that takes as input the trapdoor $\SFE.\td$ and the first message of the receiver $\msg^{(1)}$ and outputs a tuple of strings $\{ \alpha_{j}^{b} \}_{j\in [\lenx], b\in \bin }$ or $\bot$, where the latter indicates that the state recovery failed. 
We require that for all $x=x[1]\| \cdots \| x[\lenx ] \in \bin^\lenx$, the following holds: 
\begin{align}
\Pr\left[
\st = \alpha_{1}^{x[1]} \| \alpha_{1}^{x[2]} \| \cdots \| \alpha_{\lenx}^{x[\lenx]}
\ \middle |
\begin{array}{ll}
(\crs,\td) \la \CRSGen(1^\secp,2 ) \\
(\msg^{(1)},\st) \la \Receive{1}(\crs,x)  \\ 
\{ \alpha_{j}^{b} \}_{j\in [\lenx], b\in \bin } \lrun \StaRcv(\td, \msg^{(1)})
\end{array}
\right] 
=1.
\end{align}
That is, we can recover the corresponding $\st$ from $\msg^{(1)}$ and $x$ using $\td$, and moreover, the state is in a specific form, where each block only depends on the bit of $x$ at the corresponding position.  %\takashi{Slightly changed the explanation.}\synote{ok}
\item[Extractability in the extraction mode:]
For classical algorithms $\Extract$ and $\Sim$, let us consider the following experiment formalized by the experiment $\expb{\SFE,\qA, \Extract, \Sim}{ext}{sim}(1^\secp,\coin)$ between an adversary $\qA$ and the challenger:
\begin{enumerate}
            \item  The challenger runs $(\crs, \td)\sample \CRSGen(1^\secp, 1)$ and sends $(\crs, \td)$ to $\qA$. 
            \item $\qA$ sends $\msg^{(1)}$ and a circuit $C$ with input length $\lenx$ to the challenger. 
            \item The challenger runs $\Extract(\td, \msg^{(1)}) \to x$ and 
            it computes 
            \[
            \msg^{(2)} \lrun 
            \begin{cases}
              \Send(\crs, \msg^{(1)}, C), & \text{if } \coin = 0 \\ 
              \Sim(\crs,\msg^{(1)}, 1^{|C|}, C(x) ) & \text{if } \coin = 1, 
            \end{cases}
            \]
            where $|C|$ denotes the size of $C$. 
            
            Then, it gives $\msg^{(2)}$ to $\qA$.
            \takashi{I added $x$ to the input of $\Sim$.}\takashi{I removed $x$ from the input of $\Sim$ and added $\msg^{(1)}$.}
            \item $\qA$ outputs a guess $\coin^\prime$ for $\coin$. The challenger outputs $\coin'$ as the final output of the experiment.
        \end{enumerate}
        We require that there exist PPT algorithms $\Extract$ and $\Sim$ such that 
        for any QPT $\qA$, it holds that
\begin{align}
\advb{\SFE,\qA,\Extract, \Sim}{ext}{sim}(\secp) \seteq \abs{\Pr[\expb{\SFE,\qA,\Extract, \Sim}{ext}{sim} (1^\secp,0) = 1] - \Pr[\expb{\SFE,\qA,\Extract, \Sim}{ext}{sim} (1^\secp,1) = 1] }\leq \negl(\secp).
\end{align} 
%\shota{For $NC1$, the above security can be achieved against unbounded adversary. To go beyond NC1, we should weaken the security notion to be a computational one, because we do not have information theoretically secure GC for P/poly.}

% For all but negligible probability over the choice of $(\crs, \td)\sample \CRSGen(1^\secp, 1)$, for all (possibly maliciously generated) $\SFE.\msg^{(1)}$ and all $C$ with input length $w$ and depth $O(\log \secp)$, 
% the following distributions are statistically close: 
% \[
% \Send(\crs, \msg^{(1)}, C) \approx_s 
% \Sim(\crs, C(x^*)) 
% \]
%
\item[Efficient state superposition:] 
We require that there exists a QPT ``coherent version" $\qReceive{1}$ of $\Receive{1}$ that is given a quantum state $\ket{\psi}=\sum_{x\in \bin^{\lenx}}\alpha_x\ket{x}$ %that is represented by a qubits of length $\lenx$ 
and $\crs$ and works as follows:

It first generates a state that is within negligible trace distance of the following state:

\[
\sum_{x\in \bit^\lenx}\alpha_x\sum_{\msg^{(1)} \in \Supp( \Receive{1}(\crs, x) )} 
\sqrt{ p_{x \to \msg^{(1)}} }
\ket{x} 
\ket{\st_{x\to \msg^{(1)}} }\ket{\msg^{(1)}} 
\]
where $p_{x\to \msg^{(1)}}$ is the probability that $\Receive{1}(\crs, x)$ outputs $\msg^{(1)}$. 
\if0 
It first implements the following efficient isometry up to negligible error: \takashi{I added a negligible error since we cannot perfectly generate the Gaussian superposition.} 
\[
\ket{x} \mapsto 
\sum_{\msg^{(1)} \in \Supp( \Receive{1}(\crs, x) )} 
\sqrt{ p_{x \to \msg^{(1)}} }
\ket{x} 
\ket{\st_{x\to \msg^{(1)}} }\ket{\msg^{(1)}} 
~
\mbox{for all $x\in \bin^\lenx$},
\]
where $p_{x\to \msg^{(1)}}$ is the probability that $\Receive{1}(\crs, x)$ outputs $\msg^{(1)}$.  %and $\st_{x\to \msg^{(1)}}$ is the unique state such that there exists $\ernd$satisfying $(\msg^{(1)},\st) \la \Receive{1}(\crs,x; (\st, \ernd) ) $.
%
\fi
Then, it measures the last register to obtain $\msg^{(1)}$ and outputs the residual quantum state $\ket{\psi'}$.

We remark that implementing the above by simply running the classical algorithm in the superposition does not work, since it may leave the entanglement with the randomness for $\Receive{1}$.  
\end{description}
\end{definition} 

We observe that a combination of the dual-mode oblivious transfer of \cite{C:PeiVaiWat08} and garbled circuits gives special dual-mode SFE satisfying the above definition. 
\begin{theorem}\label{thm:dual-mode-SFE-from-LWE}
If the LWE assumption holds, then there exists a special dual-mode SFE scheme.
\end{theorem}
See \Cref{sec:dual-mode-SFE-from-LWE} for the proof of \Cref{thm:dual-mode-SFE-from-LWE}.

% !TEX root = main.tex

\section{Definitions of Secure Key Leasing with Classical Lessor}\label{sec-def-SKL}
Here, we define secure key-leasing cryptographic primitives with a classical lessor. 
Specifically, we consider the cases of PKE, PRF, and digital signatures. 
Our definitions closely follow those of \cite{kitagawa2025simple}, but with two key differences. 
First, in our setting, quantum keys are generated through classical communication between the lessor and the lessee, rather than being generated directly by a quantum lessor. 
Second, we strengthen the security notions: in our definitions, the adversary is additionally granted access to an evaluation oracle (for PRFs) or a signing oracle (for signatures), in contrast to \cite{kitagawa2025simple}.

%\fuyuki{The references should be \cite{kitagawa2025simple}?}
%\synote{Thanks, yes. Fixed.}

%\takashi{Below, I put the definitions of SKL with a quantum lessor, taken from \cite{EC:KitMorYam25}. They should be adapted to the classical lessor setting.}

\subsection{Public Key Encryption with Secure Key Leasing} \label{sec:pke-skl-defs}

In this subsection, we define public key encryption with secure key leasing (PKE-SKL), where the lessor is entirely classical and thus all the communication is classical as well. The lessee operates quantum computations only internally.

\begin{definition}[PKE-SKL with classical lessor]\label{def:syntax-pke-skl}
A PKE-SKL scheme $\PKESKL$ with classical lessor is a tuple of seven algorithms $(\Setup,\qIntKeyGen, \Enc, \Dec, \qDec,\qDel,\DelVrfy)$. 
Below, let $\cM$  be the message space of $\PKESKL$. 
\begin{description}
%\item[$\Setup(1^\secp,1^{\numkey})\ra\msk$:] The setup algorithm takes a security parameter $1^\lambda$ and a collusion bound $1^{\numkey}$, and outputs a master secret key $\msk$.
\item[$\Setup(1^\secp) \ra (\ek,\msk,\dvk)$:]
The setup algorithm is a classical PPT algorithm that takes as input the security parameter $1^\secp$ and outputs an encryption key $\ek$, a master secret key $\msk$, and a deletion verification key $\dvk$. 

\item[$\qIntKeyGen(\ek, \msk)\ra (\tau,\qdk) \mbox{ or } (\tau,\bot)$:] The interactive key generation algorithm is an interactive protocol with classical communication run between a classical PPT lessor on input $\msk$ and QPT lessee on input $\ek$. 
This either outputs a decryption key $\qdk$ or $\bot$, which indicates that the protocol failed because the lessee cheated. 
After the protocol, the lessee obtains $\qdk$. We write $\tau$ to denote the transcript of the execution of $\qIntKeyGen$.

\item[$\Enc(\ek,m)\ra\ct$:] The encryption algorithm is a PPT algorithm that takes an encryption key $\ek$ and a message $m \in \cM$, and outputs a ciphertext $\ct$.

\item[$\Dec(\msk,\ct)\ra\tilde{m}$:] 
 The lessor's decryption algorithm is a PPT algorithm that takes a master secret key $\msk$ and a ciphertext $\ct$, and outputs a value $\tilde{m}$.
 
\item[$\qDec(\qdk,\ct)\ra\tilde{m}$:] The lessee's decryption algorithm is a QPT algorithm that takes a decryption key $\qdk$ and a ciphertext $\ct$, and outputs a value $\tilde{m}$.

\item[$\qDel(\qdk)\ra\cert$:] The deletion algorithm is a QPT algorithm that takes a decryption key $\qdk$, and outputs a classical string $\cert$.

\item[$\DelVrfy(\tau,\dvk,\cert)\ra\top/\bot$:] The deletion verification algorithm is a deterministic classical polynomial-time algorithm that takes a transcript $\tau$ for the interactive key generation, a deletion verification key $\dvk$, and a deletion certificate $\cert$, and outputs $\top$ or $\bot$.

\item[Decryption correctness:]For every $m \in \cM$, we have
\begin{align}
\Pr\left[
\begin{array}{ll}
\qDec(\qdk, \ct) \allowbreak = m 
\\
\qquad \quad \wedge 
\\
\Dec(\msk, \ct) \allowbreak = m 
\end{array}
\ \middle |
\begin{array}{ll}
 (\ek,\msk,\dvk) \la \Setup(1^\secp)\\ 
(\tau,\qdk)\gets\qIntKeyGen(\ek, \msk )\\
\ct\gets\Enc(\ek,m)
\end{array}
\right] 
=1-\negl(\secp).
\end{align}

\item[Deletion verification correctness:] We have 
\begin{align}
\Pr\left[
\DelVrfy(\tau,\dvk,\cert)=\top
\ \middle |
\begin{array}{ll}
 (\ek,\msk,\dvk) \la \Setup(1^\secp)\\ 
(\tau,\qdk)\gets\qIntKeyGen(\ek, \msk )\\
\cert\gets\qDel(\qdk)
\end{array}
\right] 
=1-\negl(\secp).
\end{align}
\end{description}
\end{definition}
\begin{remark}[Syntax]\label{rem:pke-skl-syntax}
    One could consider the syntax where $\Setup$ and $\qIntKeyGen$ are merged.
    In such a setting, $\Setup$ is run by the lessor and $\ek$ is then sent to the lessee as the first message. 
    We chose the above syntax as the main one because it allows the encryption key to be generated even without the interaction with the lessee, which seems to be more aligned with the concept of the secure key leasing. 
\end{remark}

\begin{remark}[On the number of rounds in key generation]\label{rem:number-of-rounds-PKE-keygen}
For a PKE-SKL as per the above syntax, at least two rounds of communication are necessary for $\qIntKeyGen$. Otherwise, it becomes possible for an adversary to copy the quantum decryption key $\qdk$ and compromise the security (as defined below). Our construction consists of two rounds and is thus round-optimal in this sense.
\shota{Please check the discussion here.}
\end{remark}

\begin{remark}
We can assume without loss of generality that a decryption key of a PKE-SKL scheme is reusable. This means that it can be reused to decrypt arbitrarily many polynomial number of ciphertexts. 
This is because the decryption result of $\ct$ using $\qdk$ is almost deterministic by the decryption correctness and thus measuring the decryption result does not disturb the input state by the gentle measurement lemma~\cite{Winter99}.
%
% In particular, we can assume that 
% for honestly generated $\ct$ and $\qdk$, if we decrypt $\ct$ by using $\qdk$, the state of the decryption key after the decryption is negligibly close to that before the decryption in terms of trace distance. 
% This is because the output of the decryption is almost deterministic by decryption correctness, and thus such an operation can be done without almost disturbing the input state by the gentle measurement lemma~\cite{Winter99}.  
\end{remark}

We next introduce the security notions for PKE-SKL with classical lessor.
Similarly to \cite{kitagawa2025simple}, we allow the adversary to obtain the verification key after submitting a valid certificate that passes the verification, which is stronger than many existing security definitions of PKE-SKL~\cite{EC:AKNYY23} (or key-revocable PKE~\cite{TCC:AnaPorVai23}). 
\synote{Please check the statement.}

\begin{definition}[IND-VRA security]\label{def:IND-VRA_PKESKL} 
We say that a PKE-SKL scheme $\PKESKL$  with classical lessor for the message space $\cM$ is IND-VRA secure,\footnote{"VRA" stands for "\textbf{V}erification key \textbf{R}evealing \textbf{A}ttack"} if it satisfies the following requirement, formalized by the experiment $\expb{\PKESKL,\qA}{ind}{vra}(1^\secp,\coin)$ between an adversary $\qA$ and the challenger:
        \begin{enumerate}
        \item The challenger runs $\Setup(1^\secp) \ra (\ek,\msk,\dvk)$ and sends $\ek$ to $\qA$.
            \item  Then, the challenger 
            plays the role of the lessor on input $\msk$ and 
            runs $\qIntKeyGen$ with the adversary $\qA$. 
            If the protocol outputs $\bot$, the challenger outputs $0$ as the output of the game. Let $\tau$ be the transcript of the execution of $\qIntKeyGen$.
            \item $\qA$ sends $\cert$ and $(m_0^*,m_1^*)\in \cM^2$ to the challenger. If $\DelVrfy(\tau,\dvk,\cert)=\bot$, the challenger outputs $0$ as the final output of this experiment. Otherwise, the challenger generates $\ct^*\la\Enc(\ek,m_\coin^*)$, and sends $\dvk$ and $\ct^*$ to $\qA$.
            \item $\qA$ outputs a guess $\coin^\prime$ for $\coin$. The challenger outputs $\coin'$ as the final output of the experiment.
        \end{enumerate}
        For any QPT $\qA$, it holds that
\begin{align}
\advb{\PKESKL,\qA}{ind}{vra}(\secp) \seteq \abs{\Pr[\expb{\PKESKL,\qA}{ind}{vra} (1^\secp,0) = 1] - \Pr[\expb{\PKESKL,\qA}{ind}{vra} (1^\secp,1) = 1] }\leq \negl(\secp).
\end{align} 
\end{definition}

We also define the one-way variant of the above security. 
\begin{definition}[OW-VRA security]\label{def:OW-VRA_PKESKL}
We say that a PKE-SKL scheme  $\PKESKL$ with classical lessor for the message space $\cM$ is OW-VRA secure, if it satisfies the following requirement, formalized by the experiment $\expb{\PKESKL,\qA}{ow}{vra}(1^\secp)$ between an adversary $\qA$ and the challenger:
        \begin{enumerate}
        \item The challenger runs $\Setup(1^\secp) \ra (\ek,\msk,\dvk)$ and sends $\ek$ to $\qA$.
            \item  Then, the challenger 
            plays the role of the lessor on input $\msk$ and 
            runs $\qIntKeyGen$ with the adversary $\qA$. 
            If the protocol outputs $\bot$, the challenger outputs $0$ as the output of the game.  Let $\tau$ be the transcript of the execution of $\qIntKeyGen$.           
            \item $\qA$ sends $\cert$ to the challenger. If $\DelVrfy(\tau,\dvk,\cert)=\bot$, the challenger outputs $0$ as the final output of this experiment.
            Otherwise, the challenger chooses $m^*\gets \cM$, generates $\ct^*\la\Enc(\ek,m^*)$, and sends $\dvk$ and $\ct^*$ to $\qA$.
            \item $\qA$ outputs $m^\prime$. The challenger outputs $1$ if $m^\prime=m^*$ and otherwise outputs $0$ as the final output of the experiment.
        \end{enumerate}
        For any QPT $\qA$, it holds that
\begin{align}
\advb{\PKESKL,\qA}{ow}{vra}(\secp) \seteq \Pr[\expb{\PKESKL,\qA}{ow}{vra} (1^\secp) = 1]\leq \negl(\secp).
\end{align}
\end{definition}
\begin{remark}[Security as plain PKE]
It is straightforward to see that IND-VRA (resp. OW-VRA) security implies IND-CPA (resp. OW-CPA) security as a plain PKE scheme with quantum decryption keys.
\end{remark}

By the quantum Goldreich-Levin lemma~\cite{AC02,C:CLLZ21} we have the following lemma. The proof is almost identical to that of \cite[Lemma 3.12]{EC:AKNYY23}, and thus we omit it.
\begin{lemma}\label{lem:ow-ind}
If there exists a OW-VRA secure PKE-SKL scheme with classical lessor, then there exists an IND-VRA secure PKE-SKL scheme with classical lessor.
\end{lemma}

As we discuss in \cref{rem:pke-skl-syntax}, we can consider an alternative syntax where $\Setup$ is merged into $\qIntKeyGen$ and $\ek$, $\msk$, $\dvk$, and $\qdk$ are jointly generated through the interaction between the lessee and lessor.
If we consider such an alternative syntax, our main construction in \cref{sec:PKE-SKL-with-Interaction} requires $3$ rounds of interaction, which is not optimal.
However, as we show in \cref{sec:non-interactive-pke-skl}, we can achieve optimal $2$ rounds of interaction even in this setting as well by adding a small modification to our main construction.  
To capture the scheme, we define the modified version of the syntax for PKE-SKL with classical lessor and \emph{non-interactive key generation}.
\begin{definition}[PKE-SKL with classical lessor and non-interactive key generation]\label{def-pke-skl-NI}
We define PKE-SKL with classical lessor with non-interactive key generation by replacing $\Setup$ and $\qIntKeyGen$ in PKE-SKL as per \cref{def:syntax-pke-skl} with the following 2 algorithms:
\begin{description}
    \item[$\Setup_{\mathsf{NI}}(1^\secp)\to (\pp, \msk, \dvk)$:] The lessor runs the setup algorithm on input $1^\secp$ and outputs the public parameter $\pp$, master secret key $\msk$, and the deletion verification key $\dvk$. 
    \item[$\qKG_{\mathsf{NI}}(\pp) \to (\ek, \qdk)$:] 
    The key generation algorithm is a QPT algorithm that takes as input a public parameter $\pp$ and outputs an encryption key $\ek$ and a decryption key $\qdk$. 
    We assume that $\pp$ is included in $\ek$.
\end{description}
In addition, $\DelVrfy$ takes $\ek$ as input instead of the transcript $\tau$ of the interactive key generation.

One can define decryption correctness and deletion verification correctness similarly to \cref{def:syntax-pke-skl} by replacing $\Setup$ and $\qIntKeyGen$ with the above two algorithms. 
The IND-VRA and OW-VRA security are defined similarly to \cref{def:IND-VRA_PKESKL} and \ref{def:OW-VRA_PKESKL}, where the adversary receives $\pp$ from the challenger at the outset of the game and then sends $\ek$ to the challenger.
The rest of the games are the same. 
\end{definition}
In the above modified syntax, one can consider $\Setup_{\mathsf{NI}}$ combined with $\qKG_{\mathsf{NI}}(\pp)$ as an interactive key generation protocol with 2 rounds of communication.
We say the key generation of the scheme is non-interactive, since the lessor does not have to do anything after running $\Setup_{\mathsf{NI}}(1^\secp)$ and publicizing $\pp$.
The drawback of PKE-SKL as per \cref{def-pke-skl-NI} is that it requires interaction between the lessor and lessee to generate the encryption key $\ek$, unlike \cref{def:syntax-pke-skl}.   

An analogue of \cref{lem:ow-ind} in the non-interactive key generation setting holds. We therefore can focus on the construction of OW-VRA secure scheme. 
We omit the proof for the same reason as \cref{lem:ow-ind}.
\begin{lemma}\label{lem:ow-ind-non-interactive}
If there exists a OW-VRA secure PKE-SKL scheme with classical lessor and non-interactive key generation, then there exists an IND-VRA secure PKE-SKL scheme with classical lessor and non-interactive key generation.
\end{lemma}

\subsection{Pseudorandom and Unpredictable Functions with Secure Key Leasing}\label{def:PRF-SKL}
In this subsection, we define pseudorandom functions with secure key leasing (PRF-SKL) with classical lessor. 
Similarly to the case of PKE-SKL, everything is classical except that the internal computation by the lessee is quantum.

\begin{definition}[PRF-SKL with classical lessor]
A PRF-SKL scheme  $\PRFSKL$ with classical lessor for the domain $\Domprf$ and the range $\Ranprf$ is a tuple of six algorithms $(\Setup, \qIntKeyGen, \Eval, \qLEval,\qDel,\DelVrfy)$. 
\begin{description}
\item[$\Setup(1^\secp)\ra(\pp,\msk, \dvk)$:] The setup algorithm takes a security parameter $1^\lambda$ and outputs a public parameter $\pp$, master secret key $\msk$, and a deletion verification key $\dvk$.

\item[$\qIntKeyGen(  \pp, \msk )\ra (\tau,\qsk) \mbox{ or } (\tau,\bot)$:] The interactive key generation algorithm is an interactive protocol with classical communication run between a classical PPT lessor who takes as input $\msk$ and QPT lessee on input $\pp$. 
This either outputs a secret key $\qsk$ or $\bot$, which indicates that the protocol failed because the lessee cheats. 
After the protocol, the lessee obtains $\qsk$. 
We write $\tau$ to denote the transcript of the execution of $\qIntKeyGen$.

\item[$\Eval(\msk,\prfinp)\ra \prfout$:] The evaluation algorithm is a deterministic classical polynomial-time algorithm that takes a master secret key $\msk$ and an input $\prfinp \in \Domprf$, and outputs a value $\prfout$.

\item[$\qLEval(\qsk,\prfinp)\ra \prfout$:] The leased evaluation algorithm is a QPT algorithm that takes a secret key $\qsk$ and an input $\prfinp \in \Domprf$, and outputs a value $\prfout$.

\item[$\qDel(\qsk)\ra\cert$:] The deletion algorithm is a QPT algorithm that takes a secret key $\qsk$, and outputs a classical string $\cert$.

\item[$\DelVrfy(\tau,\dvk,\cert)\ra\top/\bot$:] The deletion verification algorithm is a deterministic classical polynomial-time algorithm that takes a transcript $\tau$ of the interactive key generation, a deletion verification key $\dvk$, and a deletion certificate $\cert$, and outputs $\top$ or $\bot$.

\item[Evaluation correctness:]For every $\prfinp \in \Domprf$, we have
\begin{align}
\Pr\left[
\qLEval(\qsk,\prfinp) \allowbreak = \Eval(\msk,\prfinp)
\ \middle |
\begin{array}{ll}
(\pp,\msk, \dvk) \lrun \Setup(1^\secp) \\
(\tau,\qsk)\gets\qIntKeyGen(\pp,\msk)
\end{array}
\right] 
=1-\negl(\secp).
\end{align}

\item[Deletion verification correctness:] We have 
\begin{align}
\Pr\left[
\DelVrfy(\tau,\dvk,\cert)=\top
\ \middle |
\begin{array}{ll}
(\pp,\msk,\dvk) \la \Setup(1^\secp) \\
(\tau,\qsk)\gets\qIntKeyGen(\pp,\msk)\\
\cert\gets\qDel(\qsk)
\end{array}
\right] 
=1-\negl(\secp).
\end{align}
\end{description}
\end{definition}
\begin{remark}[Syntax]
    Similarly to the case of PKE-SKL, one could consider the syntax where $\Setup$ and $\qIntKeyGen$ are merged.
    However, we chose the above syntax because it allows the PRF key to be generated without the interaction with the lessee. 
\end{remark}

\begin{remark}[On the number of rounds in key generation]
Similarly to the case of PKE-SKL, at least two rounds of communication are necessary for $\qIntKeyGen$. Otherwise, it cannot be secure as an adversary can generate multiple copies of the quantum secret key.
Our construction only requires 2 rounds and thus is round-optimal in this sense.
\shota{Please check the discussion here.}
\end{remark}

\begin{remark}[Reusability]\label{rem:reusability＿PRF}
We can assume without loss of generality that a key of a PRF-SKL scheme is reusable, i.e., it can be reused to evaluate on  (polynomially) many inputs. This is because the output of the leased evaluation algorithm is almost unique by evaluation correctness, and thus such an operation can be done without almost disturbing the input state by the gentle measurement lemma~\cite{Winter99}.    
\end{remark}

%Our security definition is similar to that in \cite{TCC:AnaPorVai23} except that we allow the adversary to receive the verification key after submitting a valid certificate.  
\begin{definition}[PR-VRA security]\label{def:PR-VRA}
We say that a PRF-SKL scheme  $\PRFSKL$ with classical lessor for the domain $\Domprf$ and the range $\Ranprf$ is PR-VRA secure, if it satisfies the following requirement, formalized by the experiment $\expb{\PRFSKL,\qA}{pr}{vra}(1^\secp,\coin)$ between an adversary $\qA$ and the challenger:
        \begin{enumerate}
        \item 
        The challenger runs $\Setup(1^\secp)\ra(\pp,\msk, \dvk)$ and sends $\pp$ to $\qA$.
        \item 
        Given $\pp$ from the challenger, $\qA$ starts to run.
        Throughout the game, $\qA$ has an oracle access to the evaluation oracle $\Eval(\msk, \cdot)$, which takes as input $s\in \Domprf$ and returns $\Eval(\msk, s)$, until it receives the challenge input at Step~\ref{Step:receiving-challenge-PRF-VRA} below.
        $\qA$ can access the oracle arbitrarily many times at any time.
        
            \item  The challenger plays the role of the lessor and 
            runs $\qIntKeyGen(\pp, \msk)$ with $\qA$. If the output is $\bot$,
            the challenger outputs $0$ as the output of the game. Let $\tau$ be the transcript of the execution of $\qIntKeyGen$.

            \item 
            $\qA$ continues to have an oracle access to the evaluation oracle.
            At some point, 
            $\qA$ sends $\cert$ to the challenger. 
            
            \item \label{Step:receiving-challenge-PRF-VRA}
            If $\DelVrfy(\tau,\dvk,\cert)=\bot$, the challenger outputs $0$ as the final output of this experiment. Otherwise, the challenger generates $\prfinp^*\la\Domprf$, $\prfout^*_0\la\Eval(\msk,\prfinp^*)$, and $\prfout^*_1\la\Ranprf$, and sends $(\dvk,\prfinp^*,\prfout^*_\coin)$ to $\qA$.
            After receiving the challenge input $\prfinp^*$, $\qA$ loses its access to the evaluation oracle.
            \item $\qA$ outputs a guess $\coin^\prime$ for $\coin$. The challenger outputs $\coin'$ as the final output of the experiment.
        \end{enumerate}
        For any QPT $\qA$, it holds that
\begin{align}
\advb{\PRFSKL,\qA}{pr}{vra}(\secp) \seteq \abs{\Pr[\expb{\PRFSKL,\qA}{pr}{vra} (1^\secp,0) = 1] - \Pr[\expb{\PRFSKL,\qA}{pr}{vra} (1^\secp,1) = 1] }\leq \negl(\secp).
\end{align} 
\end{definition}

\begin{remark}
While PR-VRA security defined as above \emph{does} imply security as a weak PRF, it does not imply the security as a standard PRF.
    However, as observed in \cite{kitagawa2025simple}, we can add standard PRF security by a very simple conversion: run a standard PRF scheme in parallel and take the XOR of the outputs. 
    Thus, we focus on PR-VRA secure construction. 
\end{remark}

As an intermediate goal towards constructing PRF-SKL, we introduce a primitive which we call unpredictable functions with secure key leasing (UPF-SKL).\footnote{Though UPF-SKL and PRF-SKL are syntactically identical, we treat them as different primitives for clarity.} 
%\synote{The notion regarding $\UPFSKL$ and $\PRFSKL$ should be consistent .}

\begin{definition}[UPF-SKL]\label{def:UPF-SKL}
A UPF-SKL scheme $\mathsf{UPFSKL}$ with classical lessor has the same syntax as PRF-SKL with classical lessor and satisfies the following security, which we call UP-VRA security, 
formalized by the experiment $\expb{\UPFSKL,\qA}{up}{vra}(1^\secp)$ between an adversary $\qA$ and the challenger:
        \begin{enumerate}
        \item 
        The challenger runs $\Setup(1^\secp)\ra(\pp,\msk, \dvk)$ and sends $\pp$ to $\qA$.
        \item 
        Given $\pp$ from the challenger, $\qA$ starts to run.
        Throughout the game, $\qA$ has an oracle access to the evaluation oracle $\Eval(\msk, \cdot)$, which takes as input $s\in \Domprf$ and returns $\Eval(\msk, s)$, until it receives the challenge input at Step~\ref{Step:receiving-challenge-upf-VRA} below.
        $\qA$ can access the oracle arbitrarily many times at any time.
        
        \item  
            The challenger plays the role of the lessor and 
            runs $\qIntKeyGen(\pp,\msk)$ with $\qA$. If the output is $\bot$,
            the challenger outputs $0$ as the output of the game. Let $\tau$ be the transcript of the execution of $\qIntKeyGen$.

            \item 
            $\qA$ continues to have an oracle access to the evaluation oracle.
            At some point, 
            $\qA$ sends $\cert$ to the challenger. 
            
            \item \label{Step:receiving-challenge-upf-VRA}
            If $\DelVrfy(\tau,\dvk,\cert)=\bot$, the challenger outputs $0$ as the final output of this experiment. Otherwise, the challenger generates $\prfinp^*\la\Domprf$, and sends $(\dvk,\prfinp^*)$ to $\qA$.
            After receiving the challenge input $\prfinp^*$, $\qA$ loses its access to the evaluation oracle.

            \item $\qA$ outputs a guess $t^\prime$ for the output on $s^*$. The challenger outputs $1$ if $t^\prime=\Eval(\msk,s^*)$ and otherwise outputs $0$. 
        \end{enumerate}
        For any QPT $\qA$, it holds that
\begin{align}
\advb{\UPFSKL,\qA}{up}{vra}(\secp) \seteq \Pr[\expb{\UPFSKL,\qA}{up}{vra} (1^\secp) = 1] \leq \negl(\secp).
\end{align} 
\end{definition}

By the quantum Goldreich-Levin lemma~\cite{AC02,C:CLLZ21} we have the following theorem. 
The proof is almost identical to that of \cite[Lemma 4.14]{kitagawa2025simple}, and thus we omit it.
\begin{lemma}\label{lem:upf-prf}
If there exists a UP-VRA secure UPF-SKL scheme with classical lessor, then there exists a PR-VRA secure PRF-SKL scheme with classical lessor.
\end{lemma}

\subsection{Digital Signatures with Secure Key Leasing}
In this subsection, we define digital signatures with secure key leasing (DS-SKL) with classical lessors. 

\begin{definition}[DS-SKL with classical lessors]
A DS-SKL scheme $\DSSKL$  with classical lessors  
is a tuple of seven algorithms $(\Setup, \qIntKeyGen,\Sign, \qSign, \SigVrfy,\qDel,\DelVrfy)$. 
Below, let $\cM$  be the message space of $\DSSKL$. 
\begin{description}
%\item[$\Setup(1^\secp,1^{\numkey})\ra\msk$:] The setup algorithm takes a security parameter $1^\lambda$ and a collusion bound $1^{\numkey}$, and outputs a master secret key $\msk$.
\item[$\Setup(1^\secp)\ra(\svk,\msk, \dvk)$:] The setup algorithm takes a security parameter $1^\lambda$ and outputs a signature verification key $\svk$, master secret key $\msk$, and a deletion verification key $\dvk$.

\item[$\qIntKeyGen(  \svk, \msk )\ra (\tau,\qsigk) \mbox{ or } (\tau,\bot)$:] The interactive key generation algorithm is an interactive protocol with classical communication run between a classical PPT lessor who takes as input $\msk$ and a QPT lessee on input $\svk$. 
This either outputs a secret key $\qsigk$ or $\bot$, which indicates that the protocol failed because the lessee cheats. 
After the protocol, the lessee obtains $\qsigk$.
We write $\tau$ to denote the transcript of the execution of $\qIntKeyGen$.

\item[$\Sign(\msk, m) \to \sigma$:]
The (lessor's) signing algorithm is a PPT algorithm that takes a master secret key $\msk$ and a message $m \in \cM$, and outputs a signature $\sigma$. 

\item[$\qSign(\qsigk,m)\ra(\qsigk',\sigma)$:] The (lessee's) signing algorithm is a QPT algorithm that takes a signing key $\qsigk$ and a message $m \in \cM$, and outputs a subsequent signing key $\qsigk'$ and a signature $\sigma$.  %\takashi{I modified the syntax to output $\qsigk'$.}

\item[$\SigVrfy(\sigvk,m,\sigma)\ra \top/\bot$:] The signature verification algorithm is a deterministic classical polynomial-time algorithm that takes a signature verification key $\sigvk$, a message $m \in \cM$, and a signature $\sigma$, and outputs $\top$ or $\bot$. 

\item[$\qDel(\qsigk)\ra\cert$:] The deletion algorithm is a QPT algorithm that takes a signing key $\qsigk$, and outputs a deletion certificate $\cert$.

\item[$\DelVrfy(\tau,\dvk,\cert)\ra\top/\bot$:] The deletion verification algorithm is a deterministic classical polynomial-time algorithm that takes 
a transcript $\tau$ of the interactive key generation, 
a deletion  verification key $\dvk$, and a deletion certificate $\cert$, and outputs $\top$ or $\bot$.

\item[Signature verification correctness:]For every $m \in \cM$, we have
\begin{align}
\Pr\left[
\begin{array}{ll}
\SigVrfy(\sigvk, m,\sigma) \allowbreak = \top \\
\qquad \qquad \wedge \\
\SigVrfy(\sigvk, m,\sigma') \allowbreak = \top
\end{array}
\ \middle |
\begin{array}{ll}
(\svk,\msk, \dvk)\la \Setup(1^\secp)\\
(\tau,\qsigk) \la \qIntKeyGen(  \svk, \msk ) \\
(\qsigk',\sigma)\gets\qSign(\qsigk,m) \\
\sigma' \gets\Sign(\msk,m)
\end{array}
\right] 
=1-\negl(\secp).
\end{align}

\item[Deletion verification correctness:] We have 
\begin{align}
\Pr\left[
\DelVrfy(\tau,\dvk,\cert)=\top
\ \middle |
\begin{array}{ll}
(\svk,\msk, \dvk)\la \Setup(1^\secp)\\
(\tau,\qsigk) \la \qIntKeyGen(  \svk, \msk ) \\
\cert\gets\qDel(\qsigk)
\end{array}
\right] 
=1-\negl(\secp).
\end{align}
\item[Reusability with static signing keys:]
Let $\qsigk$ be an honestly generated signing key and $m\in \mathcal{M}$ be any message. 
Suppose that we run $(\qsigk',\sigma)\gets\qSign(\qsigk,m)$. 
Then we have 
\begin{align}
    \|\qsigk-\qsigk'\|_{tr}=\negl(\secp).
\end{align}
\end{description}
\end{definition}

\begin{remark}[Syntax]
Similarly to the case of PRF-SKL, one could consider the syntax where $\Setup$ and $\qIntKeyGen$ are merged.
    However, we chose the above syntax because it allows the signature to be generated without the interaction with the lessee. 
\end{remark}

\begin{remark}[On the number of rounds in key generation]
Similarly to the case of PRF-SKL, at least two rounds of communication are necessary for $\qIntKeyGen$. Otherwise, it cannot be secure as an adversary can generate multiple copies of the quantum signing key.
Our construction only requires 2 rounds and thus is round-optimal in this sense.
\shota{Please check the discussion here.}
\end{remark}

\begin{remark}[Reusability]\label{rem:reusability_DS}
Following \cite{kitagawa2025simple}, our syntax makes the quantum signing algorithm output $\qsigk'$ and requires that it should be close to $\qsigk$.
We call this property reusability with static signing keys. 
This is stronger than the notion of reusability  defined in \cite{TQC:MorPorYam24}, where $\qsigk'$ is not required to be close to $\qsigk$ as long as it can be still used to generate signatures on further messages.
Note that in contrast to PKE-SKL and PRF-SKL scenarios, we cannot take for granted the (even weaker form of) reusability without a loss of generality because the signing algorithm may not produce a unique signature. 
%
% The previous work \cite{TQC:MorPorYam24} considered a weaker definition of reusability, where $\sigk'$ is not required to be close to $\sigk$ as long as it can be still used to generate signatures on further messages.
% %signing keys to be updated whenever signing is done and requires that the updated signing key can be still used to generate signatures on any message. 
% To emphasize the difference from their definition, we call the above property reusability with static signing keys. 
% Unlike the cases of PKE-SKL and PRF-SKL, we cannot assume (even the weaker version of) reusability without loss of generality since signatures generated by the signing algorithm may not be unique. Thus, we explicitly state it as a requirement.
\end{remark}

We next introduce the security notions for DS-SKL with classical lessors.

\if0
\begin{definition}[EUF-CMA Security]\label{def:EUF-CMA}
We say that a DS-SKL scheme with classical lessor $\DSSKL$  with the message space $\cM$ is EUF-CMA secure, if it satisfies the following requirement, formalized from the experiment $\expb{\DSSKL,\qA}{euf}{cma}(1^\secp)$ between an adversary $\qA$ and the challenger:
        \begin{enumerate}
            \item  The challenger runs $(\qsigk,\sigvk,\dvk)\gets\qKG(1^\secp)$ and sends $\sigvk$ to $\qA$. \takashi{Can we also give $\dvk$ to $\qA$?}  
            \item $\qA$ makes arbitrarily many \takashi{classical?} queries to the signing oracle    $\qSign(\qsigk,\cdot)$. \takashi{We may need to remark that $\qsigk$ may change after each querybut only negligibly.}
            \item $\qA$ outputs a message $\msg^*$ that is never queried to the signing oracle and a 
            signature $\sigma^*$. The challenger outputs $1$ if $\SigVrfy(\sigvk,\msg^*,\sigma^*)=\top$ and otherwise outputs $0$.  
        \end{enumerate}
        For any QPT $\qA$, it holds that
\begin{align}
\advb{\DSSKL,\qA}{euf}{cma}(\secp) \seteq \Pr[\expb{\DSSKL,\qA}{euf}{cma} (1^\secp) = 1] \leq \negl(\secp).
\end{align} 
\end{definition}
\fi

\begin{definition}[RUF-VRA security]\label{def:RUF-VRA} 
We say that a DS-SKL scheme  $\DSSKL$ with classical lessor  for the message space $\cM$ is RUF-VRA secure,\footnote{"RUF" stands for "\textbf{R}andom message \textbf{U}n\textbf{F}orgeability".} if it satisfies the following requirement, formalized by the experiment $\expb{\DSSKL,\qA}{ruf}{vra}(1^\secp)$  between an adversary $\qA$ and the challenger:
        \begin{enumerate}
        \item 
        The challenger runs $\Setup(1^\secp)\ra(\svk,\msk, \dvk)$ and sends $\svk$ to $\qA$.
        \item 
        Given $\svk$ from the challenger, $\qA$ starts to run.
        Throughout the game, $\qA$ has an oracle access to the signing oracle $\Sign(\msk, \cdot)$, which takes as input $m \in \cM$ and returns $\sigma \lrun \Sign(\msk, m)$, until it receives the challenge message $m^*$ at Step~\ref{Step:receiving-challenge-sig-VRA} below.
        $\qA$ can access the oracle arbitrary many times at any time.
        \item  
            The challenger plays the role of the lessor and 
            runs $\qIntKeyGen(\svk,\msk)$ with $\qA$. If the output is $\bot$,
            the challenger outputs $0$ as the output of the game. Let $\tau$ be the transcript of the execution of $\qIntKeyGen$.

            \item 
            $\qA$ continues to have an oracle access to the signing oracle.
            At some point, 
            $\qA$ sends $\cert$ to the challenger.

            \item \label{Step:receiving-challenge-sig-VRA}
            If $\DelVrfy(\tau,\dvk,\cert)=\bot$, the challenger outputs $0$ as the final output of this experiment. Otherwise, the challenger chooses $m^*\la\mathcal{M}$, and sends $\dvk$ and $m^*$ to $\qA$.
            After receiving the challenge input $m^*$, $\qA$ loses its access to the signing oracle.

            \item $\qA$ outputs a signature $\sigma^\prime$. The challenger outputs $1$ if $\SigVrfy(\sigvk,m^*,\sigma^\prime)=\top$ and otherwise outputs $0$.  
        \end{enumerate}
        For any QPT $\qA$, it holds that
\begin{align}
\advb{\DSSKL,\qA}{ruf}{vra}(\secp) \seteq \Pr[\expb{\DSSKL,\qA}{ruf}{vra} (1^\secp) = 1] \leq \negl(\secp).
\end{align} 
\end{definition}
\begin{remark}
    We note that RUF-VRA security defined as above does not imply the standard EUF-CMA security as a plain digital signature.
    However, as observed in \cite{kitagawa2025simple}, we can add EUF-CMA security by a very simple conversion: run the EUF-CMA secure signature scheme in parallel, signing the message with both signature schemes, and only accept verification when both signatures are valid. Thus, we focus on constructing RUF-VRA. 
\end{remark}
% !TEX root = main.tex

\section{PKE-SKL with Classical Lessor}
\ifnum\llncs=0
The construction of PKE-SKL with a classical lessor is presented in \cref{sec:PKE-SKL-with-Interaction}, following the syntax in \cref{def:syntax-pke-skl}. 
We further show how this scheme can be modified to support non-interactive key generation, as defined in \cref{def-pke-skl-NI}. 
Both constructions are proven secure under the LWE assumption.
\else
The construction of PKE-SKL with a classical lessor is presented 
\ifnum\cameraready=1
in \cref{sec:PKE-SKL-with-Interaction}. 
\else
in \cref{sec:PKE-SKL-with-Interaction}, following the syntax in \cref{def:syntax-pke-skl}.  \fi
We further show how this scheme can be modified to support 
\ifnum\cameraready=1
non-interactive key generation in the full version. 
\else
non-interactive key generation, as defined in \cref{def-pke-skl-NI}, in \cref{sec:non-interactive-pke-skl}.\fi 
Both constructions are proven secure under the LWE assumption.

\fi

\subsection{Construction with Interactive Key Generation}\label{sec:PKE-SKL-with-Interaction} 
We construct a PKE-SKL scheme with a classical lessor that satisfies OW-VRA security. 
Roughly, it requires that if an adversary submits a valid deletion certificate, then it can no longer break the one-wayness of the scheme even given the deletion verification key.  
As shown by \cite{kitagawa2025simple}, 
OW-VRA security can be easily upgraded into IND-VRA security, its indistinguishability-based counterpart, by the quantum 
\ifnum\cameraready=1
Goldreich-Levin lemma. 
\else
Goldreich-Levin lemma (\cref{lem:ow-ind}). \fi
For the construction, we use the following ingredients:
\begin{itemize}
 \item NTCF generators 
 \ifnum\cameraready=1 
 $\NTCF=\NTCF.(\FuncGen, \StateGen, \Chk, \Invert)$ as defined in the full version. 
 \else
 $\NTCF=\NTCF.(\FuncGen, \StateGen, \Chk, \Invert)$  as defined in \cref{sec:NTCF}. \fi
 It is associated with a deterministic algorithm $\NTCF.\GoodSet$ 
 \ifnum\cameraready=1 (See the full version).
 \else(See \cref{def:NTCF-generator}). \fi 
 We assume that $\NTCF.\Chk$ is a deterministic algorithm. 
 We denote the input space by $\cX$ and assume $\cX \subseteq \bin^\lenx$ for some polynomial $\lenx = \lenx(\secp)$.
 It is known that NTCF generators can be constructed from LWE.
% \takashi{This can be constructed from LWE.} 
 \item 
 Special Dual-mode SFE $\SFE=\SFE.(\CRSGen,\Receive{1}, \Send, \Receive{2})$ that supports circuits with input space $\bin^\lenx$ as defined in 
 \ifnum\cameraready=1 the full version. 
 \else \cref{sec:dual_mode_SFE}. \fi 
 It is associated with an extraction algorithm $\SFE.\Extract$, a simulation algorithm $\SFE.\Sim$, a state recovery algorithm $\SFE.\StaRcv$,
 and a ''coherent version" $\SFE.\qReceive{1}$ of $\SFE.\Receive{1}$, which has the efficient state superposition property. 
 We require that $\SFE.\Receive{2}$ and $\SFE.\Extract$ are deterministic.
% We assume that the input space of $\Receive{1}(\SFE.\crs, \cdot)$ is $\bin^\lenx$. 
 We denote the length of each block $\alpha_{i,j}^b$ output by $\SFE.\StaRcv$ by $\lenst$. 
 As we show in \ifnum\cameraready=1 the full version, \else\cref{sec:dual-mode-SFE-from-LWE},  \fi this can be constructed from LWE.
% \takashi{This can be constructed from  LWE.} 
 \item Watermarkable PKE (with parallel extraction): 
 $\WPKE=\WPKE.(\KG, \Mark, \Enc, \Dec)$ as defined in \ifnum\cameraready=1 the full version.\else \cref{sec:WPKE}. \fi
 It is associated with a parallel mark extraction algorithm $\WPKE.\qExt$. 
 We require that $\WPKE.\Mark$ and $\WPKE.\Dec$ are deterministic algorithms.
 We need the mark space to be $\bin^\lenx$ and denote the message space by $\bin^\lenm$. 
 The length of the decryption key is denoted by $\lenmsg$. 
 As we show in \ifnum\cameraready=1 the full version, \else\cref{sec:WPKE}, \fi this can be constructed from any PKE.

% \takashi{ While this is not a standard security notion for watermarkable PKE, it can be constructed from any PKE by using the technique of \cite{C:GKMWW19}. See \Cref{watermarkable_PKE}. }
\end{itemize}

Let $n=n(\secp)$ be a positive integer such that $n=\omega(\log \secp)$.

Below is the description of a OW-VRA secure scheme. 
The message space of the scheme is $\bin^{2n\lenm}$.
\begin{description}
\item[$\Setup(1^\secp) $:]
On input the security parameter, it proceeds as follows. 
\begin{itemize}
\item 
For $i\in [2n]$, run $\WPKE.\KG(1^\lambda) \to (\WPKE.\ek_i, \WPKE.\msk_i)$.
\item Choose a uniform subset $S\subseteq [2n]$ such that $|S|=n$. 
In the following, we denote $[2n]\backslash S$ by $\overline{S}$.
\item 
For $i\in [2n]$, run $(\NTCF.\pp_i,\NTCF.\td_i) \lrun \NTCF.\FuncGen(1^\lambda, \mode)$ and $(\SFE.\crs_i,\SFE.\td_i) \lrun \SFE.\CRSGen(1^\secp, \mode)$, where 
$
\mode = 
\begin{cases}
  1 & \text{if $i\in S$}\\
  2 & \text{if $i\in \overline{S}$}.
\end{cases}
$

\item Output $\ek = \{\NTCF.\pp_i,\SFE.\crs_i, \WPKE.\ek_i\}_{i\in [2n]}$,  
$\dvk=(S,\{\NTCF.\td_i,\SFE.\td_i\}_{i\in \overline{S} } )$,
and $\msk = \{ \WPKE.\msk_i \}_{i\in [2n]}$. 

\end{itemize}
\item[$\qIntKeyGen(\ek, \msk)$:]
The lessor and lessee run the following interactive protocol to generate a decryption key $\qdk$. 

{\bf Lessee's first operation:}
On input $\ek = \{\NTCF.\pp_i,\SFE.\crs_i, \WPKE.\ek_i\}_{i\in [2n]}$, the lessee runs the following algorithms.
\begin{itemize}
\item For each $i\in [2n]$, run $\NTCF.\StateGen(\NTCF.\pp_i) \rrun (y_i, \ket{\psi_i})$.
\item For each $i\in [2n]$, run $\SFE.\qReceive{1}(\SFE.\crs_i,\ket{\psi_i}) \to (\ket{\psi'_i},\SFE.\msg_i^{(1)})$.

% \item 
% For each $i\in [2n]$, 
% prepare a uniform superposition over randomness for SFE and 
% run the first-message-generation algorithm of the SFE under CRS $\SFE.\crs_i$ in superposition, where the second register of $\psi_i$ is used as a message, and measure the ciphertext $\SFE.\ct_i$.  
% This results in the following state:
%     \begin{itemize}
%     \item If $i\in S$, $\ket{\psi''_i}=\ket{b_i}\ket{x_i}\ket{\st_i}$, where $\st_i$ is the stae corresponding $(x_i,\SFE.\ct_i)$. 
%     \item If $i\notin S$, $\ket{\psi''_i}=\ket{0}\ket{x_{i,0}}\ket{\st_{i,0}}+\ket{1}\ket{x_{i,1}}\ket{\st_{i,1}}$, where $\st_{i,b}$ is the state corresponding to $(x_{i,b},\SFE.\ct_{i})$. 
%     \end{itemize}

\item Send $\{y_i,\SFE.\msg^{(1)}_i\}_{i\in [2n]}$ to the lessor. 
\end{itemize}

{\bf Lessor's operation:}
Given the message $\{y_i,\SFE.\msg^{(1)}_i\}_{i\in [2n]}$ from the lessee, the lessor proceeds as follows.
\begin{itemize}

% \item For $i\in S$, run $\NTCF.\Invert(\NTCF.\td_i, y_i) \to x_i$ and $\SFE.\Extract(\SFE.\td_i, \SFE.\ct_i) \to (x'_i,\st_i)$.
%  Set $\valid =1$ if $x_i = x'_i \neq \bot$ holds for all $i\in S$. Otherwise, set $\valid =0$.\footnote{Even if the lessor finds that the lessee is cheating at this point (i.e., $\valid=0$), it continues with the protocol and defers catching the lessee until the deletion verification algorithm is executed. This point is important when we convert the protocol into a non-interactive one. See \cref{rem:non-interactive-version} for further discussion.
%  } 
\item 
For $i\in [2n]$, define $C[\WPKE.\msk_i,y_i]$ to be a circuit that takes $x$ as input and computes the following:
\begin{align}
C[\WPKE.\msk_i,y_i](x)=
\begin{cases} 
\WPKE.\Mark(\WPKE.\msk_i, x) & \text{if $\NTCF.\Chk(\NTCF.\pp_i, x, y_i) = \top$} \\ 
\bot & \text{if $\NTCF.\Chk(\NTCF.\pp_i, x, y_i) = \bot$}. \\ 
\end{cases}
\end{align}
Without loss of generality, we assume that the size of $C[\WPKE.\msk_i,y_i]$ does not depend on $\WPKE.\msk_i$ and $y_i$, and denote it by $|C|$.

%\footnote{Recall that we require $\WPKE.\Mark$ to be deterministic.} 
\item 
For $i\in [2n]$, run $\SFE.\Send( \SFE.\crs_i, \SFE.\msg^{(1)}_i, C[\WPKE.\msk_i,y_i] ) \rrun \SFE.\msg^{(2)}_i$.
\item Send $\{\SFE.\msg^{(2)}_{i} \}_{i\in [2n]}$ to the lessee.
\end{itemize}

{\bf Lessee's second operation:}
Given the message $\{\SFE.\msg^{(2)}_{i} \}_{i\in [2n]}$, the lessee proceeds as follows.
It uses the secret state $\ket{\psi'_i}$ in the following.
\begin{itemize}
\item Construct the unitary $U[\SFE.\crs_i, \SFE.\msg_i^{(2)}]$ defined by the following map:
\[
\ket{x}\ket{\SFE.\st} \ket{0^{\lenmsg}} \mapsto 
\ket{x}\ket{\SFE.\st} \ket{ \SFE.\Receive{2}(\SFE.\crs_i, x, \SFE.\st, \SFE.\msg_i^{(2)}) }.
\]
%where $\lenmsg$ is the length of the decryption key of $\WPKE$.\footnote{Recall that we assumed that $\SFE.\Receive{2}$ is a deterministic algorithm.}

\item 
For $i\in [2n]$, apply the unitary to compute 
\[
\ket{\phi_i}\defeq U[\SFE.\crs_i, \SFE.\msg_i^{(2)}]\left( \ket{\psi'_i}\ket{0^{\lenmsg}} \right).  
\]

% \item 
% For $i\in [2n]$, run the output derivation algorithm of SFE in superposition to obtain
% $\ket{\psi'_i}$, where:
%     \begin{itemize}
%     \item If $i\in S$, $\ket{\psi'_i}=\ket{b_i}\ket{x_i}\ket{\st_i}\ket{\sk_i(x_i)}$, 
%     \item If $i\notin S$, $\ket{\psi'_i}=\ket{0}\ket{x_{i,0}}\ket{\st_{i,0}}\ket{\sk_i(x_{i,0})}+\ket{1}\ket{x_{i,1}}\ket{\st_{i,1}}\ket{\sk_i(x_{i,1})}$.
%     \end{itemize}
 \end{itemize}

{\bf Final output:}
At the end of the protocol, 
the lessee privately obtains a decryption key 
\[
\qdk=\bigotimes_{i\in[2n]}\ket{\phi_i}\ket{\SFE.\crs_i}
\ket{\SFE.\msg^{(2)}_{i}}.\footnote{$\SFE.\crs_i$ and $\SFE.\msg^{(2)}_{i}$ are not needed for decryption. They are used only for deletion.}\] 

We write $\tau$ to denote the transcript of the execution of $\qIntKeyGen$.

\item[$\Enc(\ek,m):$]
Upon receiving $\ek=\{\WPKE.\ek_i\}_{i\in [2n]}$ and a message $m=m_1\|m_2\|\ldots \|m_{2n} \in \bin^{2n\lenm}$, 
generate $\WPKE.\ct_i \la \WPKE.\Enc(\WPKE.\ek_i,m_i)$ for $i\in[2n]$ and output $\ct=\{\WPKE.\ct_i\}_{i\in [2n]}$.
\item[$\Dec(\msk,\ct):$]
Given $\msk = \{ \WPKE.\msk_i \}_{i\in [2n]}$ and a ciphertext $\ct=\{\WPKE.\ct_i\}_{i\in [2n]}$, the decryption algorithm runs 
$\WPKE.\Dec(\WPKE.\msk_i, \WPKE.\ct_i) =m_i$ for $i\in [2n]$ and outputs $m=m_1\|m_2\|\ldots \|m_{2n}$. 
\item[$\qDec(\qdk,\ct)$:]
Given $\qdk = \bigotimes_{i\in[2n]}\ket{\phi_i}\ket{\SFE.\crs_i}
\ket{\SFE.\msg^{(2)}_{i}}$ and a ciphertext $\ct=\{\WPKE.\ct_i\}_{i\in [2n]}$, the decryption algorithm performs the following for each $i\in [2n]$.
\begin{itemize}
    \item Construct the unitary $V[\WPKE.\ct_i]$ defined by the following map:
    \[
    \ket{x}\ket{\SFE.\st}\ket{\WPKE.\dk}\ket{0^{\lenm}}
    \mapsto 
    \ket{x}\ket{\SFE.\st}\ket{\WPKE.\dk}\ket{ \WPKE.\Dec(\WPKE.\dk,\WPKE.\ct_i) }.
    \]
%    where $\lenm$ is the length of the message space of $\WPKE$.
    \item Compute $V[\WPKE.\ct_i] \left( \ket{\phi_i} \ket{0^{\lenm}} \right)$
     and measure the last $\lenm$ qubits to obtain $m_i$.
%     Let the residual state be $\ket{\phi'_i} $. 
\end{itemize}
It finally outputs $m=m_1\|m_2\|\ldots \|m_{2n}$.
%and $\qdk' = \bigotimes_{i\in[2n]} \ket{\phi'_i} \ket{\SFE.\crs_i}  \ket{\SFE.\msg^{(2)}_{i}}$. 

\item[$\Del(\qdk)$:]
Given $\qdk=\bigotimes_{i\in[2n]}\ket{\phi_i}\ket{\SFE.\crs_i}
\ket{\SFE.\msg^{(2)}_{i}}$, the deletion algorithm performs the following for each $i\in [2n]$.
\begin{itemize}
    \item Measure the corresponding registers to recover the classical strings $\SFE.\crs_i$ and $\SFE.\msg^{(2)}_{i}$.
    \item Compute $ U[\SFE.\crs_i,\SFE.\msg_i^{(2)}]^\dagger \ket{\phi_i} $
    and remove the last $\lenmsg$ qubits to obtain $\ket{\psi''_i}$.
    \item Measure $\ket{\psi''_i}$ in Hadamard basis to obtain $(d_i,c_i)$, where each $d_i\in \bin^{\lenx} $ is obtained by measuring the register corresponding to an element in $\cX$
and each $c_i \in \bin^{\lenst \lenx }$ is obtained by measuring the register corresponding to the SFE state.
\end{itemize}
The final output of the algorithm is given by $\cert = \{d_i, c_i\}_{i \in [2n]}$,

\item[$\DelVrfy(\tau,\dvk,\cert)$:]
Given the 
transcript $\tau$ of $\qIntKeyGen$, 
the deletion verification key $\dvk=\left( S, \{ \NTCF.\td_i ,  \SFE.\td_i\}_{i\notin S} \right)$,  and the deletion certificate $\cert = \{d_i, c_i\}_{i \in [2n]}$, 
the deletion verification algorithm works as follows.
\begin{itemize}
    \item Retrieve $\{y_i,\SFE.\msg^{(1)}_i\}_{i\in \overline{S}}$ from $\tau$. 
    \item For $i\in \overline{S}$, run $\NTCF.\Invert(\NTCF.\td_i, y_i) \to (x_{i,0},x_{i,1})$. 
    If the inversion fails for any of $i$, output $\bot$.
    \item For $i\in \overline{S}$, run $\SFE.\StaRcv(\SFE.\td_i, \SFE.\msg^{(1)}_i) \to \{ \alpha_{i,j}^{b} \}_{j\in [\lenx], b\in \bin } $, where $\alpha_{i,j}^{b} \in \bin^{\lenst }$. 
    \item For $i\in \overline{S}$, compute the vector $d^*_i = (d^*_i[1],\ldots, d^*_i[\lenx]) \in \bin^\lenx$ as 
    \[
    d^*_i[j] \seteq 
    \langle c_{i,j} , (\alpha^0_{i,j} \oplus \alpha^1_{i,j})\rangle ,
    \]
    where $c_i$ is parsed as $c_i = c_{i,1}\| \cdots \| c_{i,\lenx} $, where each block is in $\bin^\lenst$.
    \item For $i\in \overline{S}$, check whether 
    \[
    \langle d_i \oplus d^*_i , (x_{i,0}\oplus x_{i,1}) \rangle
    =0
    ~
    \land
    ~
    \NTCF.\GoodSet(\NTCF.\td_i, y_i, d_i \oplus d^*_i  )=\top
    \]
    holds. 
    If it holds for all $i\in \overline{S}$, it outputs $\top$.
    Otherwise, it outputs $\bot$.
\end{itemize}

% For $i\in S$, recover the corresponding message to $\SFE.\ct_i$ using the trapdoor $\SFE.\td_i$,\footnote{This is possible since the CRSs of SFE on $i\in S$ is generated in the extraction mode.} and checks that it is equal to $x_i$.
% %
% Given $\vk=(x_{i,0},x_{i,1})_{i\notin S}$ and $\cert=(e_i,d_i,c_i)_{i\in [2n]}$,
% for each $i\notin S$, 
% accepts iff
% $$e_i= d_i\cdot(x_{i,0}\oplus x_{i,1})\oplus c_i\cdot(\st_{i,0}\oplus \st_{i,1})$$ and $d_i\ne d^*_i$ for all $i\notin S$, where $d^*_i$ is defined later. 
% %
\end{description}
The following theorem asserts the decryption correctness and the deletion verification correctness of the above scheme.
\begin{theorem}\label{th:pke-correctness}
The above construction satisfies the decryption correctness if 
$\NTCF$ satisfies the state generation correctness, $\SFE$ satisfies evaluation correctness, and $\WPKE$ satisfies the decryption correctness.
Furthermore, the above construction satisfies the deletion verification correctness if $\NTCF$ satisfies the state generation correctness, $\SFE$ satisfies the statistical security against malicious senders,
and $\SFE$ supports state recovery in the hiding mode.
\end{theorem}
\begin{proof}
To show the decryption correctness and deletion correctness, 
we first observe that all the algorithms run the subsystems indexed by $i\in [2n]$, which run in parallel. 
Therefore, it suffices to show that each subsystem succeeds in decryption (i.e., $m'_i = m_i$) and deletion verification (i.e., $(d_i,c_i)$ passes the verification) with overwhelming probability assuming that the lessee follows the protocol honestly.
%We note that this parallelization is required for the security proof and not necessary for the functionality.
Here, we show this for the case of $i\in \overline{S}$. The case of $i\in S$ is easier to handle, since in this case, $\ket{\phi_i}$ collapses to a classical state corresponding to $x_i = \NTCF.\Invert(\NTCF.\td_i, y_i)$ with overwhelming probability and repeating the same discussion as the following in its collapsed form suffices. 
%In the following argument, we fix $i\not\in S$ and drop the subscript $i$ when it is clear. Namely, we denote $y$, $\ket{\psi}$, $x_0$ rather than $y_i$, $\ket{\psi_i}$, and $x_{i,0}$ for example. 
%
%In the following, we drop the subscript $i$ to simplify the notation (e.g., writing $y$ and $\ket{\psi}$ instead of $y_i$ and $\ket{\psi_i}$), noting that the same reasoning applies to all of $i\in [2n]$.

We first analyze the quantum state that the lessee has right after the interactive key generation protocol ends. 
We observe that by the correctness of state generation of $\NTCF$, 
with overwhelming probability, $y_i \in \cY_\pp$.
In addition, by the same property, 
$\ket{\psi_i}$ is negligibly close to $\frac{1}{\sqrt{2}}(\ket{x_{i,0}} + \ket{x_{i,1}}) $ in the trace distance, with overwhelming probability,
where $(x_{i,0}, x_{i,1}) = \NTCF.\Invert(\NTCF.\td_i, y_i)$. 
In the following, we replace $\ket{\psi_i}$ with $\frac{1}{\sqrt{2}}(\ket{x_{i,0}} + \ket{x_{i,1}}) $ and continue the analysis, since this change alters the probability of the decryption/verification succeeds only with negligible amount.  
Then, $\SFE.\qReceive{1}$ is applied to the state.
In the first step, we obtain a state that is negligibly close to the state obtained by applying the isometry defined by the following map to $\frac{1}{\sqrt{2}}(\ket{x_{i,0}} + \ket{x_{i,1}}) $:
\[
\ket{x_{i,b}} 
\mapsto 
\sum_{\SFE.\msg^{(1)}_i \in \Supp( \Receive{1}(\crs_i, x_{i,b}) )} 
\sqrt{ p_{x_{i,b} \to \SFE.\msg_i^{(1)}} }
\ket{x_{i,b}} 
\ket{\SFE.\st_{x_{i,b}\to \SFE.\msg_i^{(1)}} }
\ket{\SFE.\msg_i^{(1)}},
\]
where $p_{x_{i,b} \to \SFE.\msg_i^{(1)}}$ is the probability that $\SFE.\Receive{1}(\SFE.\crs_i, x_{i,b})$ outputs $\SFE.\msg_i^{(1)}$ and $\SFE.\st_{x_{i,b}\to \SFE.\msg^{(1)}_i}$ is the unique state corresponding to $x_{i,b}$ and $\SFE.\msg^{(1)}_i$.
%\synote{Will mention the negligible error.}
%satisfying $(\msg^{(1)},\st) \la \Receive{1}(\crs,x; (\st, \ernd) ) $.
We again ignore this negligible difference and proceed with the analysis, as this difference affects the probability of decryption/verification succeeding only by a negligible amount. 
We then have the following state:
\[
\sum_{z\in \calZ} \sqrt{\frac{p_{i,0,z}}{2} }\ket{x_{i,0}}\ket{\st_{i,0,z}}\ket{z} + 
\sqrt{\frac{p_{i,1,z}}{2} } \ket{x_{i,1}}\ket{\st_{i,1,z}}\ket{z},
\]
where we simplify the notation by denoting $\SFE.\msg^{(1)}_i$ by $z$, 
$\SFE.\st_{x_{i,b}\to \SFE.\msg^{(1)}_i}$ by $\st_{i,b,z}$,
$ p_{x_{i,b}\to \SFE.\msg^{(1)}_i}$ by $p_{i,b,z}$, 
and $\Supp( \Receive{1}(\crs_i, x_{i,0}) )\cup \Supp( \Receive{1}(\crs_i, x_{i,1}) )$ by $\calZ$.
We then measure the last register to obtain $z^*$, where the residual state is 
\[
\ket{\psi'_i} = 
q_{i,0,z^*}\ket{x_{i,0}}\ket{\st_{i,0,z^*}} 
+
q_{i,1,z^*}\ket{x_{i,1}}\ket{\st_{i,1,z^*}}. 
\]
In the above, we define
$q_{i,b,z^*} \seteq \sqrt{ p_{i,b,z^*}/(p_{i,0,z^*} + p_{i,1,z^*} ) }$.
We note that the probability that we obtain $z^*$ by the measurement is $(p_{i,0,z^*}+p_{i,1,z^*})/2$.

We then prove the decryption correctness and the deletion verification correctness separately.
\begin{description}
\item[Decryption correctness:]
We continue the analysis of the state that the lessee obtains. 
After $z^*$ is sent to the lessor, $\SFE.\Send( \SFE.\crs_i, z^*, C[\WPKE.\msk_i,y_i] ) \rrun \SFE.\msg^{(2)}_i$ is sent back to the lessee.
We observe that by the evaluation correctness of $\SFE$, we have 
$\SFE.\Receive{2}(\SFE.\crs_i, x_{i,b}, \st_{i,b,z^*}, \allowbreak \SFE.\msg_i^{(2)}) 
=\WPKE.\dk(x_{i,b})$.
By the definition of $\ket{\phi_i}$, we have
\begin{align}
\ket{\phi_i}
& = 
\sum_{b\in \bin}
q_{i,b,z^*}
\ket{x_{i,b}}\ket{\st_{i,b,z^*}} \ket{ \SFE.\Receive{2}(\SFE.\crs_i, x_{i,b}, \st_{i,b,z^*}, \SFE.\msg_i^{(2)}) }
\nonumber
\\
&= 
\sum_{b\in \bin}
q_{i,b,z^*}
\ket{x_{i,b}}\ket{\st_{i,b,z^*}} \ket{ \WPKE.\dk_i(x_{i,b}) }.
\end{align}
After applying $V[\WPKE.\ct_i]$ to $\ket{\phi_i} \ket{0^{\lenm}}$, we get 
\begin{align}
&
\sum_{b\in \bin}
q_{i,b,z^*}
\ket{x_{i,b}}\ket{\st_{i,b,z^*}} \ket{ \WPKE.\dk_i(x_{i,b}) }
\ket{ 
 \WPKE.\Dec(\WPKE.\dk_i(x_{i,b}),\WPKE.\ct_i)
}
\\
&=
\sum_{b\in \bin}
q_{i,b,z^*}
\ket{x_{i,b}}\ket{\st_{i,b,z^*}} \ket{ \WPKE.\dk_i(x_{i,b}) }
\ket{ 
 m_i
}
\end{align}
by the decryption correctness of $\WPKE$.
We can therefore obtain the decryption result $m_i$ by measuring the last register.

\item[Deletion verification correctness:] 
We first observe that $\ket{\psi''_i} = \ket{\psi'_i}$, since $\ket{\psi''_i}$ is obtained by applying a unitary and then its conjugate to $\ket{\psi'_i}$. 
The rest of the proof involves 3 main steps: (i) We first show that the lessee obtains $(d_i,c_i)$ such that
\begin{align}\label{eq:checking-equality}
& \langle d_i , (x_{i,0} \oplus x_{i,1}) \rangle \oplus 
\langle c_i , (\st_{i,0,z^*} \oplus \st_{i,1,z^*}) \rangle = 0 
\end{align}
with overwhelming probability,
(ii) $\NTCF.\GoodSet(\NTCF.\td_i,y_i, d_i \oplus d^*_i ) = \top$ with overwhelming probability,
and then (iii) \cref{eq:checking-equality} implies $\langle d_i \oplus d^*_i, x_{i,0}\oplus x_{i,1}\rangle=0$. 

We first show (i). By applying the Hadamard gates to $\ket{\psi'_i}$, 
we obtain 
\begin{align}
&
\frac{1}{2^{L/2}} 
\sum_{(d_i,c_i)\in \bin^L }
\left(
q_{i,0,z^*} \, (-1)^{\langle d_i, x_{i,0}\rangle \oplus \langle c_i, \st_{i,0,z^*}\rangle}
+
q_{i,1,z^*} \, (-1)^{\langle d_i, x_{i,1}\rangle \oplus \langle c_i, \st_{i,1,z^*}\rangle}
\right)
\ket{d_i}\ket{c_i}
\\
=
\frac{1}{2^{L/2}} 
&\bigg(\sum_{(d_i,c_i)\in T} (-1)^{\langle d_i, x_{i,0}\rangle \oplus \langle c_i, \st_{i,0,z^*}\rangle} \left( q_{i,0,z^*} + q_{i,1,z^*} \right) \ket{d_i}\ket{c_i} \\ 
&\quad+ 
\sum_{(d_i,c_i)\not\in T} (-1)^{\langle d_i, x_{i,0}\rangle \oplus \langle c_i, \st_{i,0,z^*}\rangle} 
\left( q_{i,0,z^*} - q_{i,1,z^*}\right)\ket{d_i}\ket{c_i}\bigg),
\label{eq:sum-of-T-and-non-T}
\end{align}

where 
$L=\lenst\lenx +\lenx$ is the length of each string $(d_i,c_i)$ and
$T$ is the set of $(d_i,c_i)$ that satisfies \cref{eq:checking-equality}.
Note that since $x_{i,0} \neq x_{i,1}$, we have $|T|=2^{L-1}$.
Therefore, 
the probability that \cref{eq:checking-equality} does not hold,
where the probability is taken over the choice of $z^*$, is:
\begin{align}
\Pr[(d_i,c_i)\not\in T] &= 
\sum_{z^*\in \cZ}
\frac{p_{i,0,z^*}+p_{i,1,z^*}}{2}
\sum_{(d_i,c_i)\in T}
\frac{1}{2^L} (q_{i,0,z^*} - q_{i,1,z^*})^2
\\
&=
\frac{1}{4} \sum_{z^*\in \cZ}
(p_{i,0,z^*}+p_{i,1,z^*}) \cdot \frac{ (\sqrt{p_{i,0,z^*}}-\sqrt{p_{i,1,z^*}})^2 }{p_{i,0,z^*}+p_{i,1,z^*}}
\\
& \leq 
\frac{1}{4}\sum_{z^*\in \cZ}
|p_{i,0,z^*} - p_{i,1,z^*}|
\end{align}
where the second line follows from $|T|=2^{L-1}$ and by the definition of $q_{i,b,z^*}$ and
the third line holds since we have $|a-b|^2 \leq |a^2 - b^2|$ for any $a,b\geq 0$.
Finally, the last quantity is negligible by the statistical security of $\SFE$ against the malicious senders, since the statistical distance between the first output $\SFE.\msg^{(1)}$ of
$\SFE.\Receive{1}(\SFE.\crs_i, x_{i,0})$ and that of $\SFE.\Receive{1}(\SFE.\crs_i, x_{i,1})$
is $\frac{1}{2}\sum_{z^*\in \cZ} |p_{i,0,z^*} - p_{i,1,z^*}|=\negl(\secp)$.

We then prove (ii).
We first fix $y_i$ and $z^*$. These fix the values of $(x_{i,0},x_{i,1})$ and $\{\alpha_{i,j}^b\}_{j\in [\lenx], b\in \bin}$.
We then assume $(d_i,c_i)\in T$, since this happens with overwhelming probability as we saw in (i).
We then fix arbitrary $c_i$. This determines the value of $d^*_i$.
By \cref{eq:sum-of-T-and-non-T}, we can see that $d_i$ is distributed uniformly at random under the condition that $(c_i,d_i)\in T$. 
Furthermore, for any $c_i$, half of $d_i \in \bin^\lenx$ satisfies $(c_i, d_i) \in T$.
We therefore have 
\begin{align}
&
\Pr_{d_i \text{ s.t. } (d_i,c_i)\in T} \left[  \NTCF.\GoodSet(\NTCF.\td_i,y_i, d_i \oplus d^*_i ) = \bot \right] 
\\
= & 
2\cdot \Pr_{d_i \lrun \bin^w }\left[
\NTCF.\GoodSet(\NTCF.\td_i,y_i, d_i \oplus d^*_i ) = \bot
\right] 
\\
= & 
2\cdot \Pr_{d_i \lrun \bin^w }\left[
\NTCF.\GoodSet(\NTCF.\td_i,y_i, d_i ) = \bot 
\right] 
\\
= & \negl(\secp).
\end{align}

We finally prove (iii). 
Recall that by the state recovery property in the hiding mode of $\SFE$, we have 
\[
\st_{i,b,z^*} = \alpha_{i,1}^{x_{i,b}[1]}\| \cdots \|\alpha_{i,\lenx}^{x_{i,b}[\lenx]},
\]
where $x_{i,b}[j]$ is the $j$-th bit of $x_{i,b}$.
We therefore have
\begin{align}
\left\langle c_{i}, \st_{i,0,z^*} \oplus \st_{i,1,z^*} \right\rangle 
&= \sum_{j \in [\lenx]} \left\langle c_{i,j}, \alpha_{i,j}^{x_{i,0}[j]} \oplus \alpha_{i,j}^{x_{i,1}[j]} \right\rangle \\
&= \sum_{j \in [\lenx]} \left\langle c_{i,j}, (x_{i,0}[j] \oplus x_{i,1}[j]) (\alpha_{i,j}^0 \oplus \alpha_{i,j}^1) \right\rangle \\
&= \sum_{j \in [\lenx]} (x_{i,0}[j] \oplus x_{i,1}[j]) d_i^*[j] \\
&= \left\langle d_i^*, x_{i,0} \oplus x_{i,1} \right\rangle,
\end{align}
where the second equation follows because $\alpha_{i,j}^{x_{i,0}[j]} \oplus \alpha_{i,j}^{x_{i,1}[j]}$ equals to $0^\lenst$ if $x_{i,0}[j] = x_{i,1}[j]$
and $\alpha_{i,j}^{0} \oplus \alpha_{i,j}^{1}$ otherwise.
Therefore, \cref{eq:checking-equality} implies $\langle d_i^*\oplus d_i,  x_{i,0} \oplus x_{i,1} \rangle =0$.
\end{description}
This completes the proof of \cref{th:pke-correctness}.
\end{proof}
%
%
% \shota{The following remarks will be moved/removed/edited later.}
% \begin{remark}[On the number of rounds in key generation.]\label{rem:non-interactive-version}
% The above protocol has interactive key generation whereas existing constructions have non-interactive ones. However, this could be made non-interactive if we treat $\{y_i,\sfe.\ct_i\}_{i\in [2n]}$ as a public key and 
% the lessor's second message as part of a ciphertext. We described the protocol in the above way to follow our basic idea. Note that we do not know how to remove the interaction for PRF and signatures. 
% \end{remark}
% \begin{remark}[A simplification using OT]
% In the above construction, we rely on (dual-mode) SFE, which is constructed from (dual-mode) OT and garbled circuit. 
% However, we could actually remove the garbled circuit and directly use OT if we additionally observe that a specific instantiation of $\WPKE$ in \cref{sec:WPKE} has a decomposable watermarked key, i.e., $\WPKE.\dk(x)$ is of the form:
% \[
% \WPKE.\dk(x)=\PKE.\dk_{1,x[1]}\|\PKE.\dk_{2,x[2]}\|\ldots\|\PKE.\dk_{\lenx,x[\lenx]}
% \]
% where $\ell$ is the bit-length of $m$ and 
% $\WPKE.\msk=\{\PKE.\dk_{i,b}\}_{i\in [\lenx],b\in\bit}$.
% In this case, it suffices to send the secret keys of the underlying PKE via OT and the garbled circuit is not necessary.
% While this observation would significantly improve the concrete efficiency, this does not change the assumption, so we use stronger primitive of SFE in the above. 
% \end{remark}

\subsection{Security Proof}
\begin{theorem}\label{thm-proof-PKESKL}
If $\NTCF$ satisfies cut-and-choose adaptive hardcore bit security, 
$\SFE$ satisfies the mode-indistinguishability and the extractability in the extraction mode, and $\WPKE$ satisfies one-wayness and parallel mark extractability, 
then our construction is OW-VRA secure.  
%\shota{Make it more precise.}
\end{theorem}
\begin{proof}
    We prove the theorem by a sequence of hybrids.
    For the sake of the contradiction, let us assume an adversary $\A$ that wins the OW-VRA security game with non-negligible probability $\epsilon$.
    In the following, the event that $\qA$ wins in $\Hyb_{\rm xx}$ is denoted by $\event_{\rm xx}$.

\begin{description}
  \item[$\Hyb_0$:] This is the original OW-VRA security game between the challenger and an adversary $\A$.  
  Namely, the game proceeds as follows.
\begin{enumerate}
    \item \label{Step:pke-proof-setup}
    The challenger computes $\ek$ as follows.
    \begin{itemize}
    \item It samples $(\WPKE.\ek_i, \WPKE.\msk_i)$ for $i\in [2n]$.
        \item It chooses random subset $S$ of $[2n]$ with size $n$.
        \item It runs $(\NTCF.\pp_i,\NTCF.\td_i) \lrun \NTCF.\FuncGen(1^\lambda, \mode)$ and $(\SFE.\crs_i,\SFE.\td_i) \lrun \SFE.\CRSGen(1^\secp, \mode)$ for $i\in [2n]$, where $\mode =1$ if $i\in S$ and $\mode =2$ if $i\in \overline{S}$.
    \end{itemize}
    Finally, it sends $\ek= \{\NTCF.\pp_i,\SFE.\crs_i, \WPKE.\ek_i \}_{i\in [2n]}$ to $\A$.
    \item The adversary $\A$ sends $\{y_i,\SFE.\msg^{(1)}_i\}_{i\in [2n]}$ to the challenger.
    \item \label{inner:pke-second-msg-to-adv}
    The challenger computes the message to the adversary as follows. 
    It first runs 
    \[\SFE.\Send(\SFE.\crs_i, \SFE.\msg^{(1)}_i, C[\WPKE.\msk_i,y_i] ) \rrun \SFE.\msg^{(2)}_i\] 
    for $i\in [2n]$.
    Then, it sends $\{\SFE.\msg^{(2)}_{i} \}_{i\in [2n]}$ to $\A$.
    
    \item Then, $\A$ submits the deletion certificate 
    $\cert = \{d_i, c_i\}_{i \in [2n]} $.

    \item Then, the challenger runs $\DelVrfy$ as follows.
    \label{Step:Del-verify-in-pke-proof}
    \begin{itemize}
%    \item If $\valid = 0$, it outputs $\bot$. 
%   
    \item For $i\in \overline{S}$, it runs $\SFE.\StaRcv(\SFE.\td_i, \SFE.\msg^{(1)}_i) \to \{ \alpha_{i,j}^{b} \}_{j\in [\lenx], b\in \bin } $. 
    \item For $i\in \overline{S}$, it computes 
    $d^*_i = (d^*_i[1],\ldots, d^*_i[\lenx]) \in \bin^\lenx$ as 
    $
    d^*_i[j] \seteq \langle c_{i,j}, (\alpha^0_{i,j} \oplus \alpha^1_{i,j})\rangle
    $.  
    \item For $i\in \overline{S}$, it runs $(x_{i,0},x_{i,1}) \lrun \NTCF.\Invert(\NTCF.\td_i, y_i) $.
    If the inversion fails for any of $i\in \overline{S}$, the output of $\DelVrfy$ is $\bot$.
    \item For $i\in \overline{S}$, it checks whether 
    \begin{equation}\label{eq:PKE-deletion-verify-in-the-proof}
    ( d_i \oplus d^*_i ) \cdot(x_{i,0}\oplus x_{i,1})=0
    ~
    \land
    ~
    \NTCF.\GoodSet(\NTCF.\td_i,y_i, d_i \oplus d^*_i ) = \top
    \end{equation}
    holds using $\NTCF.\td_i$. 
    If \cref{eq:PKE-deletion-verify-in-the-proof} holds for all $i\in \overline{S}$, the output is $\top$ and otherwise, $\bot$.
    \end{itemize}
    If the output of the algorithm is $\bot$, then it indicates that the output of the game is $0$ (i.e., $\A$ failed to win the game). 
    Otherwise, the challenger returns 
    $\dvk=\left( S, \{ \NTCF.\td_i ,  \SFE.\td_i\}_{i\in \overline{S}} \right)$ to $\A$
    and continues the game.
    
    \item \label{Step:PKE-step-6}
    Then, the challenger chooses a random message $m=m_1\|m_2\|\cdots \|m_{2n} \sample \bin^{2n\lenm}$ and computes $\WPKE.\ct_i \la \WPKE.\Enc(\WPKE.\ek_i,m_i)$ for $i\in [2n]$. Then, it sends $\ct=\{\WPKE.\ct_i\}_{i\in [2n]}$ to $\A$.
    \item \label{Step:PKE-step-7}
    $\A$ then outputs $m'=m'_1\| \cdots \|m'_{2n}$.
    The adversary wins the game if $m'=m$.
\end{enumerate}
By the definition, we have $\Pr[\event_{0}]= \epsilon$.
\item[$\Hyb_1$:]  
  In this hybrid, the challenger computes $\SFE.\msg^{(2)}_i$ differently for $i\in S$ in Step \ref{inner:pke-second-msg-to-adv} of the game. Namely, it is replaced with the following: 
\begin{enumerate}[label=\arabic*', start=3]
    \item
    \label{Step:PKE-step-3prime}
    The challenger computes the message to the adversary as follows. 
    \begin{itemize}
    \item It runs 
  $\SFE.\Extract(\SFE.\td_i, \SFE.\msg^{(1)}_i) \to x_i$ 
  and $\NTCF.\Chk(\NTCF.\pp_i, x_i, y_i) = \verdict_i$ for $i\in S$.
        
        \item It computes 
          \[
              \SFE.\msg^{(2)}_i \lrun 
              \begin{cases}
              \SFE.\Sim(\SFE.\crs_i, \SFE.\msg^{(1)}_i,1^{|C|},\gamma_i )
                  & \text{if }  i\in S \\
             \SFE.\Send(\SFE.\crs_i, \SFE.\msg^{(1)}_i, C[\WPKE.\msk_i,y_i] )
                  & \text{if } i\in \overline{S} 
              \end{cases},
        \]       
        where
  \[
  \gamma_i = 
  \begin{cases}
  \WPKE.\Mark(\WPKE.\msk_i, x_i)=
    \WPKE.\dk_i(x_i)  
      & \text{if }  \verdict_i = \top \\
    \bot 
      & \text{if } \verdict_i  = \bot. 
  \end{cases}
  \]
    \end{itemize}
    Finally, it sends $\{\SFE.\msg^{(2)}_{i} \}_{i\in [2n]}$ to $\A$.
\end{enumerate}

%  Note that by the change introduced in the previous hybrid, if the game proceeds to the point where $\SFE.\msg^{(2)}_i$ is computed, we have $x'_i = x_i$ for $i\in S$. 
  Since we have $\gamma_i = C[\WPKE.\msk_i,y_i](x_i)$, it follows that this game is indistinguishable from the previous game by a straightforward reduction to the extractability of $\SFE$ in the extraction mode. 
  Thus, we have $|\Pr[\event_{0}] - \Pr[\event_{1}]| = \negl(\secp)$.
  Note that in this game, $\WPKE.\msk_i$ is not used by the game for $i$ such that $i\in S$ and $\verdict_i = \bot$.
  
   \item[$\Hyb_2$:]  
    In this hybrid, the challenger aborts the game and sets the output of the game to be $0$ (i.e., $\qA$ loses the game), once it turns out that there is $i\in S$ such that $\NTCF.\Chk(\NTCF.\pp_i, x_i, y_i)=\bot$.
    Namely, we replace Step \ref{Step:PKE-step-3prime} in the previous hybrid with the following: 
\begin{enumerate}[label=\arabic*'', start=3]
    \item \label{Step:PKE-step-3''}
    The challenger computes the message to the adversary as follows. 
    \begin{itemize}
    \item It runs 
  $\SFE.\Extract(\SFE.\td_i, \SFE.\msg^{(1)}_i) \to x_i$ 
  and $\NTCF.\Chk(\NTCF.\pp_i, x_i, y_i) = \verdict_i$ for $i\in S$.
  \item If $\verdict_i = \bot$ for any of $i\in S$, it aborts the game and sets the output of the game to be $0$.
        \item It computes 
          \[
              \SFE.\msg^{(2)}_i \lrun 
              \begin{cases}
              \SFE.\Sim(\SFE.\crs_i, \SFE.\msg^{(1)}_i,1^{|C|}, 
              \WPKE.\dk_i(x_i)   )
                  & \text{if }  i\in S \\
             \SFE.\Send(\SFE.\crs_i, \SFE.\msg^{(1)}_i, C[\WPKE.\msk_i,y_i] )
                  & \text{if } i\in \overline{S}. 
              \end{cases}
        \]       
    \end{itemize}
    Finally, it sends $\{\SFE.\msg^{(2)}_{i} \}_{i\in [2n]}$ to $\A$.
\end{enumerate}

We claim $\Pr[\event_{2}]\geq  \Pr[\event_{1}] -\negl(\secp)$.
To show this, we observe that the only scenario where the adversary $\qA$ wins the game in the previous hybrid but not in the current one occurs when $\qA$ chooses $y_i$ such that $\verdict_i =\bot$ for some $i\in S$, yet still succeeds in inverting the ciphertext to get $\{m_j\}_{j\in [2n]}$. 
However, this implies that $\qA$ has inverted $\WPKE.\ct_i$ to obtain $m_i$ without getting any secret information corresponding to $\WPKE.\ek_i$.  
By a straightforward reduction to the one-wayness of $\WPKE$, the claim follows.

  \item[$\Hyb_3$:]  
  In this hybrid, we change the winning condition of the adversary. 
  Namely, we replace Step \ref{Step:PKE-step-6} and Step \ref{Step:PKE-step-7} of the game with the following. 
  \begin{enumerate}[label=\arabic*', start=6]
      \item
      The challenger parses the adversary $\qA$ right before Step~\ref{Step:PKE-step-6} of the game as a unitary circuit $U_\qA$ and a quantum state $\rho_\qA$.
      It then constructs a quantum decryption algorithm $\qD$ that takes as input $\{ \WPKE.\ct_i \}_{i\in S}$ as input and outputs $\{m'_i\}_{i\in S}$ 
      from $\{\WPKE.\ek_i\}_{i\in \overline{S}}$, $S$, and 
      $\A$ as follows.
     \begin{description}
         \item[$\qD( \{ \WPKE.\ct_i \}_{i\in S} )$:]
         It chooses $m_i \gets \cM$ for $i\in \overline{S}$ and computes $\WPKE.\ct_i \la \WPKE.\Enc(\WPKE.\ek_i,m_i)$ for $i\in \overline{S}$.
         Then, it runs $\A$ on input $ \{ \WPKE.\ct_i \}_{i\in [2n]} $ to obtain $m' = m'_1\|\cdots \| m'_{2n}$.
         Finally, it outputs $\{ m'_i\}_{i\in S}$.  
    \end{description}
     \item 
     The challenger parses
     $\qD$ as a pair of unitary $U_\qD$ and a quantum state $\rho_\qD$.
     Then, it runs the extraction algorithm as 
         $\WPKE.\qExt( \{\WPKE.\ek_i\}_{i\in S} , (U_\qD, \rho_\qD ) ) \to \{x'_i\}_{i\in S}$.
    \item \label{Step:PKE-step-8prime}
    If $x_i = x'_i$ holds for all $i\in S$, 
    the challenger sets the output of the game to be $1$. 
    Otherwise, it sets it to be $0$.
 \end{enumerate}
\ifnum\llncs=0
  As we show in \cref{lem:reduction-to-parallel-extractability}, 
  we have $\Pr[\event_{3}] \geq \Pr[\event_{2}]-\negl(\secp)$
  by the parallel extractability of $\WPKE$.
 \else
  As we show in \ifnum\cameraready=1 the full version, \else\cref{sec-omitted-proofs-PKESKL},\fi 
  we have $\Pr[\event_{3}] \geq \Pr[\event_{2}]-\negl(\secp)$
  by the parallel extractability of $\WPKE$.
 
 \fi
\item[$\Hyb_4$:]  
    In this game, we replace Step \ref{Step:PKE-step-8prime} of the above game with the following:
    \begin{enumerate}[label=\arabic*'', start=8]
        \item If $\NTCF.\Chk(\NTCF.\pp_i, x'_i, y_i)=\top$ holds for all $i\in S$, 
    the challenger sets the output of the game to be $1$. 
    Otherwise, it sets it to be $0$. 
    \end{enumerate}
    We observe that $x'_i = x_i$ in Step \ref{Step:PKE-step-8prime} implies $\NTCF.\Chk(\NTCF.\pp_i, x'_i, y_i) = \NTCF.\Chk(\NTCF.\pp_i, x_i, y_i) =\top$.
    To see the latter equality holds, recall that the game is set to abort if $\NTCF.\Chk(\NTCF.\pp_i, x_i, y_i) = \bot$ for any $i \in S$ 
    in Step \ref{Step:PKE-step-3''}. Thus, if the game continues to Step \ref{Step:PKE-step-8prime}, the equality holds.   
    We therefore have $\Pr[\event_{4}] \geq  \Pr[\event_{3}]$.
    (Actually, $\Pr[\event_{4}] = \Pr[\event_{3}]$ holds due to the uniqueness of $x_i$ that passes the check. However, showing the inequality suffices for our proof.)
\item[$\Hyb_5$:]  
    In this game, we switch Step \ref{Step:PKE-step-3''} in the previous hybrid back to Step \ref{Step:PKE-step-3prime}. 
    The difference from the previous hybrid is that it continues the game even if $\verdict_i = \bot$ for some $i\in S$.
    Since the adversary had no chance in winning the game in such a case in the previous hybrid, this change only increases the chance of the adversary winning. 
    We therefore have $\Pr[\event_{5}] \geq \Pr[\event_{4}]$.
  
\item[$\Hyb_6$:]  
    In this hybrid, we switch Step \ref{Step:PKE-step-3prime} in the previous hybrid back to Step \ref{inner:pke-second-msg-to-adv}.
    Namely, it computes $\SFE.\Send( \SFE.\crs_i, \allowbreak\SFE.\msg^{(1)}_i, C[\WPKE.\msk_i,y_i] ) \rrun \SFE.\msg^{(2)}_i$ for all $i\in [2n]$ again.
    Similarly to the change from $\Hyb_0$ to $\Hyb_1$, the winning probability of $\A$ in this game does not change significantly by the extractability of $\SFE$ in the extraction mode. %\shota{The name of the security notion should be fixed.}
    Namely, we have $|\Pr[\event_{6}] - \Pr[\event_{5}]| = \negl(\secp)$.
    Note that we do not need $\{\SFE.\td_i\}_{i\in S}$ any more to run the game.

\item[$\Hyb_7$:] 
    In this hybrid, we change the way $\SFE.\crs_i$ is sampled for $i\in S$.
    Namely, it is sampled as $(\SFE.\crs_i,\SFE.\td_i) \lrun \SFE.\CRSGen(1^\secp, 2)$ for all $i\in [2n]$. 
    Concretely, we replace Step \ref{Step:pke-proof-setup} with the following:
\begin{enumerate}[label=\arabic*']
    \item The challenger computes the first message to the adversary as follows.
    \begin{itemize}
    \item It samples $(\WPKE.\ek_i, \WPKE.\msk_i)$ for $i\in [2n]$.
        \item It chooses random subset $S$ of $[2n]$ with size $n$.
        \item It runs $(\NTCF.\pp_i,\NTCF.\td_i) \lrun \NTCF.\FuncGen(1^\lambda, \mode)$ and $(\SFE.\crs_i,\SFE.\td_i) \lrun \SFE.\CRSGen(1^\secp, 2)$ for $i\in [2n]$, where $\mode =1$ if $i\in S$ and $\mode =2$ if $i\in \overline{S}$.
    \end{itemize}
    Finally, it sends $\ek= \{\NTCF.\pp_i,\SFE.\crs_i, \WPKE.\ek_i\}_{i\in [2n]}$ to $\A$.
\end{enumerate}  
    Since we do not need $\SFE.\td_i$ for $i\in S$ to run $\Hyb_6$ and $\Hyb_7$, a straightforward reduction to the mode indistinguishability of $\SFE$ implies that $\Hyb_6$ and $\Hyb_7$ are computationally indistinguishable. 
    Therefore, we have $|\Pr[\event_{6}] - \Pr[\event_{7}]| = \negl(\secp)$.
    \item[$\Hyb_8$:] 
    In this game, we change the way how the challenger runs $\DelVrfy$.
    In particular, the challenger computes $\SFE.\StaRcv(\SFE.\td_i, \SFE.\msg^{(1)}_i) \to \{ \alpha_{i,j}^{b} \}_{j, b} $ and $d^*_i$ for all $i\in [2n]$, 
    instead of only for $i\in \overline{S}$.  
    Still, the inversion of $y_i$ and 
    the check of \cref{eq:PKE-deletion-verify-in-the-proof} are performed only for $i\in \overline{S}$. 
    Concretely, we replace Step \ref{Step:Del-verify-in-pke-proof} with the following.
\begin{enumerate}[label=\arabic*', start=5]
        \item Then, the challenger runs (the modified version of) $\DelVrfy$ as follows.
    \begin{itemize}
%    \item If $\valid = 0$, it outputs $\bot$. 
%   
    \item For $i\in [2n]$, it runs $\SFE.\StaRcv(\SFE.\td_i, \SFE.\msg^{(1)}_i) \to \{ \alpha_{i,j}^{b} \}_{j\in [\lenx], b\in \bin } $. 
    \item For $i\in [2n]$, it computes 
    $d^*_i = (d^*_i[1],\ldots, d^*_i[\lenx]) \in \bin^\lenx$ as 
    $
    d^*_i[j] \seteq \langle c_{i,j}, (\alpha^0_{i,j} \oplus \alpha^1_{i,j})\rangle
    $.  
    \item For $i\in \overline{S}$, it runs $(x_{i,0},x_{i,1}) \lrun \NTCF.\Invert(\NTCF.\td_i, y_i) $.
    If the inversion fails for any of $i\in \overline{S}$, the output of $\DelVrfy$ is $\bot$.
    \item For $i\in \overline{S}$, it checks whether 
    \begin{equation}\label{eq:PKE-deletion-verify-in-the-proof-2}
    ( d_i \oplus d^*_i ) \cdot(x_{i,0}\oplus x_{i,1})=0
    ~
    \land
    ~
    \NTCF.\GoodSet(\NTCF.\td_i,y_i, d_i \oplus d^*_i ) = \top
    \end{equation}
    holds using $\NTCF.\td_i$. 
    If \cref{eq:PKE-deletion-verify-in-the-proof-2} holds for all $i\in \overline{S}$, the output is $\top$ and otherwise, $\bot$.
    \end{itemize}
    If the output of the algorithm is $\bot$, then it indicates that the output of the game is $0$ (i.e., $\A$ failed to win the game). 
    Otherwise, the challenger returns 
    $\dvk=\left( S, \{ \NTCF.\td_i,  \SFE.\td_i \}_{i\in \overline{S}} \right)$ to $\A$
    and continues the game.
\end{enumerate}
    
    We can see that this change is only conceptual, since the outcome of the verification is unchanged.  
    We may therefore have
    $\Pr[\event_{8}] = \Pr[\event_{7}] $.
\end{description}
\ifnum\llncs=0
    Based on the above discussion, it is established that 
    $\Pr[\event_{8}] \geq \epsilon - \negl(\secp)$, which is non-negligible by our assumption.
    However, as we show in \cref{lem:PKE-reduction-to-ada-HC}, we have $\Pr[\event_{8}]=\negl(\secp)$ assuming cut-and-choose adaptive hardcore bit property of $\NTCF$.
    This results in a contradiction as desired.
\else
    Based on the above discussion, it is established that 
    $\Pr[\event_{8}] \geq \epsilon - \negl(\secp)$, which is non-negligible by our assumption.
    However, as we show in \ifnum\cameraready=1 the full version, \else \cref{sec-omitted-proofs-PKESKL},\fi we have $\Pr[\event_{8}]=\negl(\secp)$ assuming cut-and-choose adaptive hardcore bit property of $\NTCF$.
    This results in a contradiction as desired.
\fi
\end{proof}
\ifnum\llncs=0
It remains to prove \cref{lem:reduction-to-parallel-extractability} and \cref{lem:PKE-reduction-to-ada-HC}.
\ifnum\llncs=0
\else
\section{Omitted Proofs in Proof of \cref{thm-proof-PKESKL}}\label{sec-omitted-proofs-PKESKL}
\fi

\begin{lemma}\label{lem:reduction-to-parallel-extractability}
    If $\WPKE$ satisfies parallel extractability, we have $\Pr[\event_{3}] \geq \Pr[\event_{2}]-\negl(\secp)$. 
\end{lemma}
\begin{proof}
    We consider an adversary $\qB$ who plays the role of the challenger in $\Hyb_2$ for $\A$, while working as an adversary against the parallel one-way inversion experiment with the number of instances $n$. It proceeds as follows: 
    \begin{enumerate}
    \item Given $\{ \WPKE.\ek'_j \}_{j\in [n]}$ from the challenger, $\qB$ proceeds as follows. 
    \begin{itemize}
           \item $\qB$ chooses $S=\{i_1,\ldots, i_n\}\subset [2n]$ and samples $\{\NTCF.\pp_i,\NTCF.\td_i, \SFE.\crs_i, \SFE.\td_i \}_{i\in [2n]}$ as specified in $\Hyb_2$. 
           \item It sets $\WPKE.\ek_{i_j} \seteq \WPKE.\ek'_j$  for $j\in [n]$.
           \item It then samples $(\WPKE.\ek_{i}, \WPKE.\msk_{i}) \lrun \WPKE.\KG(1^\secp)$ for $i\in \overline{S}$.
           
           \item Then, it sends $\ek= \{\NTCF.\pp_i,\SFE.\crs_i, \WPKE.\ek_i \}_{i\in [2n]}$ to $\A$.
           \item Then, the adversary $\A$ sends $\{y_i,\SFE.\msg^{(1)}_i\}_{i\in [2n]}$ to $\qB$.
           \item $\qB$ then computes $x_i$ for $i\in S$ and $\verdict_i$ using the trapdoors for $\SFE$ and $\NTCF$ as in $\Hyb_2$. If $\verdict_i=\bot$ for any of $i\in S$, it aborts. 
        \item It then sets $\{x'_1,\ldots, x'_n\}\seteq \{x_{i_1},\ldots, x_{i_n}\}$ and sends it to its challenger.
       \end{itemize} 
       
        \item Given $\{\WPKE.\dk'_j(x'_{j})\}_{j\in [n]}$ from the challenger, $\qB$ proceeds as follows. 
        \begin{itemize}
        \item It sets $
        \WPKE.\dk_{i_j}(x_{i_j})\seteq \WPKE.\dk'_j(x'_{j})$ for $j\in [n]$.
         
        \item It computes $\SFE.\msg^{(2)}_i$ as follows: 
        \[
        \SFE.\msg^{(2)}_i \lrun 
        \begin{cases}
          \SFE.\Sim(\SFE.\crs_i, \SFE.\msg^{(1)}_i,1^{|C|}, \WPKE.\dk_i(x_i) ) & \text{if } i\in S \\
          \SFE.\Send(\SFE.\crs_i, \SFE.\msg^{(1)}_i, C[\WPKE.\msk_i, y_i] ) & \text{if } i\in \overline{S}.
        \end{cases}
        \]
        \item $\qB$ then sends $\{\SFE.\msg^{(2)}_{i} \}_{i\in [2n]}$ to $\qA$. 
        \item Then, $\A$ submits the deletion certificate $\cert$ to $\qB$.
        $\qB$ runs $\DelVrfy$ using the trapdoors for $\NTCF$ and $\SFE$.
        If the output is $\bot$, $\qB$ also outputs $\bot$ and aborts.
    Otherwise, $\qB$ returns 
    $\dvk=\left( S, \{ \NTCF.\td_i \}_{i\in \overline{S}} \right)$ to $\A$.
        \end{itemize}

        \item Then, the challenger sends $\{\WPKE.\ct'_j\}_{j\in [n]}$ to $\qB$.  
        \item $\qB$ computes its output as follows. 
        \begin{itemize}
        \item It sets $\WPKE.\ct_{i_j} \seteq \WPKE.\ct'_{j}$ for $j\in [n]$.
        \item It samples $m_i\lrun \cM$ and computes $\WPKE.\ct_i \lrun \WPKE.\Enc(\WPKE.\ek_i, m_i)$ for $i\in \overline{S}$.
        \item It then sends $\{\WPKE.\ct_{i} \}_{i\in [2n]}$ to $\A$.
        \item $\A$ output $(m'_1,\ldots, m'_{2n})$. 
        It then outputs $(m'_{i_1},\ldots, m'_{i_{n}})$ as its output.
        \end{itemize}
    \end{enumerate}
We can see that $\qB$ defined above perfectly simulates $\Hyb_2$ for $\A$ and $\qB$ wins the game if so does $\A$.
We therefore have $\advb{\WPKE,\qB,n}{par}{ow}(\secp) \geq \Pr[\event_2]$.
Furthermore, by the parallel extractability of $\WPKE$, we have 
$\advb{\WPKE,\qB,\qExt,n}{par}{ext}(\secp) \geq \advb{\WPKE,\qB,n}{par}{ow}(\secp) -\negl(\secp)$.
We can also see that $\Pr[\event_3]=\advb{\WPKE,\qB,\qExt,n}{par}{ext}(\secp)$ holds, because the way $\qB$ inverts the ciphertexts is exactly the same as the way $\qD$ does in $\hyb_3$.
By combining these equations, the lemma follows.
\end{proof}
\begin{lemma}\label{lem:PKE-reduction-to-ada-HC}
    If $\NTCF$ satisfies the cut-and-choose adaptive hardcore bit property, we have $\Pr[\event_8]=\negl(\lambda)$.
\end{lemma}
\begin{proof}
    To show the lemma, we construct an adversary $\qB$ who plays the role of the challenger in $\hyb_8$ for $\qA$ and tries to break the cut-and-choose adaptive hardcore bit property of $\NTCF$. 
    $\qB$ proceeds as follows.
    \begin{enumerate}
    \item Given $\{\NTCF.\pp_i\}_{i\in[2n]}$ from its challenger, $\qB$ computes the first message to its challenger as follows.
    \begin{itemize}
        \item It runs 
        $(\WPKE.\ek_i, \WPKE.\msk_i) \la \WPKE(1^\secp)$ and
        $(\SFE.\crs_i,\SFE.\td_i) \lrun \SFE.\CRSGen(1^\secp, 2)$ for $i\in [2n]$ and sends $\ek= \{\NTCF.\pp_i,\SFE.\crs_i, \WPKE.\ek_i\}_{i\in [2n]}$ to $\A$.
        \item $\qA$ returns $\{y_i,\SFE.\msg^{(1)}_i\}_{i\in [2n]}$ to $\qB$.
        \item $\qB$ computes $ \SFE.\Send( \SFE.\msg^{(1)}_i, C[\WPKE.\msk_i,y_i] ) \rrun \SFE.\msg^{(2)}_i $ for $i\in [2n]$. 
        It then sends $\{ \SFE.\msg^{(2)} \}_{i\in [2n]}$ to $\qA$.
        \item $\A$ then submits the deletion certificate $\cert = \{ d_i, c_i\}_{i \in [2n]} $ to $\qB$.
        \item $\qB$ runs $\SFE.\StaRcv(\SFE.\td_i, \SFE.\msg^{(1)}_i) \to \{ \alpha_{i,j}^{b} \}_{j\in [\lenx], b\in \bin } $ for $i\in [2n]$.  
    It then computes 
    $d^*_i$ from $c_i$ and $\{ \alpha_{i,j}^{b} \}_{j\in [\lenx], b\in \bin } $ for $i\in [2n]$ and sends $\{y_i, d_i\oplus d^*_i\}_{i\in [2n]}$ to the challenger.
    \end{itemize}
    \item 
    If $\qB$ passes the verification, the challenger gives  
    $S$ and $\{ \NTCF.\td_i \}_{i\in \overline{S}}$ to $\qB$. 
    Then, $\qB$ gives $\dvk = (S, \{ \NTCF.\td_i ,  \SFE.\td_i\}_{i\in \overline{S}})$ to $\A$.
    \item 
    $\qB$ then constructs a quantum decryption algorithm $\qD=(U_\qD, \rho_\qD)$ 
      from $\{\WPKE.\ek_i\}_{i\in \overline{S}}$, $S$, and 
      $\A$ as specified in $\hyb_3$. Then, it runs the extraction algorithm as $\WPKE.\qExt( \{\WPKE.\ek_i\}_{i\in S}, (U_\qD, \rho_\qD) ) \to \{x'_i\}_{i\in S}$.
      Then, $\qB$ submits $\{x'_i\}_{i\in S}$ to its challenger.
    \end{enumerate}
    We can see that $\qB$ perfectly simulates $\hyb_8$, where a part of the deletion verification algorithm (i.e., inversion of $y_i$ and checking of \cref{eq:PKE-deletion-verify-in-the-proof} for $i\in \overline{S}$) is delegated to its challenger. 
    Furthermore, it can be checked that $\qB$ wins the cut-and-choose adaptive hardcore bit game if and only if $\qA$ wins in $\hyb_8$.
    We therefore have $\advc{\NTCF,\qB}{cut}{and}{choose}(\secp,n)=\Pr[\event_8]$.
    Hence, the lemma follows.
\end{proof}

\else
\fi

\ifnum\llncs=0
\ifnum\llncs=0
\subsection{Construction with Non-Interactive Key Generation}
\label{sec:non-interactive-pke-skl}
\else
\section{Classical-Lessor PKE-SKL with Non-Interactive Key Generation}
\label{sec:non-interactive-pke-skl}
\fi
Here, we discuss that the construction of PKE-SKL with interactive key generation that we showed in \cref{sec:PKE-SKL-with-Interaction} can be readily modified into a construction with non-interactive key generation as per \cref{def-pke-skl-NI}.
The idea is to merging the lessor's operation with the encryption algorithm. 
For this change to be possible, we need that the encryptor has to know the decryption keys of $\WPKE$. To accomplish this, we postpone the generation of the $\WPKE$ encryption keys until the execution of the encryption algorithm.

More concretely, we give the description of the construction in the following. Let $\intPKESKL = \intPKESKL.(\Setup, \allowbreak \qIntKeyGen, \allowbreak \Enc, \Dec, \qDec, \Del, \DelVrfy)$ be the construction in \cref{sec:PKE-SKL-with-Interaction}.
\begin{description}
\item[$\Setup_{\mathsf{NI}}(1^\secp)$:] 
It samples $\{\NTCF.\pp_i,\SFE.\crs_i\}_{i\in [2n]}$ as in $\intPKESKL.\Setup(1^\secp)$. However, it does \emph{not} sample $\WPKE$ encryption keys.  
It sets $\pp = \{\NTCF.\pp_i,\SFE.\crs_i\}_{i\in [2n]}$, $\msk=\bot$ and 
$\dvk = \{\NTCF.\td_i,\SFE.\td_i\}_{i\in \overline{S}}$.
\item[$\qKG_{\mathsf{NI}}(\pp)$:] 
Upon receiving the input $\pp =\{\NTCF.\pp_i,\SFE.\crs_i\}_{i\in [2n]}$, it runs the first operation of the lessee in $\intPKESKL.\qIntKeyGen$ to obtain 
$\{y_i,\SFE.\msg^{(1)}_i\}_{i\in [2n]}$ and $\ket{\psi'}$.
Note that it is possible to run the lessee's algorithm even though encryption keys of $\WPKE$ are not generated. 
It sets $\ek = \{y_i,\SFE.\msg^{(1)}_i\}_{i\in [2n]}$ and $\qdk = \ket{\psi'}$.
\item[$\Enc(\ek,m)$:] 
It takes as input $\ek = \{y_i,\SFE.\msg^{(1)}_i\}_{i\in [2n]}$ and proceeds as follows. 
\begin{itemize}
    \item It first samples $\WPKE.\KG(1^\lambda) \to (\WPKE.\ek_i, \WPKE.\msk_i)$ for $i\in [2n]$. 
    \item It then runs the lessor's operation of $\intPKESKL.\qIntKeyGen$ to obtain $\{\SFE.\msg^{(2)}_{i} \}_{i\in [2n]}$. Note that it is possible to run the lessor's operation since it knows $\WPKE.\msk_i$ for all $i$. 
    \item It then runs $\intPKESKL.\ct\lrun \intPKESKL.\Enc(\intPKESKL.\ek, m)$.
    \item Finally, it outputs the ciphertext $\ct=(\{\WPKE.\ek_i,\SFE.\msg^{(2)}_{i} \}_{i\in [2n]}, \intPKESKL.\ct )$.
\end{itemize}
\item[$\qDec(\qdk,\ct)$:] 
It takes as input a ciphertext $\ct=(\{\WPKE.\ek_i,\SFE.\msg^{(2)}_{i} \}_{i\in [2n]}, \intPKESKL.\ct )$ and the decryption key $\qdk= \ket{\psi'}$, it runs the lessee's second operation on $\{\WPKE.\ek_i,\SFE.\msg^{(2)}_{i} \}_{i\in [2n]}$ and $\ket{\psi'}$ to obtain $\otimes_{i\in [2n]} \ket{\phi_i}$. 
It then sets $\intPKESKL.\qdk=\bigotimes_{i\in[2n]}\ket{\phi_i}\ket{\SFE.\crs_i}
\ket{\SFE.\msg^{(2)}_{i}}$ and runs $\intPKESKL.\qDec(\intPKESKL.\qdk, \allowbreak \intPKESKL.\ct ) =m$.
It finally outputs $m$.

\item[$\Del(\qdk)$:] 
It measures $\qdk= \ket{\psi'}$ in Hadamard basis to obtain $\{d_i,c_i\}_{i\in [2n]}$.
It then outputs $\cert = \{d_i, c_i\}_{i \in [2n]}$. 
\item[$\DelVrfy(\ek,\dvk,\cert)$:]
It parses $\ek = \{y_i,\SFE.\msg^{(1)}_i\}_{i\in [2n]}$. 
It runs $\intPKESKL.\DelVrfy(\tau,\dvk,\cert)$,
where $\tau$ is a transcript that starts from $\{y_i,\SFE.\msg^{(1)}_i\}_{i\in [2n]}$,\footnote{While this does not fully specify $\tau$, this is sufficient to run $\intPKESKL.\DelVrfy$ since it only uses $\{y_i,\SFE.\msg^{(1)}_i\}_{i\in \overline{S}}$ of $\tau$.} 
and outputs what it outputs.
\end{description}
The decryption correctness and deletion verification correctness of the scheme readily follows from that of $\intPKESKL$.\footnote{The above scheme does not support the decryption algorithm $\Dec$ using $\msk$. However, adding this is straightforward: We have $\Setup_{\mathsf{NI}}$ generate an encryption key $\PKE.\ek$ and a decryption key $\PKE.\dk$ for the plain PKE. We then include $\PKE.\ek$ in $\pp$, and $\PKE.\dk$ in $\msk$. The encryption algorithm encrypts $m$ using $\PKE.\ek$, and incorporates the resulting ciphertext $\PKE.\ct$ into $\ct$. Consequently, the decryption algorithm can decrypt $\PKE.\ct$ when provided with $\msk$.}  
The following theorem asserts the security of the above construction.
\begin{theorem}
If $\intPKESKL$ is OW-VRA secure, then so is the above construction.
\end{theorem}
\begin{proof}
    We observe that the OW-VRA security game for the above construction is the same as that for $\intPKESKL$, except that the adversary has to submit the deletion certificate 
    in the former before getting $\{\WPKE.\ek_i,\SFE.\msg^{(2)}_{i} \}_{i\in [2n]}$.
    Since the former restricts the adversary more than the latter, the security of $\intPKESKL$ immediately implies that of the above construction. 
\end{proof}
Using \cref{lem:ow-ind-non-interactive}, we can upgrade the construction to be IND-VRA secure. 
\else
\fi

% !TEX root = main.tex

\section{PRF-SKL with Classical Lessor}\label{sec-PRFSKL}
In this section, we show the construction of PRF-SKL with a classical lessor from the LWE assumption. 
\subsection{Construction}
We construct a UPF-SKL scheme with a classical lessor that satisfies UP-VRA security as per \cref{def:UPF-SKL}. 
Roughly, it requires that if an adversary submits a valid deletion certificate, then it can no longer break unpredictability of the scheme even given the deletion verification key.  
As shown by \cite{kitagawa2025simple}, 
UP-VRA security can be easily upgraded into PR-VRA security, its indistinguishability-based counterpart, by the quantum Goldreich-Levin lemma (\cref{lem:ow-ind}).
The construction is very similar to our PKE-SKL construction, where we replace the watermarkable PKE with watermarkable UPF.
Concretely, we use the following ingredients:
\begin{itemize}
 \item NTCF generators $\NTCF=\NTCF.(\FuncGen, \StateGen, \Chk, \Invert)$ defined as in \cref{sec:NTCF}.
 It is associated with a deterministic algorithm $\NTCF.\GoodSet$ (See \cref{def:NTCF-generator}). We assume that $\NTCF.\Chk$ is a deterministic algorithm. 
 We denote the input space by $\cX$ and assume $\cX \subseteq \bin^\lenx$ for some polynomial $\lenx = \lenx(\secp)$.
 It is known that NTCF generators can be constructed from LWE.
% \takashi{This can be constructed from LWE.} 
 \item 
 Special Dual-mode SFE $\SFE=\SFE.(\CRSGen,\Receive{1}, \Send, \Receive{2})$ that supports circuits with input space $\bin^\lenx$ as defined in \cref{sec:dual_mode_SFE}.
 It is associated with an extraction algorithm $\SFE.\Extract$, a simulation algorithm $\SFE.\Sim$, a state recovery algorithm $\SFE.\StaRcv$,
 and a ``coherent version" $\SFE.\qReceive{1}$ of $\SFE.\Receive{1}$, which has the efficient state superposition property. 
 We require that $\SFE.\Receive{2}$ and $\SFE.\Extract$ are deterministic.
% We assume that the input space of $\Receive{1}(\SFE.\crs, \cdot)$ is $\bin^\lenx$. 
 We denote that the length of each block $\alpha_{i,j}^b$ output by $\SFE.\StaRcv$ by $\lenst$. 
 As we show in \cref{sec:dual-mode-SFE-from-LWE}, this can be constructed from LWE.
% \takashi{This can be constructed from  LWE.} 
 \item Watermarkable UPF (with parallel extraction): 
 $\WPRF=\WPRF.(\KG, \Mark, \Eval)$ as defined in \cref{watermarkable_UPF}.
 Without loss of generality, 
 we assume that it also satisfies security as a plain UPF (See \cref{rem:security-as-plain-UPF}).
 It is associated with a parallel mark extraction algorithm $\WPRF.\qExt$. 
 We require that $\WPRF.\Mark$ and $\WPRF.\Eval$ are deterministic algorithms.
 We need the mark space to be $\bin^\lenx$ and denote the input space (resp., output space) by $\bin^\lenm$ (resp., $\bin^{\ell'}$). 
 The bit-length of the marked key is denoted by $\lenmsg$. 
 As we show in \cref{watermarkable_UPF}, this can be constructed from any OWF.
\end{itemize}

Let $n=n(\secp)$ be a positive integer such that $n=\omega(\log \secp)$.

Below is the description of a UP-VRA secure UPF-SKL scheme. 
The input space of the scheme is $\bin^{2n\lenm}$.
\begin{description}
\item[$\Setup(1^\secp)$:] 
On input the security parameter, it proceeds as follows. 
\begin{itemize}
\item 
For $i\in [2n]$, run $\WPRF.\KG(1^\lambda) \to (\WPRF.\msk_i, \WPRF.\extk_i)$.
\item Choose a uniform subset $S\subseteq [2n]$ such that $|S|=n$. 
In the following, we denote $[2n]\backslash S$ by $\overline{S}$.
\item 
For $i\in [2n]$, run $(\NTCF.\pp_i,\NTCF.\td_i) \lrun \NTCF.\FuncGen(1^\lambda, \mode)$ and $(\SFE.\crs_i,\SFE.\td_i) \lrun \SFE.\CRSGen(1^\secp, \mode)$, where 
$
\mode = 
\begin{cases}
  1 & \text{if $i\in S$}\\
  2 & \text{if $i\in \overline{S}$}.
\end{cases}
$

\item Output $\pp = \{\NTCF.\pp_i,\SFE.\crs_i\}_{i\in [2n]}$, 
$\msk = \{ \WPRF.\msk_i \}_{i\in [2n]}$, and 
$\dvk=(S,\{\NTCF.\td_i,\SFE.\td_i\}_{i\in \overline{S} } )$. 

\end{itemize}
\item[$\qIntKeyGen(\pp, \msk )$:]
The lessor and lessee run the following interactive protocol to generate $\qsk$.
Here, the lessor takes $\msk$ as input and the lessee takes $\pp$ as input. 

{\bf Lessee's first operation:}
The lessee parses $\pp = \{\NTCF.\pp_i,\SFE.\crs_i\}_{i\in [2n]}$ and runs the following algorithms.
\begin{itemize}
\item For each $i\in [2n]$, run $\NTCF.\StateGen(\NTCF.\pp_i) \rrun (y_i, \ket{\psi_i})$.
\item For each $i\in [2n]$, run $\SFE.\qReceive{1}(\SFE.\crs_i,\ket{\psi_i}) \to (\ket{\psi'_i},\SFE.\msg_i^{(1)})$.

% \item 
% For each $i\in [2n]$, 
% prepare a uniform superposition over randomness for SFE and 
% run the first-message-generation algorithm of the SFE under CRS $\SFE.\crs_i$ in superposition, where the second register of $\psi_i$ is used as a message, and measure the ciphertext $\SFE.\ct_i$.  
% This results in the following state:
%     \begin{itemize}
%     \item If $i\in S$, $\ket{\psi''_i}=\ket{b_i}\ket{x_i}\ket{\st_i}$, where $\st_i$ is the stae corresponding $(x_i,\SFE.\ct_i)$. 
%     \item If $i\notin S$, $\ket{\psi''_i}=\ket{0}\ket{x_{i,0}}\ket{\st_{i,0}}+\ket{1}\ket{x_{i,1}}\ket{\st_{i,1}}$, where $\st_{i,b}$ is the state corresponding to $(x_{i,b},\SFE.\ct_{i})$. 
%     \end{itemize}

\item Send $\{y_i,\SFE.\msg^{(1)}_i\}_{i\in [2n]}$ to the lessor. 
\end{itemize}

{\bf Lessor's operation:}
Given the message $\{y_i,\SFE.\msg^{(1)}_i\}_{i\in [2n]}$ from the lessee, the lessor proceeds as follows.
%Note that this operation requires the master secret key $\msk=\{ \WPRF.\msk_i \}_{i\in [2n]}$ that the lessor has.
\begin{itemize}

% \item For $i\in S$, run $\NTCF.\Invert(\NTCF.\td_i, y_i) \to x_i$ and $\SFE.\Extract(\SFE.\td_i, \SFE.\ct_i) \to (x'_i,\st_i)$.
%  Set $\valid =1$ if $x_i = x'_i \neq \bot$ holds for all $i\in S$. Otherwise, set $\valid =0$.\footnote{Even if the lessor finds that the lessee is cheating at this point (i.e., $\valid=0$), it continues with the protocol and defers catching the lessee until the deletion verification algorithm is executed. This point is important when we convert the protocol into a non-interactive one. See \cref{rem:non-interactive-version} for further discussion.
%  } 

\item 
For $i\in [2n]$, define $C[\WPRF.\msk_i,y_i]$ to be a circuit that takes $x$ as input and computes the following:
\begin{align}
C[\WPRF.\msk_i,y_i](x)=
\begin{cases} 
\WPRF.\Mark(\WPRF.\msk_i, x) & \text{if $\NTCF.\Chk(\NTCF.\pp_i, x, y_i) = \top$} \\ 
\bot & \text{if $\NTCF.\Chk(\NTCF.\pp_i, x, y_i) = \bot$}. \\ 
\end{cases}
\end{align}
%\footnote{Recall that we require $\WPRF.\Mark$ to be deterministic.} 

Without loss of generality, we assume that the size of $C[\WPRF.\msk_i,y_i]$ does not depend on $\WPRF.\msk_i$ and $y_i$, and denote it by $|C|$.
\item 
For $i\in [2n]$, run $\SFE.\Send( \SFE.\crs_i, \SFE.\msg^{(1)}_i, C[\WPRF.\msk_i,y_i] ) \rrun \SFE.\msg^{(2)}_i$.
\item Send $\{\SFE.\msg^{(2)}_{i} \}_{i\in [2n]}$ to the lessee.
\end{itemize}

{\bf Lessee's second operation:}
Given the message $\{\SFE.\msg^{(2)}_{i} \}_{i\in [2n]}$, the lessee proceeds as follows.
It uses the secret state $\ket{\psi'_i}$ in the following.
\begin{itemize}
\item Construct the unitary $U[\SFE.\crs_i, \SFE.\msg_i^{(2)}]$ defined by the following map:
\[
\ket{x}\ket{\SFE.\st} \ket{0^{\lenmsg}} \mapsto 
\ket{x}\ket{\SFE.\st} \ket{ \SFE.\Receive{2}(\SFE.\crs_i, x, \SFE.\st, \SFE.\msg_i^{(2)}) }.
\]
%where $\lenmsg$ is the length of the decryption key of $\WPRF$.\footnote{Recall that we assumed that $\SFE.\Receive{2}$ is a deterministic algorithm.}

\item 
For $i\in [2n]$, apply the unitary to compute 
\[
\ket{\phi_i}\defeq U[\SFE.\crs_i, \SFE.\msg_i^{(2)}]\left( \ket{\psi'_i}\ket{0^{\lenmsg}} \right).  
\]

% \item 
% For $i\in [2n]$, run the output derivation algorithm of SFE in superposition to obtain
% $\ket{\psi'_i}$, where:
%     \begin{itemize}
%     \item If $i\in S$, $\ket{\psi'_i}=\ket{b_i}\ket{x_i}\ket{\st_i}\ket{\sk_i(x_i)}$, 
%     \item If $i\notin S$, $\ket{\psi'_i}=\ket{0}\ket{x_{i,0}}\ket{\st_{i,0}}\ket{\sk_i(x_{i,0})}+\ket{1}\ket{x_{i,1}}\ket{\st_{i,1}}\ket{\sk_i(x_{i,1})}$.
%     \end{itemize}
 \end{itemize}

{\bf Final output:}
The lessee privately obtains a secret key 
\[
\qsk=\bigotimes_{i\in[2n]}\ket{\phi_i}\ket{\SFE.\crs_i}
\ket{\SFE.\msg^{(2)}_{i}}.\footnote{$\SFE.\crs_i$ and $\SFE.\msg^{(2)}_{i}$ are not needed for the evaluation. They are used only for the deletion.}\] 

We write $\tau$ to denote the transcript of the execution of $\qIntKeyGen$.

\item[$\Eval(\msk,s):$]
Upon receiving $\msk=\{\WPRF.\msk_i\}_{i\in [2n]}$ and an input $s=s_1\|s_2\|\ldots \|s_{2n} \in \bin^{2n\lenm}$, 
compute $t_i \seteq \WPRF.\Eval(\WPRF.\msk_i,s_i)$ for $i\in[2n]$ and output $t=t_1\| \cdots \|t_{2n}$.
\item[$\qLEval(\qsk,s)$:]
Given $\qsk = \bigotimes_{i\in[2n]}\ket{\phi_i}\ket{\SFE.\crs_i}
\ket{\SFE.\msg^{(2)}_{i}}$ and an input $s=s_1\|s_2\|\ldots \|s_{2n}$, the evaluation algorithm performs the following for each $i\in [2n]$.
\begin{itemize}
    \item Construct the unitary $V[s_i]$ defined by the following map:
    \[
    \ket{x}\ket{\SFE.\st}\ket{\WPRF.\key}\ket{0^{\ell'}}
    \mapsto 
    \ket{x}\ket{\SFE.\st}\ket{\WPRF.\key}\ket{ \WPRF.\Eval(\WPRF.\key,s_i) }.
    \]
%    where $\lenm$ is the length of the message space of $\WPRF$.
    \item Compute $V[s_i] \left( \ket{\phi_i} \ket{0^{\ell'}} \right)$
     and measure the last $\ell'$ qubits to obtain $t_i$.
%     Let the residual state be $\ket{\phi'_i} $. 
\end{itemize}
It finally outputs $t=t_1\|t_2\|\ldots \|t_{2n}$.
%and $\qdk' = \bigotimes_{i\in[2n]} \ket{\phi'_i} \ket{\SFE.\crs_i}  \ket{\SFE.\msg^{(2)}_{i}}$. 

\item[$\Del(\qdk)$:]
Given $\qdk=\bigotimes_{i\in[2n]}\ket{\phi_i}\ket{\SFE.\crs_i}
\ket{\SFE.\msg^{(2)}_{i}}$, the deletion algorithm performs the following for each $i\in [2n]$.
\begin{itemize}
    \item Measure the corresponding registers to recover the classical strings $\SFE.\crs_i$ and $\SFE.\msg^{(2)}_{i}$.
    \item Compute $ U[\SFE.\crs_i,\SFE.\msg_i^{(2)}]^\dagger \ket{\phi_i} $
    and remove the last $\lenmsg$ qubits to obtain $\ket{\psi''_i}$.
    \item Measure $\ket{\psi''_i}$ in Hadamard basis to obtain $(d_i,c_i)$, where each $d_i\in \bin^{\lenx} $ is obtained by measuring the register corresponding to an element in $\cX$
and each $c_i \in \bin^{\lenst \lenx }$ is obtained by measuring the register corresponding to the SFE state.
\end{itemize}
The final output of the algorithm is given by $\cert = \{d_i, c_i\}_{i \in [2n]}$,

\item[$\DelVrfy(\tau,\dvk,\cert)$:]
Given the transcript $\tau$ of $\qIntKeyGen$,  the deletion verification key $\dvk=\left( S, \{ \NTCF.\td_i ,  \SFE.\td_i\}_{i\in \overline{S}} \right)$, and the deletion certificate $\cert = \{d_i, c_i\}_{i \in [2n]}$, 
the deletion verification algorithm works as follows.
\begin{itemize}
    \item Retrieve $\{y_i,\SFE.\msg^{(1)}_i\}_{i\in \overline{S}}$ from $\tau$. 
    \item For $i\in \overline{S}$, run $\NTCF.\Invert(\NTCF.\td_i, y_i) \to (x_{i,0},x_{i,1})$. 
    If the inversion fails for any $i$, output $\bot$.
    \item For $i\in \overline{S}$, run $\SFE.\StaRcv(\SFE.\td_i, \SFE.\msg^{(1)}_i) \to \{ \alpha_{i,j}^{b} \}_{j\in [\lenx], b\in \bin } $, where $\alpha_{i,j}^{b} \in \bin^{\lenst }$. 
    \item For $i\in \overline{S}$, compute the vector $d^*_i = (d^*_i[1],\ldots, d^*_i[\lenx]) \in \bin^\lenx$ as 
    \[
    d^*_i[j] \seteq 
    \langle c_{i,j} , (\alpha^0_{i,j} \oplus \alpha^1_{i,j})\rangle ,
    \]
    where $c_i$ is parsed as $c_i = c_{i,1}\| \cdots \| c_{i,\lenx} $, where each block is in $\bin^\lenst$.
    \item For $i\in \overline{S}$, check whether 
    \[
    \langle d_i \oplus d^*_i , (x_{i,0}\oplus x_{i,1}) \rangle
    =0
    ~
    \land
    ~
    \NTCF.\GoodSet(\NTCF.\td_i, y_i, d_i \oplus d^*_i  )=\top
    \]
    holds. 
%    Note that the latter condition can be checked using $\NTCF.\td_i$.
    If it holds for all $i\in \overline{S}$, it outputs $\top$.
    Otherwise, it outputs $\bot$.
\end{itemize}
\end{description}

\paragraph{Correctness}
The following theorem asserts the evaluation correctness and the deletion verification correctness of the above scheme.
\begin{theorem}\label{th:prf-correctness}
The above construction satisfies the evaluation correctness if 
$\NTCF$ satisfies the state generation correctness, $\SFE$ satisfies evaluation correctness, and $\WPRF$ satisfies the evaluation correctness.
Furthermore, the above construction satisfies the deletion verification correctness if $\NTCF$ satisfies the state generation correctness, $\SFE$ satisfies the statistical security against malicious senders,
and $\SFE$ supports state recovery in the hiding mode.
\end{theorem}
\begin{proof}
    We omit the proof of deletion verification correctness, since it is completely the same as that for PKE-SKL shown in 
    \cref{th:pke-correctness}. 
    We then prove the evaluation correctness. 
    As observed in the proof for \cref{th:pke-correctness}, 
    all the algorithms run the subsystems indexed by $i\in [2n]$, which run in parallel. 
Therefore, it suffices to show that each subsystem succeeds in evaluation with overwhelming probability assuming that the lessee follows the protocol honestly.
By the same argument as \cref{th:pke-correctness}, we can represent 
$\ket{\psi'_i}$ as
$\ket{\psi'_i} = 
q_{i,0,z^*}\ket{x_{i,0}}\ket{\st_{i,0,z^*}} 
+
q_{i,1,z^*}\ket{x_{i,1}}\ket{\st_{i,1,z^*}}$,
where $(x_{i,0}, x_{i,1}) = \NTCF.\Invert(\NTCF.\td_i, y_i)$, $z^*= \SFE.\msg^{(1)}_i$, and $\st_{i,b,z^*}$ is the unique state corresponding to
$x_{i,b}$ and $z^*=\SFE.\msg^{(1)}_i$. 
The definition of $q_{i,b,z^*}$ is not relevant here, but it can be found in the proof of \cref{th:pke-correctness}. 
By the evaluation correctness of $\SFE$, we have\footnote{
Precisely, the equation does not hold in strict sense, since we ignore negligible difference in trace distance in the proof of \cref{th:pke-correctness}. However, this can be ignored as this difference only negligibly affects the probability that the verification passes.
} 
\begin{align}
\ket{\phi_i}
&= 
\sum_{b\in \bin}
q_{i,b,z^*}
\ket{x_{i,b}}\ket{\st_{i,b,z^*}} \ket{ \SFE.\Receive{2}(\SFE.\crs_i, x_{i,b}, \st_{i,b,z^*}, \SFE.\msg_i^{(2)}) }
\\
&=
\sum_{b\in \bin}
q_{i,b,z^*}
\ket{x_{i,b}}\ket{\st_{i,b,z^*}} \ket{
\WPRF.\key_i(x_{i,b})
}.
\end{align}
After applying  $V[s_i]$, the above state becomes 
\begin{align}
&
\sum_{b\in \bin}
q_{i,b,z^*}
\ket{x_{i,b}}\ket{\st_{i,b,z^*}} \ket{
\WPRF.\key_i(x_{i,b})
}
\ket{ \WPRF.\Eval(\WPRF.\key_i(x_{i,b}),s_i) }
\\
&= 
\sum_{b\in \bin}
q_{i,b,z^*}
\ket{x_{i,b}}\ket{\st_{i,b,z^*}} \ket{
\WPRF.\key_i(x_{i,b})
}
\ket{ \WPRF.\Eval(\WPRF.\msk,s_i) }.
\end{align}
Here, we use the fact that $\WPRF.\Eval(\WPRF.\key_i(x_{i,b}),s_i) = { \WPRF.\Eval(\WPRF.\msk_i,s_i) }$ holds for \emph{all} $x$ with overwhelming probability over the randomness of $\WPRF.\KG$ by the evaluation correctness of $\WPRF$.
By measuring the last register, one can get $\WPRF.\Eval(\WPRF.\msk,s_i)$ as desired. 
\end{proof}

\subsection{Security Proof}
\begin{theorem}
\label{thm-proof-skl-prf}
If $\NTCF$ satisfies cut-and-choose adaptive hardcore bit security, 
$\SFE$ satisfies the mode-indistinguishability and the extractability in the extraction mode, and $\WPRF$ satisfies unpredictability as a plain UPF and parallel mark extractability, 
then our construction is UP-VRA secure.  
%\shota{Make it more precise.}
\end{theorem}
\begin{proof}
    We prove the theorem by a sequence of hybrids.
    For the sake of the contradiction, let us assume an adversary $\A$ that wins the UP-VRA security game with non-negligible probability $\epsilon$.
    In the following, the event that $\qA$ wins in $\Hyb_{\rm xx}$ is denoted by $\event_{\rm xx}$.

\begin{description}
  \item[$\Hyb_0$:] This is the original UP-VRA security game between the challenger and an adversary $\A$.  
  Namely, the game proceeds as follows.
\begin{enumerate}
    \item \label{Step:prf-proof-setup}
    The challenger computes the first message to the adversary as follows.
    \begin{itemize}
    \item 
    For $i\in [2n]$, it runs $\WPRF.\KG(1^\lambda) \to (\WPRF.\msk_i, \WPRF.\extk_i)$.
        \item It chooses a random subset $S$ of $[2n]$ with size $n$.
        \item It runs $(\NTCF.\pp_i,\NTCF.\td_i) \lrun \NTCF.\FuncGen(1^\lambda, \mode)$ and $(\SFE.\crs_i,\SFE.\td_i) \lrun \SFE.\CRSGen(1^\secp,  \mode)$ for $i\in [2n]$, where $\mode =1$ if $i\in S$ and $\mode =2$ if $i\in \overline{S}$.
    \end{itemize}
    Finally, it sends $\pp=\{\NTCF.\pp_i,\SFE.\crs_i\}_{i\in [2n]}$ to $\A$.
    It privately keeps the master secret key $\msk = \{ \WPRF.\msk_i \}_{i\in [2n]}$ and the deletion verification key $\dvk=\left( S, \{ \NTCF.\td_i ,  \SFE.\td_i\}_{i\in \overline{S}} \right)$.
    \item 
    Until $\qA$ receives the challenge input at Step~\ref{Step:PRF-challenge-is-given} of the game, $\qA$ has access to the evaluation oracle $\Eval(\msk,\cdot )$, which takes as input $\hats=\hats_1\| \cdots \|\hats_{2n} \in \bin^{2n \lenm }$ and returns $\hatt=\hatt_1\| \cdots \|\hatt_{2n}$, where $\hatt_i=\WPRF.\Eval(\WPRF.\msk_i,\hats_i)$. 
    \item The adversary $\A$ sends $\{y_i,\SFE.\msg^{(1)}_i\}_{i\in [2n]}$ to the challenger.
    \item \label{inner:PRF-second-msg-to-adv}
    The challenger computes the message to the adversary as follows. 
    It first runs 
    \[
    \SFE.\Send(\SFE.\crs_i, \SFE.\msg^{(1)}_i, C[\WPRF.\msk_i,y_i] ) \rrun \SFE.\msg^{(2)}_i 
    \]
    for $i\in [2n]$. Then, it sends $\{\SFE.\msg^{(2)}_{i} \}_{i\in [2n]}$ to $\A$.
    
    \item Then, $\A$ submits the deletion certificate 
    $\cert = \{d_i, c_i\}_{i \in [2n]} $.

    \item Then, the challenger runs $\DelVrfy$ as follows.
    \label{Step:PRF-Del-verify-in-prf-proof}
    \begin{itemize}
%    \item If $\valid = 0$, it outputs $\bot$. 
%   
    \item For $i\in \overline{S}$, it runs $\SFE.\StaRcv(\SFE.\td_i, \SFE.\msg^{(1)}_i) \to \{ \alpha_{i,j}^{b} \}_{j\in [\lenx], b\in \bin } $. 
    \item For $i\in \overline{S}$, it computes 
    $d^*_i = (d^*_i[1],\ldots, d^*_i[\lenx]) \in \bin^\lenx$ as 
    $
    d^*_i[j] \seteq \langle c_{i,j}, (\alpha^0_{i,j} \oplus \alpha^1_{i,j})\rangle
    $.  
    \item For $i\in \overline{S}$, it runs $(x_{i,0},x_{i,1}) \lrun \NTCF.\Invert(\NTCF.\td_i, y_i) $.
    If the inversion fails for any $i\in \overline{S}$, the output of $\DelVrfy$ is $\bot$.
    \item For $i\in \overline{S}$, it checks whether 
    \begin{equation}\label{eq:PRF-deletion-verify-in-the-proof}
    \langle d_i \oplus d^*_i  , x_{i,0}\oplus x_{i,1} \rangle =0 
    ~
    \land
    ~
    \NTCF.\GoodSet(\NTCF.\td_i,y_i, d_i \oplus d^*_i ) = \top
    \end{equation}
    holds using $\NTCF.\td_i$. 
    If \cref{eq:PRF-deletion-verify-in-the-proof} holds for all $i\in \overline{S}$, the output is $\top$ and otherwise, $\bot$.
    \end{itemize}
    If the output of the algorithm is $\bot$, then it indicates that the output of the game is $0$ (i.e., $\A$ failed to win the game). 
    Otherwise, the challenger returns 
    $\dvk$ to $\A$
    and continues the game.
    
    \item \label{Step:PRF-challenge-is-given}
    Then, the challenger chooses a random input $s=s_1\|s_2\|\cdots \|s_{2n} \sample \bin^{2n\lenm}$ and sends $s$ to $\A$.
    After receiving the challenge input, $\qA$ loses its oracle access to the evaluation oracle.
    \item \label{Step:PRF-winning-condition}
    $\A$ then outputs $t'=t'_1\| \cdots \|t'_{2n}$.
    The adversary wins the game if $t'=t$.
\end{enumerate}
By the definition, we have $\Pr[\event_{0}]= \epsilon$.
\item[$\Hyb_1$:]  
  In this hybrid, the challenger computes $\SFE.\msg^{(2)}_i$ differently for $i\in S$ in Step \ref{inner:PRF-second-msg-to-adv} of the game. Namely, it is replaced with the following: 
\begin{enumerate}[label=\arabic*', start=4]
    \item
    \label{Step:PRF-step-4prime}
    The challenger computes the second message to the adversary as follows. 
    \begin{itemize}
    \item It runs 
  $\SFE.\Extract(\SFE.\td_i, \SFE.\msg^{(1)}_i) \to x_i$ 
  and $\NTCF.\Chk(\NTCF.\pp_i, x_i, y_i) = \verdict_i$ for $i\in S$.
        
        \item It computes 
          \[
              \SFE.\msg^{(2)}_i \lrun 
              \begin{cases}
              \SFE.\Sim(\SFE.\crs_i, \SFE.\msg^{(1)}_i,1^{|C|}, \gamma_i )
                  & \text{if }  i\in S \\
             \SFE.\Send(\SFE.\crs_i, \SFE.\msg^{(1)}_i, C[\WPRF.\msk_i,y_i] )
                  & \text{if }  i\in \overline{S} 
              \end{cases},
        \]       
        where
  \[
  \gamma_i = 
  \begin{cases}
  \WPRF.\Mark(\WPRF.\msk_i, x_i)=
    \WPRF.\key_i(x_i)  
      & \text{if }  \verdict_i = \top \\
    \bot 
      & \text{if } \verdict_i  = \bot. 
  \end{cases}
  \]
    \end{itemize}
    Finally, it sends $\{\SFE.\msg^{(2)}_{i} \}_{i\in [2n]}$ to $\A$.
\end{enumerate}

%  Note that by the change introduced in the previous hybrid, if the game proceeds to the point where $\SFE.\msg^{(2)}_i$ is computed, we have $x'_i = x_i$ for $i\in S$. 
  Since we have $\gamma_i = C[\WPRF.\msk_i,y_i](x_i)$, it follows that this game is indistinguishability from the previous game by the straightforward reduction to the extractability of $\SFE$ in the extraction mode. 
  Thus, we have $|\Pr[\event_{0}] - \Pr[\event_{1}]| = \negl(\secp)$.
  Note that in this hybrid, no marked key of $\WPRF$ is necessary for running the game for $i$ such that $i\in S$ and $\verdict_i = \bot$.
  
   \item[$\Hyb_2$:]  
    In this hybrid, the challenger aborts the game and sets the output of the game to be $0$ (i.e., $\qA$ loses the game), once it turns out that there is $i\in S$ such that $\NTCF.\Chk(\NTCF.\pp_i, x_i, y_i)=\bot$.
    Namely, we replace Step \ref{Step:PRF-step-4prime} in the previous hybrid with the following: 
\begin{enumerate}[label=\arabic*'', start=4]
    \item \label{Step:PRF-step-4''}
    The challenger computes the second message to the adversary as follows. 
    \begin{itemize}
    \item It runs 
  $\SFE.\Extract(\SFE.\td_i, \SFE.\msg^{(1)}_i) \to x_i$ 
  and $\NTCF.\Chk(\NTCF.\pp_i, x_i, y_i) = \verdict_i$ for $i\in S$.
  \item If $\verdict_i = \bot$ for any $i\in S$, it aborts the game and sets the output of the game to be $0$.
        \item It computes 
          \[
              \SFE.\msg^{(2)}_i \lrun 
              \begin{cases}
              \SFE.\Sim(\SFE.\crs_i, \SFE.\msg^{(1)}_i, 1^{|C|},  \WPRF.\key_i(x_i)   )
                  & \text{if }  i\in S \\
             \SFE.\Send(\SFE.\crs_i, \SFE.\msg^{(1)}_i, C[\WPRF.\msk_i,y_i] )
                  & \text{if } i\in \overline{S}. 
              \end{cases}
        \]       
    \end{itemize}
    Finally, it sends $\{\SFE.\msg^{(2)}_{i} \}_{i\in [2n]}$ to $\A$.
\end{enumerate}

We claim $\Pr[\event_{2}]\geq  \Pr[\event_{1}] -\negl(\secp)$.
To show this, we observe that the only scenario where the adversary $\qA$ wins the game in the previous hybrid but not in the current one occurs when $\qA$ chooses $y_i$ such that $\verdict_i =\bot$ for some $i\in S$, yet still succeeds in predicting $t$. 
However, this implies that $\qA$ has predicted the value of $\WPRF.\Eval(\WPRF.\msk_i,s_i)$ without getting any marked key for the $i$-th instance of $\WPRF$. 
Note that the probability that $\qA$ has made an evaluation query to $\WPRF.\Eval(\WPRF.\msk_i, \cdot )$ on input $s_i$ is negligible, since $s_i$ is chosen uniformly at random and $\qA$ is not allowed to make evaluation queries once it receives $s_i$.
By a straightforward reduction to the unpredicatability of $\WPRF$, the claim follows.
%\synote{We may assume standard PRF security for $\WPRF$.}

  \item[$\Hyb_3$:]  
  In this hybrid, we change the winning condition of the adversary. 
  Namely, we replace Step \ref{Step:PRF-challenge-is-given} and Step \ref{Step:PRF-winning-condition} of the game with the following. 
  \begin{enumerate}[label=\arabic*', start=7]
      \item
      The challenger parses the adversary $\qA$ right before Step~\ref{Step:PRF-challenge-is-given} of the game as a unitary circuit $U_\qA$ and a quantum state $\rho_\qA$.
      It then constructs a quantum predictor $\qP$ for $\WPRF$ that takes as input $\{s_i\}_{i\in S}$ as input and outputs $\{ t'_i \}_{i\in S}$ 
      from $\A$ as follows.
     \begin{description}
         \item[$\qP( \{s_i\}_{i\in S} )$:]
         It chooses $s_i \gets \bin^\lenm$ for $i\in \overline{S}$ and runs $\A$ on input $ \{ s_i \}_{i\in [2n]} $ to obtain $\{t'_i\}_{i\in 2n}$.
         Finally, it outputs $\{ t'_i\}_{i\in S}$.  
    \end{description}
     \item 
     The challenger parses
     $\qP$ as a pair of unitary $U_\qP$ and a quantum state $\rho_\qP$.
     Then, it runs the extraction algorithm as 
         $\WPRF.\qExt( \{\WPRF.\extk_i\}_{i\in S} , (U_\qP, \rho_\qP ) ) \to \{x'_i\}_{i\in S}$.
    \item \label{Step:PRF-step-9prime}
    If $x_i = x'_i$ holds for all $i\in S$, 
    the challenger sets the output of the game to be $1$. 
    Otherwise, it sets it to be $0$.
 \end{enumerate}
  As we show in \cref{lem:reduction-to-par-ext-PRF}, 
  we have $\Pr[\event_{3}] \geq \Pr[\event_{2}]-\negl(\secp)$
  by the parallel mark extractability of $\WPRF$.
\item[$\Hyb_4$:]  
    In this game, we replace Step \ref{Step:PRF-step-9prime} of the above game with the following:
    \begin{enumerate}[label=\arabic*'', start=9]
        \item If $\NTCF.\Chk(\NTCF.\pp_i, x'_i, y_i)=\top$ holds for all $i\in S$, 
    the challenger sets the output of the game to be $1$. 
    Otherwise, it sets it to be $0$. 
    \end{enumerate}
    We observe that $x'_i = x_i$ in Step \ref{Step:PRF-step-9prime} implies $\NTCF.\Chk(\NTCF.\pp_i, x'_i, y_i) = \NTCF.\Chk(\NTCF.\pp_i, x_i, y_i) =\top$.
    To see the latter equality holds, recall that the game is set to abort if $\NTCF.\Chk(\NTCF.\pp_i, x_i, y_i) = \bot$ for any $i \in S$ 
    in Step~\ref{Step:PRF-step-4''}. Thus, if the game continues to Step \ref{Step:PRF-step-9prime}, the equality holds.   
    We therefore have $\Pr[\event_{4}] \geq  \Pr[\event_{3}]$.
    (Actually, $\Pr[\event_{4}] = \Pr[\event_{3}]$ holds due to the uniqueness of $x_i$ that passes the check. However, showing the inequality suffices for our proof.)
\item[$\Hyb_5$:]  
    In this game, we switch Step \ref{Step:PRF-step-4''} in the previous hybrid back to Step \ref{Step:PRF-step-4prime}. 
    The difference from the previous hybrid is that it continues the game even if $\verdict_i = \bot$ for some $i\in S$.
    Since the adversary had no chance in winning the game in such a case in the previous hybrid, this change only increases the chance of the adversary winning. 
    We therefore have $\Pr[\event_{5}] \geq \Pr[\event_{4}]$.
  
\item[$\Hyb_6$:]  
    In this hybrid, we switch Step \ref{Step:PRF-step-4prime} in the previous hybrid back to Step \ref{inner:PRF-second-msg-to-adv}.
    Namely, it computes $\SFE.\Send( \SFE.\crs_i, \allowbreak\SFE.\msg^{(1)}_i, C[\WPRF.\msk_i,y_i] ) \rrun \SFE.\msg^{(2)}_i$ for all $i\in [2n]$ again.
    Similarly to the change from $\Hyb_0$ to $\Hyb_1$, the winning probability of $\A$ in this game does not change significantly by the extractability of $\SFE$ in the extraction mode. %\shota{The name of the security notion should be fixed.}
    Namely, we have $|\Pr[\event_{6}] - \Pr[\event_{5}]| = \negl(\secp)$.
    Note that we do not need $\{\SFE.\td_i\}_{i\in S}$ any more to run the game.

\item[$\Hyb_7$:] 
    In this hybrid, we change the way $\SFE.\crs_i$ is sampled for $i\in S$.
    Namely, it is sampled as $(\SFE.\crs_i,\SFE.\td_i) \lrun \SFE.\CRSGen(1^\secp, 2)$ for all $i\in [2n]$. 
    Concretely, we replace Step \ref{Step:prf-proof-setup} with the following:
\begin{enumerate}[label=\arabic*']
    \item \label{Step:prf-proof-setup-1-prime}
    The challenger computes the first message to the adversary as follows.
    \begin{itemize}
        \item 
        For $i\in [2n]$, it runs $\WPRF.\KG(1^\lambda) \to (\WPRF.\msk_i, \WPRF.\extk_i)$.
        \item It chooses a random subset $S$ of $[2n]$ with size $n$.
        \item It runs $(\NTCF.\pp_i,\NTCF.\td_i) \lrun \NTCF.\FuncGen(1^\lambda, \mode)$ and $(\SFE.\crs_i,\SFE.\td_i) \lrun \SFE.\CRSGen(1^\secp, 2)$ for $i\in [2n]$, where $\mode =1$ if $i\in S$ and $\mode =2$ if $i\in \overline{S}$.
    \end{itemize}
    Finally, it sends $\pp=\{\NTCF.\pp_i,\SFE.\crs_i\}_{i\in [2n]}$ to $\A$.
    It privately keeps the master secret key $\msk = \{ \WPRF.\msk_i \}_{i\in [2n]}$ and the deletion verification key $\dvk=\left( S, \{ \NTCF.\td_i ,  \SFE.\td_i\}_{i\in \overline{S}} \right)$.  
\end{enumerate}  
    Since we do not need $\SFE.\td_i$ for $i\in S$ to run $\Hyb_6$ and $\Hyb_7$, a straightforward reduction to the mode indistinguishability of $\SFE$ implies that $\Hyb_6$ and $\Hyb_7$ are computationally indistinguishable. 
    Therefore, we have $|\Pr[\event_{6}] - \Pr[\event_{7}]| = \negl(\secp)$.
    \item[$\Hyb_8$:] 
    In this game, we change the way how the challenger runs $\DelVrfy$.
    In particular, the challenger computes $\SFE.\StaRcv(\SFE.\td_i, \SFE.\msg^{(1)}_i) \to \{ \alpha_{i,j}^{b} \}_{j, b} $ and $d^*_i$ for all $i\in [2n]$, 
    instead of only for $i\in \overline{S}$.  
    Still, the inversion of $y_i$ and 
    the check of \cref{eq:PRF-deletion-verify-in-the-proof} are performed only for $i\in \overline{S}$. 
    Concretely, we replace Step \ref{Step:PRF-Del-verify-in-prf-proof} with the following.
\begin{enumerate}[label=\arabic*', start=6]
        \item Then, the challenger runs (the modified version of) $\DelVrfy$ as follows.
    \begin{itemize}
%    \item If $\valid = 0$, it outputs $\bot$. 
%   
    \item For $i\in [2n]$, it runs $\SFE.\StaRcv(\SFE.\td_i, \SFE.\msg^{(1)}_i) \to \{ \alpha_{i,j}^{b} \}_{j\in [\lenx], b\in \bin } $. 
    \item For $i\in [2n]$, it computes 
    $d^*_i = (d^*_i[1],\ldots, d^*_i[\lenx]) \in \bin^\lenx$ as 
    $
    d^*_i[j] \seteq \langle c_{i,j}, (\alpha^0_{i,j} \oplus \alpha^1_{i,j})\rangle
    $.  
    \item For $i\in \overline{S}$, it runs $(x_{i,0},x_{i,1}) \lrun \NTCF.\Invert(\NTCF.\td_i, y_i) $.
    If the inversion fails for any $i\in \overline{S}$, the output of $\DelVrfy$ is $\bot$.
    \item For $i\in \overline{S}$, it checks whether 
    \begin{equation}
    e_i= ( d_i \oplus d^*_i ) \cdot(x_{i,0}\oplus x_{i,1})
    ~
    \land
    ~
    \NTCF.\GoodSet(\NTCF.\td_i,y_i, d_i \oplus d^*_i ) = \top
    \end{equation}
    holds using $\NTCF.\td_i$. 
    If the above holds for all $i\in \overline{S}$, the output is $\top$ and otherwise, $\bot$.
    \end{itemize}
    If the output of the algorithm is $\bot$, then it indicates that the output of the game is $0$ (i.e., $\A$ failed to win the game). 
    Otherwise, the challenger returns 
    $\dvk=\left( S, \{ \NTCF.\td_i,  \SFE.\td_i \}_{i\in \overline{S}} \right)$ to $\A$
    and continues the game.
\end{enumerate}
    
    We can see that this change is only conceptual, since the outcome of the verification is unchanged.  
    We may therefore have
    $\Pr[\event_{8}] = \Pr[\event_{7}] $.
\end{description}
    Based on the above discussion, it is established that 
    $\Pr[\event_{8}] \geq \epsilon - \negl(\secp)$, which is non-negligible by our assumption.
    However, we can show $\Pr[\event_{8}]=\negl(\secp)$ assuming the cut-and-choose adaptive hardcore bit property of $\NTCF$, by a very similar reduction to \cref{lem:PKE-reduction-to-ada-HC}.
    This results in a contradiction as desired.
\end{proof}
It remains to prove \cref{lem:reduction-to-par-ext-PRF}.
\begin{lemma}\label{lem:reduction-to-par-ext-PRF}
    If $\WPRF$ satisfies parallel mark extractability, we have $\Pr[\event_{3}] \geq \Pr[\event_{2}]-\negl(\secp)$. 
\end{lemma}
\begin{proof}
    We consider an adversary $\qB$ who plays the role of the challenger in $\Hyb_2$ for $\A$, while working as an adversary against the parallel prediction experiment with the number of instances $n$. It proceeds as follows: 
    \begin{enumerate}
    \item At the beginning of the game, the challenger runs $\WPRF.\KG(1^\secp)\ra (\WPRF.\msk'_j,\WPRF.\extk'_j)$ for $j\in [n]$. 
    $\qB$ is then given oracle access to $\WPRF.\Eval(\WPRF.\msk'_j,\cdot)$ for each $j\in [n]$.
    $\qB$ proceeds as follows:
    \begin{itemize}
        \item $\qB$ chooses a random subset $S=\{i_1,\ldots, i_n\}$ of $[2n]$ with size $n$.
        \item It implicitly sets 
    $\WPRF.\msk_{i_j} = \WPRF.\msk'_{j} $ and $\WPRF.\extk_{i_j} = \WPRF.\extk'_{j} $ for $j\in [n]$.
    \item It runs $\WPRF.\KG(1^\secp)\ra (\WPRF.\msk_i,\WPRF.\extk_i)$ for $i\in \overline{S}$. 
    \item It chooses $\{\NTCF.\pp_i,\NTCF.\td_i\}_{i\in [2n]}$ and $\{\SFE.\crs_i,\SFE.\td_i\}_{i\in [2n]} $ as in $\hyb_2$.
    Finally, $\qB$ sends $\pp = \{\NTCF.\pp_i,\SFE.\crs_i\}_{i\in [2n]}$ to $\qA$.
    \end{itemize}
 \item 
    Until $\qA$ receives the challenge input from $\qB$, $\qA$ has access to the evaluation oracle $\Eval(\msk,\cdot )$ throughout the game.
    To answer the query on input $\hats=\hats_1\| \cdots \|\hats_{2n} \in \bin^{2n \lenm }$,
    $\qB$ computes $\hatt_i =\WPRF.\Eval(\WPRF.\msk_i,\hats_i)$ by itself if $i\in \overline{S}$ and by making an evaluation query to the oracle $\WPRF.\Eval(\WPRF.\msk_i,\cdot )$ on input $\hats_i$ if $i\in S$.
    Then it returns $\hatt=\hatt_1\| \cdots \|\hatt_{2n}$ to $\qA$.
        
       \item 
       At some point, $\qA$ sends $\{y_i,\SFE.\msg^{(1)}_i\}_{i\in [2n]}$ to $\qB$.
       $\qB$ then chooses its initial message to its challenger by the following steps. 
       \begin{itemize}
           \item It then computes $x_j$ and $\verdict_j$ for $j\in S$ using the trapdoors for $\NTCF$ and $\SFE$ as in $\Hyb_2$. If $\verdict_j=\bot$ for any $j\in S$, it aborts. 
        \item It then submits $\{x'_1,\ldots, x'_n\}\seteq \{x_{i_1},\ldots, x_{i_n}\}$ to its challenger.
       \end{itemize} 
        
        \item Given $\{\WPRF.\key'_j(x'_{j})\}_{j\in [n]}$ from the challenger, $\qB$ proceeds as follows. 
        \begin{itemize}
        \item It sets $
        \WPRF.\key_{i_j}(x_{i_j})\seteq \WPRF.\key'_j(x'_{j})$ for $j\in [n]$.
        \item It computes $\SFE.\msg^{(2)}_i$ as follows: 
        \[
        \SFE.\msg^{(2)}_i \lrun 
        \begin{cases}
          \SFE.\Sim(\SFE.\crs_i,\SFE.\msg^{(1)}_i, 1^{|C|}, \WPRF.\key_i(x_i) ) & \text{if } i\in S \\
          \SFE.\Send(\SFE.\crs_i, \SFE.\msg^{(1)}_i, C[\WPRF.\msk_i, y_i] ) & \text{if } i\in \overline{S}.
        \end{cases}
        \]
        \item $\qB$ then sends $\{ \SFE.\msg^{(2)}_{i} \}_{i\in [2n]}$ to $\qA$. 
        \item Then, $\A$ submits the deletion certificate $\cert$ to $\qB$.
        $\qB$ runs $\DelVrfy$ using the trapdoors for $\NTCF$ and $\SFE$.
        If the output is $\bot$, $\qB$ also outputs $\bot$ and aborts.
    Otherwise, $\qB$ returns 
    $\dvk=\left( S, \{ \NTCF.\td_i,\SFE.\td_i \}_{i\in \overline{S}} \right)$ to $\A$.
        \end{itemize}

        \item Then, the challenger sends $\{s'_j\}_{j\in [n]}$ to $\qB$.  
        \item $\qB$ computes its output as follows. 
        \begin{itemize}
        \item It sets $s_{i_j} \seteq s'_{j}$ for $j\in [n]$ and samples $s_i\lrun \bin^\ell$ for $i\in \overline{S}$.
        \item It then sends $\{s_{i} \}_{i\in [2n]}$ to $\A$.
        \item $\A$ outputs $(t'_1,\ldots, t'_{2n})$. 
        It then outputs $(t'_{i_1},\ldots, t'_{i_{n}})$ as its output.
        \end{itemize}
    \end{enumerate}
We can see that $\qB$ defined above perfectly simulates $\Hyb_2$ for $\A$ above
and $\qB$ wins the game whenever $\A$ does.
We therefore have $\advb{\WPRF,\qB,n}{par}{pre}(\secp) \geq \Pr[\event_2]$.
Furthermore, by the parallel mark extractability of $\WPRF$, we have 
$\advb{\WPRF,\qB,\qExt,n}{par}{ext}(\secp) \geq \advb{\WPRF,\qB,n}{par}{pre}(\secp) -\negl(\secp)$.
We can also see that $\Pr[\event_3]=\advb{\WPRF,\qB,\qExt,n}{par}{ext}(\secp)$ holds, because the way $\qB$ predicts the outputs is exactly the same as the way $\qP$ does in $\hyb_3$.
By combining these equations, the lemma follows.
\end{proof}
% !TEX root = main.tex

\section{DS-SKL with Classical Lessor}\label{sec-DSSKL}
In this section, we show the construction of DS-SKL with a classical lessor from the LWE assumption. 
\subsection{Construction}
We construct a DS-SKL scheme with a classical lessor that satisfies RUF-VRA security as per \cref{def:RUF-VRA}. 
The construction is very similar to our PKE-SKL construction, where we replace the watermarkable PKE with watermarkable DS.
Concretely, we use the following ingredients:
\begin{itemize}
 \item NTCF generators $\NTCF=\NTCF.(\FuncGen, \StateGen, \Chk, \Invert)$ defined as in \cref{sec:NTCF}.
 It is associated with a deterministic algorithm $\NTCF.\GoodSet$ (See \cref{def:NTCF-generator}). 
 We assume that $\NTCF.\Chk$ is a deterministic algorithm. 
 We denote the input space by $\cX$ and assume $\cX \subseteq \bin^\lenx$ for some polynomial $\lenx = \lenx(\secp)$.
 It is known that NTCF generators can be constructed from LWE.
% \takashi{This can be constructed from LWE.} 
 \item 
 Special Dual-mode SFE $\SFE=\SFE.(\CRSGen,\Receive{1}, \Send, \Receive{2})$ that supports circuits with input space $\bin^\lenx$ as defined in \cref{sec:dual_mode_SFE}.
 It is associated with an extraction algorithm $\SFE.\Extract$, a simulation algorithm $\SFE.\Sim$, a state recovery algorithm $\SFE.\StaRcv$,
 and a ''coherent version" $\SFE.\qReceive{1}$ of $\SFE.\Receive{1}$, which has the efficient state superposition property. 
 We require that $\SFE.\Receive{2}$ and $\SFE.\Extract$ are deterministic.
% We assume that the input space of $\Receive{1}(\SFE.\crs, \cdot)$ is $\bin^\lenx$. 
 We denote that the length of each block $\alpha_{i,j}^b$ output by $\SFE.\StaRcv$ by $\lenst$. 
 As we show in \cref{sec:dual-mode-SFE-from-LWE}, this can be constructed from LWE.
% \takashi{This can be constructed from  LWE.} 
 \item Watermarkable DS (with parallel extraction): 
 $\WDS=\WDS.(\KG, \Mark, \Sign, \Vrfy)$ as defined in \cref{watermarkable_DS}.
 Without loss of generality, 
 we assume that this satisfies EUF-CMA security as a plain DS (See \cref{rem:WDS-security-as-plain-DS}).
 It is associated with a parallel mark extraction algorithm $\WDS.\qExt$. 
We require that $\WDS.\Mark$ be a deterministic algorithm. This can be assumed without loss of generality because the randomness needed to run 
$\WDS.\Mark$ can be included into $\WDS.\msk$.
 We assume that it has coherent-signability and thus it is associated with a coherent signing quantum algorithm $\WDS.\QSign$.
 We need the mark space to be $\bin^\lenx$ and denote the message space by $\bin^\lenm$. 
 The bit length of the marked signing key is denoted by $\lenmsg$. 
 As we show in \cref{watermarkable_DS}, this can be constructed from the SIS assumption.

\end{itemize}
Let $n=n(\secp)$ be a positive integer such that $n=\omega(\log \secp)$.

Below is the description of the construction.
\begin{description}
\item[$\Setup(1^\secp)$:] 
On input the security parameter, it proceeds as follows. 
\begin{itemize}
\item 
For $i\in [2n]$, run $\WDS.\KG(1^\lambda) \to (\WDS.\vk_i, \WDS.\msk_i, \WDS.\extk_i)$.
\item Choose a uniform subset $S\subseteq [2n]$ such that $|S|=n$. 
In the following, we denote $[2n]\backslash S$ by $\overline{S}$.
\item 
For $i\in [2n]$, run $(\NTCF.\pp_i,\NTCF.\td_i) \lrun \NTCF.\FuncGen(1^\lambda, \mode)$ and $(\SFE.\crs_i,\SFE.\td_i) \lrun \SFE.\CRSGen(1^\secp, \mode)$, where 
$
\mode = 
\begin{cases}
  1 & \text{if $i\in S$}\\
  2 & \text{if $i\in \overline{S}$}.
\end{cases}
$

\item Output $\svk = \{\NTCF.\pp_i,\SFE.\crs_i, \WDS.\vk_i\}_{i\in [2n]}$, 
$\msk = \{ \WDS.\msk_i \}_{i\in [2n]}$, and 
$\dvk=(S,\{\NTCF.\td_i, \allowbreak \SFE.\td_i\}_{i\in \overline{S} } )$. 

\end{itemize}
\item[$\qIntKeyGen(\svk, \msk )$:]
The lessor and lessee run the following interactive protocol to generate $\qsigk$.
Here, the lessor takes $\msk$ as input and the lessee takes $\svk$ as input. 

{\bf Lessee's first operation:}
The lessee parses $\svk = \{\NTCF.\pp_i,\SFE.\crs_i, \WDS.\vk_i\}_{i\in [2n]}$ and runs the following algorithms.
\begin{itemize}
\item For each $i\in [2n]$, run $\NTCF.\StateGen(\NTCF.\pp_i) \rrun (y_i, \ket{\psi_i})$.
\item For each $i\in [2n]$, run $\SFE.\qReceive{1}(\SFE.\crs_i,\ket{\psi_i}) \to (\ket{\psi'_i},\SFE.\msg_i^{(1)})$.

% \item 
% For each $i\in [2n]$, 
% prepare a uniform superposition over randomness for SFE and 
% run the first-message-generation algorithm of the SFE under CRS $\SFE.\crs_i$ in superposition, where the second register of $\psi_i$ is used as a message, and measure the ciphertext $\SFE.\ct_i$.  
% This results in the following state:
%     \begin{itemize}
%     \item If $i\in S$, $\ket{\psi''_i}=\ket{b_i}\ket{x_i}\ket{\st_i}$, where $\st_i$ is the stae corresponding $(x_i,\SFE.\ct_i)$. 
%     \item If $i\notin S$, $\ket{\psi''_i}=\ket{0}\ket{x_{i,0}}\ket{\st_{i,0}}+\ket{1}\ket{x_{i,1}}\ket{\st_{i,1}}$, where $\st_{i,b}$ is the state corresponding to $(x_{i,b},\SFE.\ct_{i})$. 
%     \end{itemize}

\item Send $\{y_i,\SFE.\msg^{(1)}_i\}_{i\in [2n]}$ to the lessor. 
\end{itemize}

{\bf Lessor's operation:}
Given the message $\{y_i,\SFE.\msg^{(1)}_i\}_{i\in [2n]}$ from the lessee, the lessor proceeds as follows.
\begin{itemize}

% \item For $i\in S$, run $\NTCF.\Invert(\NTCF.\td_i, y_i) \to x_i$ and $\SFE.\Extract(\SFE.\td_i, \SFE.\ct_i) \to (x'_i,\st_i)$.
%  Set $\valid =1$ if $x_i = x'_i \neq \bot$ holds for all $i\in S$. Otherwise, set $\valid =0$.\footnote{Even if the lessor finds that the lessee is cheating at this point (i.e., $\valid=0$), it continues with the protocol and defers catching the lessee until the deletion verification algorithm is executed. This point is important when we convert the protocol into a non-interactive one. See \cref{rem:non-interactive-version} for further discussion.
%  } 

\item 
For $i\in [2n]$, define $C[\WDS.\msk_i,y_i]$ to be a circuit that takes $x$ as input and computes the following:
\begin{align}
C[\WDS.\msk_i,y_i](x)=
\begin{cases} 
\WDS.\Mark(\WDS.\msk_i, x) & \text{if $\NTCF.\Chk(\NTCF.\pp_i, x, y_i) = \top$} \\ 
\bot & \text{if $\NTCF.\Chk(\NTCF.\pp_i, x, y_i) = \bot$}. \\ 
\end{cases}
\end{align}
%\footnote{Recall that we require $\WPRF.\Mark$ to be deterministic.} 

Without loss of generality, we assume that the size of $C[\WDS.\msk_i,y_i]$ does not depend on $\WDS.\msk_i$ and $y_i$, and denote it by $|C|$.

\item 
For $i\in [2n]$, run $\SFE.\Send( \SFE.\crs_i, \SFE.\msg^{(1)}_i, C[\WDS.\msk_i,y_i] ) \rrun \SFE.\msg^{(2)}_i$.
\item Send $\{\SFE.\msg^{(2)}_{i} \}_{i\in [2n]}$ to the lessee.
\end{itemize}

{\bf Lessee's second operation:}
Given the message $\{\SFE.\msg^{(2)}_{i} \}_{i\in [2n]}$, the lessee proceeds as follows.
It uses the secret state $\ket{\psi'_i}$ in the following.
\begin{itemize}
\item Construct the unitary $U[\SFE.\crs_i, \SFE.\msg_i^{(2)}]$ defined by the following map:
\[
\ket{x}\ket{\SFE.\st} \ket{0^{\lenmsg}} \mapsto 
\ket{x}\ket{\SFE.\st} \ket{ \SFE.\Receive{2}(\SFE.\crs_i, x, \SFE.\st, \SFE.\msg_i^{(2)}) }.
\]
%where $\lenmsg$ is the length of the decryption key of $\WPRF$.\footnote{Recall that we assumed that $\SFE.\Receive{2}$ is a deterministic algorithm.}

\item 
For $i\in [2n]$, apply the unitary to compute 
\[
\ket{\phi_i}\defeq U[\SFE.\crs_i, \SFE.\msg_i^{(2)}]\left( \ket{\psi'_i}\ket{0^{\lenmsg}} \right).  
\]

% \item 
% For $i\in [2n]$, run the output derivation algorithm of SFE in superposition to obtain
% $\ket{\psi'_i}$, where:
%     \begin{itemize}
%     \item If $i\in S$, $\ket{\psi'_i}=\ket{b_i}\ket{x_i}\ket{\st_i}\ket{\sk_i(x_i)}$, 
%     \item If $i\notin S$, $\ket{\psi'_i}=\ket{0}\ket{x_{i,0}}\ket{\st_{i,0}}\ket{\sk_i(x_{i,0})}+\ket{1}\ket{x_{i,1}}\ket{\st_{i,1}}\ket{\sk_i(x_{i,1})}$.
%     \end{itemize}
 \end{itemize}

{\bf Final output:}
The lessee privately obtains a secret key 
\[
\qsigk=\bigotimes_{i\in[2n]}\ket{\phi_i}\ket{\SFE.\crs_i}
\ket{\SFE.\msg^{(2)}_{i}}.\footnote{$\SFE.\crs_i$ and $\SFE.\msg^{(2)}_{i}$ are not needed for signing. They are used only for deletion.}\] 

We write $\tau$ to denote the transcript of the execution of $\qIntKeyGen$.

\item[$\Sign(\msk,m):$]
Upon receiving $\msk=\{\WDS.\msk_i\}_{i\in [2n]}$ and an input $m=m_1\|m_2\|\ldots \|m_{2n} \in \bin^{2n\lenm}$, 
compute $\WDS.\sigma_i \gets \WDS.\Sign(\WDS.\msk_i,m_i)$ for $i\in[2n]$ and output $\sigma = ( \WDS.\sigma_1, \ldots, \WDS.\sigma_{2n} ) $.
\item[$\qSign(\qsigk,m)$:]
Given $\qsigk = \bigotimes_{i\in[2n]}\ket{\phi_i}\ket{\SFE.\crs_i}
\ket{\SFE.\msg^{(2)}_{i}}$ and an input $m=m_1\|m_2\|\ldots \|m_{2n}$, the signing algorithm performs the following for each $i\in [2n]$.
Then, it runs 
\[
\WDS.\QSign(\ket{\phi_i}, m_i) \to ( \ket{\phi'_i}, \WDS.\sigma_i )
\]
for $i\in [2n]$.
It finally outputs 
$\qsigk' = \bigotimes_{i\in[2n]}\ket{\phi'_i}\ket{\SFE.\crs_i}
\ket{\SFE.\msg^{(2)}_{i}}$
and $\sigma = ( \WDS.\sigma_1, \ldots, \WDS.\sigma_{2n} ) $.
%and $\qdk' = \bigotimes_{i\in[2n]} \ket{\phi'_i} \ket{\SFE.\crs_i}  \ket{\SFE.\msg^{(2)}_{i}}$. 

\item[$\SigVrfy(\sigvk,m,\sigma)\ra \top/\bot$:]
Given $\svk = \{\NTCF.\pp_i,\SFE.\crs_i, \WDS.\vk_i\}_{i\in [2n]}$, $m=m_1\|m_2\|\ldots \|m_{2n}$, and $\sigma = ( \WDS.\sigma_1, \ldots, \WDS.\sigma_{2n} ) $, 
it outputs $\top$ if $\WDS.\Vrfy(\WDS.\vk_i, m_i, \WDS.\sigma_i) = \top$ holds for all $i\in [2n]$. Otherwise, it outputs $\bot$.

\item[$\Del(\qsigk)$:]
Given $\qsigk=\bigotimes_{i\in[2n]}\ket{\phi_i}\ket{\SFE.\crs_i}
\ket{\SFE.\msg^{(2)}_{i}}$, the deletion algorithm performs the following for each $i\in [2n]$.
\begin{itemize}
    \item Measure the corresponding registers to recover the classical strings $\SFE.\crs_i$ and $\SFE.\msg^{(2)}_{i}$.
    \item Compute $ U[\SFE.\crs_i,\SFE.\msg_i^{(2)}]^\dagger \ket{\phi_i} $
    and remove the last $\lenmsg$ qubits to obtain $\ket{\psi''_i}$.
    \item Measure $\ket{\psi''_i}$ in Hadamard basis to obtain $(d_i,c_i)$, where each $d_i\in \bin^{\lenx} $ is obtained by measuring the register corresponding to an element in $\cX$
and each $c_i \in \bin^{\lenst \lenx }$ is obtained by measuring the register corresponding to the SFE state.
\end{itemize}
The final output of the algorithm is given by $\cert = \{d_i, c_i\}_{i \in [2n]}$,

\item[$\DelVrfy(\tau,\dvk,\cert)$:]
Given the transcript $\tau$ of $\qIntKeyGen$, the deletion verification key $\dvk=\left( S, \{ \NTCF.\td_i ,  \SFE.\td_i\}_{i\in \overline{S}} \right)$, and the deletion certificate $\cert = \{d_i, c_i\}_{i \in [2n]}$, 
the deletion verification algorithm works as follows.
\begin{itemize}
    \item Retrieve $\{y_i,\SFE.\msg^{(1)}_i\}_{i\in \overline{S}}$ from $\tau$. 
    \item For $i\in \overline{S}$, run $\NTCF.\Invert(\NTCF.\td_i, y_i) \to (x_{i,0},x_{i,1})$. 
    If the inversion fails for any $i$, output $\bot$.
    \item For $i\in \overline{S}$, run $\SFE.\StaRcv(\SFE.\td_i, \SFE.\msg^{(1)}_i) \to \{ \alpha_{i,j}^{b} \}_{j\in [\lenx], b\in \bin } $, where $\alpha_{i,j}^{b} \in \bin^{\lenst }$. 
    \item For $i\in \overline{S}$, compute the vector $d^*_i = (d^*_i[1],\ldots, d^*_i[\lenx]) \in \bin^\lenx$ as 
    \[
    d^*_i[j] \seteq 
    \langle c_{i,j} , (\alpha^0_{i,j} \oplus \alpha^1_{i,j})\rangle ,
    \]
    where $c_i$ is parsed as $c_i = c_{i,1}\| \cdots \| c_{i,\lenx} $, where each block is in $\bin^\lenst$.
    \item For $i\in \overline{S}$, check whether 
    \[
    \langle d_i \oplus d^*_i , (x_{i,0}\oplus x_{i,1}) \rangle
    =0
    ~
    \land
    ~
    \NTCF.\GoodSet(\NTCF.\td_i, y_i, d_i \oplus d^*_i  )=\top
    \]
    holds. 
%    Note that the latter condition can be checked using $\NTCF.\td_i$.
    If it holds for all $i\in \overline{S}$, it outputs $\top$.
    Otherwise, it outputs $\bot$.
\end{itemize}
\end{description}

\paragraph{Correctness}
The following theorem asserts the signature verification correctness,
the reusability with static signing keys, and the deletion verification correctness of the above scheme.
\begin{theorem}\label{th:sig-skl-correctness}
The above construction satisfies the signature verification correctness if 
$\NTCF$ satisfies the state generation correctness, $\SFE$ satisfies evaluation correctness, and $\WDS$ satisfies the correctness and the coherent signability.
Under the same condition, the construction satisfies the reusability with static signing keys. 
Furthermore, the construction satisfies the deletion verification correctness if $\NTCF$ satisfies the state generation correctness, $\SFE$ satisfies the statistical security against malicious senders,
and $\SFE$ supports state recovery in the hiding mode.
\end{theorem}
\begin{proof}
    We omit the proof of deletion verification correctness, since it is completely the same as that for PKE-SKL shown in 
    \cref{th:pke-correctness}. 
    We therefore show the signature verification correctness and the reusability with static signing keys.
    It is straightforward to see that a signature generated by $\msk$ passes the verification. 
    We therefore consider the case where a signature is generated by $\qsigk$.
    As observed in the proof for \cref{th:pke-correctness}, 
    all the algorithms run the subsystems indexed by $i\in [2n]$, which run in parallel. 
Therefore, it suffices to show that each $\WDS.\sigma_i$ passes the verification with overwhelming probability and generation of $\WDS.\sigma_i$ almost does not change the state $\ket{\phi_i}$, assuming that the lessee follows the protocol honestly.

By the same argument as \cref{th:pke-correctness}, we can represent 
$\ket{\psi'_i}$ as
$\ket{\psi'_i} = 
q_{i,0,z^*}\ket{x_{i,0}}\ket{\st_{i,0,z^*}} 
+
q_{i,1,z^*}\ket{x_{i,1}}\ket{\st_{i,1,z^*}}$,
where $(x_{i,0}, x_{i,1}) = \NTCF.\Invert(\NTCF.\td_i, y_i)$, where $z^*= \SFE.\msg^{(1)}_i$, and $\st_{i,b,z^*}$ is the unique state corresponding to
$x_{i,b}$ and $z^*=\SFE.\msg^{(1)}_i$. 
The definition of $q_{i,b,z^*}$ is not relevant here, but it can be found in the proof of \cref{th:pke-correctness}. 
By the evaluation correctness of $\SFE$, we have\footnote{
Precisely, the equation does not hold in the strict sense, since we ignore negligible difference in trace distance in the proof of \cref{th:pke-correctness}. However, this can be ignored as this difference only negligibly affects the probability that the verification passes.
} 
%\shota{Underconstruction}
\begin{align}
\ket{\phi_i}
&= 
\sum_{b\in \bin}
q_{i,b,z^*}
\ket{x_{i,b}}\ket{\st_{i,b,z^*}} \ket{ \SFE.\Receive{2}(\SFE.\crs_i, x_{i,b}, \st_{i,b,z^*}, \SFE.\msg_i^{(2)}) }
\\
&=
\sum_{b\in \bin}
q_{i,b,z^*}
\ket{x_{i,b}}\ket{\st_{i,b,z^*}} \ket{
\WDS.\sigk_i(x_{i,b})
}.
\end{align}
Then, $\WDS.\QSign$ is applied on $\ket{\phi_i}$ and $m_i$. 
Let us consider $L=\lenx+\lenst\lenx$ and consider $\{\alpha_t\in \CC \}_{t\in \bin^L}$ and $\{x'_t\in \bin^\lenx\}_{t\in \bin^L}$ indexed by $t$ such that
\begin{align}
\alpha_t
=
  \begin{cases}
    q_{i,b,z^*} & \text{if } t= x_{i,b}\|\st_{i,b,z^*}~ \text{for some } b\in\bin \\
    0 & \text{otherwise}
  \end{cases},
  \quad
  x'_t=
  \begin{cases}
    x_{i,b} & \text{if } t= x_{i,b}\|\st_{i,b,z^*}~ \text{for some } b\in\bin \\
    x^* & \text{otherwise}
  \end{cases}
\end{align}
where $x^*$ is an arbitrary string.
Then, we can write $\ket{\phi_i}$ as $\ket{\phi_i}=\sum_{t\in \bin^L}\alpha_t \ket{t}\ket{\WDS.\sigk_i(x'_{t})}$.
By the first property of coherent signability satisfied by $\WDS$, 
$\WDS.\sigma_i$ is in $\cup_{t}\Supp(\WDS.\Sign(\WDS.\sigk_i(x'_t),m_i))$.
By the correctness of $\WDS$, $\WDS.\sigma_i$ passes the verification. 
Furthermore, by the second property of coherent signability satisfied by $\WDS$, 
$\ket{\phi'_i}$ obtained by applying $\WDS.\QSign$ is negligibly close to $\ket{\phi_i}$ in trace distance. 
\end{proof}

\subsection{Security Proof}
The following theorem asserts the security of the above construction.
\begin{theorem}
\label{thm-proof-skl-ds}
If $\NTCF$ satisfies cut-and-choose adaptive hardcore bit security, 
$\SFE$ satisfies the mode-indistinguishability and the extractability in the extraction mode, and $\WDS$ satisfies parallel mark extractability and EUF-CMA security as a plain DS, 
then our construction is RUF-VRA secure.  
%\shota{Make it more precise.}
\end{theorem}
\begin{proof}
    We prove the theorem by a sequence of hybrids.
    For the sake of the contradiction, let us assume an adversary $\A$ that wins the RUF-VRA security game with non-negligible probability $\epsilon$.
    In the following, the event that $\qA$ wins in $\Hyb_{\rm xx}$ is denoted by $\event_{\rm xx}$.

\begin{description}
  \item[$\Hyb_0$:] This is the original RUF-VRA security game between the challenger and an adversary $\A$.  
  Namely, the game proceeds as follows.
\begin{enumerate}
    \item \label{Step:DS-proof-setup}
    The challenger computes the first message to the adversary as follows.
    \begin{itemize}
    \item 
    For $i\in [2n]$, it runs $\WDS.\KG(1^\lambda) \to (\WDS.\vk_i, \WDS.\msk_i, \WDS.\extk_i)$.
        \item It chooses random subset $S$ of $[2n]$ with size $n$.
        \item It runs $(\NTCF.\pp_i,\NTCF.\td_i) \lrun \NTCF.\FuncGen(1^\lambda, \mode)$ and $(\SFE.\crs_i,\SFE.\td_i) \lrun \SFE.\CRSGen(1^\secp, \mode)$ for $i\in [2n]$, where $\mode =1$ if $i\in S$ and $\mode =2$ if $i\in \overline{S}$.
    \end{itemize}
    Finally, it sends $\svk=\{\NTCF.\pp_i,\SFE.\crs_i, \WDS.\vk_i \}_{i\in [2n]}$ to $\A$.
    It privately keeps the master secret key $\msk = \{ \WDS.\msk_i \}_{i\in [2n]}$ and the deletion verification key $\dvk=\left( S, \{ \NTCF.\td_i ,  \SFE.\td_i\}_{i\in \overline{S}} \right)$.

    \item 
    Until $\qA$ receives the challenge input at Step~\ref{Step:DS-challenge-is-given} of the game, $\qA$ has an access to the signing oracle $\Sign(\msk,\cdot )$, which takes as input $\hatm=\hatm_1\| \cdots \|\hatm_{2n} \in \bin^{2n \lenm }$ and returns $\hatsigma =\WDS.\hatsigma_1\| \cdots \|\WDS.\hatsigma_{2n}$, where $\WDS.\hatsigma_i \lrun \WDS.\Sign(\WDS.\msk_i,\hatm_i)$. 
    \item The adversary $\A$ sends $\{y_i,\SFE.\msg^{(1)}_i\}_{i\in [2n]}$ to the challenger.
    \item \label{inner:DS-second-msg-to-adv}
    The challenger computes the message to the adversary as follows. 
    It first runs 
    \[
    \SFE.\Send(\SFE.\crs_i, \SFE.\msg^{(1)}_i, C[\WDS.\msk_i,y_i] ) \rrun \SFE.\msg^{(2)}_i
    \]       
    for $i\in [2n]$.
    Then, it sends $\{\SFE.\msg^{(2)}_{i} \}_{i\in [2n]}$ to $\A$.
    
    \item Then, $\A$ submits the deletion certificate 
    $\cert = \{d_i, c_i\}_{i \in [2n]} $.

    \item Then, the challenger runs $\DelVrfy$ as follows.
    \label{Step:DS-Del-verify-in-prf-proof}
    \begin{itemize}
%    \item If $\valid = 0$, it outputs $\bot$. 
%   
    \item For $i\in \overline{S}$, it runs $\SFE.\StaRcv(\SFE.\td_i, \SFE.\msg^{(1)}_i) \to \{ \alpha_{i,j}^{b} \}_{j\in [\lenx], b\in \bin } $. 
    \item For $i\in \overline{S}$, it computes 
    $d^*_i = (d^*_i[1],\ldots, d^*_i[\lenx]) \in \bin^\lenx$ as 
    $
    d^*_i[j] \seteq \langle c_{i,j}, (\alpha^0_{i,j} \oplus \alpha^1_{i,j})\rangle
    $.  
    \item For $i\in \overline{S}$, it runs $(x_{i,0},x_{i,1}) \lrun \NTCF.\Invert(\NTCF.\td_i, y_i) $.
    If the inversion fails for any $i\in \overline{S}$, the output of $\DelVrfy$ is $\bot$.
    \item For $i\in \overline{S}$, it checks whether 
    \begin{equation}\label{eq:DS-deletion-verify-in-the-proof}
    ( d_i \oplus d^*_i ) \cdot(x_{i,0}\oplus x_{i,1})=0
    ~
    \land
    ~
    \NTCF.\GoodSet(\NTCF.\td_i,y_i, d_i \oplus d^*_i ) = \top
    \end{equation}
    holds using $\NTCF.\td_i$. 
    If \cref{eq:DS-deletion-verify-in-the-proof} holds for all $i\in \overline{S}$, the output is $\top$ and otherwise, $\bot$.
    \end{itemize}
    If the output of the algorithm is $\bot$, then it indicates that the output of the game is $0$ (i.e., $\A$ failed to win the game). 
    Otherwise, the challenger returns 
    $\dvk$ to $\A$
    and continues the game.
    
    \item \label{Step:DS-challenge-is-given}
    Then, the challenger chooses a random input $m=m_1\|m_2\|\cdots \|m_{2n} \sample \bin^{2n\lenm}$ and sends $m$ to $\A$.

    After receiving the challenge input, $\qA$ loses its oracle access to the signing oracle.
    \item \label{Step:DS-winning-condition}
    $\A$ then outputs $\sigma'=\WDS.\sigma'_1\| \cdots \|\WDS.\sigma'_{2n}$.
    The challenger sets the output of the game to be $1$ (i.e., indicating that $\qA$ wins the game) if $\WDS.\Vrfy(\WDS.\vk_i, m_i, \WDS.\sigma'_i) = \top$ holds for all $i\in [2n]$. Otherwise, it outputs $\bot$.
\end{enumerate}
By the definition, we have $\Pr[\event_{0}]= \epsilon$.
\item[$\Hyb_1$:]  
  In this hybrid, the challenger computes $\SFE.\msg^{(2)}_i$ differently for $i\in S$ in Step \ref{inner:DS-second-msg-to-adv} of the game. Namely, it is replaced with the following: 
\begin{enumerate}[label=\arabic*', start=4]
    \item
    \label{Step:DS-step-4prime}
    The challenger computes the second message to the adversary as follows. 
    \begin{itemize}
    \item It runs 
  $\SFE.\Extract(\SFE.\td_i, \SFE.\msg^{(1)}_i) \to x_i$ 
  and $\NTCF.\Chk(\NTCF.\pp_i, x_i, y_i) = \verdict_i$ for $i\in S$.
        
        \item It computes 
          \[
              \SFE.\msg^{(2)}_i \lrun 
              \begin{cases}
              \SFE.\Sim(\SFE.\crs_i, \SFE.\msg^{(1)}_i, 1^{|C|}, \gamma_i )
                  & \text{if }  i\in S \\
             \SFE.\Send(\SFE.\crs_i, \SFE.\msg^{(1)}_i, C[\WDS.\msk_i,y_i] )
                  & \text{if } i\in \overline{S} 
              \end{cases},
        \]       
        where
  \[
  \gamma_i = 
  \begin{cases}
  \WDS.\Mark(\WDS.\msk_i, x_i)=
    \WDS.\sigk_i(x_i)  
      & \text{if }  \verdict_i = \top \\
    \bot 
      & \text{if } \verdict_i  = \bot. 
  \end{cases}
  \]
    \end{itemize}
    Finally, it sends $\{\SFE.\msg^{(2)}_{i} \}_{i\in [2n]}$ to $\A$.
\end{enumerate}

%  Note that by the change introduced in the previous hybrid, if the game proceeds to the point where $\SFE.\msg^{(2)}_i$ is computed, we have $x'_i = x_i$ for $i\in S$. 
  Since we have $\gamma_i = C[\WDS.\msk_i,y_i](x_i)$, it follows that this game is indistinguishable from the previous game by the straightforward reduction to the extractability of $\SFE$ in the extraction mode. 
  Thus, we have $|\Pr[\event_{0}] - \Pr[\event_{1}]| = \negl(\secp)$.
  Note that the simulation of this game does not require any marked signing key for the $i$-th instance of $\WDS$ if $i\in S$ and $\verdict_i =\bot$.

   \item[$\Hyb_2$:]  
    In this hybrid, the challenger aborts the game and sets the output of the game to be $0$ (i.e., $\qA$ loses the game), once it turns out that there is $i\in S$ such that $\NTCF.\Chk(\NTCF.\pp_i, x_i, y_i)=\bot$.
    Namely, we replace Step \ref{Step:DS-step-4prime} in the previous hybrid with the following: 
\begin{enumerate}[label=\arabic*'', start=4]
    \item \label{Step:DS-step-4''}
    The challenger computes the second message to the adversary as follows. 
    \begin{itemize}
    \item It runs 
  $\SFE.\Extract(\SFE.\td_i, \SFE.\msg^{(1)}_i) \to x_i$ 
  and $\NTCF.\Chk(\NTCF.\pp_i, x_i, y_i) = \verdict_i$ for $i\in S$.
  \item If $\verdict_i = \bot$ for any $i\in S$, it aborts the game and sets the output of the game to be $0$.
        \item It computes 
          \[
              \SFE.\msg^{(2)}_i \lrun 
              \begin{cases}
              \SFE.\Sim(\SFE.\crs_i, \SFE.\msg^{(1)}_i, 1^{|C|},  \WDS.\sigk_i(x_i)   )
                  & \text{if }  i\in S \\
             \SFE.\Send(\SFE.\crs_i, \SFE.\msg^{(1)}_i, C[\WDS.\msk_i,y_i] )
                  & \text{if } i\in \overline{S}. 
              \end{cases}
        \]       
    \end{itemize}
    Finally, it sends $\{\SFE.\msg^{(2)}_{i} \}_{i\in [2n]}$ to $\A$.
\end{enumerate}

We claim $\Pr[\event_{2}]\geq  \Pr[\event_{1}] -\negl(\secp)$.
To show this, we observe that the only scenario where the adversary $\qA$ wins the game in the previous hybrid but not in the current one occurs when $\qA$ chooses $y_i$ such that $\verdict_i =\bot$ for some $i\in S$, yet still succeeds in signing on $\{m_j\}_{j\in [2n]}$. 
However, this implies that $\qA$ has successfully forged a signature $\sigma_i$ satisfying $\WDS.\Vrfy(\WDS.\vk_i,m_i,\sigma_i)$ without getting marked signing key for the $i$-th instance of $\WDS$. 
Note that the probability that $\qA$ has made a signing query to $\WDS.\Sign(\WDS.\msk_i, \cdot )$ on message $m_i$ is negligible, since $m_i$ is chosen uniformly at random and $\qA$ is not allowed to make signing queries once it receives $m_i$.
By a straightforward reduction to the EUF-CMA security of $\WDS$, the claim follows.  
%\synote{Alternatively, we may assume $\WDS$ is EUF-CMA, which is simpler.}

  \item[$\Hyb_3$:]  
  In this hybrid, we change the winning condition of the adversary. 
  Namely, we replace Step \ref{Step:DS-challenge-is-given} and Step \ref{Step:DS-winning-condition} of the game with the following. 
  \begin{enumerate}[label=\arabic*', start=7]
      \item
      The challenger parses the adversary $\qA$ right before Step~\ref{Step:DS-challenge-is-given} of the game as a unitary circuit $U_\qA$ and a quantum state $\rho_\qA$.
      It then constructs a quantum forger $\qF$ for $\WDS$ that takes $\{m_i\}_{i\in S}$ as input and outputs $\{ \WDS.\sigma'_i \}_{i\in S}$ 
      from $\A$ as follows.
     \begin{description}
         \item[$\qF( \{m_i\}_{i\in S} )$:]
         It chooses $m_i \gets \bin^\lenm$ for $i\in \overline{S}$ and runs $\A$ on input $ \{ m_i \}_{i\in [2n]} $ to obtain $\{\WDS.\sigma'_i\}_{i\in [2n]}$.
         Finally, it outputs $\{ \WDS.\sigma'_i\}_{i\in S}$.  
    \end{description}
     \item 
     The challenger parses
     $\qF$ as a pair of unitary $U_\qF$ and a quantum state $\rho_\qF$.
     Then, it runs the extraction algorithm as 
         $\WDS.\qExt( \{\WDS.\extk_i\}_{i\in S} , (U_\qF, \rho_\qF ) ) \to \{x'_i\}_{i\in S}$.
    \item \label{Step:DS-step-9prime}
    If $x_i = x'_i$ holds for all $i\in S$, 
    the challenger sets the output of the game to be $1$. 
    Otherwise, it sets it to be $0$.
 \end{enumerate}
  By the parallel extractability of $\WDS$,  
  we have $\Pr[\event_{3}] \geq \Pr[\event_{2}]-\negl(\secp)$.
  The proof is almost the same as 
  \cref{lem:reduction-to-par-ext-PRF} and thus omitted.
\item[$\Hyb_4$:]  
    In this game, we replace Step \ref{Step:DS-step-9prime} of the above game with the following:
    \begin{enumerate}[label=\arabic*'', start=9]
        \item If $\NTCF.\Chk(\NTCF.\pp_i, x'_i, y_i)=\top$ holds for all $i\in S$, 
    the challenger sets the output of the game to be $1$. 
    Otherwise, it sets it to be $0$. 
    \end{enumerate}
    We observe that $x'_i = x_i$ in Step \ref{Step:DS-step-9prime} implies $\NTCF.\Chk(\NTCF.\pp_i, x'_i, y_i) = \NTCF.\Chk(\NTCF.\pp_i, x_i, y_i) =\top$.
    To see the latter equality holds, recall that the game is set to abort if $\NTCF.\Chk(\NTCF.\pp_i, x_i, y_i) = \bot$ for any $i \in S$ 
    in Step~\ref{Step:DS-step-4''}. Thus, if the game continues to Step \ref{Step:DS-step-9prime}, the equality holds.   
    We therefore have $\Pr[\event_{4}] \geq  \Pr[\event_{3}]$.
    (Actually, $\Pr[\event_{4}] = \Pr[\event_{3}]$ holds due to the uniqueness of $x_i$ that passes the check. However, showing the inequality suffices for our proof.)
\item[$\Hyb_5$:]  
    In this game, we switch Step \ref{Step:DS-step-4''} in the previous hybrid back to Step \ref{Step:DS-step-4prime}. 
    The difference from the previous hybrid is that it continues the game even if $\verdict_i = \bot$ for some $i\in S$.
    Since the adversary had no chance in winning the game in such a case in the previous hybrid, this change only increases the chance of the adversary winning. 
    We therefore have $\Pr[\event_{5}] \geq \Pr[\event_{4}]$.
  
\item[$\Hyb_6$:]  
    In this hybrid, we switch Step \ref{Step:DS-step-4prime} in the previous hybrid back to Step \ref{inner:DS-second-msg-to-adv}.
    Namely, it computes $\SFE.\Send( \SFE.\crs_i, \allowbreak\SFE.\msg^{(1)}_i, C[\WDS.\msk_i,y_i] ) \rrun \SFE.\msg^{(2)}_i$ for all $i\in [2n]$ again.
    Similarly to the change from $\Hyb_0$ to $\Hyb_1$, the winning probability of $\A$ in this game does not change significantly by the extractability of $\SFE$ in the extraction mode. %\shota{The name of the security notion should be fixed.}
    Namely, we have $|\Pr[\event_{6}] - \Pr[\event_{5}]| = \negl(\secp)$.
    Note that we do not need $\{\SFE.\td_i\}_{i\in S}$ any more to run the game.

\item[$\Hyb_7$:] 
    In this hybrid, we change the way $\SFE.\crs_i$ is sampled for $i\in S$.
    Namely, it is sampled as $(\SFE.\crs_i,\SFE.\td_i) \lrun \SFE.\CRSGen(1^\secp, 2)$ for all $i\in [2n]$. 
    Concretely, we replace Step \ref{Step:DS-proof-setup} with the following:
\begin{enumerate}[label=\arabic*']
    \item \label{Step:DS-proof-setup-1-prime}
    The challenger computes the first message to the adversary as follows.
    \begin{itemize}
        \item 
        For $i\in [2n]$, it runs $\WDS.\KG(1^\lambda) \to (\WDS.\vk_i, \WDS.\msk_i, \WDS.\extk_i)$.
        \item It chooses random subset $S$ of $[2n]$ with size $n$.
        \item It runs $(\NTCF.\pp_i,\NTCF.\td_i) \lrun \NTCF.\FuncGen(1^\lambda, \mode)$ and $(\SFE.\crs_i,\SFE.\td_i) \lrun \SFE.\CRSGen(1^\secp, 2)$ for $i\in [2n]$, where $\mode =1$ if $i\in S$ and $\mode =2$ if $i\in \overline{S}$.
    \end{itemize}
    Finally, it sends $\svk=\{\NTCF.\pp_i,\SFE.\crs_i, \WDS.\vk_i\}_{i\in [2n]}$ to $\A$.
    It privately keeps the master secret key $\msk = \{ \WDS.\msk_i \}_{i\in [2n]}$ and the deletion verification key $\dvk=\left( S, \{ \NTCF.\td_i ,  \SFE.\td_i\}_{i\in \overline{S}} \right)$.  
\end{enumerate}  
    Since we do not need $\SFE.\td_i$ for $i\in S$ to run $\Hyb_6$ and $\Hyb_7$, a straightforward reduction to the mode indistinguishability of $\SFE$ implies that $\Hyb_6$ and $\Hyb_7$ are computationally indistinguishable. 
    Therefore, we have $|\Pr[\event_{6}] - \Pr[\event_{7}]| = \negl(\secp)$.
    \item[$\Hyb_8$:] 
    In this game, we change the way how the challenger runs $\DelVrfy$.
    In particular, the challenger computes $\SFE.\StaRcv(\SFE.\td_i, \SFE.\msg^{(1)}_i) \to \{ \alpha_{i,j}^{b} \}_{j, b} $ and $d^*_i$ for all $i\in [2n]$, 
    instead of only for $i\in \overline{S}$.  
    Still, the inversion of $y_i$ and 
    the check of \cref{eq:DS-deletion-verify-in-the-proof} are performed only for $i\in \overline{S}$. 
    Concretely, we replace Step \ref{Step:DS-Del-verify-in-prf-proof} with the following.
\begin{enumerate}[label=\arabic*', start=6]
        \item Then, the challenger runs (the modified version of) $\DelVrfy$ as follows.
    \begin{itemize}
%    \item If $\valid = 0$, it outputs $\bot$. 
%   
    \item For $i\in [2n]$, it runs $\SFE.\StaRcv(\SFE.\td_i, \SFE.\msg^{(1)}_i) \to \{ \alpha_{i,j}^{b} \}_{j\in [\lenx], b\in \bin } $. 
    \item For $i\in [2n]$, it computes 
    $d^*_i = (d^*_i[1],\ldots, d^*_i[\lenx]) \in \bin^\lenx$ as 
    $
    d^*_i[j] \seteq \langle c_{i,j}, (\alpha^0_{i,j} \oplus \alpha^1_{i,j})\rangle
    $.  
    \item For $i\in \overline{S}$, it runs $(x_{i,0},x_{i,1}) \lrun \NTCF.\Invert(\NTCF.\td_i, y_i) $.
    If the inversion fails for any $i\in \overline{S}$, the output of $\DelVrfy$ is $\bot$.
    \item For $i\in \overline{S}$, it checks whether 
    \begin{equation}
    ( d_i \oplus d^*_i ) \cdot(x_{i,0}\oplus x_{i,1})=0
    ~
    \land
    ~
    \NTCF.\GoodSet(\NTCF.\td_i,y_i, d_i \oplus d^*_i ) = \top
    \end{equation}
    holds using $\NTCF.\td_i$. 
    If the above holds for all $i\in \overline{S}$, the output is $\top$ and otherwise, $\bot$.
    \end{itemize}
    If the output of the algorithm is $\bot$, then it indicates that the output of the game is $0$ (i.e., $\A$ failed to win the game). 
    Otherwise, the challenger returns 
    $\dvk=\left( S, \{ \NTCF.\td_i,  \SFE.\td_i \}_{i\in \overline{S}} \right)$ to $\A$
    and continues the game.
\end{enumerate}
    
    We can see that this change is only conceptual, since the outcome of the verification is unchanged.  
    We therefore have
    $\Pr[\event_{8}] = \Pr[\event_{7}] $.
\end{description}
    Based on the above discussion, it is established that 
    $\Pr[\event_{8}] \geq \epsilon - \negl(\secp)$, which is non-negligible by our assumption.
    However, we can show $\Pr[\event_{8}]=\negl(\secp)$ assuming the cut-and-choose adaptive hardcore bit property of $\NTCF$, by a very similar reduction to \cref{lem:PKE-reduction-to-ada-HC}.
    This results in a contradiction as desired.
\end{proof}

% \begin{lemma}\label{lem:reduction-to-par-ext-DS}
%     If $\WDS$ satisfies parallel extractability, we have $\Pr[\event_{3}] \geq \Pr[\event_{2}]-\negl(\secp)$. 
% \end{lemma}
% \begin{proof}
%     ddd
% \end{proof}
\else

\fi
\noindent
\textbf{Acknowledgment.}  
Generative AI tools were employed to improve the readability and presentation of this manuscript. We thank Duo Xu for answering our questions about \cite{takeuchi2025computationalcertifieddeletionproperty}. We thank anonymous reviewers of CRYPTO 2026 and TQC 2026 for their helpful comments. 
Shota Yamada is supported by JPMJCR22M1 and JPMJKP24U3.

	\ifnum\llncs=1
\bibliographystyle{splncs04}
\bibliography{abbrev3,crypto,siamcomp_jacm,other}
	\else
\bibliographystyle{alpha} 
\bibliography{abbrev3,crypto,siamcomp_jacm,other}
	\fi

\ifnum\cameraready=0
	\ifnum\llncs=0
	%%%%%% Full version region %%%%%%
	\appendix
% !TEX root = main.tex

\newcommand{\tlC}{\widetilde{C}}

\section{Special Dual-Mode SFE from LWE}\label{sec:dual-mode-SFE-from-LWE}
In this section, we show a construction of special dual-mode SFE from the LWE assumption.

In general, by relying on well-known Yao's protocol, SFE can be constructed from oblivious transfer (OT) by additionally using garbled circuits which can be based on OWFs.
We can prove that the resulting scheme becomes special dual-mode SFE if the underlying OT satisfies an analogous special dual-mode property.
In this section, we first formally define special dual-mode OT and provide its realization assuming the LWE assumption based on the construction by \cite{C:PeiVaiWat08}.
Then, we provide a sketch on how to achieve dual-mode SFE from dual-mode OT using GC.

\subsection{Definition of Special Dual-mode OT}\label{sec:def-dual-mode-OT}
In this subsection, we define special dual-mode OT.  
\begin{definition}[Special Dual-mode OT]\label{def:dual-mode-OT}
A special dual-mode OT scheme $\OT$ is a tuple of four algorithms $\OT=\OT.(\CRSGen, \Receive{1},\allowbreak \Send, \Receive{2})$.  
Below, let $\bit^\ell$ be the string  space of $\OT$. %\shota{How about string?} \takashi{okay}
\begin{description}
%\item[$\Setup(1^\secp,1^{\numkey})\ra\msk$:] The setup algorithm takes a security parameter $1^\lambda$ and a collusion bound $1^{\numkey}$, and outputs a master secret key $\msk$.
\item[$\CRSGen(1^\secp,\mode )\ra (\crs,\td)$:] The CRS generation algorithm is a PPT algorithm that takes a security parameter $1^\lambda$ and $\mode\in \{1,2\}$, and outputs strings $\crs$ and $\td$.
$\mode=1$ indicates that the system in the ``extractable mode", while $\mode=2$ indicates ``hiding mode".

\item[$\Receive{1}(\crs,b)\ra(\msg^{(1)},\st)$:] This is a classical PPT algorithm supposed to be run by a receiver to generate the first message of the protocol. 
It takes as input the CRS $\crs$ and a bit $b\in \bin$ and outputs the first message $\msg^{(1)}$ and a secret state $\st$. 
\item[$\Send(\crs,\msg^{(1)}, (z_0,z_1) )\ra\msg^{(2)}$:] This is a classical PPT algorithm supposed to be run by a sender to generate the second message of the protocol. 
It takes as input the CRS $\crs$, a message $\msg^{(1)}$ from the receiver, 
and a pair of strings $(z_0,z_1)\in\bit^{\ell\times 2}$ and outputs the second message $\msg^{(2)}$.
\item[$\Receive{2}(\crs,b, \st, \msg^{(2)})\ra z$:] This is a classical PPT algorithm supposed to be run by a receiver to derive the output. 
It takes as input the CRS $\crs$, a bit $b\in \bin$, a state $\st$, and a message $\msg^{(2)}$ from the sender, and outputs a string $z$. 
\end{description}
We require dual-mode OT to satisfy the following properties.
\begin{description}
\item[Evaluation correctness:]For all $b \in \bin$, $\mode \in \{1,2\}$, and a pair of strings $(z_0,z_1)\in \bit^{\ell\times 2}$, we have
\begin{align}
\Pr\left[
\Receive{2}(\crs,b, \st, \msg^{(2)}) = z_b
\ \middle |
\begin{array}{ll}
(\crs,\td) \la \CRSGen(1^\secp,\mode ) \\
(\msg^{(1)},\st) \la \Receive{1}(\crs,b)  \\ 
\msg^{(2)} \la \Send(\crs,\msg^{(1)}, (z_0,z_1) )
\end{array}
\right] 
=1.
\end{align}
\item[Unique state:]
For all 
$\mode\in \{1,2\}$, 
$(\crs,\td) \in \Supp(\CRSGen(1^\secp,\mode))$, $b\in \bin$,  and $(\msg^{(1)},\st) \in \Supp(\Receive{1}(\crs,b))$, $\st$ is the unique state that corresponds to 
$\crs$, $b$, and $\msg^{(1)}$,  that is, there is no $\st^\prime\ne \st$ such that $(\msg^{(1)},\st^\prime) \in \Supp(\Receive{1}(\crs,b))$. 
We denote the unique state $\st$ such that  $(\msg^{(1)},\st) \in \Supp(\Receive{1}(\crs,b))$ by $\st_{b\to \msg^{(1)}}$.\footnote{The state also depend on $\crs$, but we omit it from the notation.}
\item[Mode indistinguishability:]
For all QPT algorithm $\cA$, we have 
    \begin{align}
    \abs{
    \Pr[\cA(\crs)=1: (\crs,\td) \la \CRSGen(1^\secp,1) ]-
    \Pr[\cA(\crs)=1: (\crs,\td) \la \CRSGen(1^\secp,2) ]
    }\leq  \negl(\secp).
    \end{align}
\item[Statistical security against malicious senders in the hiding mode:]
For all $\crs$ output by $\CRSGen(1^\secp, 2)$,  we require the following statistical indistinguishability:
\[
\msg^{(1)}_0 \approx_s 
\msg^{(1)}_1
\quad
\mbox{where}
\quad 
(\msg^{(1)}_b,\st)\lrun \Receive{1}(\crs,b) \quad \mbox{for} \quad b\in \bin.
\]
\item[State recoverability in the hiding mode:]
%We require that the randomness used by $\Receive{1}$ is divided into two parts $\st\in \bin^{\lenst \lenx}$ for some polynomial $u=\poly(\secp)$ and $\ernd$. While the former is output by $\Receive{1}$ as its secret state and reused by $\Receive{2}$, the latter is ephemeral and is discarded.  
%Namely, we have $\Receive{1}(\crs, x; (\st, \ernd)) = (\msg^{(1)},\st)$. 
%Furthermore, we require that 
There exists a classical PPT algorithm $\StaRcv$ that takes as input the trapdoor $\td$ and the first message of the receiver $\msg^{(1)}$ and outputs a pair of strings $(\alpha^0,\alpha^1)$  or $\bot$, where the latter indicates that the state recovery failed. 
We require that for all $b\in \bin$, the following holds: 
\begin{align}
\Pr\left[
\st = \alpha^{b} 
\ \middle |
\begin{array}{ll}
(\crs,\td) \la \CRSGen(1^\secp,2 ) \\
(\msg^{(1)},\st) \la \Receive{1}(\crs,b)  \\ 
(\alpha^0,\alpha^1) \lrun \StaRcv(\td, \msg^{(1)})
\end{array}
\right] 
=1.
\end{align}
%Note that the above requirement implies that for all $\msg^{(1)}$ and $b\in \bin$, there exists at most one $\st \in \bin^{\lenst \lenx}$ such that there exists randomness $\ernd$satisfying $(\msg^{(1)},\st) \la \Receive{1}(\crs,x; \ernd)$.  
%
\item[Extractability in the extraction mode:]
For classical algorithms $\Extract$ and $\Sim$, let us consider the following experiment formalized by the experiment $\expb{\OT,\qA, \Extract, \Sim}{ext}{sim}(1^\secp,\coin)$ between an adversary $\qA$ and the challenger:
\begin{enumerate}
            \item  The challenger runs $(\crs, \td)\sample \CRSGen(1^\secp, 1)$ and sends $(\crs, \td)$ to $\qA$. 
            \item $\qA$ sends $\msg^{(1)}$ and a pair of messages $(z_0,z_1)\in\bit^{\ell\times 2}$.  
            \item The challenger runs $\Ext(\td, \msg^{(1)}) \to b$ 
            and 
            it computes 
            \[
            \msg^{(2)} \lrun 
            \begin{cases}
              \Send(\crs, \msg^{(1)}, (z_0,z_1)), & \text{if } \coin = 0 \\ 
              \Sim(\crs, \msg^{(1)},z_b ) & \text{if } \coin = 1. 
            \end{cases}
            \]
            Then, it gives $\msg^{(2)}$ to $\qA$.
            
            \item $\qA$ outputs a guess $\coin^\prime$ for $\coin$. The challenger outputs $\coin'$ as the final output of the experiment.
        \end{enumerate}
        For any QPT $\qA$, it holds that
\begin{align}
\advb{\OT,\qA, \Extract, \Sim}{ext}{sim}(\secp) \seteq \abs{\Pr[\expb{\OT,\qA, \Extract, \Sim}{ext}{sim} (1^\secp,0) = 1] - \Pr[\expb{\OT,\qA, \Extract, \Sim }{ext}{sim} (1^\secp,1) = 1] }\leq \negl(\secp).
\end{align} 
\takashi{In fact we can achieve security against unbounded adv?}

\item[Efficient state superposition:]
We require that there exists a QPT ``coherent version" $\qReceive{1}$ of $\Receive{1}$ that is given a quantum state $\ket{\psi}=\sum_{b\in \bin}\alpha_b\ket{b}$ 
and $\crs$ and works as follows:

It first generates a state that is within negligible trace distance of the following state:

\[
\sum_{b\in \bit}\alpha_b\sum_{\msg^{(1)} \in \Supp( \Receive{1}(\crs, b) )} 
\sqrt{ p_{b \to \msg^{(1)}} }
\ket{b} 
\ket{\st_{b\to \msg^{(1)}} }\ket{\msg^{(1)}} 
\]
where $p_{b\to \msg^{(1)}}$ is the probability that $\Receive{1}(\crs, b)$ outputs $\msg^{(1)}$. 
Then, it measures the last register to obtain $\msg^{(1)}$ and outputs the residual quantum state $\ket{\psi'}$.

\if0
We require that there exists a ``coherent version" $\qReceive{1}$ of $\Receive{1}$ that is given a quantum state $\ket{\psi}$ that is represented by a qubits of length $\ell$ and $\crs$ and works as follows:
It first implements the following efficient isometry up to negligible error: 
\[
\ket{b} \mapsto 
\sum_{\msg^{(1)} \in \Supp( \Receive{1}(\crs, b) )} 
\sqrt{ p_{b \to \msg^{(1)}} }
\ket{b} 
\ket{\st_{b\to \msg^{(1)}} }\ket{\msg^{(1)}} 
~
\mbox{for all $b\in \bin$},
\]
where $p_{b\to \msg^{(1)}}$ is the probability that $\Receive{1}(\crs, b)$ outputs $\msg^{(1)}$.  %and $\st_{b\to \msg^{(1)}}$ is the unique state such that there exists randomness $\ernd$satisfying $(\msg^{(1)},\st) \la \Receive{1}(\crs,x; \ernd) $.
%
Then, it measures the last register to obtain $\msg^{(1)}$ and outputs the residual quantum state $\ket{\psi'}$.
\fi

We remark that implementing the above by simply running the classical algorithm in the superposition does not work, since it may leave the entanglement with the randomness for $\Receive{1}$. 
\end{description}
\end{definition}

\subsection{Dual-mode OT from LWE}\label{sec:dual-mode_OT_LWE}
In this subsection, we show that the dual-mode OT based on LWE in \cite{C:PeiVaiWat08} satisfies our requirements of special dual-mode OT as in \Cref{def:dual-mode-OT} under a certain parameter regime.
%As already mentioned, this suffices to give special dual-mode SFE as in \Cref{def:dual-mode-SFE} using Yao's protocol. 
\paragraph{\bf Lattice preliminaries}\label{sec:lattice_pre} 
%A lattice $\Lambda\subseteq \RR^m$ is a discrete additive subgroup of $\RR^m$. 
%A shifted lattice is a subset that can be written as $\Lambda+\vt=\{\vx+\vt:\vx\in \Lambda\}$ by using a lattice $\Lambda$ and a vector $\vt\in \RR^n$. 
We review definitions and lemmata on lattice-based cryptography. 

For $q\in \NN$, we write $\ZZ_q$ to mean the  additive group of integers modulo $q$. We represent elements of $\ZZ_q$ by elements in $\ZZ\cap (-q/2,q/2]$ and often identify the elements of $\ZZ_q$ and their representatives. 
%For $\mA\in \ZZ_q^{n\times m}$ and $\vy\in \ZZ_q^n$, we define  $\Lambda_\mA^\bot:=\{\vx\in \ZZ^m:\mA\vx=\vzero \mod q\}$ and $\Lambda_\mA^\vy:=\{\vx\in \ZZ^m:\mA\vx=\vy \mod q\}$. It is easy to see that $\Lambda_\mA^\bot$ is a lattice and $\Lambda_\mA^\vy$ is a shifted lattice written as $\Lambda_\mA^\vy=\Lambda_\mA^\bot+ \vt$ for any $\vt\in \Lambda_\mA^\vy$.   
For a vector $\vx\in \mathbb{R}^n$, $\|\vx\|$ denotes the  Euclidean norm of $\vx$ and $\|\vx\|_{\infty}$ denotes the infinity norm of $\vx$, that is, $\|\vx\|_{\infty}\seteq \max_{i\in [n]}|x_i|$. 
%we define $\|\vx\|_{\infty}\seteq \max_{i\in [n]}|x_i|$.  Similarly, for a matrix $\mA=(a_{i,j})_{i\in[n],j\in[m]}\in \mathbb{R}^{n\times m}$, we define $\|\mA\|_\infty\seteq \max_{i\in[n],j\in[m]}a_{i,j}$.  We often use the inequality $\|\mA\mB\|_\infty\le m\|\mA\|_\infty\|\mB\|_\infty$ for $\mA\in \RR^{n\times m}$ and $\mB\in\RR^{m\times \ell}$. 

For $m\in \NN$, the Gaussian function $\rho_\sigma:\RR^m \ra [0,1]$ with a scaling parameter $\sigma>0$ is defined as follows:
\[
\rho_{\sigma}(\vx)\seteq \exp(-\pi\|\vx\|^2/\sigma^2).
\]
%For any coset  $\Lambda-\vt$, we define 
%\[
%\rho_{\sigma}(\Lambda-\vt):=\rho_\sigma
%\]
For any $q\in \mathbb{N}$ and $\sigma>0$,  
we define the truncated discrete Gaussian distribution $D_{\ZZ_q,\sigma}$ with the support $\{x\in\ZZ_q:|x|\le \sigma\}$
by the following probability density function: 
\[
D_{\ZZ_q,\sigma}(x)=\frac{\rho_\sigma(x)}{\sum_{y\in \ZZ_q:|y|\le \sigma}\rho_\sigma(y)}.
\] 
More generally, for any $q,m\in \mathbb{N}$ and $\sigma>0$,  we define the truncated discrete Gaussian distribution $D_{\ZZ_q^m,\sigma}$ with the support $\{\vx\in\ZZ_q^m:\|\vx\|_\infty\le \sigma\}$
by the following probability density function: 
\[
D_{\ZZ_q^m,\sigma}(x_1,\ldots,x_m)=\prod_{i\in[m]}D_{\ZZ_q,\sigma}(x_i).
\]

\begin{lemma}[Noise smudging~\cite{JACM:BCMVV21}]\label{lem:smudging} 
For $q,m\in \mathbb{N}$, $\sigma>0$, and $\ve\in \ZZ_q^m$ such that $\|\ve\|\le \sigma\sqrt{m}$,   the statistical distance between $D_{\ZZ_q^m,\sigma}$ and $D_{\ZZ_q^m,\sigma}+\ve$ is at most $\sqrt{2\left(1-\exp\left(-2\pi\sqrt{m}\|\ve\|/\sigma\right)
\right)}$ \takashi{Is there a better citation?}
\end{lemma}

\begin{definition}[Learning with errors]
Let $n,m,q\in\mathbb{N}$ and $\chi$ be a distribution over $\ZZ_q^m$, where $n,m,q,\chi$ depends on the security parameter $\secp$.    
We say that the learning with errors (LWE) problem $\LWE_{n,m,q,\chi}$ is quantumly hard if for any QPT adversary $\qA$, 
\[
\abs{\Pr[\qA(\mA,\vs^\trans \mA+\ve^\trans)= 1]-
\Pr[\qA(\mA,\vu)= 1]
}\le \negl(\secp)
\]
where $\mA\gets \ZZ_q^{n\times m}$, $\vs\gets \ZZ_q^n$, and $\ve\gets \chi$. 
\end{definition}

\begin{lemma}[Lattice trapdoors~\cite{EC:MicPei12}]\label{lem:lattice_trapdoor}
    Let $n,m,q\in \NN$ with $m=\Omega(n\log q)$. There exists a pair of PPT algorithms $(\TrapGen,\Invert)$ satisfying the following. 
    \begin{itemize}
        \item $\TrapGen(1^n,1^m,q)\rightarrow (\mA,\mT):$ On input $1^n$, $1^m$, and $q$, it outputs a matrix $\mA\in \ZZ_q^{n\times m}$ and a trapdoor $\mT$. The statistical distance between the distribution of $\mA$  generated by $\TrapGen(1^n,1^m,q)$ and the uniform distribution over $\ZZ_q^{n\times m}$ is negligible in $n$. 
        \item $\Invert(\mA,\mT,\vv)\rightarrow (\vs',\ve'):$ On input $\mA$, $\mT$, and $\vv$, it returns vectors $\vs' \in \ZZ_q^n$ and $\ve'\in \ZZ_q^m$. 
        The following holds for some constant $C>0$: 
        If $(\mA,\mT)$ is generated by $\TrapGen(1^n,1^m,q)$ and $\vv=\vs^\trans\mA+\ve^\trans$ for $\vs\in \ZZ_q^n$ and $\ve\in \ZZ_q^m$ such that $\|\ve\|\le q/(C\sqrt{n\log q})$, $\Invert(\mA,\mT,\vv)$ outputs $(\vs,\ve)$ with an overwhelming probability in $n$.
    \end{itemize}
\end{lemma}

\begin{lemma}[\cite{STOC:GenPeiVai08,C:PeiVaiWat08}\footnote{In \cite{STOC:GenPeiVai08,C:PeiVaiWat08}, they consider the untruncated discrete Gaussian distribution but this does not affect the statement since the truncated and untruncated discrete Gaussian distributions are negligibly close.}]\label{lem:messy}
Let $n,m\in \mathbb{N}$ and $q$ be a prime. For any $\delta,\tau>0$, let $M_{\delta,\tau}\subseteq \ZZ_q^{n\times m}\times \ZZ_q^m$ be the subset consisting of $(\mA,\vv)\in \ZZ_q^{n\times m}\times \ZZ_q^m$ such that the distribution $\{(\mA\vt,\vv^\trans \vt):\vt\gets D_{\ZZ_q^m,\tau}\}$  has statistical distance within at most $\delta$ from the uniform distribution.    There is a PPT algorithm $\IsMessy$ with  single-bit outputs that satisfies the following. 
If $m\ge 2(n+1)\log q$ and $\tau\ge \sqrt{qm}\log^2 m$, there is $\delta=\negl(n)$ such that  for an overwhelming fraction of $(\mA,\mT)\gets \TrapGen(1^n,1^m,q)$,  the following properties hold.
\begin{itemize}
%\item We have $\Pr[(\mA,\vv)\in M_{\delta,\tau} \mid (\mA,\vv)\gets \ZZ_q^{n\times m}\times \ZZ_q^m]\ge 1-2q^{-n}$.
%\item There is a PPT algorithm $\IsMessy$ with  single-bit outputs such that for an overwhelming fraction of $(\mA,\mT)\gets \TrapGen(1^n,1^m,q)$, the following hold
%\begin{enumerate}
\item For all but an at most $(1/2\sqrt{q})^m$ fraction of $\vv\in \ZZ_q^m$, $\IsMessy(\mA,\mT,\vv)$ returns $1$ with overwhelming probability.  
\item For any $\vv\in \ZZ_q^m$ such that $(\mA,\vv)\notin M_{\delta,\tau}$, $\IsMessy(\mA,\mT,\vv)$ returns $0$ with an overwhelming probability. 
%\end{enumerate}
\end{itemize}
\end{lemma}

\paragraph{\bf Construction} 
It is easy to see that all the properties required for special dual-mode OT are preserved under parallel repetition. Therefore, it suffices to construct it for the case of $\ell=1$. 
The scheme is described below, where the parameters are specified below the description. 
\begin{description} 
\item[$\CRSGen(1^\secp,\mode )$:]
Given the security parameter $1^\secp$ and $\mode\in \{1,2\}$, generate $(\mA,\mT)\gets \TrapGen(1^n,1^m,q)$. 
\begin{itemize}
        \item If $\mode=1$ (i.e., in the extractable mode), choose $\vv\gets \ZZ_q^m$. 
        \item If $\mode=2$ (i.e., in the hiding mode), choose $\vs\gets \ZZ_q^n$ and 
        $\ve\gets D_{\ZZ_q^m,\sigma}$
        and compute $\vv^\trans=\vs^\trans\mA+\ve^\trans$.
\end{itemize}
Output $\crs=(\mA,\vv)$ and $\td=(\mA,\mT,\vv)$.  

\item[$\Receive{1}(\crs,b)$:] 
On input $\crs=(\mA,\vv)$ and $b\in \bit$, choose $\vr\gets \ZZ_q^n$
and   $\ve^\prime\gets D_{\ZZ_q^m,\sigma^\prime}$
, compute $\vu_b^\trans=\vr^\trans\mA+{\ve^\prime}^\trans$
and $\vu_{1-b}^\trans=\vu_b^\trans+(-1)^b\vv^\trans$
, and output $\msg^{(1)}=\vu_0$ and $\st=\vr$.  
\item[$\Send(\crs,\msg^{(1)}, (z_0,z_1) )\ra\msg^{(2)}$:] 
On input $\crs$, $\msg^{(1)}=\vu_0$,  and $(z_0,z_1)\in\bit^2$, let $\vu_1=\vu_0+\vv$, 
choose $\vt_b\gets D_{\ZZ_q^m,\tau}$,  compute $\vw_b=\left(\begin{matrix}
\mA\\
\vu_b^\trans
\end{matrix}\right)\vt_b+
\left(\begin{matrix}
\mathbf{0}\\
z_b\lfloor q/2\rfloor \mod q
\end{matrix}\right)
$ for $b\in \bit$, and output $\msg^{(2)}=(\vw_0,\vw_1)$.  
\item[$\Receive{2}(\crs,b, \st, \msg^{(2)})\ra z$:] 
On input $\crs$, $b\in\bit$, $\st=\vr$, and $\msg^{(2)}=(\vw_0,\vw_1)$, 
let $\vw_{b,[1,n]}\in \ZZ_q^n$ be the vector consisting of the first $n$ coordinates of $\vw_b$ and $w_{b,n+1}\in \ZZ_q$ be the remaining coordinate, 
output $0$ if $\left|w_{b,n+1}-\vr^\trans\vw_{b,[1,n]} \mod q\right|<q/4$ and output $1$ otherwise. 
\end{description}
\paragraph{\bf Parameters}
We choose the parameters $n,m,q,\sigma,\sigma^\prime,\tau$ so that the following hold.
\begin{itemize}
\item $m\ge 2(n+1)\log q$
\item $\LWE_{n,m,q,D_{\ZZ_q^m,\sigma}}$ is quantumly hard.
We need $O(q/\sigma) < 2^{n^\theta}$ for some $0<\theta<1$. 
\item $\sigma'/\sigma=2^{\omega(\log \secp)}$
\item $\sigma^\prime\le q/(C\sqrt{nm\log q})$ where $C$ is the constant in \Cref{lem:lattice_trapdoor}. 
\item $\tau\ge \sqrt{qm}\log^2 m$
\item $m^2\sigma'\tau\le q/8$
\end{itemize}
For example, we can choose the parameters as follows (for sufficiently large $\secp$):
\[
n=\secp^c, \quad 
m=3n^2,\quad q\approx 2^\secp,  \quad 
\sigma=n^2\quad
\sigma^\prime=2^{\secp /6}, \quad
\tau=2^{2\secp /3}.
\]
where $c>1$ is some constant. 
%\synote{Taking $q=2^n$ may not be secure due to LLL. I changed the parameters so that the noise-to-modulus ratio is subexpo in $n$ rather than exp in $n$. }\takashi{You are right! I somehow misunderstood that $2^n/n$ was subexponential.}

\paragraph{\bf Security}
\begin{theorem}
If $\LWE_{n,m,q,D_{\ZZ_q^m,\sigma}}$ is quantumly hard, 
    then the above scheme satisfies the requirements of special dual-mode OT (\Cref{def:dual-mode-OT}).  
\end{theorem}

\begin{description}
\item[Evaluation correctness:]For any $b \in \bin$, $\mode \in \{1,2\}$, and $(z_0,z_1)\in \bit^{2}$,
suppose that $(\crs,\td) \la \CRSGen(1^\secp,\mode )$, 
$(\msg^{(1)},\st) \la \Receive{1}(\crs,b)$, 
$\msg^{(2)} \la \Send(\crs,\msg^{(1)}, (z_0,z_1) )$. 
Using the notations in the description of the scheme, we have 
\begin{align}
&\left|w_{b,n+1} -\vr^\trans\vw_{b,[1,n]}\mod q\right|\\
&=\left|\vu_b^\trans\vt_b+z_b\lfloor q/2\rfloor- 
\vr^\trans \mA \vt_b \mod q\right|\\
&=\left|(\vr^\trans\mA+{\ve^\prime}^\trans)\vt_b+z_b\lfloor q/2\rfloor- 
\vr^\trans \mA \vt_b \mod q\right|\\
&=\left|{\ve^\prime}^\trans\vt_b+z_b\lfloor q/2\rfloor \mod q \right|
\end{align}
Since $\ve^\prime$ is sampled from $D_{\ZZ_q^m,\sigma'}$ and $\vt$ is sampled from $D_{\ZZ_q^m,\tau}$, we have $\|\ve^\prime\|_\infty\le \sqrt{m}\sigma^\prime$ and $\|\vt_b\|_\infty\le \sqrt{m}\tau$.  
Thus, 
$|{\ve^\prime}^\trans\vt_b|\le m^2\sigma^\prime \tau\le q/8$. 
Thus, when $z_b=0$, we have  
$
\left|w_{b,n+1} -\vr^\trans\vw_{b,[1,n]}\mod q\right|\le q/8<q/4.  
$
Similarly, 
 when $z_b=1$, we have  
$
\left|w_{b,n+1} -\vr^\trans\vw_{b,[1,n]}\mod q\right|\ge \lfloor q/2\rfloor -q/8>q/4.  
$
The above together imply the evaluation correctness. 
\item[Unique state:]
Let 
$\mode\in \{1,2\}$, 
$(\crs=(\mA,\vv),\td=(\mA,\mT,\vv)) \in \Supp(\CRSGen(1^\secp,\mode))$, $b\in \bin$,  and $(\msg^{(1)}=\vu_0,\st=\vr) \in \Supp(\Receive{1}(\crs,b))$.
Then we have $\vu_0=\vr^\trans\mA+{\ve^\prime}^\trans-b\vv^\trans$ for some $\ve^\prime$ such that $\|\ve^\prime\|_\infty\le \sigma^\prime\le q/(C\sqrt{nm\log q})$. 
By \Cref{lem:lattice_trapdoor}, there is a PPT algorithm that recovers $\st=\vr$ from $\mA$, $\mT$, $b$, $\vv$, and $\msg^{(1)}=\vu_0$. This implies the uniqueness of the state. 
\item[Mode indistinguishability:]
The only difference between the two modes is that $\vv$ is uniformly random when $\mode=1$ and is an LWE instance when $\mode=2$. Thus, the mode indistinguishability directly follows from the assumed quantum hardness of $\LWE_{n,m,q,D_{\ZZ_q^m,\sigma}}$. 
\item[Statistical security against malicious senders in the hiding mode:]
Fix $\crs=(\mA,\vv)$ in the support of $\CRSGen(1^\secp,2)$. 
It satisfies $\vv^\trans=\vs^\trans\mA+\ve^\trans$ for some $\vs\in \ZZ_q^n$ and 
        $\ve\in \ZZ_q^m$ such that $\|\ve\|_\infty\le \sigma$.  
Suppose $(\msg^{(1)}_b,\st)\lrun \Receive{1}(\crs,b)$ for $b\in \bin$. 
Then we can write 
$$\msg^{(1)}_0=\vr_0^\trans\mA+{\ve_0^\prime}^\trans$$ 
and 
$$\msg^{(1)}_1=\vr_1^\trans\mA+{\ve_1^\prime}^\trans-\vv
=(\vr_1^\trans-\vs^\trans)\mA+({\ve_1^\prime}^\trans-\ve^\trans)$$  
for 
$\vr_0,\vr_1\gets \ZZ_q^n$
and  $\ve_0^\prime,\ve_1^\prime\gets D_{\ZZ_q^m,\sigma^\prime}$. 
Since $\vr_1^\trans$ is uniformly random, $(\vr_1^\trans-\vs^\trans)$ is also uniformly random. Moreover, 
since $\|\ve\|\le \sqrt{m}\sigma$, 
the statistical distance between the distributions of ${\ve_0^\prime}^\trans$ and $({\ve_1^\prime}^\trans-\ve^\trans)$ is at most 
$\sqrt{2\left(1-\exp\left(-2\pi m \sigma/\sigma^\prime\right)
\right)}\le \sqrt{4\pi m\sigma/\sigma'} =\negl(\secp)$
by \Cref{lem:smudging}  
and $\sigma'/\sigma=2^{\omega(\log \secp)}$.  Thus, the distributions of $\msg_0^{(1)}$ and $\msg_1^{(1)}$ are statistically negligibly close in $\secp$. 
\item[State recoverability in the hiding mode:]
Let $\StaRcv$ be an algorithm that works as follows. 
\begin{description}
\item[$\StaRcv(\td,\msg^{(1)})$:]
On input $\td=(\mA,\mT,\vv)$ and $\msg^{(1)}=\vu_0$, 
compute $\vu_1=\vu_0+\vv$, 
run $(\vr_b,\ve_b^\prime) \gets \Invert(\mA,\mT,\vu_b)$ for $b\in\bit$ and output $(\vr_0,\vr_1)$.  
\end{description}
Suppose $(\crs,\td) \la \CRSGen(1^\secp,2 )$ and 
$(\msg^{(1)},\st) \la \Receive{1}(\crs,b)$.
Then $\msg^{(1)}=\vu_0=\vr^\trans\mA+{\ve^\prime}^\trans-b \vv$  and $\st=\vr$ where $\vr\gets \ZZ_q^n$
and   $\ve^\prime\gets D_{\ZZ_q^m,\sigma^\prime}$. 
Then we have $\vu_b^\trans=\vu_0^\trans+b \vv=\vr^\trans\mA+{\ve^\prime}^\trans$.
Thus, for $(\vr_0,\vr_1)$ output by $\StaRcv(\td,\msg^{(1)})$, we have $\vr_b=\vr$ by \Cref{lem:lattice_trapdoor} and  $\|\ve^\prime\|\le \sqrt{m}\sigma^\prime \le q/(C\sqrt{n\log q})$. 
\item[Extractability in the extraction mode:]
Let $\Ext$ and $\Sim$ be algorithms that work as follows. 
\begin{description}
    \item[$\Ext(\td,\msg^{(1)})$:]
    On input $\td=(\mA,\mT,\vv)$ and $\msg^{(1)}=\vu_0$, compute $\vu_1=\vu_0+\vv$ 
    and run $\beta_b\gets \IsMessy(\mA,\mT,\vu_b)$. 
    If $\beta_0=1$, output $1$, else if $\beta_1=1$, output $0$, otherwise output $\bot$. 
       \item[$\Sim(\crs,\msg^{(1)},z_b)$:]
On input $\crs=(\mA,\vv)$, $\msg^{(1)}=\vu_0$,  %, %$b\in \bit$, 
and $z_b\in \bit$,  
Run $(\vw_0,\vw_1)\gets \Send(\crs,\msg^{(1)}, (z_b,z_b) )$ and output $(\vw_0,\vw_1)$.  
%That is, let $\vu_1=\vu_0+\vv$,  choose $\vt_0,\vt_1\gets D_{\ZZ_q^m,\tau}$,  compute $\vw_b=\left(\begin{matrix}\mA\\\vu_b^\trans\end{matrix}\right)\vt_b+\left(\begin{matrix}\mathbf{0}\\z\lfloor q/2\rfloor \mod q\end{matrix}\right)$  
%for $b\in \bit$,   
%choose $\vw_{1-b}\gets \ZZ_q^n$, 
%and output $\msg^{(2)}=(\vw_0,\vw_1)$.  
Note that it uses the same $z_b$ for generating both $\vw_0$ and $\vw_1$. 
\end{description}
Fix $(\mA,\mT)$ that satisfies the two items of \Cref{lem:messy}.  
Let $G\subseteq  \ZZ_q^m$ be the subset in the first item of \Cref{lem:messy}, that is, for 
any $\vv\notin G$, 
$\IsMessy(\mA,\mT,\vv)$ returns $1$ with overwhelming probability.  
Since $G$ consists of at most $(1/2\sqrt{q})^m$ fraction of elements of $\ZZ_q^m$,  $|G|\le (\sqrt{q}/2)^m$. 
Therefore, we have 
\[
\#\{\vv\in \ZZ_q^m: \exists \vu_0\in \ZZ_q^m\text{~s.t.~} \vu_0\in G \land \vu_1\in G \}\le |G|^2\le (q/4)^{m}
\]
where $\vu_1=\vu_0+\vv$. 
Thus, we have 
\[
\Pr_{\vv\gets \ZZ_q^m}[\exists \vu_0\in \ZZ_q^m\text{~s.t.~} \vu_0\in G \land \vu_0+\vv\in G ]\le 4^{-m}. 
\]
Thus, when $\td$ is honestly generated, for any $\msg^{(1)}$, $\Ext(\td,\msg^{(1)})$ returns a non-$\bot$ bit with overwhelming probability. 
 By the second item of \Cref{lem:messy}, 
 when $\Ext(\td,\msg^{(1)})$ returns $b\in \bit$, 
$(\mA,\vu_{1-b})\in M_{\delta,\tau}$ with overwhelming probability where $M_{\delta,\tau}$ is as defined in \Cref{lem:messy}. 
Note that the only difference between $\Send(\crs,\msg^{(1)},(z_0,z_1))$ and $\Sim(\crs,\msg^{(1)},z_b)$ is that $\vw_{1-b}$ is generated using $z_{1-b}$ in the former but using $z_b$ in the latter. 
By the definition of $M_{\delta,\tau}$, when $(\mA,\vu_{1-b})\in M_{\delta,\tau}$, the distribution of $\vw_{1-b}$ generated in either way  %generated by $\Send(\crs,\msg^{(1)},(z_0,z_1))$ 
has statistical distance at most $\delta=\negl(n)$ from the uniform distribution. 
%On the other hand, $\Sim(\td,b,z_b)$ runs $\Send(\crs,\msg^{(1)},(z_0,z_1))$
%Since the only difference between $\Send(\crs,\msg^{(1)},(z_0,z_1))$ and $\Sim(\td,b,z_b)$ is that the way of generating $\vw_{1-b}$ is changed to uniformly sampling from $\ZZ_q^n$. 
Thus, the above implies  extractability in the extraction mode.

\item[Efficient state superposition:]
It suffices to show that 
there is a QPT algorithm that is given $b\in\bit$, $\mA$, and $\vv$, and generates a state that is within negligible trace distance of the following state: 
%the following isometry is implementable in QPT with a negligible error given $\mA$ and $\vv$: 
\[
\sum_{\vr\in \ZZ_q^n,\ve^\prime\in \Supp(D_{\ZZ_q^m,\sigma^\prime})} 
\sqrt{ p_{\vr,\ve^\prime} }
\ket{b} 
\ket{\vr} \ket{\vr^\trans\mA+{\ve^\prime}^\trans+(-1)^b\vv^\trans} 
~
\]
where $p_{\vr,\ve^\prime}$ is the probability that $\Receive{1}(\crs,b)$ chooses the randomness $(\vr,\ve^\prime)$. 

\takashi{I expanded the explanation below.}
This can be done as follows. First, as shown in \cite{JACM:Regev09,JACM:BCMVV21}, the Grover-Rudolph algorithm \cite{GR02} enables us to generate the following  Gaussian superposition up to negligible error:
\[
\sum_{\ve^\prime\in \Supp(D_{\ZZ_q^m,\sigma^\prime})} 
\sqrt{ D_{\ZZ_q^m,\sigma^\prime}(\ve^\prime) }\ket{\ve^\prime}.
\]
We prepare $\ket{b}$ in another register and 
generate
a uniform superposition over $\vr\in \ZZ_q^n$ in yet another register. Then, the whole system can be written as follows:
\begin{align}
&\ket{b}\otimes q^{-n/2}\sum_{\vr\in \ZZ_q^n}\ket{\vr}\otimes \sum_{\ve^\prime\in \Supp(D_{\ZZ_q^m,\sigma^\prime})} 
\sqrt{ D_{\ZZ_q^m,\sigma^\prime}(\ve^\prime) }\ket{\ve^\prime}\\
=&\sum_{\vr\in \ZZ_q^n,\ve^\prime\in \Supp(D_{\ZZ_q^m,\sigma^\prime})} 
\sqrt{ p_{\vr,\ve^\prime}}\ket{b}\ket{\vr}\ket{\ve^\prime}
\end{align}
where the equality follows from $p_{\vr,\ve^\prime}=q^{-n}D_{\ZZ_q^m,\sigma^\prime}(\ve^\prime)$.

By using $\mA$ and $\vv$, we can coherently compute  $\vr^\trans\mA+{\ve^\prime}^\trans+(-1)^b\vv^\trans$ in another register. This yields the following state:
\begin{align}
\sum_{\vr\in \ZZ_q^n,\ve^\prime\in \Supp(D_{\ZZ_q^m,\sigma^\prime})} 
\sqrt{ p_{\vr,\ve^\prime}}\ket{b}\ket{\vr}\ket{\ve^\prime}\ket{\vr^\trans\mA+{\ve^\prime}^\trans+(-1)^b\vv^\trans}.
\end{align}
Finally, uncomputing the third register results in the desired state.

\if0
It is easy to see that this is reduced to preparing the Gaussian superposition below:
\[
\sum_{\ve^\prime\in \Supp(D_{\ZZ_q^m,\sigma^\prime})} 
\sqrt{ D_{\ZZ_q^m,\sigma^\prime}(\ve^\prime) }\ket{\ve^\prime}.
\]
As shown in \cite{JACM:Regev09,JACM:BCMVV21}, this is efficiently doable with a negligible error by using the Grover-Rudolph algorithm \cite{GR02}. \fi
\end{description}

\subsection{From OT to SFE via Garbled Circuit}

We explain how to transform special dual-mode OT into special dual-mode SFE using garbled circuits.
We only provide a sketch since the technique is fairly standard.

\paragraph{Garbled circuit.}
We first quickly review the notion of garbled circuit.
A circuit garbling scheme $\GC$ for circuits is a tuple $(\Garble, \Eval)$ of two PPT algorithms with the following syntax.
\begin{itemize}
\item A garbling algorithm $\Garble$, given a security parameter $1^\sep$ and a circuit $C$ with $n$-bit input, outputs a garbled circuit $\tlC$ and input labels $\{\lab_{i,b}\}_{i \in [n],b \in \bin}$.
\item The evaluation algorithm, given a garbled circuit $\tlC$ and $n$ input labels $\{\lab_i\}_{i \in [n]}$,
outputs $y$. 
\end{itemize}
As correctness, we require $\Eval(\Ctilde, \{\lab_{i,x[i]}\}_{i \in [n]}) = C(x)$ holds for every circuit $C $ with $n$-bit input and $x \in \bin^n$,
where $(\Ctilde,\{\lab_{i,b}\}_{i \in [n],b \in \bin}) \allowbreak \la \Garble(1^\sep, C)$.
As security, we require that for any circuit $C$ and input $x$, $(\Ctilde,\{\lab_{i,x[i]}\}_{i \in [n]})$ can be simulated by a computationally indistinguishable distribution given only $1^{|C|}$ and $C(x)$. 
A circuit garbling scheme can be realized based on OWFs~\cite{FOCS:Yao86}.

\paragraph{Transformation based on garbled circuit.}
We now sketch how to transform a special dual-mode OT scheme $\OT=\OT.(\CRSGen, \Receive{1},\allowbreak \Send, \Receive{2})$ into a special dual-mode SFE scheme $\SFE=\SFE.(\CRSGen, \Receive{1},\allowbreak \Send, \Receive{2})$ using a circuit garbling scheme $\GC=(\Garble,\Eval)$.
\begin{description}
\item[$\SFE.\CRSGen$:] It runs $\OT.\CRSGen$ $n$ times and gets $(\OT.\crs_i,\OT.\td_i)_{i\in[n]}$. It sets $\SFE.\crs:=(\OT.\crs_i)_{i\in[n]}$ and $\SFE.\td:=(\OT.\td_i)_{i\in[n]}$.

\item[$\SFE.\Receive{1}$:]Given $\SFE.\crs:=(\OT.\crs_i)_{i\in[n]}$ and $x\in\bit^n$, it executes $(\OT.\msg^{(1)}_i,\OT.\st_i)\gets\OT.\Receive{1}(\OT.\crs_i,x[i])$ for every $i\in[n]$ and outputs $\SFE.\msg^{(1)}:=(\OT.\msg^{(1)}_i)_{i\in[n]}$ and $\SFE.\st:=(\OT.\st_i)_{i\in[n]}$.

\item[$\SFE.\Send$:]Given $\SFE.\crs:=(\OT.\crs_i)_{i\in[n]}$, $\SFE.\msg^{(1)}:=(\OT.\msg^{(1)}_i)_{i\in[n]}$, and a circuit $C$, it first executes $(\tlC,(\lab_{i,b})_{i\in[n],b\in\bit})\gets\Garble(1^\secp,C)$. It then runs $\OT.\msg^{(2)}_i\gets\OT.\Send(\OT.\crs_i, \OT.\msg^{(1)}_i,(\lab_{i,0},\lab_{i,1}))$ for every $i\in[n]$ and outputs $\SFE.\msg^{(2)}:=(\tlC, (\OT.\msg^{(2)}_i)_{i\in[n]})$.

\item[$\SFE.\Receive{2}$:]Given $\SFE.\crs:=(\OT.\crs_i)_{i\in[n]}$, $x\in\bit^n$, $\SFE.\st:=(\OT.\st_i)_{i\in[n]}$, and $\SFE.\msg^{(2)}:=(\tlC, (\OT.\msg^{(2)}_i)_{i\in[n]})$, it executes $\lab_i\gets\OT.\Receive{2}(\OT.\crs_i,x[i],\OT.\st_i,\OT.\msg^{(2)}_i)$ for every $i\in[n]$ and outputs $y\gets\Eval(\tlC, (\lab_i)_{i\in[ n]})$.

\end{description}

We can easily check that the evaluation correctness of $\SFE$ follows from that of $\OT$ and $\GC$.
$\SFE$ uses $n$ instances of $\OT$ in parallel and uses $\GC$ on top of it.
Especially, the CRS, the trapdoor, and the first message of $\SFE$ consist of just $n$ copies of those of $\OT$.
By this construction, the unique state property, mode indistinguishability, statistical security against malicious senders in the hiding mode, state recoverability in the hiding mode, and efficient state superposition property of $\SFE$ directly follow from those of $\OT$, respectively.
Finally, we can prove the extractability in the extraction mode of $\SFE$ based on that of $\OT$ and the security of $\GC$.
Roughly speaking, the sender message $\SFE.\msg^{(2)}$ of $\SFE$ is composed of $(\tlC, (\OT.\msg^{(2)}_i)_{i\in[n]})$, and the first component can be simulated from $1^{|C|}$ and $C(x)$ by the security of $\GC$ and the second one can be simulated by the extractability in the extraction mode of $\OT$ together with the simulated labels of $\GC$.

\else
%%%%% LNCS-style submission version region %%%%%
	\newpage
	 	\appendix
	 	\setcounter{page}{1}
 	{
	\noindent
 	\begin{center}
	{\Large SUPPLEMENTAL MATERIALS}
	\end{center}
 	}
	\setcounter{tocdepth}{2}

	\setcounter{tocdepth}{2}
	\tableofcontents

	\fi
	%%%%%%% Camera-ready region (no appendix) %%%%%%
\fi

\end{document}